\documentclass[11pt]{article}

\usepackage{float, amssymb, amsmath, amsthm, graphicx, bbm, multirow, siunitx, placeins, authblk, geometry, fullpage}
\usepackage{subcaption}
\usepackage{caption}
\usepackage[round]{natbib}
\usepackage[dvipsnames]{xcolor}
\usepackage{diagbox}
\usepackage{tikz}

\usetikzlibrary{arrows.meta}
\usetikzlibrary{decorations.pathreplacing}
\usepackage[shortlabels]{enumitem}
\setlist{nolistsep}

\usepackage{algorithm}
\usepackage{algpseudocode}
\algnewcommand\algorithmicinput{\textbf{INPUT:}}
\algnewcommand\INPUT{\item[\algorithmicinput]}
\algnewcommand\algorithmicoutput{\textbf{OUTPUT:}}
\algnewcommand\OUTPUT{\item[\algorithmicoutput]}

\usepackage{xr-hyper}
\usepackage{hyperref}[]
\hypersetup{
    colorlinks=true,
    linkcolor=blue,
    filecolor=magenta,      
    urlcolor=cyan,
      citecolor=blue,
    }

\usepackage{cleveref}

\newtheorem{theorem}{Theorem}
\newtheorem{lemma}[theorem]{Lemma}

\newtheorem{remark}{Remark}
\newtheorem{assumption}{Assumption}

\newtheorem{definition}{Definition}

\newtheorem{subdefinition}{Definition}[definition]

\makeatletter
\newcommand{\neutralize}[1]{\expandafter\let\csname c@#1\endcsname\count@}
\makeatother

\newenvironment{thmbis}[1]
  {\renewcommand{\theassumption}{\ref*{#1}$'$}
   \neutralize{assumption}\phantomsection
   \begin{assumption} }
  {\end{assumption} }

\allowdisplaybreaks

\DeclareMathOperator*{\argsup}{arg\,sup}
\DeclareMathOperator*{\argmax}{arg\,max}

\usepackage{mathtools}

\def\E{\mathbb{E}}
\def\P{\mathbb{P}}

\def\R{\mathbb{R}}

\title{Multilayer random dot product graphs: \\ Estimation and online change point detection}
\author[1]{Fan Wang}
\author[2]{Wanshan Li}
\author[3]{Oscar Hernan Madrid Padilla}
\author[1]{Yi Yu}
\author[4]{Alessandro Rinaldo}
\affil[1]{Department of Statistics, University of Warwick}
\affil[2]{Wizard Quant Investment}
\affil[3]{Department of Statistics, University of California, Los Angeles}
\affil[4]{Department of Statistics \& Data Science, University of Texas at Austin} 

\begin{document}

\maketitle

\begin{abstract}
We study the multilayer random dot product graph (MRDPG) model, an extension of the random dot product graph to multilayer networks. To estimate the edge probabilities, we deploy a tensor-based methodology and demonstrate its superiority over existing approaches. Moving to dynamic MRDPGs, we formulate and analyse an online change point detection framework. At every time point, we observe a realization from an MRDPG. Across layers, we assume fixed shared common node sets and latent positions but allow for different connectivity matrices. We propose efficient tensor algorithms under both fixed and random latent position cases to minimize the detection delay while controlling false alarms. Notably, in the random latent position case, we devise a novel nonparametric change point detection algorithm based on density kernel estimation that is applicable to a wide range of scenarios, including stochastic block models as special cases.  Our theoretical findings are supported by extensive numerical experiments, with the code available online\footnote{\url{https://github.com/MountLee/MRDPG}}.
\end{abstract}

\section{Introduction}

Statistical network analysis is concerned with studying relationships among a set of individuals.  A network is usually encoded by sets of nodes and edges, representing the units in a population of interest and their pairwise relationships, respectively.  Multilayer networks are generalizations of networks that enable the simultaneous representation of multiple types of interactions, or {\it layers}, among a fixed set of nodes. The popularity of multilayer networks has increased over the last few years in response to the soaring demand from a multitude of application areas.  For example, an air transportation network can be represented by a multilayer network, with nodes as airports, edges as flights, and layers as the flights by the same airline companies \cite[e.g.][]{cardillo2013emergence}.  Such formulation helps to identify the commonly-shared structures across different airlines and provides insights for policy-makers for future regulations or planning.

When complex, multilayer network data are observed sequentially across time, dynamic multilayer network modelling is in order.  For example, \cite{barigozzi2010multinetwork} modelled a data set consisting of 12 international trade networks, one for each year from 1992 to 2003, among 162 countries and 97 commodities, into dynamic multilayer networks.  By treating the international trade of each commodity as a layer, the data can be modelled as a time series of multilayer networks.  
  
In general, to study the dynamics of a data sequence, it is crucial to model how the sequence changes over time.  To that effect, in the recent literature on dynamic multilayer network research, different assumptions have been considered.  Some researchers assumed the time-invariance \cite[e.g][]{chen2022time, bazzi2016community}, some allowed for slow evolution along the time \cite[e.g.][]{guo2020nonparametric, cai2021slowly}, and some considered the case where abrupt changes are present \cite[e.g.][]{zhang2019detecting, rossi2013modularity}. The stationarity and slow-varying assumptions might be reasonable over a short period of time, but empirical evidence suggests that abrupt changes do occur over time.  For instance, in the international trading networks, unexpected shocks such as the Covid-19 pandemic, the Russia-Ukraine War or the 2022 European drought may have a dramatic impact.

In this paper, we study online change point detection in dynamic multilayer networks with fixed or random latent positions.  To be specific, at each time point, we observe a realization of an $L$-layered multilayer network with nodes associated with latent and unobservable positions. When the latent positions are fixed, we assume that the weight matrices, which encode layer-specific information, may change over time.  In the more complex scenario of random latent positions, each multilayer network is characterized by an $L$-dimensional distribution, and we assume that the underlying distributions may change over time.  Our goal is to develop efficient algorithms to detect such changes in a {\it sequential fashion}, i.e.~controlling the false alarms while, at the same time, minimizing the delay in detecting the change.  An important step in our analysis is to estimate a single multilayer network using a tensor-based method, which we are going to show outperforms the state-of-the-art multilayer network estimation methods.

\subsection{Related literature}\label{sec-related-lit-intro}
In our analysis, we investigate three distinct but inter-related themes: (1) multilayer networks, (2) sequential change point analysis, and (3) spectral methods for tensor estimation.  We discuss the relevant literature in these three areas.  More detailed comparisons between our results and those found in the current literature will be provided in the sequel.

\medskip
\noindent \textbf{Statistical networks analysis.}  The Erd\H{o}s--R\'enyi random graphs \citep{erdos1960evolution} are arguably the simplest network models, where all edges are considered as independent and identically distributed Bernoulli random variables.  To incorporate more heterogeneity, a line of attack has been launched including stochastic block models \citep{holland1983stochastic}, degree-corrected block models \citep{karrer2011stochastic}, mixed membership models \citep{airoldi2008mixed}, and random dot product models \citep{young2007random}, to name but a few.  As for this paper, the closest related is the random dot product model, where each node is associated with a latent position - random or fixed.  The connectivity between a pair of nodes is consecutively determined by the distance between their latent positions.  To some extent, the random dot product graph model can be seen as a generalization covering all the aforementioned  models.  We refer the reader to \cite{athreya2017statistical} for a comprehensive survey on random dot product graph models.

Growing demand for more sophisticated approaches to statistical network analysis in many application areas has recently sparked considerable interest in multilayer network modelling.  Below, we mainly focus on models with random latent positions for the nodes.  Assigning each node to a latent position, \cite{zhang2020flexible} assumed that the logit of connectivity is a linear function of the latent positions and of layer-specific parameters.  Likelihood-based methods are then deployed for estimation.  \cite{arroyo2021inference} introduced the common subspace independent edge (COSIE) model, which assumes shared latent positions across layers. For estimation, they proposed a spectral embedding algorithm on an augmented matrix. \cite{jones2020multilayer} introduced the multilayer random dot product graph (MRDPG) model, an extension of the random dot product graph to multilayer networks and exploited spectral methods to uncover the latent positions.  To handle the multilayer nature, they assumed that the latent positions of head nodes are shared across layers and packed all the layer-specific latent positions into the tail nodes.  They develop spectral embedding type methods on an augmented matrix, by combining the adjacency matrices of all layers into a large matrix.  \cite{macdonald2022latent} proposed a model that assumes that the latent positions of the nodes not only possess shared structures across layers but also layer-specific positions. They develop a penalized loss function, along with a proximal gradient descent algorithm, to recover the multilayer connectivity.  Instead of estimating the connectivity, \cite{jing2021community} and \cite{lei2023computational} focused on community detection in multilayer networks.  Different from the aforementioned papers, \cite{jing2021community} formulated a multilayer network as a tensor and summoned tensor estimation methods.  \cite{lei2023computational} investigated both the computational and statistical network density thresholds in multilayer stochastic block models.  It should be noted that all the papers discussed in this paragraph focus on a single multilayer network.

\medskip
\noindent \textbf{Change point analysis.}  Having its debut more than seven decades ago, change point analysis is recently going through a resurgence.  Depending on the data availability, change point analysis can be roughly categorized as online and offline change point analysis, where the full data set is being collected during or before the analysis is conducted.  Online change point analysis is the main focus of this paper.  In the online change point analysis framework, recently univariate mean change \cite[e.g.][]{yu2020note}, multivariate distributional change \cite[e.g.][]{berrett2021locally}, inhomogeneous network change \cite[e.g.][]{yu2021optimal}, dynamic treatment change \cite[e.g.][]{padilla2022dynamic}, high-dimensional Gaussian mean change \cite[e.g.][]{chen2020high} and more general changes \cite[e.g.][]{chu2018sequential, he2018sequential}, among many others have been studied. 

Regarding change point analysis for dynamic networks, in addition to the aforementioned, \cite{wang2021optimal} considered an offline change point analysis problem in dynamic inhomogeneous Bernoulli networks. \cite{li2022network} studied the dynamic network change point analysis problems under local differential privacy.  \cite{padilla2019change} derived an offline change point detection mechanism for dynamic random dot product models, with temporal dependence.  

\medskip
\noindent \textbf{Tensor estimation.}  The emergence of increasingly complex data structures has led to a growing interest in statistical analysis of tensors.  Different methods have been considered, including maximum likelihood estimation- \cite[e.g.][]{richard2014statistical} and singular value decomposition-based methods \cite[e.g.][]{de2000best, de2000multilinear}.  In particular, under the low-rank assumption and with Gaussian noise, \cite{zhang2018tensor} showed that the higher-order orthogonal iteration method, a singular value decomposition type method, is minimax rate-optimal in the task of recovering a low-rank tensor.  In the case of certain sparsity structures, \cite{zhang2019optimal} proposed the sparse tensor alternating thresholding for singular value decomposition methods.  More recently, \cite{han2022optimal} developed the non-convex low-rank tensor optimization method accompanied by a proximal gradient descent algorithm for generalized tensors such as sub-Gaussian and Poisson tensors. They have demonstrated that their proposed method achieves the minimax rate of convergence up to a logarithmic factor.  We refer readers to \cite{auddy2024tensor} - a very recent survey paper, for more discussions.

\subsection{List of contributions}

We investigate the multilayer random dot product graph model introduced in \cite{jones2020multilayer}, where the deterministic or random latent positions are shared across all layers and the layer-specific information is encoded in a weight matrix. This is shown in \Cref{sec-def-mrdpg}.  We then move from static to dynamic MRDPGs and focus on online change point detection. When the latent positions are fixed, we assume that the change, if there is any, occurs in the weight matrix sequence; with random latent positions, the connectivity is instead characterized by an $L$-dimensional distribution, and the change affects the distribution itself.  Our goal is to detect the potential change in an online fashion, while controlling false alarms and minimizing the detection delay.  For the case of fixed and random latent positions, this generic setting allows the model parameters to vary as functions of the change point location, shown in Sections~\ref{sec-dynamic-f} and \ref{sec-dy-non-framework}, respectively.  The inclusion of both fixed and random latent positions, along with discussions on the difference between these two cases, yields a comprehensive picture not only for multilayer network online change point detection problems, but also for random dot product graph problems in general. 

Secondly, to maximize model flexibility, in the case of random latent positions, we do not require that the distribution of the latent positions admits a Lebesgue density, a general setting allowing, in particular, for distributions supported on manifolds. This is achieved by characterizing the distributional change using the expectation of a kernel density estimator \cite[e.g.][]{fasy2014confidence}.  Such quantity can be viewed as a generalization of density functions to distributions supported on sets of Lebesgue measure zero, and has been recently studied in the nonparametric estimation literature \cite[e.g.][]{kim2019uniform}.  The need for such general approach to modelling latent positions, which to the best of our knowledge is new to the network literature can be intuitively explained as follows.  Stochastic block models with random membership can be seen as a special case of the random dot product graph model with random latent positions whose distributions do not possess Lebesgue densities. Further detailed discussions on this topic can be found in \Cref{sec-dy-non-framework}.

Thirdly, in terms of theoretical contributions, with proper controls on false alarms, we provide a sharp upper bound on the detection delay.  In particular, compared to the existing work with theoretical guarantees \cite[e.g.][]{padilla2019change}, we improve the estimation error and the signal-to-noise ratio by a factor of sample-size and $\sqrt{\mbox{sample-size}}$, respectively.  More details will be discussed in \Cref{dynamic-theory}.  At a high level, in dynamic network analysis, when spectral methods are used, eigen-decomposition is typically conducted on every adjacency matrix due to the difficulty in studying the eigenspace of a linear combination of low-rank matrices.   This leads to high computational complexity and, more importantly, a sub-optimal rate in terms of the number of time points.  In this paper, by making appropriate model assumptions and employing more refined analysis, we are able to analyse the averages of matrices/tensors, resulting in improved rates.

To the best of our knowledge, the use of tensor methods for estimating a single multilayer network is a novel contribution. Note that, \citealt{jing2021community} exploited tensor tools but focused on community detection rather than estimation, though we do provide comparisons with their estimation performances in this paper.  We provide in-depth discussions on the rationale of using tensor-based methods and thorough comparisons with other non-tensor-based methods in Sections~\ref{sec_estimation_single} and \ref{sec-comp-single}.

Lastly, the theoretical findings of our paper are accompanied by extensive numerical results in \Cref{sec-numerical}, where we evaluate the performance of our proposed methods against other state-of-the-art algorithms.

\subsection{Notation and organization of the paper}

For any $a, b \in \R$, let $a \vee b = \max\{a, b\}$, $a \wedge b = \min \{a, b\}$ and $(a)_{+} = \max\{a, 0\}$.  For any positive integer $p$, let $[p] = \{1, \ldots, p\}$. For any vector $\mathbf{v}$, let $\|\mathbf{v}\|_2$ and $\| \mathbf{v}\|_{\infty}$ be its $\ell_2$- and entrywise supremum norms.

For any matrix $A \in \R^{p_1 \times p_2}$, let $A_i$ and $A^j$ be the $i$th column and $j$th row of $A$.  Let $A_{i, j}$ denote the entry of $A$ in the $i$th row and $j$th column and let $\sigma_1(A) \geq \dots \geq  \sigma_{p_1 \wedge p_2} (A)\geq 0$ be its singular values. Let $\|A\|$ and $\|A\|_\mathrm{F}$  be the spectral and Frobenius norms of $A$, respectively.  

For any order-3 tensor $\mathbf{P} \in \R^{p_1 \times p_2 \times p_3}$, let 
$\|\mathbf{P}  \|_{\mathrm{F}}^2 = \sum_{i = 1}^{p_1} \sum_{j=2}^{p_2} \sum_{l = 3}^{p_3}   \mathbf{P} _{i, j, l }^2$ and 
for any $l \in [p_3]$, let $\mathbf{P} _{:, ;, l} \in \R^{p_1 \times p_2 }$ with $\left(  \mathbf{P} _{:, ;, l} \right)_{i, j} = \mathbf{P} _{i, j, l}$ for any $i \in [p_1]$ and $j \in [p_2]$.

The rest of the paper is organized as follows. \Cref{sec:background} provides a brief explanation of the tensor representation of multilayer networks and tensor algebra.  Formulation and estimation of single multilayer networks are collected in \Cref{sec-single-mrdpg}. The dynamic counterpart and online change point detection with fixed latent positions and random latent positions are discussed in Sections~\ref{sec-dynamic-f} and \ref{sec-dynamic}, respectively. In the main text, we focus on undirected edges and with directed edge cases that accommodate bipartite graphs collected in \Cref{sec-directed-edges}. Extensive numerical results are available in \Cref{sec-numerical}, with all the technical details deferred to Appendices.  

\subsection{Tensor representation of multilayer networks and tensor algebra} \label{sec:background}

Statistical network analysis usually defines a network as collections of nodes and edges, which can be formulated in an adjacency matrix, with each entry encoding the connection between a pair of nodes.  As for a multilayer network, we formally define its associated quantities below.

\begin{definition}[Multilayer network] \label{def-multi-layer-network}
A multilayer network $\mathcal{G}$ is written as $\mathcal{G} = (\mathcal{V}_1, \mathcal{V}_2, \mathcal{E}, \mathcal{L})$, where $\mathcal{V}_1 = [n_1]$ is the collection of head nodes, $\mathcal{V}_2 = [n_2]$ is the collection of tail nodes, $\mathcal{L} = [L]$ is the collection of layers and $\mathcal{E} \subset \{(i, j, l): i \in \mathcal{V}_1, j \in \mathcal{V}_2, l \in \mathcal{L}\}$ is the collection of edges.  We let $\mathbf{A} \in \mathbb{R}^{n_1 \times n_2 \times L}$ be the associated adjacency tensor, i.e.
\[
    \mathbf{A} _{i, j, l} = 
    \begin{cases}
        1, & (i, j, l) \in \mathcal{E},\\
        0, & \mbox{otherwise},
    \end{cases} \quad i \in \mathcal{V}_1, \, j \in \mathcal{V}_2, \, l \in \mathcal{L}.
\]
\end{definition}

\Cref{def-multi-layer-network} allows for directed multilayer networks, where each edge is associated with a head and a tail node.  We allow for the two node sets to be different, but the undirected case can be seen as a special case of \Cref{def-multi-layer-network}. For the rest of the main text, we focus on the undirected edge cases, with the directed edge cases that accommodate bipartite graphs in \Cref{sec-directed-edges}.  The methodology and theory of the directed case is a simplification of that of the undirected case. Considering the multilayer nature, we encode a multilayer network in an order-3 adjacency tensor.  

The formulation in terms of tensors is further explained in \Cref{sec-def-mrdpg}.  In the rest of this subsection, we introduce some necessary and standard tensor algebra for completeness.  For more detailed tensor algebra, we refer readers to \cite{kolda2009tensor}.  For any order-3 tensor $\mathbf{M} \in \mathbb{R}^{p_1 \times p_2 \times p_3}$ and $s \in [3]$, letting $p_4 = p_1$ and $p_5 = p_2$, following the notation in \cite{zhang2018tensor}, we define the mode-$s$ matricization of $\mathbf{M}$ as $\mathcal{M}_s(\mathbf{M}) \in \mathbb{R}^{p_{s} \times (p_{s+1} p_{s+2})}$, with entries
\begin{equation} \label{eq-matricisation}
    \mathcal{M}_s(\mathbf{M} )_{i_1, (i_2 - 1)p_{s+2} +i_3 } = \begin{cases}
        \mathbf{M} _{i_1, i_2, i_3}, & s = 1,\\
        \mathbf{M} _{i_3, i_1, i_2}, & s = 2,\\
        \mathbf{M} _{i_2, i_3, i_1}, & s = 3,
    \end{cases} \quad i_1 \in [p_s], \, i_2 \in [p_{s+1}], i_3 \in [p_{s+2}].
\end{equation}

An important quantity in tensor algebra is the Tucker ranks.  For any tensor $\mathbf{M} \in \mathbb{R}^{p_1 \times p_2 \times p_3}$, its Tucker ranks $(r_1, r_2, r_3)$ are defined as $r_s = \mathrm{rank}(\mathcal{M}_s(\mathbf{M}))$, $s \in [3]$.  We further let $\sigma_* (\mathbf{M} ) = \min \{\sigma_{r_s} (\mathcal{M}_s(\mathbf{M} )), \, s \in [3]\}$.  For any $s \in [3]$ and $U_s \in \mathbb{R}^{q_s \times p_s}$, the marginal multiplication operator is defined as
\begin{equation}\label{eq-times-def}
    \mathbf{M}  \times_s U_s = \begin{cases}
        \big(\sum_{k = 1}^{p_{1}} \mathbf{M} _{k, j, l} U_{i, k} \big)_{i \in [q_1], \, j \in [p_2], \, l \in [p_3]} \in \mathbb{R}^{q_1 \times p_2 \times p_3}, & s = 1,\\
        \big(\sum_{k = 1}^{p_{2}} \mathbf{M} _{i, k, l} U_{j, k} \big)_{i \in [p_1], \, j \in [q_2], \, l \in [p_3]} \in \R^{p_1 \times q_2  \times p_3}, & s = 2,\\ 
        \big(\sum_{k = 1}^{p_{3}} \mathbf{M} _{i, j, k} U_{l, k} \big)_{i \in [p_1], \, j \in [p_2], \, l \in [q_3]} \in \R^{p_1 \times p_2  \times q_3}, & s = 3.
    \end{cases}
\end{equation}

\section{Multilayer random dot product graphs}\label{sec-single-mrdpg}

In this section, we consider the estimation of a single multilayer network, under multilayer random dot product graph (MRDPG) models \citep{jones2020multilayer}.  For completeness, in \Cref{sec-def-mrdpg}, we detail two sub-model assumptions.  A tensor-based estimation procedure is proposed in \Cref{sec_estimation_single}, with its theoretical guarantees and comparisons with existing literature in Sections~\ref{subsec-theoretical-guarantee-Alg1} and \ref{sec-comp-single}, respectively.   

\subsection{Model description}\label{sec-def-mrdpg}

Recall that a random dot product graph assumes that each node is associated with a latent position $X \in \mathcal{X} \subset \mathbb{R}^d$.  Conditioning on the latent positions, the connectivity between a pair of nodes is assumed to be a function of the inner product of their latent positions.  We extend such latent structure from a single network to a multilayer network.  To formalize this extension, consistent with the notation introduced in \Cref{def-multi-layer-network}, we present two detailed model assumptions in Definitions~\ref{umrdpg-f} and \ref{umrdpg}, depending on the randomness of the latent positions.  

We highlight that, for single multilayer networks, different model assumptions hardly make any difference in terms of estimation and theoretical guarantees.  This is especially seen in \Cref{random_theorem}.  We include them here, since when it comes to dynamic MRDPGs discussed in Sections~\ref{sec-dynamic-f} and \ref{sec-dynamic}, the difference will be prominent.

\begin{definition}[Multilayer random dot product graphs with fixed latent positions, $\mbox{MRDPG}$-Fix]\label{umrdpg-f}
 Given a sequence of deterministic matrices $\{W_{(l)}\}_{l = 1}^L \subset \R^{d \times d}$, let $\{X_i\}_{i=1}^{n} \subset \mathbb{R}^d$ be fixed vectors satisfying for all $i, j \in [n]$  and $l \in [L]$, $X_i^{\top} W_{(l)} X_j \in [0, 1]$.
 
Tensor $\mathbf{A}  \in \R^{n \times n  \times L}$ is an adjacency tensor of an undirected multilayer random dot product graph with fixed latent positions $\{X_i\}_{i=1}^{n}$ if the distribution satisfies that
\begin{align*}
    \mathbb{P} \{\mathbf{A} \} & = \prod_{l = 1}^{L} \prod_{ 1 \leq i\leq j \leq n}\mathbf{P}_{i, j, l}^{\mathbf{A}_{i, j, l}} (1- \mathbf{P}_{i, j, l})^{1- \mathbf{A} _{i, j, l}} \\
    & = \prod_{l = 1}^{L} \prod_{ 1 \leq i \leq j \leq n} \big(X_i^{\top} W_{(l)} X_j\big)^{\mathbf{A} _{i, j, l}} \big(1 - X_i^{\top} W_{(l)} X_j\big)^{1 - \mathbf{A} _{i, j, l}}.
\end{align*}
We write $\mathbf{A} \sim \mathrm{MRDPG}\mbox{-}\mathrm{Fix}(\{X_i\}_{i=1}^{n}, \{W_{(l)}\}_{l\in [L]})$ and write $\mathbf{P} = (\mathbf{P}_{i, j, l}) \in \mathbb{R}^{n \times n \times L}$ as the probability tensor. 
\end{definition}

Moving on from fixed to random latent positions, one usually imposes some additional modelling to enforce that the connectivity probability is between zero and one. This definition is presented in \Cref{ipd}.

\begin{definition}[Inner product distribution pair] \label{ipd}
For positive integers $d$ and $L$, given a sequence of deterministic matrices $\{W_{(l)}\}_{l = 1}^L \subset \R^{d \times d}$, let $\mathcal{F}$ and $\widetilde{\mathcal{F}}$ be two distributions with support $\mathcal{X}, \widetilde{\mathcal{X}} \subset \R^d$, respectively.  We say that $(\mathcal{F}, \widetilde{\mathcal{F}})$ is an inner product distribution pair with weight matrices $\{W_{(l)}\}_{l\in [L]}$, if for all $l \in [L]$, $\mathbf{x} \in \mathcal{X}$ and $\tilde{\mathbf{x}} \in \widetilde{\mathcal{X}}$, it holds that $\mathbf{x}^{\top} W_{(l)} \tilde{\mathbf{x}} \in [0, 1]$. 
\end{definition}

Compared to the assumptions imposed in the existing literature \cite[e.g.][]{athreya2017statistical, padilla2019change}, \Cref{ipd} is more general on two fronts.  First, since we allow the edges to be directed in \Cref{sec-directed-edges}, \Cref{ipd} concerns a pair of distributions rather than a single one, as used in the undirected networks.  Second, due to the multilayer nature, additional weight matrices come to play.  The latter will be clearer after \Cref{umrdpg}.

With latent positions drawn from distributions defined in \Cref{ipd}, the undirected multilayer random dot product graphs with random latent positions are defined in \Cref{umrdpg}.

\begin{definition}[Multilayer random dot product graphs with random latent positions, MRDPG]\label{umrdpg}
Let $(\mathcal{F}, \mathcal{F})$ be an inner product distribution pair with weight matrices $\{W_{(l)}\}_{l\in [L]}$ satisfying \Cref{ipd}.  Let $\{X_i\}_{i=1}^{n} \subset \mathbb{R}^d$ be mutually independent random vectors generated from~$\mathcal{F}$.  

Tensor $\mathbf{A} \in \R^{n \times n  \times L}$ is an adjacency tensor of an undirected multilayer random dot product graph with random latent positions $\{X_i\}_{i=1}^{n}$ if the conditional distribution satisfies that
\begin{align*}
    \mathbb{P} \big\{\mathbf{A} \vert \{X_i\}_{i = 1}^{n}\big\} & = \prod_{l = 1}^{L} \prod_{ 1 \leq i \leq j \leq n} \mathbf{P}_{i, j, l}^{\mathbf{A}_{i, j, l}} (1- \mathbf{P}_{i, j, l})^{1- \mathbf{A}_{i, j, l}} \\
    & = \prod_{l = 1}^{L} \prod_{ 1 \leq i \leq j \leq n} \big(X_i^{\top} W_{(l)} X_j\big)^{\mathbf{A}_{i, j, l}} \big(1 - X_i^{\top} W_{(l)} X_j\big)^{1 - \mathbf{A}_{i, j, l}}.
\end{align*}
We write $\mathbf{A} \sim \mathrm{MRDPG}(\mathcal{F}, \{W_{(l)}\}_{l\in [L]}, n)$ and write $\mathbf{P} = (\mathbf{P}_{i, j, l}) \in \mathbb{R}^{n \times n \times L}$ as the probability tensor. 
\end{definition}

Using the marginal multiplication operator defined in \eqref{eq-times-def}, the probability tensor $\mathbf{P}$ introduced in Definitions~\ref{umrdpg-f} and \ref{umrdpg} can be written in a Tensor representation, i.e.
\begin{equation}\label{Y_tucker_rep}
    \mathbf{P} = \mathbf{S} \times_1 X \times_2 X \times_3 Q,
\end{equation}
where $X = (X_1, \ldots, X_{n})^{\top} \in \mathbb{R}^{n \times d}$, $\mathbf{S} \in \R^{d \times d \times d^2}$ with 
\begin{equation}\label{matrix_S}
    \mathbf{S}_{i, j, l} = \mathbbm{1}\{l = (i-1)d + j\}
\end{equation}  
and 
\begin{equation}\label{matrix Q}
    Q = \begin{pmatrix}
        (W_{(1)})^1 & \cdots & (W_{(1)})^d \\
        \vdots & \vdots & \vdots \\ 
        (W_{(L)})^1 & \cdots & (W_{(L)})^d 
   \end{pmatrix} \in \R^{L \times d^2}.
\end{equation}
For $l \in [L]$, the $l$th layer of the probability tensor is $\mathbf{P}_{:,:,l} = XW_{(l)}X^{\top}$.  

Note that the flexibility of Definitions~\ref{umrdpg-f} and \ref{umrdpg} lies in at least two aspects: (\romannumeral 1)  both fixed and random latent positions are allowed, especially noting that we do not enforce $\mathcal{F}$ to possess density functions; and (\romannumeral 2) with the introduction of layer-specific connectivity matrices, we in fact allow both layer-specific and layer-shared latent positions.  As for (\romannumeral 2), see for instance, when $L = 2$ and $d = 3$, let
\[
    W_{(1)} = \left(\begin{array}{ccc}
         1 & 0 & 0  \\
         0 & 0 & 0  \\
         0 & 0 & 1
    \end{array} \right) \quad \mbox{and} \quad W_{(2)} = \left(\begin{array}{ccc}
         0 & 0 & 0  \\
         0 & 1 & 0  \\
         0 & 0 & 1
    \end{array} \right).
\]
Then for the latent position vector $X_1 = (X_{1,1}, X_{1,2}, X_{1,3})^{\top} \in \mathbb{R}^3$, it holds that $X_1^{\top}W_{(1)}X_1 = X_{1,1}X_{1,1} + X_{1,3}X_{1,3}$ and $X_1^{\top}W_{(2)}X_1 = X_{1,2}X_{1,2} + X_{1,3}X_{1,3}$.  This means that the first two coordinates are layer-specific and the third coordinate is layer-shared.

\begin{remark}\label{msbm}
Since random dot product graphs can be seen as a generalization of the stochastic block model, the MRDPG model in \Cref{umrdpg} may also be seen as a generalization of a multilayer stochastic block model with $d$-communities.  To be specific, let the distributions $\mathcal{F}$ be discrete distributions, supported on the sets $\{\mathbf{e}^1, \ldots, \mathbf{e}^d\}$ - standard basis vectors of $\R^d$.  The latent position matrix $X$ can therefore be seen as the membership matrix and $\{ W_{(l)} \}_{l\in [L]}$ the connectivity matrices encoding the edge probabilities between communities.  This example amplifies the importance of allowing for distributions not admitting densities.
\end{remark}

\begin{remark}[Model comparisons with existing literature]
As mentioned in \Cref{sec-related-lit-intro}, there are a few existing papers modelling multilayer networks assuming latent positions. We elaborate on the model comparisons here.

\cite{zhang2020flexible} formulated the probability tensor $\mathbf{P} \in \R^{n \times n \times L}$ as  $\mathrm{logit}(\mathbf{P}_{i, j, l}) = \alpha^{(l)}_{i} + \alpha^{(l)}_{j} + ( X_i )^{\top}W_{(l)} X_j$ for each $i, j \in [n]$ and $l \in [L]$.  Compared with ours, they assumed the latent positions to be fixed and exploited the logit function to avoid the inner product distribution assumption.
\cite{arroyo2021inference} proposed the COSIE model for analysing multiple heterogeneous networks.  Translated to our notation, each layer in a COSIE model represents a network and the probability tensor $\mathbf{P} \in \R^{n \times n \times L}$ is expressed as $\mathbf{P}_{:, :, l} = X W_{(l)}X^{\top}$ for each $l \in [L]$, aligning with \Cref{umrdpg-f}. 
\cite{jones2020multilayer} introduced the MRDPG model and considered packing all the layer-specific latent positions into the tail nodes. Specifically, they assumed $\mathbf{P} \in \mathbb{R}^{ n_1\times n_2 \times L } $ and for each $l \in [L]$, $\mathbf{P}_{:,:,l} = XW_{(l)}(Y^{(l)})^{\top}$ with $X \in \R^{n_1 \times d}$ and $\{Y^{(l)}\}_{l=1}^{L} \subset \R^{n_2 \times d}$. This approach enabled them to expand a tensor into a large matrix by stacking all layers together. 
\cite{macdonald2022latent} separated the layer-shared and layer-specific contributions from the latent positions by writing the $l$th layer, $l \in [L]$, probability matrix as $\mathbf{P}_{:,:,l} = F + G_{(l)}$, where both $F$ and $G_{(l)}$ are low-rank matrices.  \cite{jing2021community} considered a multilayer stochastic block model and represented it in a tensor form.  To be specific, they let the probability tensor be $\mathbf{P} = \mathbf{B} \times_1 X \times_2 X \times_3 Q$, where $\mathbf{B}$ is the connectivity tensor, $X$ is the global membership matrix and $Q$ is the layer-specific membership matrix.  The matrices  $X$ and $Q$ are treated as fixed.  This can be seen as a special case of \Cref{umrdpg-f}. 

We acknowledge that, in terms of the model, we do share similarities with the existing literature, especially \cite{arroyo2021inference} and \cite{jones2020multilayer}. Fundamental differences in optimization and nonparametric methodology will be unveiled in the subsequent sections.
\end{remark}

\subsection{Estimation of multilayer random dot product graphs} \label{sec_estimation_single}

\begin{algorithm}[!t]
\caption{TH-PCA$(\mathbf{A}, (r_1, r_2, r_3))$} \label{thpca}
\begin{algorithmic}
    \INPUT{tensor $\mathbf{A} \in \mathbb{R}^{p_1 \times p_2 \times p_3}$, ranks $r_1, r_2, r_3 \in \mathbb{N}^*$.}
    \For{each mode $s \in [3]$}
        \State{Compute the mode-$s$ matricization $\mathcal{M}_s(\mathbf{A})$ of tensor $\mathbf{A}$. \Comment{\small{See \eqref{eq-matricisation} for $\mathcal{M}_s(\mathbf{A})$}}}
        \State{Apply H-PCA on the matrix $\mathcal{M}_s(\mathbf{A})\mathcal{M}_s(\mathbf{A})^\top$ with rank $r_s$ to obtain the matrix $\widehat{U}_s$:  
        \[
        \widehat{U}_s \leftarrow \mbox{H-PCA}( \mathcal{M}_s (\mathbf{A})  \mathcal{M}_s (\mathbf{A})^{\top}, r_s).  
        \]
        \Comment{See \small\Cref{hpca} for H-PCA}
        }
    \EndFor
    \State{Construct the tensor approximation $\widetilde{\mathbf{P}}$ by applying marginal multiplications: \Comment{See \eqref{eq-times-def} for $\times_s$}
    \[
    \widetilde{\mathbf{P}} \leftarrow \mathbf{A} \times_1 \widehat{U}_1 \widehat{U}_1^{\top} \times_2 \widehat{U}_2 \widehat{U}_2^{\top} \times_3 \widehat{U}_3  \widehat{U}_3^{\top}.
    \]}
    \For{each entry $\{i, j, l\} \in [p_1] \times [p_2] \times [p_3]$}
        \State{Truncate $\widetilde{\mathbf{P}}$ to obtain $\widehat{\mathbf{P}}$:
        \[
        \widehat{\mathbf{P}}_{i, j, l} \leftarrow  
        \begin{cases}
         1,  & 
         \widetilde{\mathbf{P}}_{i, j, l} > 1,  \\
      \widetilde{\mathbf{P}}_{i, j, l}, & \widetilde{\mathbf{P}}_{i, j, l} \in [0, 1], \\
       0, &\widetilde{\mathbf{P}}_{i, j, l} < 0. 
          \end{cases}
        \]}
    \EndFor
    \OUTPUT{$\widehat{\mathbf{P}} \in \mathbb{R}^{p_1 \times p_2 \times p_3}$}
\end{algorithmic}
\end{algorithm}

\begin{algorithm}[!t] 
\caption{H-PCA$(\Sigma, r)$} \label{hpca}
\begin{algorithmic}
    \INPUT{matrix $\Sigma \in \mathbb{R}^{n \times n}$, rank $r \in \mathbb{N}^*$.}
    \State{Initialize $\widehat{\Sigma}^{(0)}$ by copying $\Sigma$ and setting its diagonals to zero:  $\widehat{\Sigma}^{(0)} \leftarrow \Sigma$, $\mathrm{diag}\big(\widehat{\Sigma}^{(0)}\big) \leftarrow 0.$}
    \State{Set the maximum number of iterations $T$ as $T \leftarrow 5\log\{\sigma_{\min}(\Sigma)/n\}$.}
    \For{each iteration $t \in \{0\} \cup [T-1]$}
        \State{Perform singular value decomposition on $\widehat{\Sigma}^{(t)}$:
        \[
        \widehat{\Sigma}^{(t)} = \sum_{i=1}^n\sigma^{i, (t)} \mathbf{u}^{i,(t)} ( \mathbf{v}^{i, (t)})^{\top},  \quad  \sigma^{1, (t)} \geq \cdots \geq  \sigma^{n, (t)}  \geq 0. 
        \]}
        \State{Construct $\widetilde{\Sigma}^{(t)}$ as the best rank-$r$ approximation of $\widehat{\Sigma}^{(t)}$:
        \[
        \widetilde{\Sigma}^{(t)} \leftarrow \sum_{i = 1}^r \sigma^{i, (t)} \mathbf{u}^{i, (t)} \big(\mathbf{v}^{i, (t)}\big)^{\top}.
        \]}
        \State{Update $\widehat{\Sigma}^{(t+1)}$ by setting its diagonals to those of $\widetilde{\Sigma}^{(t)}$:
        \[
        \widehat{\Sigma}^{(t+1)} \leftarrow \widehat{\Sigma}^{(t)},  \quad \mathrm{diag}\big(\widehat{\Sigma}^{(t+1)}\big) \leftarrow \mathrm{diag}\big(\widetilde{\Sigma}^{(t)}\big).
        \]}
    \EndFor
    \State{Let $\{\mathbf{u}^i\}_{i = 1}^r$ denote left singular vectors of $\widehat{\Sigma}^{(T)}$.}
    \OUTPUT{$U\leftarrow (\mathbf{u}^1, \ldots, \mathbf{u}^r) \in \mathbb{R}^{n \times r}$}
\end{algorithmic}
\end{algorithm}

Recall our close cousin random dot product graphs.  With the involvement of latent positions, one usually exploits the low-rank structure and deploys spectral methods to estimate the probability matrix.  Moving on to the MRDPGs defined in \Cref{sec-def-mrdpg}, we first note that for a tensor, its low-rank structure lies in the values of its Tucker ranks (see \Cref{sec:background}).  With the formulation~\eqref{Y_tucker_rep}, in this subsection, we summon the tensor heteroskedastic principal component analysis (TH-PCA) algorithm introduced in \cite{han2022optimal} to estimate the MRDPG.  For completeness, we include TH-PCA in \Cref{thpca}, with its subroutine heteroskedastic principle component analysis (H-PCA) algorithm in \Cref{hpca}, which is applied to low-rank matrices and is developed in \cite{zhang2018heteroskedastic}.  The core of TH-PCA as well as H-PCA is to iteratively conduct SVD on matrices - this can be heavily parallelized. 

 The entries of the adjacency tensor follow that $\mathbf{A}_{i, j, l} \vert  X \sim \mathrm{Bernoulli}(\mathbf{P}_{i, j, l})$, i.e.~conditional on~$X$, each entry follows a Bernoulli distribution, as detailed in \Cref{sec-def-mrdpg}.  The node- and layer-varying nature of the Bernoulli parameters $\{\mathbf{P}_{i, j, l}\}$ leads us to the choice of TH-PCA, which was put forward by \cite{han2022optimal} to handle the heterogeneity in tensor entries when estimating a low-rank tensor.  As pointed out by \cite{han2022optimal}, compared to the maximum likelihood estimation, the higher order singular value decomposition \citep{de2000multilinear} and the higher order orthogonal iteration \citep{de2000best} methods, TH-PCA is minimax rate-optimal estimating of~$\mathbf{P}$, with Guassianity assumptions on the noise, measured by the tensor Frobenius norm.  We inherit such superiority in the multilayer network estimation, demonstrated in \Cref{sec-comp-single}.

\Cref{thpca} opens with estimating a low-rank singular subspace of each mode of the matricization of $\mathbf{P}$, with a deployment of the sub-routine H-PCA detailed in \Cref{hpca}.  It then glues the low-rank subspaces back to a low-rank tensor, with a final truncation step.  The core of \Cref{thpca} and the secret ingredient in dealing with the heterogeneity is \Cref{hpca}.  It iteratively imputes the diagonal entries by the diagonals of its low-rank approximation.   The tuning parameter selection, especially the rank inputs, is deferred to \Cref{simu_denoising} for practical guidance.

\subsubsection[]{Theoretical guarantees for \Cref{thpca}}\label{subsec-theoretical-guarantee-Alg1}

In this subsection alone, we allow all model parameters, including the latent position dimension~$d$ and the number of layers $L$ to diverge as the size of the network $n$.  Before stating theoretical guarantees on \Cref{thpca} on adjacency tensors defined in \Cref{sec-def-mrdpg}, we collect some necessary assumptions below. 

\begin{assumption}[Fixed latent positions]\label{ass_X_Y_u_f}
Consider an $\mathrm{MRDPG}\mbox{-}\mathrm{Fix}(\{X_i\}_{i=1}^n, \{W_{(l)}\}_{l = 1}^L)$ defined in \Cref{umrdpg-f}. Let $X = (X_1, \ldots, X_{n})^{\top} \in \mathbb{R}^{n \times d}$.
$(a)$ Let a singular value decomposition of $X$ be $X = U_X D_X V_X^{\top}$, with $U_X \in \mathbb{R}^{n \times d}$, $D_X, V_X \in \mathbb{R}^{d \times d}$.  Assume that  $\mathrm{rank}(X) = d$, $\sigma_1(X)/\sigma_d(X) \leq C_{\sigma}$ and  $\sigma_{d}(X) \geq C_{\mathrm{gap}} \sqrt{n}$ with  absolute constants $C_{\mathrm{gap}}, C_{\sigma}>0$. $(b)$ For $Q \in \mathbb{R}^{L \times d^2}$ defined in \eqref{matrix Q}, let $m = \mathrm{rank}(Q)$ and a singular value decomposition of~$Q$ be $Q = U_Q D_Q V_Q^{\top}$, with $U_Q \in \mathbb{R}^{L \times m}$, $D_Q \in \mathbb{R}^{m \times m}$ and $V_Q \in \mathbb{R}^{d^2 \times m}$.  Assume that $ \sigma_1(Q)/\sigma_m(Q) \leq C_{\sigma}$ and  $\sigma_{m}(Q) \geq C_{\mathrm{gap}}$, where $C_{\mathrm{gap}}, C_{\sigma}$ are the same as those in $(a)$.
\end{assumption}

\begin{assumption}[Random latent positions]\label{ass_X_Y}
Consider an $\mathrm{MRDPG}(\mathcal{F}, \{W_{(l)}\}_{l = 1}^L, n)$ defined in \Cref{umrdpg}. Let $X_1 \sim \mathcal{F}$ be a sub-Gaussian random vector with $\Sigma_X = \mathbb{E}(X_1X_1^{\top}) \in \mathbb{R}^{d \times d}$.  Let $\mu_{X, 1} \geq \cdots \geq \mu_{X, d} > 0$ be the eigenvalues of $\Sigma_X$.   $(a)$ Assume that there exist absolute constants $C_{\mathrm{gap}}, C_{\sigma} > 0$, such that $\mu_{X, d} > C_{\mathrm{gap}}$ and $\mu_{X, 1}/\mu_{X, d}\leq C_{\sigma}$.  $(b)$ Further let $\theta_X = \E(X_1)$.  Assume that there exists an absolute constant $C_{\theta} > 0$ such that $\mu_{X, 1}/\| \theta_X\|^2 \leq C_{\theta}$.
    $(c)$ For $Q \in \mathbb{R}^{L \times d^2}$ defined in \eqref{matrix Q}, let $m = \mathrm{rank}(Q)$ and a singular value decomposition of $Q$ be $Q = U_Q D_Q V_Q^{\top}$, with $U_Q \in \mathbb{R}^{L \times m}$, $D_Q \in \mathbb{R}^{m \times m}$ and $V_Q \in \mathbb{R}^{d^2 \times m}$.  Assume that $ \sigma_1(Q)/\sigma_m(Q) \leq C_{\sigma}$ and  $\sigma_{m}(Q) \geq C_{\mathrm{gap}}$, where $C_{\mathrm{gap}}, C_{\sigma}$ are the same as those in $(a)$.
\end{assumption}

\Cref{ass_X_Y_u_f}$(a)$ assumes the full-rankness of $X$, with its smallest eigenvalue lower bounded by the order $\sqrt{n}$ and all the eigenvalues of the same order. This is essentially the same as imposing a constant lower bound on $\mu_{X, d}$ in \Cref{ass_X_Y}$(a)$. 
We note that $X$ encodes latent positions rather than any observed data. The full-rankness condition on $X$ is essentially a condition imposed on the knowledge of the intrinsic dimension $d$.  To be specific, assuming the full-rankness of $X$ is in fact assuming the input ranks $r_1$ and $r_2$ in \Cref{thpca} are no less than the intrinsic dimension $d$.  Further discussion on rank selection is provided at the end of \Cref{subsec-theoretical-guarantee-Alg1}. 

\Cref{ass_X_Y}$(a)$ and $(b)$ impose regularity conditions on the inner product distribution $\mathcal{F}$, especially, both requiring the full-rankness of $\Sigma_X$. \Cref{ass_X_Y}$(a)$ in addition assumes that the smallest eigenvalues are lower bounded by an absolute constant and that all the eigenvalues are of the same order.  This is a condition regularly seen in the RDPG literature when the latent positions are assumed to be random \cite[e.g.][]{athreya2017statistical}.  \Cref{ass_X_Y}$(b)$ essentially requires that the first and second moments of $\mathcal{F}$ exist and are of the same order.  Such a requirement is necessary for the nonparametric analysis conducted in our paper for change point detection.  More discussions are available in \Cref{sec-dynamic}.  We remark that for inner product distributions defined in \Cref{ipd}, a lower bound on the $\ell_2$-norms of mean vectors is reasonable.

Regarding the weight matrices,  both Assumptions~\ref{ass_X_Y_u_f}$(b)$ and \ref{ass_X_Y}$(c)$ impose a low-rank condition on $Q$, with its smallest eigenvalue lower bounded by a constant order and all the eigenvalues of the same order, with the matrix $Q$ being a collection of the weight matrices.  Note that, we assume the rank of $Q$ is $m \leq \min \{L, \, d^2\}$, but we do allow $m, L, d$ to be functions of the node sizes as for now. 

With \Cref{ass_X_Y_u_f} for the case of fixed latent positions and \Cref{ass_X_Y} for the cases of random latent positions in hand, we come in sight of theoretical guarantees for the output of TH-PCA.

\begin{theorem}\label{random_theorem}
Suppose that \Cref{ass_X_Y_u_f} holds with fixed latent positions, or \Cref{ass_X_Y}$(a)$ and $(c)$ hold with random latent positions.
Let $\widehat{\mathbf{P}}$  be the output of TH-PCA (\Cref{thpca}) with inputs
\begin{itemize}
    \item an adjacency tensor $\mathbf{A}$ satisfying either
        $(i)$ $\mathbf{A} \sim \mathrm{MRDPG}\mbox{-}\mathrm{Fix}( \{X_i \}_{i=1}^n,  \{W_{(l)}\}_{l = 1}^L)$ described in \Cref{umrdpg-f},
        $(ii)$ or $\mathbf{A} \sim \mathrm{MRDPG}(\mathcal{F},  \{W_{(l)}\}_{l = 1}^L, n)$ described in \Cref{umrdpg}, and
    \item Tucker ranks $(r_1, r_2, r_3) = (d, d, m)$, where $d$ and $m$ are defined in Assumptions~\ref{ass_X_Y_u_f} or \ref{ass_X_Y}.
\end{itemize}

Letting $\mathbf{P}$ be the probability tensor of the adjacency tensor $\mathbf{A}$, it then holds that 
\begin{align}\label{upper_bound_theorem_unified}
    \mathbb{P}\big\{\|\widehat{\mathbf{P}} - \mathbf{P}\|_{\mathrm{F}}^2 \leq  C(d^2m + nd + Lm)\big\} > 1 - C(n \vee L)^{-c},
\end{align}
where $C, c > 0$  are constants depending on the constants $C_{\mathrm{gap}}, C_{\sigma}>0 $, which are defined in Assumptions~\ref{ass_X_Y_u_f} or \ref{ass_X_Y}.
\end{theorem}

\Cref{random_theorem} is an application of Theorem 4.1 in \cite{han2022optimal}, where the probability tensor $\mathbf{P}$ and therefore its Tucker ranks are considered as fixed, in contrast to the randomness introduced by the random latent positions we consider here.  The proof of \Cref{random_theorem} is deferred to \Cref{proof-sec2} where efforts have been made to control the randomness in the eigenvalues and ranks of relevant quantities.  As pointed out in \cite{zhang2018tensor} and \cite{han2022optimal}, the high-probability upper bound in \eqref{upper_bound_theorem_unified} is in fact minimax rate-optimal when $\mathbf{P}$ is considered as fixed and $\mathbf{A} - \mathbf{P}$ have Gaussian entries.

We emphasize that in \Cref{random_theorem}, the number of layers $L$, the latent position dimension $d$ and the rank of the matrix $Q$, denoted by $m$, are all allowed to diverge as the network size $n$ diverges. In contrast, \cite{jones2020multilayer} require that $L$ is either fixed or grows at a rate much slower than that of $n$. 
Although \cite{macdonald2022latent}  allow $L$, $d$ and $m$ to grow, they additionally assumed that $dL^{-1} = o(1)$ and $2^ddm^2n^{1-2c} = o(1)$, for some constant $c \in (1/2, 1]$. 
These aforementioned results are hence stricter than those in \Cref{random_theorem}.  
\cite{arroyo2021inference} do allow for diverging parameters without additional constraints, but achieving a larger error control, as detailed in \Cref{sec-comp-single}.
Despite the positive results, \Cref{random_theorem} does rely on the right choice of the Tucker ranks $(d, d$, $m)$. Generally speaking, estimating the ranks of low-rank matrices or tensors with theoretical justification still remains a challenging task. Using values of $d$ and $m$ larger than necessary will inflate the error bound. On the other hand, underestimating $d$ and $m$ is potentially more consequential, as it means missing out on important eigenspaces,  structural properties of the underlying distribution. For this reason, in our numerical experiments, we choose relatively large~$d$ and $m$ as inputs despite the increase in the computational cost. 

\subsection{Comparisons with the existing literature}\label{sec-comp-single}

We have introduced two models: an MRDPG with fixed latent positions in \Cref{umrdpg-f} and an MRDPG with random latent positions in \Cref{umrdpg}.  We proposed to exploit low-rank tensor estimation to estimate the probability tensor in \Cref{thpca} and provide an error bound in \Cref{random_theorem}. A generalization of these results to directed MRDPG with random positions is given in \Cref{sec-directed-edges}.  We have already compared our work with the existing multilayer network literature regarding modelling choices and theoretical guarantees.  In this subsection, we provide more justification for the use of tensors, and highlight the differences of our approach with current approaches in the literature. 

\cite{jing2021community} resorted to low-rank tensor estimation to study a multilayer network, and proposed the Tucker decomposition with integrated singular value decomposition transformation (TWIST) algorithm.  They focused on clustering and node membership estimation, rather than on recovering the underlying connectivity probability.  Because their goals are markedly different from ours, we do not provide direct theoretical comparisons, but only numerical ones.

\cite{jones2020multilayer} proposed the unfolded adjacency spectral embedding (UASE) algorithm, by stacking the adjacency matrices from all layers into a large adjacency matrix and applying spectral methods to recover a low-rank approximation.  Note that the theoretical results in \cite{jones2020multilayer} are on the recovery of the latent positions rather than the underlying connectivity probabilities. Taking inspiration from this approach, we also consider scaled adjacency spectral embedding (SASE) estimators \cite[e.g.][]{sussman2012consistent}.  In detail, the adjacency matrices of each layer are estimated separately and then aggregated into one probability tensor.  Using the universal singular value thresholding algorithm \citep{xu2018rates} as layer-level estimators, Theorem~1 in \cite{xu2018rates} implies that, with high probability, $\|\widehat{\mathbf{P}}^{\mathrm{SASE}} - \mathbf{P}\|_{\mathrm {F}}^2 \lesssim Lnd$, yielding a larger error bound compared to the one obtained in \Cref{random_theorem}.
 
\cite{arroyo2021inference} introduced the multiple adjacency spectral embedding (MASE) algorithm, an alternative method for applying spectral methods to an augmented matrix.
Adapting their notation to ours, Theorems 7 and 11 in \cite{arroyo2021inference} imply that, with high probability, $\|\widehat{\mathbf{P}}^{\mathrm{MASE}} - \mathbf{P} \|_{\mathrm{F}}^2  \lesssim Lnd$. This is a looser upper bound than the one in \Cref{random_theorem}.
Moving away from spectral methods, \cite{zhang2020flexible} considered a maximum likelihood estimator (MLE) for the logistic model they proposed.  They showed that, with high probability, $\|\widehat{\mathbf{P}}^{\mathrm{MLE}}- \mathbf{P}\|_{\mathrm {F}}^2 \lesssim  Ln + d^2(n +L)$.  For any $n \geq m$, \Cref{random_theorem} provides a sharper bound. 
\cite{zhang2020flexible} proposed a projected gradient descent algorithm to provide an approximate evaluation of the maximum likelihood estimators. 
\cite{macdonald2022latent} studied a similar model and proposed a convex optimization approach (COA) based on the nuclear norm penalty, without any tensor formulation.  The authors show that, with high probability, $\|\widehat{\mathbf{P}}^{\mathrm{COA}} - \mathbf{P} \|_{\mathrm{F}}^2  \lesssim Lnd$, which again is looser than \Cref{random_theorem}.

To conclude these comparisons, we provide a simulation study to numerically compare TH-PCA (\Cref{thpca}) with the aforementioned procedures (TWIST, UASE, SASE, MASE, MLE and COA), as well as two additional methodologies, namely higher-order singular value decomposition \cite[HOSVD,][]{de2000multilinear} and higher-order orthogonal iteration \cite[HOOI,][]{zhang2018tensor}.  The last two are popular low-rank tensor estimation methods, but do not take data heterogeneity into consideration.  We show the results in \Cref{Fig_simulation_denoising}; see \Cref{simu_denoising} for details about the simulation settings.  With varying numbers of layers and nodes, we can see clearly that TH-PCA outperforms all algorithms in almost all settings, with the only exception of COA \citep{macdonald2022latent} when the numbers of layers and nodes are relatively small.  

A few more remarks are in order. Firstly, except for the MLE and COA algorithms, all the other methods can be categorized as matrix-based (SASE, UASE and MASE) or tensor-based methods (TH-PCA, HOSVD, HOOI and TWIST). It is interesting to see that, except for MASE, tensor-based methods prominently outperform matrix-based methods.  Secondly, the estimation errors of the matrix-based methods, as well as MLE and COA, scale linearly in terms of $L$ and $n$, as we discussed.  Thirdly, tensor-based methods are robust with respect to the increasing number of layers.  We conjecture this desirable feature stems from the exploitation of the shared layer structures.  Last but not least, we see the superiority of TH-PCA, an algorithm taking the heterogeneity into consideration, over other low-rank tensor estimation methods in this specific scenario.

\begin{figure}[t] 
\centering 
\includegraphics[width=0.9 \textwidth]{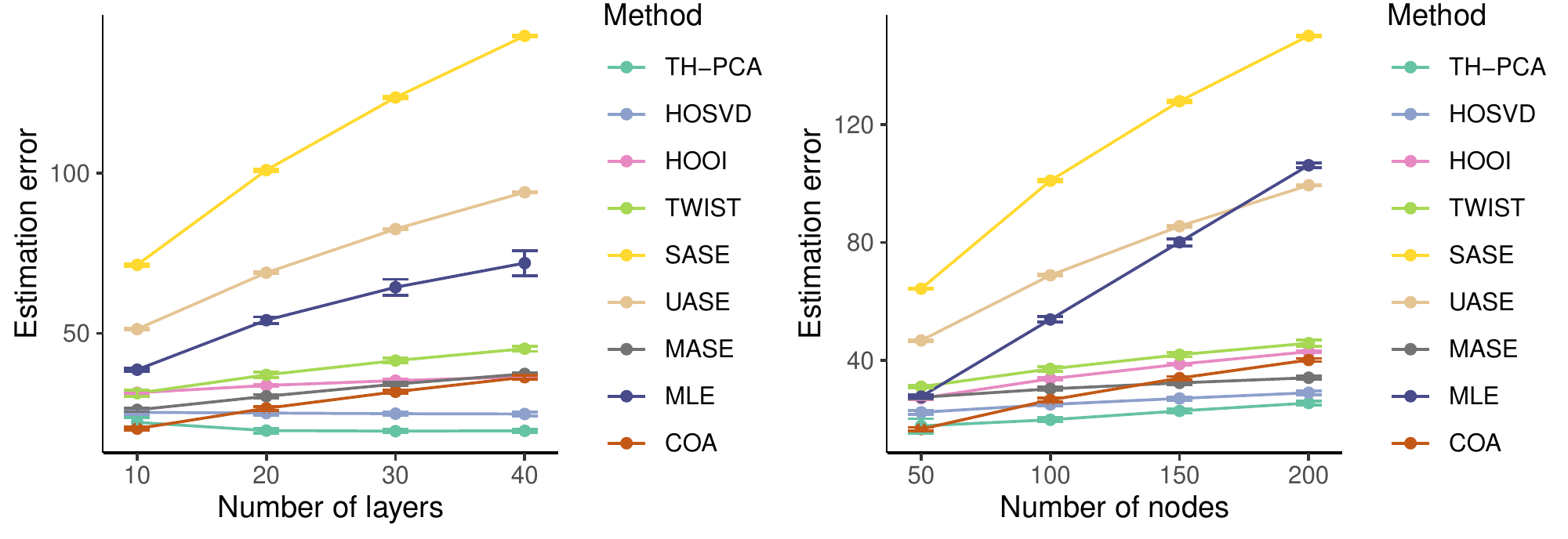}
\caption{Results of estimating a probability tensor with details in Scenario 1 in \Cref{simu_denoising}.  Left panel: $n  = 100$ and $L \in \left\{ 10, 20, 30, 40\right\}$. Right panel: $L = 20$ and $n \in \left\{ 50, 100, 150, 200\right\}$. In both panels, $y$-axis values correspond to the values of estimation error $\| \widehat{\mathbf{P}} - \mathbf{P}\|_{\mathrm{F}}$ of different estimators~$\widehat{\mathbf{P}}$.  The values are in the form of mean and standard deviation over $100$ Monte Carlo trials.} \label{Fig_simulation_denoising}
\end{figure}

\section{Online change point analysis in dynamic multilayer random dot product graphs with fixed latent positions }\label{sec-dynamic-f}

We now consider dynamic (i.e.~time-dependent) MRDPGs and focus on the online change point detection problems.  We begin our analysis with the simplest setting, in which the latent positions are fixed and time-invariant. Thus, the change point is the result of a change in the weight matrix sequence.  This can be seen as a direct extension of the network online change point detection problems \cite[e.g.][]{yu2021optimal}.  

With \Cref{umrdpg-f} as the building block, we formalize the model in \Cref{def-umrdpg-f-dynamic} and the change point structure in Assumptions~\ref{ass_no_change_point_u_f} and \ref{ass_change_point_u_f}.

\begin{definition}[Dynamic multilayer random dot product graphs with fixed latent positions, dynamic MRDPG-Fix]\label{def-umrdpg-f-dynamic}
Let $\{X_i\}_{i=1}^{n} \subset \mathbb{R}^d$ be latent positions and $\{W_{(l)}(t)\}_{l\in[L], t\in \mathbb{N}^*} \subset \mathbb{R}^{d\times d}$ be a weight matrix sequence.  The sequence of mutually independent adjacency tensors  $\{\mathbf{A}(t)\}_{t \in \mathbb{N}^*}$ follow a dynamic undirected MRDPGs with fixed latent positions $\{X_i\}_{i=1}^{n}$ and weight matrix sequence $\{W_{(l)}(t)\}_{l\in[L], t\in \mathbb{N}^*}$, if 
\[
    \mathbf{A}(t) \sim \mathrm{MRDPG}\mbox{-}\mathrm{Fix}( \{X_i\}_{i=1}^{n}, \{W_{(l)}(t)\}_{l \in [L]}), \quad t \in \mathbb{N}^*,
\]
as defined in \Cref{umrdpg-f}.
\end{definition}

\begin{assumption}[No change point]\label{ass_no_change_point_u_f}
Let $\{\mathbf{A}(t)\}_{t \in \mathbb{N}^*}$ be as in \Cref{def-umrdpg-f-dynamic}.  Assume that, for any pair of positive integers $t$ and $s$,
$
 \{W_{(l)}(t)\}_{l = 1}^L =  \{W_{(l)}(s)\}_{l = 1}^L.
$   
\end{assumption}     

\begin{assumption}[One change point] \label{ass_change_point_u_f}
Let $\{\mathbf{A}(t)\}_{t \in \mathbb{N}^*}$ be defined in \Cref{def-umrdpg-f-dynamic}.  Assume that there exists an integer $\Delta \geq 1$ such that 
$\{W_{(l)}(t)\}_{l = 1}^L \neq  \{W_{(l)}(t+1)\}_{l = 1}^L$ if and only if $t = \Delta$.
Let the jump size be $\kappa = \| \mathbf{P}(\Delta) - \mathbf{P}(\Delta+1)\|_{\mathrm {F}} > 0$.
\end{assumption}             

With the change point framework detailed in \Cref{def-umrdpg-f-dynamic}, Assumptions~\ref{ass_no_change_point_u_f} and \ref{ass_change_point_u_f}, we are now ready to describe the online change point detection task and the algorithms we propose to tackle it.  

At a high level, our goal is to detect change points with a control on the false alarm and with a minimum detection delay.  This is mathematically formalized in different ways in the existing literature.  We direct interested readers to \cite{yu2020note} for discussions on the relationship between different criteria.  In this paper, we aim to find a change point estimator $\widehat{\Delta}$ such that (i) for a prescribed false alarm tolerance $\alpha \in (0, 1)$,
\[
    \mathbb{P}_{\infty}\{\widehat{\Delta} < \infty\} \leq \alpha \quad \mbox{and} \quad \sup_{\Delta =1, 2, \ldots} \mathbb{P}_{\Delta}\{\widehat{\Delta} < \Delta \} \leq \alpha,
\]
where $\mathbb{P}_{\infty}$ and $\mathbb{P}_{\Delta}$ indicate the probabilities under the distributions described in Assumptions~\ref{ass_no_change_point_u_f} and~\ref{ass_change_point_u_f}, respectively; and (ii) $(\widehat{\Delta} - \Delta)_+$ possesses a sharp upper bound with high probability, under \Cref{ass_change_point_u_f}.  In the sequel, we will refer to $(\widehat{\Delta} - \Delta)_+$ as the {\it detection delay.}

We propose an online change point detection procedure in \Cref{online_hpca}, which is commonly used in the change point detection literature \cite[e.g.][]{yu2020note, berrett2021locally}.  The main tool of this algorithm is the scan statistic $\widehat{D}_{\cdot, \cdot}$  defined next in \Cref{def-cusum-f}.  With a sequence of pre-specified, time-dependent, threshold values $\{\tau_{s, t} \}_{1 \leq s <t}$, a change point is detected as soon as the scan statistic exceeds its corresponding threshold value. 

\begin{algorithm}[!t] 
\caption{Dynamic MRDPG online change point detection} \label{online_hpca}
\begin{algorithmic}
\INPUT{tensor sequence $\{\mathbf{A}(t)\}_{t \in \mathbb{N}^*} \subset \mathbb{R}^{n \times n \times L}$,  threshold sequence $\{\tau_{s, t}\}_{1 \leq  s < t} \subset \mathbb{R}$, tolerance level $\alpha \in (0, 1)$.}
\State{Initialize time $t \leftarrow  1$ and the detection flag $\mathrm{FLAG} \leftarrow 0$.}
\While{$\mbox{FLAG} = 0$} 
        \State{Increment time $t  \leftarrow  t+1$.}
        \State{Compute scan statistics $\{\widehat{D}_{s, t}\}_{s \in [t-1]}$. \Comment{See \Cref{def-cusum-f}  or \Cref{def-cusum}}}
       \State{Evaluate change point detection $\mbox{FLAG}\leftarrow 1 -\prod_{s = 1}^{t-1} \big\{\widehat{D}_{s, t} \leq \tau_{s, t} \big\}$.}
\EndWhile
\OUTPUT{$t$.}
\end{algorithmic}
\end{algorithm}

\begin{definition}\label{def-cusum-f}
For any integer pair $(s, t)$, $ 1 \leq s < t$, let the scan statistics be $\widehat{D}_{s, t} = \|\widehat{\mathbf{P}}^{0, s} - \widehat{\mathbf{P}}^{s, t}\|_{\mathrm{F}}$, where $\widehat{\mathbf{P}}^{s, t} = \mathrm{TH}\mbox{-}\mathrm{PCA} \big((t-s)^{-1}\sum_{u=s+1}^t \mathbf{A}(u), (d, d, m)\big)$ (see \Cref{thpca}) and the rank parameters~$d$ and $m$ are introduced in \Cref{ass_X_Y_u_f}.  
\end{definition}

Following common heuristics developed in the change point analysis literature, we refer to the quantity 
\[
    \mbox{jump size} \times \sqrt{\mbox{pre-change sample size}}
\]
as the signal-to-noise ratio.  As discussed in the recent literature \citep{yu2020note, berrett2021locally, yu2021optimal}, the above quantity captures the fundamental difficulty of online change point detection problems.  In offline change point analysis problems the quantity $\mbox{jump size} \times \sqrt{\mbox{minimal spacing}}$ plays an analogous role \citep{padilla2019optimal, wang2020univariate, wang2021optimal}.

\begin{assumption}[Signal-to-noise ratio condition]\label{snr_ass_change_point_f} 
Under \Cref{ass_X_Y_u_f}, for any $\alpha  \in (0, 1)$, assume that there exists a large enough absolute constant $C_{\mathrm {SNR}}>0$ such that
\[
\kappa \sqrt{\Delta} \geq C_{\mathrm {SNR}}  \sqrt{ (d^2m  + nd + Lm) \log(\Delta/\alpha)}.
\]
\end{assumption}

To understand \Cref{snr_ass_change_point_f}, we provide a few scenarios.
\begin{itemize}
    \item \textbf{Small pre-change sample size.} Consider $d, m, \alpha \asymp 1$. If $\kappa \asymp n\sqrt{L}$, indicating changes across all layer weight matrices, then \Cref{snr_ass_change_point_f} holds provided $\Delta \gtrsim \log(\Delta)/ (n^2 \wedge nL)$. This implies that, for a large network size $n$ or a large number of layers $L$, \Cref{snr_ass_change_point_f} is satisfied even with a small number of pre-change samples.  Alternatively, if $\kappa \asymp n$, suggesting a change in a single layer weight matrix, then \Cref{snr_ass_change_point_f} holds when $\Delta \gtrsim L \log(\Delta)/ (n^2 \wedge nL)$. This shows that for a large network size~$n$,  \Cref{snr_ass_change_point_f} holds even with a small number of pre-change samples.
    \item \textbf{Large rank tensors.} If $ \alpha \asymp 1$ and $\kappa \asymp n\sqrt{L}$, then \Cref{snr_ass_change_point_f} holds for $d \asymp n$ and $m \asymp L \wedge d^2$, provided that $\Delta \gtrsim \{1 \vee (L \wedge d^2)/n^2\} \log(\Delta)$, allowing for large ranks. 
    \item \textbf{Small jump size.}  Assuming $ d, m, \alpha, \Delta \asymp 1$, \Cref{snr_ass_change_point_f} holds provided that $\kappa \gtrsim \sqrt{n \vee L}$. This implies that the change in each entry can decrease to zero as the network size $n$ diverges.
\end{itemize}

\begin{theorem}\label{main_theorem_f}
Let $\widehat{\Delta}$ be the output of \Cref{online_hpca} applied to an adjacency tensor sequence $\{\mathbf{A}(t)\}_{t \in \mathbb{N}^*}$  following a dynamic $\mathrm{MRDPG}\mbox{-}\mathrm{Fix}$ model as defined in \Cref{def-umrdpg-f-dynamic} and satisfying \Cref{ass_X_Y_u_f}. Let $\alpha \in (0,1)$ and consider the threshold values
\begin{equation}\label{eq-tau-def-thm_f}
        \tau_{s, t} = C_{\tau}   \sqrt{(d^2m + nd +Lm) \log (t/\alpha) }  \left( \frac{1}{\sqrt{s}}+ \frac{1}{\sqrt{t- s}}\right), \quad 1 \leq s \leq t,
    \end{equation}
     where $C_{\tau} > 0$ is an absolute constant. Suppose that the scan statistic sequence used in \Cref{online_hpca} is specified as in \Cref{def-cusum-f}.
Then,
\begin{enumerate}[(i)]
    \item[$(i)$] under \Cref{ass_no_change_point_u_f}, it holds that $\mathbb{P}_{\infty} \{\widehat{\Delta} < \infty\} \leq \alpha$;
    \item[$(ii)$] under Assumptions~\ref{ass_change_point_u_f} and \ref{snr_ass_change_point_f}, it holds that, for an absolute constant  $ C_\epsilon > 0$,
    \[
        \mathbb{P}_{\Delta} \left\{\Delta < \widehat{\Delta} \leq \Delta + C_\epsilon  \frac{(d^2m + nd +Lm) \log (\Delta /\alpha )}{\kappa^2} \right\} \geq 1 - \alpha.
    \]
\end{enumerate}
\end{theorem}

The proof of \Cref{main_theorem_f} can be found in \Cref{proof-sec3}.

When there is no change point, \Cref{main_theorem_f} shows that \Cref{online_hpca} will not raise any false alarm with probability at least $1 - \alpha$. On the other hand, if there is a change point, then  the detection delay is at most of order
\begin{align}
    \frac{(d^2m + nd +Lm) \log (\Delta /\alpha )}{\kappa^2}, \nonumber
\end{align}
with probability at least $1-\alpha$. This rate follows the usual form of detection delay in the online change point literature (or the localization error rate in the offline change point literature), namely
\begin{equation}\label{eq-detection-delay-intuitive_f}
    \mbox{detection delay} \asymp \frac{\mbox{noise variance level}}{\kappa^2}.
\end{equation}
The threshold sequence $\{\tau_{s, t}\}$ was the noise level.  Without any further constraint on the smallest interval size in the change point detection procedure, it can be seen in \eqref{eq-tau-def-thm_f} that 
\begin{equation}\label{eq-sup-tau_f}
    \sup_{1 \leq s < t} \tau_{s, t}\lesssim \sqrt{(d^2m + nd +Lm) \log (t/\alpha) }.
\end{equation}
In view of \eqref{eq-detection-delay-intuitive_f} and \eqref{eq-sup-tau_f}, it can be seen that the detection delay obtained in \Cref{main_theorem_f} follows a pattern consistent with existing literature, including the matrix version of this problem presented in Theorem 1 in \cite{yu2021optimal}.

\section{Online change point analysis in dynamic multilayer random dot product graphs with random latent positions}\label{sec-dynamic}

The problem studied in \Cref{sec-dynamic-f} assumes fixed latent positions that are invariant over time. In turn, this requires a known and time-invariant set of labelled vertices. As a result, this setting rules out missingness and inclusion of new nodes and is not applicable to unlabelled networks.   To relax such stringent assumptions, in this section we consider random latent positions.  Randomness in the latent positions results in  additional complexity, including the need to treat directed and undirected networks separately. 
For ease of presentation, the undirected case is considered below, and the analysis of the directed case is given in \Cref{sec-directed-edges}. 
The general framework is introduced in \Cref{sec-dy-non-framework}.  We describe our online change point detection algorithm and present its theoretical guarantees in Sections~\ref{sec-online-detect-alg} and \ref{dynamic-theory}, respectively.

\subsection{A dynamic and nonparametric framework}\label{sec-dy-non-framework}

Recall that in \Cref{umrdpg}, the connectivity is determined by the inner product distribution $\mathcal{F}$ and the weight matrices $\{W_{(l)}\}_{l = 1}^L$.  In the dynamic setting, we allow all these quantities to be time-varying.  This is detailed in \Cref{def-umrdpg-dynamic}.

\begin{definition}[Dynamic multilayer random dot product graphs with random latent positions, dynamic MRDPG]\label{def-umrdpg-dynamic}
Let $\{W_{(l)}(t)\}_{l \in [L], t\in \mathbb{N}^*} \subset \mathbb{R}^{d\times d}$ be the weight matrix sequence and $\{\mathcal{F}_t\}_{t \in \mathbb{N}^*}$ be the latent position distribution sequence.  We say that $\{\mathbf{A}(t)\}_{t \in \mathbb{N}^*}$ is a sequence of adjacency tensors of dynamic undirected MRDPGs with random latent positions, if they are mutually independent and
\[
    \mathbf{A}(t) \sim \mathrm{MRDPG}(\mathcal{F}_t, \{W_{(l)}(t)\}_{l \in [L]}, n), \quad t \in \mathbb{N}^*,
\]
as defined in \Cref{umrdpg}.
\end{definition}

\begin{remark}
    Note that all the discussions and analysis in the sequel allow $n$ to be time-dependent.  We assume the same node size across time here for notational simplicity.
\end{remark}

We entail the dynamics with a change point component.

\begin{assumption}[No change point]\label{ass_no_change_point}
Let $\{\mathbf{A}(t)\}_{t \in \mathbb{N}^*}$ be defined in \Cref{def-umrdpg-dynamic}.  Assume that, for any pair of positive integers $t$ and $s$, $\mathcal{F}_t = \mathcal{F}_{s} $ and $\{W_{(l)}(t)\}_{l = 1}^L =  \{W_{(l)}(s)\}_{l = 1}^L $. 
\end{assumption}     
    
\begin{thmbis}{ass_change_point}[One change point] \label{ass_change_point-prime}
Let $\{\mathbf{A}(t)\}_{t \in \mathbb{N}^*}$ be defined in \Cref{def-umrdpg-dynamic}.  For $t \in \mathbb{N}^*$, let
\begin{equation}\label{eq-def-L-t}
    H(t) = \Big(\{X_1(t)\}^{\top}W_{(1)}(t)X_2(t), \ldots, \{X_1(t)\}^{\top}W_{(L)}(t)X_2(t)\Big)^{\top}
\end{equation}
denote an $L$-dimensional random vector at time point $t$, with distribution denoted by $\mathcal{H}_t$. 
Assume that there exists an integer $\Delta \geq 1$ such that $\mathcal{H}_{\Delta} \neq \mathcal{H}_{\Delta + 1}$.
\end{thmbis}            

Across time, we assume that the number of layers stays unchanged, but allow for different sets of nodes or layers.  At each time point $t \in \mathbb{N}^*$, we observe an adjacency tensor generated by an MRDPG.  The distribution is determined by $(\mathcal{F}_t,  \{W_{(l)}(t)\}_{l = 1}^L)$.  When there is no change point, we assume that this pair stays the same.  When there does exist a change point, we allow that changes root in either $\mathcal{F}$ or $\{W_{(l)}\}_{l\in[L]}$.  This general setup covers the scenario studied in \cite{padilla2019change}, where they essentially assume that $\{W_{(l)}\}_{l = 1}^L$ is time invariant.  

Despite the generality of \Cref{ass_change_point-prime}, we would like to remark that, the sole observation we have is the adjacency tensor. When the distribution $\mathcal{H}_t$ of the quantity $H(t)$, as defined in \eqref{eq-def-L-t}, is fixed, the distribution of the adjacency tensor is also fixed. This is to say, the unconditional distributions of the adjacency tensor have an equivalent class, determined by \eqref{eq-def-L-t}.  For a special case, \cite{padilla2019change} discussed this phenomenon in Section 3.1 thereof.  Given this observation, we further enforce that $\mathcal{H}_t$ changes at the change point.

In \Cref{ass_change_point-prime}, we envelop the change in the distribution $\mathcal{H}$.  To maximize the flexibility and generality of this setup, we are reluctant to impose more parametric assumptions and resort to nonparametric quantities to further quantify the change size.  In the existing literature, to deal with the multivariate nonparametric distributional change, there are roughly two categories of treatments.  One is to directly assume the density exists and summon some function norm to quantify the differences between two density functions \cite[e.g.][]{padilla2019change, padilla2021optimal}.  The other is to sweep any arbitrary space under the rug of a univariate quantity \cite[e.g.][]{garreau2018consistent, song2022new}.  Despite the generality and flexibility of the latter treatment, the change sensitivity highly depends on the choice of the transformation.  

In our context, we do have a specific $L$-dimensional distribution, each coordinate of which satisfies the inner product distribution pair condition in \Cref{ipd}.  On top of that, we have the multilayer stochastic block model as a special case, where the support of $\mathcal{F}$  is a finite set.  This means that $\mathcal{F}$ does not necessarily possess densities.  In view of these two features, we strengthen \Cref{ass_change_point-prime} by quantifying the change.

\begin{assumption}[One change point]\label{ass_change_point}
Let $\{\mathbf{A}(t)\}_{t \in \mathbb{N}^*}$ be defined in \Cref{def-umrdpg-dynamic}.
Given a kernel function  $\mathcal{K}: \mathbb{R}^L \to \mathbb{R}$ and a bandwidth $h > 0$, for $t \in \mathbb{N}^*$, let $G_t,  \widetilde{G}_t: [0, 1]^L \to \R$ be given by
\[
    G_t(\cdot) = \mathbb{E}\bigg\{h^{-L} \mathcal{K}\bigg(\frac{\cdot - \mathbf{P}_{1, 2, :}(t)}{h}\bigg)\bigg\} \quad \mbox{and} \quad
    \widetilde{G}_t(\cdot) =h^{-L} \mathcal{K}\bigg(\frac{\cdot - \E\{\mathbf{P}_{1, 2, :}(t) \} }{h}\bigg).
\]
Assume that there exists an integer $\Delta \geq 1$ such that $G_{\Delta} \neq G_{\Delta + 1}$ and let the corresponding jump size be
\begin{equation}\label{eq-kappa-def}
    \kappa = \sup_{z \in [0, 1]^L} |G_{\Delta}(z) - G_{\Delta+1}(z)|.
\end{equation}
Assume that there exists an absolute constant $C_{\mathcal{K}} > 2$ such that
\begin{equation}\label{kappa_lower_bound}
    \kappa > C_{\mathcal{K}}\max \bigg\{  \sup_{z \in [0, 1]^L} \big| G_{\Delta}(z) -  \widetilde{G}_{\Delta}(z)\big|, \sup_{z \in [0, 1]^L} \big| G_{\Delta+1}(z) -  \widetilde{G}_{\Delta+1}(z)\big| \bigg\}.
\end{equation}
\end{assumption}

Comparing Assumptions~\ref{ass_change_point-prime} and \ref{ass_change_point}, we see that \Cref{ass_change_point-prime} only asserts that the change lies in the the sequence $\{\mathcal{H}_t\}_{t \in \mathbb{N}^*}$.  If we assume that the densities of $\{\mathcal{H}_t\}_{t \in \mathbb{N}^*}$ exist, then one can directly introduce the change to the density sequences and analyse the change using kernel density estimators, as done in \cite{padilla2021optimal}.  When we want to incorporate the case when the density does not exist, one can instead consider the \emph{point-wise expectation of kernel density estimator}, as denoted by $G_t(\cdot)$ in \Cref{ass_change_point}.

The point-wise expectation of kernel density estimator $G_t(\cdot)$ can be seen as a smoothed density of $\mathbf{P}_{1, 2, :}(t)$, when its density does exist.  As discussed in \cite{fasy2014confidence} and \cite{kim2019uniform}, the importance of $G_t(\cdot)$ goes beyond the case of nonparametric estimation of multivariate distribution when the density does not exist.  We leave interested readers to the aforementioned two papers for more explanations of this quantity.

With the introduction of $G_t(\cdot)$, we naturally quantify the jump size using the supreme norm of the difference between $G_{\Delta}(\cdot)$ and $G_{\Delta+1}(\cdot)$, as a smooth version of the density supreme norm distance \cite[e.g.][]{padilla2021optimal}.  It is worth mentioning that, due to the kernel estimator involved, the jump size itself is a function of the chosen kernel $\mathcal{K}(\cdot)$ and the bandwidth $h$.  This is indeed a caveat, but this is also an improvement over the existing literature.  For instance, \cite{garreau2018consistent} use kernels to transform any type of distributional change to a univariate mean change.  More model-specific features are therefore lost and the reliance on the choice of kernels is, therefore, more severe.  This caveat, however, may also be a blessing in disguise. 
As expected, we will adopt a kernel estimator later in \Cref{def-cusum} in \Cref{sec-online-detect-alg}, where one can select the same kernel and bandwidth as those in \Cref{ass_change_point}. The only constraint for the kernel function is to satisfy the uniform Lipschitz condition defined in \Cref{kernel_function_ass},  which is met by commonly-used kernels like triangular and Gaussian kernels.  The selection of the kernel function and bandwidth is further discussed in numerical experiments in \Cref{change_point_section}.
  
Lastly, we further require $\kappa$ to be lower bounded as in \eqref{kappa_lower_bound}.  This may come across as a bizarre condition, so we elaborate from a few angles.  Firstly, it can be seen as a condition imposed on the kernel function $\mathcal{K}$ and bandwidth $h$.  For instance, if $\mathcal{K}$ is linear or if we can take $h$ to zero in the limit, then $G_t(\cdot) = \widetilde{G}_t(\cdot)$.  Intuitively speaking, the smoothing imposed on the approximation should not be too large to mask the change size in terms of $G_t(\cdot)$.  Secondly, there is an alternative way to present \eqref{kappa_lower_bound} by deriving a high-probability upper bound on $\sup_{z \in [0, 1]^L} | G_t(z) -  \widetilde{G}_t(z)|$.  Without additional assumptions, this has an $O(1)$ upper bound, since essentially it consists of upper bounding one data point's deviation from its mean.  
Having said these, we would like to admit that this is albeit a strong condition.  For commonly-used kernel functions, \eqref{kappa_lower_bound} imposes that $\kappa \gtrsim 1$.  This to a large extent is due to the artefact of our proof.

\subsection{Online change point detection} \label{sec-online-detect-alg}

We will continue to use the online change point detection procedure proposed in \Cref{online_hpca} but will deploy a different, more complex scan statistic given in \Cref{def-cusum}.   
Towards the goal, we will first introduce new quantities. 

\addtocounter{definition}{1}  
\begin{subdefinition}\label{def-cusum-S}
For any $k \in [2n-1]$, let 
\begin{equation}\label{eq-mathcal-S-def_u}
\mathcal{S}_k = \argmax_{\substack{B \subseteq    \widetilde{\mathcal{S}}_k:~\forall  (i, j), (i', j')  \in B, \\i \neq j, \ i' \neq j', \ i \neq j', \ j \neq i'}} |B|,
\end{equation}
with 
\[
     \widetilde{\mathcal{S}}_k = \begin{cases}
        \big\{(i, i + k - 1): \, i \in [n + 1 - k]\big\}, & k \in [n],\\
        \big\{(i + k - n, i): \, i \in [2n - k]\big\}, & k \in [2n -1] \setminus [n].
    \end{cases}
\]
\end{subdefinition}

The collection of subsets $\{\mathcal{S}_k\}_{k=1}^{2n-1}$ involved in \Cref{def-cusum-S} consists of  $2n - 1$  disjoints subsets. This is for technical convenience to bring independence among samples.

\addtocounter{definition}{-1}
\begin{definition}\label{def-cusum}
For any $\alpha \in (0, 1)$, let $\widehat{D}_{s, t} = \max_{z \in \mathcal{Z}_{\alpha, t}}|\widehat{D}_{s, t}(z)|$, where
\begin{enumerate}[(a)]
    \item[$(a)$] $\mathcal{Z}_{\alpha, t} = \{z_v\}_{v = 1}^{M_{\alpha, t}} \sim \mathrm{Uniform}([0, 1]^L)$  is a collection of independent random vectors  with 
        \[
            M_{\alpha, t}  =   C_M L^{-L}(tn)^{L/2}   [ \log\{ (n \vee t )/\alpha\} ]^{-L/2+1},
        \]
    for $C_M > 0$ an absolute constant;
    \item[$(b)$]  for $z \in [0, 1]^L$, 
        \[
            \widehat{D}_{s, t}(z) =  \bigg( \sum_{k \in [2n -1]} \widetilde{S}_k  \bigg)^{-1} h^{-L} \sum_{k \in [2n -1]} \widetilde{S}_k \left\{\mathcal{K} \left(\frac{z - \widehat{\mathbf{P}}^{0, s}_{\mathcal{S}_k, :}}{h} \right) - \mathcal{K} \left(\frac{z - \widehat{\mathbf{P}}^{s, t}_{\mathcal{S}_k, :}}{h} \right) \right\},
        \]
        with $\{\mathcal{S}_k, \widetilde{S}_k = |\mathcal{S}_k|\}_{k \in [2n -1]}$ defined in \Cref{def-cusum-S}; and
    \item[$(c)$]  for any integer pair $(s, t)$, $ 0 \leq s < t$, with $\mathrm{HOSVD}$ detailed in \Cref{alg-hosvd},
        \[
            \widehat{\mathbf{P}}^{s, t} = \mathrm{HOSVD} \bigg((t-s)^{-1}\sum_{u=s+1}^t \mathbf{A}(u)\bigg) \quad \mbox{and} \quad \widehat{\mathbf{P}}^{s, t}_{\mathcal{S}_k, :} = \widetilde{S}_k^{-1} \sum_{(i, j) \in \mathcal{S}_k } \widehat{\mathbf{P}}^{s, t}_{i, j, :}.
        \]
\end{enumerate}
\end{definition}

\begin{remark}
In \Cref{def-cusum}, we assume that the random vectors in $\mathcal{Z}_{\alpha, t}$, $t \in \mathbb{N}^*$, are from a uniform distribution.  In fact, any probability distribution on $[0,1]^L $ admitting a Lebesgue density bounded away from zero would work. 
\end{remark}

The scan statistics defined in  \Cref{def-cusum} is based on the HOSVD algorithm (\Cref{alg-hosvd}), instead of \Cref{thpca}, as the tensor estimation subroutine. 
This choice was made because \Cref{thpca} is computationally more expensive and incurs an additional error in estimating the Tucker ranks in an online fashion.  To overcome these issues,  we instead turn to \Cref{alg-hosvd}, which is essentially the HOSVD algorithm studied in \cite{de2000multilinear} but with two additional, technical tweaks: we require all unit ranks as an input and add a final entrywise truncation. 

\begin{algorithm}[!t]
\caption{Higher-order singular value decomposition. HOSVD$(\mathbf{A})$} 
    \begin{algorithmic}
        \INPUT{tensor $\mathbf{A} \in \R^{p_1 \times p_2 \times p_3}$.}
        \For{each mode $v \in [3]$ }
        \State{Compute the mode-$v$ matricization $\mathcal{M}_v( \mathbf{A})$ of tensor $\mathbf{A}$. \Comment{See \eqref{eq-matricisation} for $\mathcal{M}_v( \mathbf{A})$}
         \State {Let $\widehat{U}_{v}$ denote the eigenvector of $\mathcal{M}_v( \mathbf{A}) \mathcal{M}_v ( \mathbf{A} )^{\top}$ corresponding to the largest eigenvalue.}
         }
        \EndFor
        \State{Construct the tensor approximation $\widetilde{\mathbf{P}}$ by applying marginal multiplications: \Comment{See \eqref{eq-times-def} for $\times_v$} 
        \[
        \widetilde{\mathbf{P}} \leftarrow   \mathbf{A} \times_1 \widehat{U}_{1} \widehat{U}_{1}^{\top}  \times_2  \widehat{U}_2 \widehat{U}_2^{\top}  \times_3 \widehat{U}_3 \widehat{U}_3^{\top}
        \]}
        \For{each entry $\{i, j, l\} \in [p_1] \times [p_2] \times [p_3]$}
        \State{
        Truncate $\widetilde{\mathbf{P}}$ to obtain $\widehat{\mathbf{P}}$:
        \[
        \widehat{\mathbf{P}}_{i, j, l} \leftarrow  
        \begin{cases}
         1,  & 
         \widetilde{\mathbf{P}}_{i, j, l} > 1,  \\
      \widetilde{\mathbf{P}}_{i, j, l}, & \widetilde{\mathbf{P}}_{i, j, l} \in [0, 1], \\
       0, &\widetilde{\mathbf{P}}_{i, j, l} < 0. 
          \end{cases}
        \]}
        \EndFor
        \OUTPUT{$\widehat{\mathbf{P}} \in \R^{p_1\times p_2 \times p_3}$.}
    \end{algorithmic}\label{alg-hosvd}
\end{algorithm}

Below, we further elaborate on \Cref{def-cusum}.  
\begin{enumerate}[(a)]
    \item[$(a)$] In the existing dynamic network literature, it has long been a known challenge that due to the hardness of analysing the eigenspace of low-rank matrices' linear combination \cite[e.g.][]{padilla2019change, athreya2022discovering}, one has to carry out low-rank approximation for every adjacency matrix separately.  This results in an extra time factor in the rate.  We, however, are able to utilize averages of adjacency tensors and use them as inputs in the tensor estimation.  This adjustment stems from the online nature of the problem as well as the use of kernels, see \Cref{dynamic-theory}.
    \item[$(b)$]  Allowing for random latent positions increases the model flexibility but also the difficulty of the problem.  This randomness, together with the fact that we do not enforce node-correspondence across time, implies that  the \emph{de facto} high-probability Tucker rank of the average adjacency tensor is $(1, 1, 1)$; see \Cref{dynamic-theory}.  This enables us to only use the leading eigenvectors and to avoid estimating unknown Tucker ranks in \Cref{alg-hosvd}.
    \item[$(c)$]  Another difference between Algorithms~\ref{thpca} and \ref{alg-hosvd} is that no iterative estimation is required in \Cref{alg-hosvd}.  As we discussed in \Cref{sec-single-mrdpg}, the iterations used in \Cref{hpca} are designed to cope with heterogeneity.  Because of the averaging effect in \Cref{def-cusum}, the level of heterogeneity is instead dampened and does not require any algorithmic modifications.
\end{enumerate}

\subsection{Theoretical analysis of the online change point detection procedure}\label{dynamic-theory}

To analyse the theoretical performance of our methodology, we impose two additional assumptions, one on the kernel function and the other on the signal-to-noise ratio.

\begin{assumption}[The kernel function] \label{kernel_function_ass}
The kernel function $\mathcal{K}: \mathbb{R}^L \to \mathbb{R}$, used in \Cref{ass_change_point} and  \Cref{def-cusum}, is assumed to be Lipschitz, i.e.~there exists an absolute constant $C_{\mathrm{Lip}} > 0$, such that for any $\mathbf{x}, \mathbf{y} \in \mathbb{R}^L$,
   $ 
        |\mathcal{K}(\mathbf{x}) - \mathcal{K}(\mathbf{y})| \leq C_{\mathrm{Lip}} \|\mathbf{x} - \mathbf{y}\|.
    $
\end{assumption}

The above uniform Lipschitz condition is a standard assumption in the nonparametric literature \citep{linton1996estimation,raskutti2014early,amini2021concentration} and holds for commonly-used kernels, e.g.~triangular and Gaussian kernels.

\begin{assumption}[Signal-to-noise ratio condition] \label{snr_ass_change_point}
Under \Cref{ass_X_Y}, for any $\alpha \in (0, 1)$, assume that there exists a large enough absolute constant $C_{\mathrm{SNR}} > 0$ such that
\[
    \kappa\sqrt{\Delta}  > C_{\mathrm{SNR}} h^{-L-1} \sqrt{\frac{(L^2 \vee d) \log \{(n \vee \Delta)/\alpha \}}{  n}}.
\]   
\end{assumption}

Throughout \Cref{sec-dynamic}, we assume $L$ to be fixed.   
Recall that $\kappa$ defined in \eqref{eq-kappa-def} is a function of the kernel function $\mathcal{K}$ and the bandwidth $h$.  One can recast \Cref{snr_ass_change_point} as requiring the quantity
\begin{equation}\label{def-snr-normalised}
    \kappa h^{L+1} \sqrt{\Delta}
\end{equation}
to be away from zero.  In this way, we regard $\kappa h^{L+1}$ as a normalized jump size, capturing all the distributional change after being filtered by the kernel function.

We compare \Cref{snr_ass_change_point} with its counterparts in \cite{padilla2019change} and \cite{padilla2021optimal}, where nonparametric distributional changes are also studied.  For simplicity, we consider the independence case $(\rho = 0)$ in \cite{padilla2019change}.  We note that, \Cref{snr_ass_change_point} is sharper than its counterpart in \cite{padilla2019change} by a factor of $\sqrt{\Delta}$. To be more explicit, \cite{padilla2019change} deals with an offline problem and the $\sqrt{T}$, with $T$ being the total number of time points, can be regarded to play the same role as $\sqrt{\Delta}$ in the online setting. This is a direct benefit of low-rank estimation of averages of adjacency tensors, in contrast to using the averages of low-rank estimators of adjacency matrices in \cite{padilla2019change}.  

\cite{padilla2021optimal} establishes the fundamental limits of localizing nonparametric distributional change, in a multivariate setting where the density exists and in an offline manner.    Adapting their notation to ours, the signal-to-noise ratio is of the form $\kappa^{(L+2)/2}\sqrt{\Delta}$ and the optimal bandwidth is $h \asymp \kappa$.  Using the same arguments, our quantity in \eqref{def-snr-normalised} is $\kappa^{L+2}\sqrt{\Delta}$.  In the more challenging regime where $\kappa < 1$, we essentially require a stronger condition than \cite{padilla2021optimal}.  This is partly due to the more challenging nature of the framework we assume - allowing for densities not to exist -  and partly due to the fact we assume Lipschitz instead of VC-class conditions \citep{gine1999laws, kim2019uniform} for the kernel function.  We conjecture a sharper rate may be obtained if we replace \Cref{kernel_function_ass} with a VC-class condition.

\begin{theorem}\label{main_theorem}
Let $\widehat{\Delta}$ be the output of \Cref{online_hpca} applied to an  adjacency tensor sequence $\{\mathbf{A}(t)\}_{t \in \mathbb{N}^*}$ following a dynamic MRDPG model as defined in \Cref{def-umrdpg-dynamic} and satisfying \Cref{ass_X_Y}. Let $\alpha \in (0,1)$ and consider the threshold values
\begin{equation}\label{eq-tau-def-thm-2}
        \tau_{s, t} = C_{\tau}h^{-L-1} \sqrt{\frac{ (L^2 \vee d) \log\{( n  \vee t) / \alpha\}}{n}}\left( \frac{1}{\sqrt{s}}+ \frac{1}{\sqrt{t- s}}\right), \quad 1 < s < t,    
    \end{equation}
    where $C_{\tau} > 0$ is an absolute constant. Suppose that the scan statistic sequence used in \Cref{online_hpca} is specified in \Cref{def-cusum} with the kernel function $\mathcal{K}(\cdot)$ satisfying \Cref{kernel_function_ass}.
    Then,
\begin{enumerate}[(i)]
    \item[$(i)$] under \Cref{ass_no_change_point}, it holds that $\mathbb{P}_{\infty} \{\widehat{\Delta} < \infty\} \leq \alpha$;
    \item[$(ii)$] under Assumptions~\ref{ass_change_point} and \ref{snr_ass_change_point}, it holds that, for an absolute constant  $ C_\epsilon > 0$,
    \[
        \mathbb{P}_{\Delta} \left\{\Delta < \widehat{\Delta} \leq \Delta + C_\epsilon  \frac{(L^2 \vee d) \log \{(n \vee \Delta) /\alpha \}}{\kappa^2 h^{2L+2} n} \right\} \geq 1 - \alpha.
    \]
\end{enumerate}
\end{theorem}

The proof of \Cref{main_theorem} is given in \Cref{proof-sec4}.
 
With a pre-specified false alarm tolerance $\alpha \in (0, 1)$, \Cref{main_theorem} shows that \Cref{online_hpca} enjoys a uniform control on the false alarms with probability at least $1 - \alpha$.  When there indeed exists a change point, with probability at least $1 - \alpha$, the detection delay is at most of the order
    \begin{equation}\label{eq-detection-delay-rate}
        \frac{(L^2 \vee d) \log \{(n \vee \Delta) /\alpha \}}{\kappa^2   h^{2L+2} n}.
    \end{equation}
Before we conclude this section, a few remarks of \Cref{main_theorem} are in order.

\textbf{The detection delay.}   The detection delay rate presented in \eqref{eq-detection-delay-rate} conforms to the standard form of detection delay in the online change point literature (or the localization error rate in the offline change point literature), as shown in \eqref{eq-detection-delay-intuitive_f}.
The values of threshold sequence $\{\tau_{s, t}\}$ are designed to provide time uniform control on the noise level.  Without any further constraint on the smallest interval size in the change point detection procedure, it can be seen in \eqref{eq-tau-def-thm-2} that 
\[
    \sup_{1 \leq s < t} \tau_{s, t}\lesssim h^{-L-1} \sqrt{\frac{ (L^2 \vee d) \log\{( n  \vee t) / \alpha\}}{n}}.    
\]
The detection delay established in \Cref{main_theorem}, therefore, follows the same pattern as that in the existing literature.  

Compared to the results for the dynamic network change point detection under low-rank assumptions \cite[e.g.][]{padilla2019change}, the detection delay in \Cref{main_theorem} offers an improvement by a factor of $\mbox{sample size}$. 
Specifically, the detection delay upper bound in \cite{padilla2019change}, when translated into our notation, is of order $\Delta \max\{d \log(n \vee \Delta), d^3\}/(\kappa^2 n)$ when considering the independent case $(\rho=0)$. The main ingredient behind the improvement is that we are able to conduct low-rank approximation to averages of tensors/matrices.   

It is also worth mentioning \cite{padilla2021optimal}, where a multivariate nonparametric offline change point localization problem is studied.  As we discussed, due to the more challenging scenario, especially that we do not assume the existence of densities, we have $h^{-(2L+2)}$ instead of $h^{-L}$ (the dependence on $h$ in the detection delay in \citealt{padilla2021optimal}) in the detection delay rate. 
We note that if the underlying distributions possess densities, one could directly replace the kernel density estimator and associated analysis with their counterparts in \cite{padilla2021optimal}, sacrificing some model flexibility to achieve a sharper rate.

\textbf{Rank-1 approximation of averages of adjacency tensors.}  To conclude this section, we elaborate on the difficulty of analysing averages of adjacency tensors and why we switched to rank-1 approximations in constructing the scan statistics \Cref{def-cusum}, instead of using the Tucker rank $(d, d, m)$ as in \Cref{thpca}.  This is a direct consequence of using the averages of adjacency tensors, each of which is associated with low-rank but independent random latent positions.  
In fact, at each time point $t \in \mathbb{N}^*$, with high probability, the probability tensor $\mathbf{P}(t)$ is a low-rank tensor; however that is not the case for the average $\widehat{\mathbf{P}}_{s, t} = (t-s)^{-1} \sum_{u=s+1}^t \mathbf{P}(u)$, for general integer pair $(s, t)$.  We instead approximate $\widehat{\mathbf{P}}_{s, t}$ by $\widetilde{\widehat{\mathbf{P}}}^{s, t} = S \times_1 \overline{X}^{s, t} \times_2 \overline{X}^{s, t} \times_3 Q$, with $\overline{X}^{s, t} = (t-s)^{-1}\sum_{u = s+1}^t X(u)$ and show that, with high probability, $\widetilde{\widehat{\mathbf{P}}}^{s, t}$ has Tucker ranks $(d, d, m)$.  Recall that $X(\cdot)$ is a random latent vector, with a non-zero mean.  As a result, the second-moment matrices
 \[
 \Sigma^{s, t}_X = \mathbb{E}[ \{(\overline{X}^{s, t})^1\}^{\top} (\overline{X}^{s, t})^{1} ]  \in \mathbb{R}^{d \times d}
 \]
possess dominating largest eigenvalues.  To be specific, as shown in \Cref{lemm_eigen} in \Cref{sec-C4},
\begin{equation}\label{eq-sigma-x-decomp}
    \mathrm{rank}(\Sigma^{s, t}_{X}) = d \quad \mbox{and} \quad  \Sigma^{s, t}_{X} = (t-s)^{-1} \Sigma_X + (t-s)^{-1}(t-s-1)   \theta_{X} \left( \theta_{X} \right)^{\top},
\end{equation}
where $\Sigma_X \in \R^{d \times d}$ and $\theta_X \in \R^{d}$ are defined in \Cref{ass_X_Y}.  It is easy to see that the leading eigenvalue of $\Sigma^{s, t}_{X}$ is driven by $\|\theta_{X}\|^2$ and with a large eigen-gap with the second largest one.  On the other hand, due to the average effect, the heterogeneity level in the tensor entries of  $\overline{X}^{s, t}$ is diminished.  This is also seen from the decomposition in \eqref{eq-sigma-x-decomp}. 

These observations enable the use of rank-1 approximation to avoid estimating Tucker ranks, and avoid the necessity of using iterative methods in \Cref{hpca}.  As a final observation, the experiments of \Cref{change_point_section}  indicate a comparable numerical performance when using \Cref{alg-hosvd} instead of \Cref{thpca}.

\section{Numerical experiments}\label{sec-numerical}

In this section, we summarize the results of extensive numerical experiments.  In \Cref{sec-simulation}, we use simulated data while in \Cref{sec-real-data} we analyze two real data sets. The code and data sets for all experiments are available online\footnote{\url{https://github.com/MountLee/MRDPG}}.

\subsection{Simulation studies}\label{sec-simulation}

Sections~\ref{simu_denoising}, \ref{change_point_section_f} and~\ref{change_point_section} corroborate the theoretical findings of Sections~\ref{sec-single-mrdpg}, \ref{sec-dynamic-f} and \ref{sec-dynamic}, respectively.

While in the main text, we only cover undirected MRDPGs with random latent positions and leave the analysis of the directed case that accommodates bipartite graphs to the appendix, in this simulation study we investigate both scenarios. Throughout, we use $n$ to denote the size of an undirected network and  $(n_1, n_2)$ for the size of a directed network.

\subsubsection{Estimation of a single probability tensor}\label{simu_denoising}

To assess the performance of TH-PCA (\Cref{thpca}) in estimating a single probability tensor, we compare it with several other methodologies: higher-order singular value decomposition \cite[HOSVD,][]{de2000multilinear}, higher-order orthogonal iteration \cite[HOOI,][]{zhang2018tensor}, Tucker decomposition with integrated SVD transformation \cite[TWIST,][]{jing2021community}, scaled adjacency spectral embedding estimators \cite[SASE,][]{sussman2012consistent} applied to each layer, unfolded adjacency spectral embedding \cite[UASE,][]{jones2020multilayer}, multiple adjacency spectral embedding \cite[MASE,][]{arroyo2021inference}, maximum likelihood estimators \cite[MLE,][]{zhang2020flexible} and convex optimization approach \cite[COA,][]{macdonald2022latent}. 
For the implementation of  HOSVD, HOOI, TWIST, MASE and COA, we use the R \citep{R} packages \texttt{rTensor} \citep{rTensor}, \texttt{STATSVD} \citep{STATSVD}, \texttt{rMultiNet} \citep{rMultiNet}, \texttt{MASE} \citep{MASE}   and  \texttt{multiness} \citep{multiness}, respectively. As there is no publically available code in \cite{zhang2020flexible} for MLE, we implement the projected gradient descent algorithm.

Note that TH-PCA, HOOI and HOSVD require Tucker ranks of probability tensors as inputs.  As discussed in \Cref{sec_estimation_single}, for models defined in Definitions~\ref{umrdpg-f},  \ref{umrdpg} or \ref{mrdpg}, the probability tensor has Tucker ranks $(r_1, r_2, r_3) = (d, d, m)$ with high probability, where $d$ is the dimension of latent positions and $m$ is the rank of the matrix $Q \in \R^{L \times d^2}$ defined in \eqref{matrix Q} with the number of layers~$L$.  To ensure robustness, we choose relatively large Tucker ranks as inputs, setting $(\widehat{r}_1, \widehat{r}_2, \widehat{r}_3) = (10, 10, 10)$.
TWIST is based on slightly different modelling and we set the input as $ (r_1, r_2, r_3) = (\hat{d}, \hat{d}, L)$, again with $\hat{d} = 10$.
For SASE, UASE and MASE, we set the input dimension of the latent positions as $\hat{d} =10$.  COA requires the estimated dimension of latent positions to obtain the initial estimators, so we let $\hat{d} =10$ and choose its tuning parameters in an adaptive way as instructed in \cite{macdonald2022latent}. MLE also requires the dimension of the latent space as an input, which we set to $\hat{d} = 10$ in all experiments.  To visualize the effect of dimensions, we have also conducted a small-scale simulation where the true dimensions are used. For the step size and the number of iterations in the projected gradient descent algorithm solving for MLE, we set the value that has the best overall empirical performance among our attempts, as there is no official recommendation in \cite{zhang2020flexible}.

We consider two simulation settings below.  In each scenario, each combination of parameters is repeated over $100$ Monte Carlo trials.

\textbf{Scenario 1. Multilayer stochastic block models.} 
To construct multilayer stochastic block models, let the probability tensor $\mathbf{P} \in [0,1]^{n \times n \times L}$ be with
\begin{align}\label{simulation_tensor_P}
    \mathbf{P}_{i, j, l} = \begin{cases}
        p_{1, l}, & i, j \in \mathcal{B}_b, \, b \in [4], \\
        p_{2, l}, & \mbox{otherwise},
   \end{cases}
\end{align}
where $ \mathcal{B}_b, b \in [4]$, are evenly-sized communities of nodes that form a partition of $[n]$, and $p_{1, l}$ and $p_{2, l}$ satisfy that for $l \in [L]$,
\[
    p_{1, l} \stackrel{\mbox{ind.}}{\sim} \mbox{Uniform} \left(\frac{3L+l-1}{4L}, \frac{3L+l}{4L} \right) \quad \mbox{and} \quad p_{2, l} \stackrel{\mbox{ind.}}{\sim} \mbox{Uniform} \left(\frac{2L+l-1}{4L}, \frac{2L+l}{4L} \right).
\]
The data are then generated such that for $i,j \in [n]$, $i \leq j$ and $l \in [L]$, $\mathbf{A}_{i, j, l} \stackrel{\mbox{ind.}}{\sim} \mbox{Bernoulli}\left(\mathbf{P}_{i, j, l}\right)$, and for all $l \in L$, $\mathbf{A}_{:, :, l}$ are symmetric.  We consider the node set size $n \in \{ 50, 100, 150, 200 \}$ and the number of layers $L \in \{ 10, 20, 30, 40 \}$.

\textbf{Scenario 2. Dirichlet distribution models.}  To generate the weight matrices $W_{(l)} \in \R^{d \times d}$, for $l \in [L]$, let
\begin{align}\label{simulation_matrix_W}
    \left( W_{(l)} \right)_{i, j} \stackrel{\mbox{ind.}}{\sim} \mbox{Uniform} \left(\frac{L + l - 1}{2L}, \frac{L + l}{2L}\right), \quad i, j \in[d].
\end{align}
Fix $\big\{ W_{(l)}\big\}_{l \in [L]}$ for all Monte Carlo trials within the same set of parameter combinations.  Let $\mathbf{1}_{d} \in \mathbb{R}^d$ be an all-one vector.  Let $\{X_i\}_{i \in [n_1]}$ be independent and identically distributed as Dirichlet($\mathbf{1}_d$) and $\{Y_j\}_{j \in [n_2]}$ be independent and identically distributed as Dirichlet($10 \times \mathbf{1}_d$).  Let the probability tensor $\mathbf{P} \in [0, 1]^{n_1 \times n_2 \times L}$ be with entries $\mathbf{P}_{i, j, l} = X_i^{\top} W_{(l)} Y_j$, $i \in [n]$, $j \in [n_2]$ and $l \in [L]$. The data are generated as $\mathbf{A}_{i, j, l} \stackrel{\mbox{ind.}}{\sim}\mbox{Bernoulli} \left( \mathbf{P}_{i, j, l}\right)$, $i \in [n_1]$, $j \in [n_2]$ and $l \in [L]$.  We consider  the dimension of the latent positions $d \in \{ 2, 4, 6, 8 \}$, the number of nodes $n_1, n_2  \in \{ 50, 100, 150, 200\}$ and the number of layers  $L \in \{ 10, 20, 30, 40\}$.  Note that this is a directed network example.

\medskip

The simulation results for \textbf{Scenarios 1} and \textbf{2} are depicted in Figures~\ref{Fig_simulation_denoising} and \ref{Fig_simulation_denoising_2}, respectively.  In both figures, we report on the $y$-axis unnormalized estimation errors $\| \widehat{\mathbf{P}} - \mathbf{P}\|_{\mathrm{F}}$, where $\widehat{\mathbf{P}}$ is an estimated probability tensor, presented in the form of mean and standard deviation over $100$ Monte Carlo trials.  \textbf{Scenario 1} produces undirected adjacency tensors having the same block structures but different weight matrices across layers.  More variation is included in \textbf{Scenario 2}. Note that the available code for MASE only handles undirected multilayer networks, making it not applicable to  \textbf{Scenario 2}.

In both scenarios, TH-PCA has the best overall performance, with a few odd exceptions when the numbers of layers and/or nodes are relatively small.  COA is the only one which has better performances in these exceptions. Across both scenarios, the tensor-based methods (TH-PCA, HOOI, HOSVD and TWIST) outperform the matrix-based methods (SASE, UASE), except for MASE in \textbf{Scenario 1}.  Due to the more variation contained in \textbf{Scenario 2}, we see a clearer lead of TH-PCA over the other two tensor-based methods (HOOI and HOSVD).  This aligns with the intuition that TH-PCA is proposed to handle the heterogeneity within entries. 

In Panels (A) and (B) in \Cref{Fig_simulation_denoising_2}, we show the performances of different methods with varying numbers of layers.  Panel (A) is according to the setting described in \textbf{Scenario 2} and  Panel (B) has a small tweak.  Instead of generating layer-specific weight matrices, all entries of all weight matrices are independently drawn from Uniform$(0, 1)$.  It follows from \Cref{random_theorem} and our discussions in \Cref{sec-comp-single} that, one should expect the performances of tensor-based methods to deteriorate as the number of layers increases.  This is however not the case in Panel (A), due to the dependence on the number of layers $L$ in the distribution of weight matrices.  As $L$ increases, the variation of weight matrices entries decreases, which consequently decreases the hardness of the tensor estimation problem.  To single out the effect of $L$, we conduct the analysis summarized in  Panel (B), where we can see the performance of tensor-based methods indeed deteriorates as $L$ grows.  It is also interesting to point out that, as we turn down the heterogeneity in the right panel, by letting all entries of $W$'s have the same distributions, the performances of TH-PCA and HOSVD are almost identical.  This is also our motivation for going for HOSVD for cheaper computational cost when the heterogeneity is less of a concern in \Cref{sec-dynamic}.  Lastly, the performances of two matrix-based methods (SASE and UASE), as well as COA and MLE, do not change much with the twist on weight matrices' distributions.

Panels (C) and (D) in \Cref{Fig_simulation_denoising_2} show the performance of different algorithms with varying latent position dimensions.  As we discussed in \Cref{sec_estimation_single}, if we over-estimate the latent position dimension, then the estimation error will be inflated.  This is apparent from Panel (C), where although true dimensions vary, the inputs stay fixed and over-estimated.  To see the true effect of dimensions, we include Panel (D), where we use the true dimension as the inputs.  The deterioration kicks in for all methods (except COA), but TH-PCA shows more resilience compared to other matrix- or tensor-based methods.  
Another noticeable fact in Panels (C) and (D) is that the estimation errors of COA do not show an increasing pattern as other methods. However, according to \cite{macdonald2022latent}, the upper bound of the estimation error of COA involves a factor $d$. We conjecture that the distribution of tensor entries has a significant impact on the performance of COA, and under some distributions, the upper bound can be loose for a factor up to $d$, as is illustrated by Figure 8 of \cite{macdonald2022latent}.

Lastly, we see the impact of the node sizes $n_1$ and $n_2$ for all methods in Panels (E) and (F) of \Cref{Fig_simulation_denoising_2}. In Panel (E) we set $n_1 = n_2$, while in Panel (F) we fix $n_1$ ad let $n_2$ vary. Because the available code of TWSIT and COA, and the code we developed for MLE can only handle the case when $n_1 = n_2$, TWSIT, COA and MLE are not applicable for the experiments in Panel (F).

\begin{figure}[ht] 
\centering 
\includegraphics[width=0.9 \textwidth]{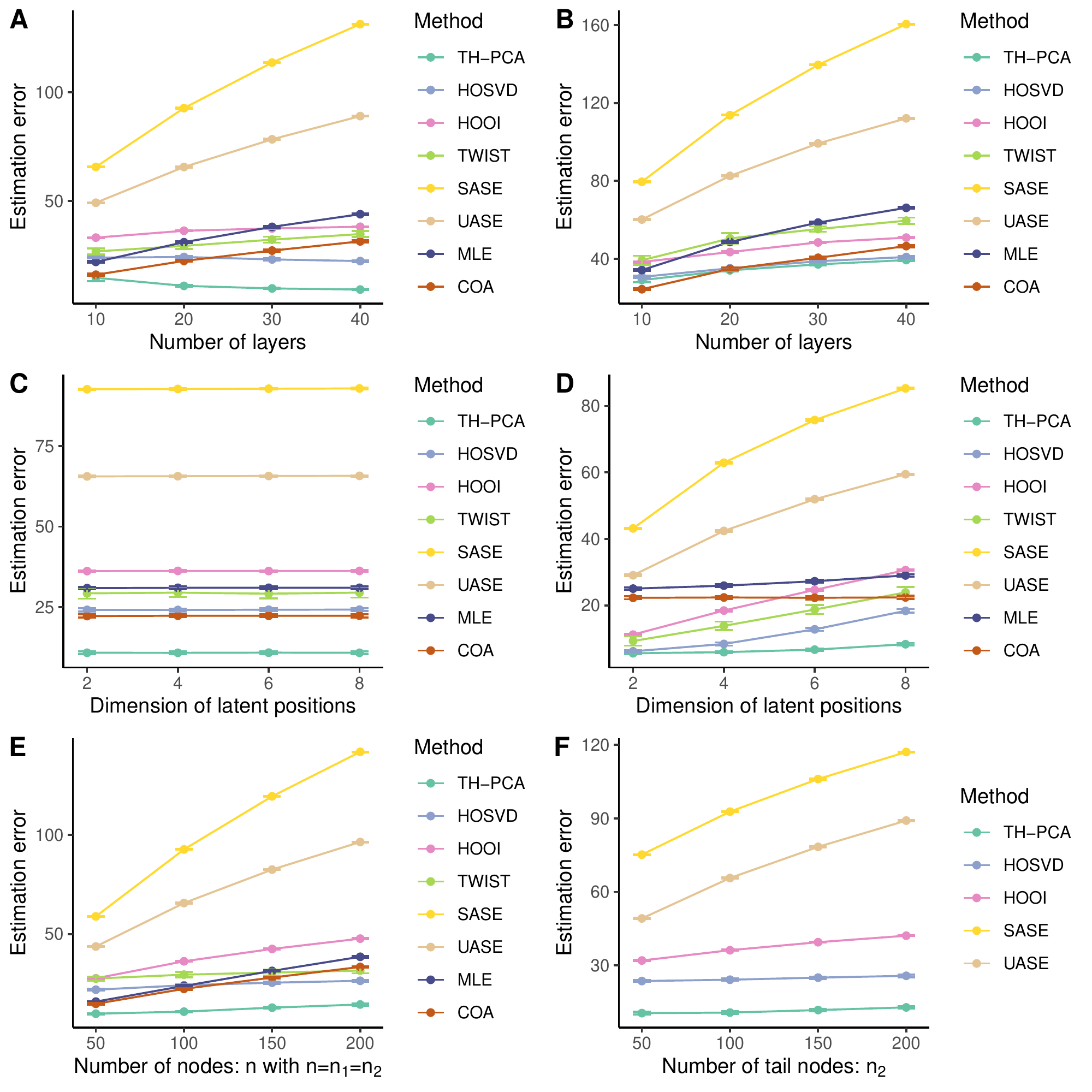}
\caption{Results of estimating a probability tensor in Scenario 2 in \Cref{simu_denoising} and beyond. (A) and (B): $n_1 =n_2 = 100$, $d=4$ and $L \in \{ 10, 20, 30, 40\}$.  (A): Scenario 2.  (B): with a twist of the weight matrices. (C) and (D): $n_1 = n_2 = 100$, $L=20$ and $d \in \{2, 4, 6, 8\}$.  (C): Scenario 2. (D): with the true dimension as the inputs. (E) and (F): Scenario 2 with $d = 4$ and $L=20$. (E): $n_1 = n_2 \in \{50, 100, 150, 200\}$.  (F): $n_1=100$ and $n_2 \in \{50, 100, 150, 200\}$.  In all panels, $y$-axis values correspond to the values of estimation error $\| \widehat{\mathbf{P}} - \mathbf{P}\|_{\mathrm{F}}$ of different estimators $\widehat{\mathbf{P}}$.  The values are in the form of mean and standard deviation over $100$ Monte Carlo trials.}\label{Fig_simulation_denoising_2}
\end{figure}

\subsubsection{Online change point detection with fixed latent positions} \label{change_point_section_f}

We are not aware of any procedure for online change point analysis for multilayer networks. As a result, for this simulation study, we compare the performance of our algorithms with a series of \emph{ad hoc} procedures obtained by modifying existing methods not designed for this problem.  We change the tensor estimation 
routine in \Cref{def-cusum-f} from TH-PCA (\Cref{thpca}) to HOSVD \citep{de2000multilinear}, TWIST \citep{jing2021community}, UASE \citep{jones2020multilayer} and COA \citep{macdonald2022latent}. Furthermore, we also adopt the $k$-nearest neighbours ($k$-NN) method as proposed by \cite{chen2019sequential}. The four alternative methods developed by us are referred to as on-HOSVD, on-TWIST, on-UASE and on-COA, respectively. TWIST, COA and $k$-NN are implemented using the R \citep{R} packages \texttt{rMultiNet} \citep{rMultiNet}, \texttt{multiness} \citep{multiness} and \texttt{gStream} \citep{gStream}, respectively.

To evaluate the performance of different change point detection algorithms we consider the following metrics.  Let $N$ be the number of Monte Carlo trials.  For each trail $u \in [N]$, let $\hat{t}_u$ be the estimated change point location. If there is no change point detected in $[T]$, we let $\hat{t}_u = T$.
Define
\[
    \mbox{Delay} = \frac{\sum_{u=1}^N \mathbbm{1}_{\{ \hat{t}_u \geq \Delta \}}  \left(\hat{t}_u - \Delta \right)}{\sum_{u=1}^N \mathbbm{1}_{\{ \hat{t}_u \geq \Delta \}}}, \,\, \mbox{PFA} = \frac{\sum_{u=1}^N \mathbbm{1}_{\{ \hat{t}_u <\Delta \}}}{N} \,\, \mbox{and} \,\, \mbox{PFN} = \frac{\sum_{u=1}^N \mathbbm{1}_{\{ \hat{t}_u = T\}}}{N},
\]
where PFA and PFN stand for the proportion of false alarms and the proportion of false negatives, respectively.   Throughout this section, we let $\alpha = 0.01$.  Since all settings we consider have a change point, PFA and PFN correspond to Type-I and Type-II errors respectively.

For the tuning parameters, we assume access to an additional sample of size $\widetilde{T}$ from the pre-change distribution, referred to as the training set.   We train the thresholds $\{\tau_{s, t}\}$ defined in~\eqref{eq-tau-def-thm_f} by permuting the training set $B = 100$ times.  To be specific, for the unspecified absolute constant~$C_{\tau}$, we choose the minimal value from a range, such that the proportion of false alarms in these $B$ permutations is upper bounded by pre-specified $\alpha$.  The same threshold selection method is used for on-HOSVD, on-TWIST, on-UASE and on-COA.
To show the sensitivity of the choice of the tuning parameter, we consider a set of values for  $C_{\tau} \in \{0.16, 0.17, 0.18, 0.19, 0.20\}$. The simulation results can be found in \Cref{simu-sec5.2}.
Based on the discussion in  \Cref{simu_denoising}, we let the rank input be $(\hat{d}, \hat{d}, \hat{m}) = (10, 10, L)$ for both \Cref{thpca} and HOSVD, $(\hat{d}, \hat{d}, L) = (10, 10, L)$ for TWIST, and $10$ for both UASE and COA.  For COA, we adopt the same tuning parameter selection strategy as in \cite{macdonald2022latent}.
For the $k$-NN method, we use the max-type edge-count scan statistic \cite[see][]{chen2019sequential}.

We consider a simulation setting where each combination of parameters is repeated over $N = 100$ Monte Carlo trials.  The time horizon is set to be $T = 100$, the change point occurs at $\Delta = 50$ and the size of the training set is $\widetilde{T} = 75$.

\textbf{Multilayer stochastic block models.}
Let $\mathbf{P} \in [0,1]^{n \times n \times L}$ denote the probability tensor before the change point and be defined as \eqref{simulation_tensor_P}. Let $\mathbf{Q} \in [0,1]^{n \times n \times L}$ denote the probability tensor after the change point, which will be specified in each scenario.

The data are then generated such that for all $t \in [T]$ and $l \in [L]$, $\mathbf{A}_{;,;,l}$ are symmetric and  for $t \in [T]$, $i, j \in [n]$, $i\leq j$ and $l \in [L]$, 
\[
    \mathbf{A}_{i, j, l}(t) \stackrel{\mbox{ind.}}{\sim} \begin{cases}
        \mbox{Bernoulli}(\mathbf{P}_{i, j, l}), & t \in [\Delta], \\ 
        \mbox{Bernoulli}(\mathbf{Q}_{i, j, l}), & t \in [T] \setminus [\Delta].
    \end{cases}
\]
We consider the number of nodes $n \in \{25, 50, 75 \}$ and number of layers $L \in \{ 2, 3, 4 \}$.
\begin{itemize}
    \item \textbf{Scenario 1: Layers flipped.}  Let $\mathbf{Q}$ be with entries $\mathbf{Q}_{i, j, l} = \mathbf{P}_{i, j, L-l+1}, i, j \in [n], \, l \in [L]$.
    \item \textbf{Scenario 2: Nodes permuted.}  Apply the same random permutation each time to the rows and columns of each layer of $\mathbf{P}$ to obtain $\mathbf{Q}$.
    \item \textbf{Scenario 3: Dimension changed.}  Let $\mathbf{Q}$ be defined as \eqref{simulation_tensor_P} but with 8 evenly-sized communities instead of 4.
\end{itemize}

The change point specified in \textbf{Scenario 1} follows \Cref{ass_change_point_u_f}, but not for the ones in \textbf{Scenarios 2} and \textbf{3}.  Such designs aim to investigate the robustness of our proposed algorithms.

The simulation results can be found in \Cref{simu-sec5.2}. 
We present a few examples in \Cref{table_1} to convey key messages.  From \Cref{table_1}, we conclude that in this scenario, Alg.~\ref{online_hpca}-Def.~\ref{def-cusum-f}, \Cref{online_hpca} with the scan statistics defined in \Cref{def-cusum-f} and the tensor estimation routine specified in \Cref{thpca},  as well as on-HOSVD, the competitor we created by replacing the tensor estimation routine with HOSVD, demonstrate nearly equal effectiveness. These two algorithms exhibit robustness in \textbf{Scenarios 2} and \textbf{3}, where \Cref{ass_change_point_u_f} does not hold. We highlight that, although on-TWIST has the smallest detection delay, it struggles to control the proportion of false alarms effectively. The runners-up in terms of performance are then on-UASE and $k$-NN.

\setlength\tabcolsep{3pt}
\begin{table}[t]
\caption{Simulation results for \Cref{change_point_section_f}, with the number of nodes $n = 50$ and the number of layers $L = 4$. PFA, the proportion of false alarms; PFN, the proportion of false negatives; Delay, the average of the detection delays; 
Alg.~\ref{online_hpca}-Def.~\ref{def-cusum-f}, on-HOSVD, on-TWIST, on-UASE and  on-COA, \Cref{online_hpca} with the scan statistics defined in \Cref{def-cusum-f} and the tensor estimation routine using \Cref{thpca}, HOSVD \citep{de2000multilinear},  TWIST \citep{jing2021community}, UASE \citep{jones2020multilayer} and  COA \citep{macdonald2022latent} respectively; $k$-NN, \cite{chen2019sequential}.}
\begin{center}
\begin{tabular}{cccccccccc} 
\hline
  Scenario &  Metric & Alg.~\ref{online_hpca}-Def.~\ref{def-cusum-f} & on-HOSVD & on-TWIST & on-UASE & on-COA & $k$-NN \\ \hline
 1  & PFA    & 0.00 & 0.00 & 0.17 &  0.01 &  0.00 & 0.02 \\
    & PFN   & 0.00 & 0.00 & 0.00 &  0.00 &  0.98 & 0.00\\
    &  Delay & 1.00 & 1.00 & 0.00 &  1.00 &  48.27 & 2.36\\ [3pt]
 2 & PFA   & 0.04 & 0.01  & 0.23  & 0.02 & 0.06  &   0.99 \\
   & PFN   & 0.00 & 0.00 & 0.00 & 0.00 &  0.78 &  0.00 \\
   &  Delay & 0.93 & 0.99 & 0.00  & 1.00 &  43.80  &  2.00 \\ [3pt]
 3 & PFA   & 0.00  & 0.00  & 0.26  & 0.01 &  0.00  & 0.02   \\
   & PFN   &0.00  & 0.00 & 0.00  & 0.00 &  0.00  & 0.00  \\
   &  Delay &  1.28  &  1.27 & 0.04  & 3.52  & 0.29   & 25.65   \\ \hline
\end{tabular}\label{table_1}
\end{center}
\end{table}

\subsubsection{Online change point detection with random latent positions} \label{change_point_section}

To assess the numerical performance of the online change point detection algorithm stated in \Cref{online_hpca} using the scan statistics described in \Cref{def-cusum} and tensor estimation routine specified in \Cref{alg-hosvd}, and applied to the case of random latent positions, we sought to compare it against alternative methods. We change the tensor estimation routine in \Cref{def-cusum} from  \Cref{alg-hosvd} to  TH-PCA (\Cref{thpca}), TWIST \citep{jing2021community},
UASE \citep{{jones2020multilayer}} and COA  \citep{macdonald2022latent}. Furthermore,
we also take into account the $k$-nearest neighbours ($k$-NN) method as proposed by \cite{chen2019sequential}. For the
three competitors created by us, we name them on-TH-PCA, on-TWIST, on-UASE and on-COA, respectively.

For the tuning parameters, for each algorithm, we assume access to an additional sample from the pre-change distribution of size $\widetilde{T}$, which we refer to as the training set.  For \Cref{online_hpca} with scan statistics defined in \Cref{def-cusum}, we adopt the Gaussian kernel and set the bandwidth $h = \{500 \log ( T n_1 n_2 )/(Tn_1 n_2)\}^{1/L}$,
a choice guided by \cite{padilla2019optimal}.  We permute the training set $B = 100$ times the use the same threshold selection method as described in \Cref{change_point_section_f} to train the thresholds $\{\tau_{s, t}\}$ defined in \eqref{eq-tau-def-thm-2}.  The choice of the kernels and the threshold chosen methods are inherited for the competitors created by us, i.e.~on-TH-PCA, on-TWIST, on-UASE and on-COA.  Following the observations from \Cref{simu_denoising}, we let the rank inputs in TH-PCA and TWIST be $(10, 10, L)$,  $\hat{d} = 10$ for UASE and COA.  For COA, we adopt the same tuning parameter selection strategy as that in \cite{macdonald2022latent}.  For $k$-NN method, we use the max-type edge-count scan statistic \cite[see][]{chen2019sequential}.

We consider three simulation scenarios below with each combination of parameters repeated over $N = 100$ Monte Carlo trials.  The time horizon is set to be $T = 100$, the change point occurs at $\Delta = 50$ and the size of the training set is $\widetilde{T} = 75$.

\textbf{Scenario 1. Multilayer stochastic block models with unlabelled nodes across time.}  
Let $\mathbf{P}, \mathbf{Q} \in [0,1]^{n \times n \times L}$ be the probability tensors defined in \textbf{Scenario~1} in \Cref{change_point_section_f}.  For $t_1 \in [\Delta]$, we apply the same permutation to the rows and columns of each layer in $\mathbf{P}$ and obtain~$\widetilde{\mathbf{P}}(t_1)$.  The permutations across time are mutually independent.  The same procedure is done on~$\mathbf{Q}$ to obtain $\{\widetilde{\mathbf{Q}}(t_2)\}_{t_2 \in [T] \setminus [\Delta]}$.

Let the entries of the adjacency tensors be generated such that for all $t \in [T]$ and $l \in [L]$, $\mathbf{A}_{:,:,l}(t)$ are symmetric and for any $t \in [T]$, $i, j \in [n]$, $i\leq j$ and $l \in [L]$,
\[
    \mathbf{A}_{i, j, l}(t) \stackrel{\mbox{ind.}}{\sim} \begin{cases}
        \mbox{Bernoulli}\big(\widetilde{\mathbf{P}}_{i, j, l}(t) \big), & t \in [\Delta], \\ 
        \mbox{Bernoulli}\big(\widetilde{\mathbf{Q}}_{i, j, l}(t)\big), & t \in [T] \setminus [\Delta].
    \end{cases}
\]
We consider the number of nodes $n \in \{25, 50, 75 \}$ and number of layers $L \in \{ 2, 3, 4 \}$. 

\textbf{Scenario 2. Undirected Dirichlet distribution models.}  For any $l \in [L]$, let $W_{(l)} \in \R^{d \times d}$ be defined in \eqref{simulation_matrix_W}. Fix $\{ W_{(l)}\}_{l \in [L]}$ for all time points and all Monte Carlo trials.  Let $\mathbf{1}_{d} \in \mathbb{R}^d$ be an all-one vector.  For $t \in [\Delta]$, let $\{X_i(t)\}_{i \in [n]}$ be independent and identically distributed as Dirichlet($\mathbf{1}_d$).  For $t \in [T]\setminus [\Delta]$, let $\{X_i(t)\}_{i \in [n]}$ be independent and identically distributed as Dirichlet($500 \times \mathbf{1}_d$).  For $t \in [T]$, let the probability tensor $\mathbf{P}(t) \in [0, 1]^{n \times n \times L}$ be with entries $(\mathbf{P}(t))_{i, j, l} = X_i(t)^{\top} W_{(l)} X_j(t)$, $i, j \in [n]$ and $l \in [L]$. The data are generated such that for all $t \in [T]$ and $l \in [L]$, $\mathbf{A}_{:,:,l}(t)$ are symmetric and for any $t \in [T]$, $i, j 
\in [n]$, $i\leq j$ and $l \in [L]$, $(\mathbf{A}(t))_{i, j, l} \stackrel{\mbox{ind.}}{\sim}\mbox{Bernoulli} \left( (\mathbf{P}(t))_{i, j, l}\right)$.  We consider the dimension of the latent positions $d \in \left\{ 2, 4, 6\right\}$, the number of nodes $n \in \{25, 50, 75 \}$ and the number of layers $L \in \{ 2, 3, 4 \}$.

\textbf{Scenario 3. Directed Dirichlet distribution models.} For any $l \in [L]$, let $W_{(l)} \in \R^{d \times d}$ be defined in \eqref{simulation_matrix_W}. Fix $\{ W_{(l)}\}_{l \in [L]}$ for all time points and all Monte Carlo trials.  For $t \in [\Delta]$, let $\{X_i(t)\}_{i \in [n_1]}$ be independent and identically distributed as Dirichlet($\mathbf{1}_d$) and $\{Y_j(t)\}_{j \in [n_2]}$ be independent and identically distributed as Dirichlet($10 \times \mathbf{1}_d$).  For $t \in [T]\setminus [\Delta]$, let $\{X_i(t)\}_{i \in [n_1]}$ be independent and identically distributed as Dirichlet($500 \times \mathbf{1}_d$) and $\{Y_j(t)\}_{j \in [n_2]}$ be independent and identically distributed as Dirichlet($1000 \times \mathbf{1}_d$).  For $t \in [T]$, let the probability tensor $\mathbf{P}(t) \in [0, 1]^{n_1 \times n_2 \times L}$ be with entries $(\mathbf{P}(t))_{i, j, l} = X_i(t)^{\top} W_{(l)} Y_j(t)$, $i \in [n_1]$, $j \in [n_2]$ and $l \in [L]$. The data are generated as $(\mathbf{A}(t))_{i, j, l} \stackrel{\mbox{ind.}}{\sim}\mbox{Bernoulli} \left( (\mathbf{P}(t))_{i, j, l}\right)$, $i \in [n_1]$, $j \in [n_2]$ and $l \in [L]$.  We consider  the dimension of the latent positions $d \in \left\{ 2, 4, 6\right\}$, the number of nodes $n_1 \in \{25, 50, 75 \}$, $n_2= 50$ and the number of layers $L \in \{ 2, 3, 4 \}$.

\setlength\tabcolsep{3pt}
\begin{table}[t]
\caption{Simulation results for Scenario 3 in \Cref{change_point_section}: $n_1$ and $n_2$, the number of nodes; $L$, the number of layers; $d$, the dimension of the latent position; PFA, the proportion of false alarm; PFN, the proportion of false negatives; Delay, the average of the detection delays; Alg.~\ref{online_hpca}-Def.~\ref{def-cusum}, on-TH-PCA, on-TWIST
on-UASE and on-COA, \Cref{online_hpca} with the scan statistics defined in \Cref{def-cusum} and the tensor estimation routine using \Cref{alg-hosvd},   \Cref{thpca}, TWIST \citep{jing2021community}, UASE \citep{jones2020multilayer}  and COA \citep{macdonald2022latent}, respectively; $k$-NN, \cite{chen2019sequential}.}
\begin{center}
\begin{tabular}{ccccccccccc}
\hline
$n_1$ & $n_2$ & $L$ & $d$ & Metric & Alg.~\ref{online_hpca}-Def.~\ref{def-cusum} & on-TH-PCA  & on-TWIST&  on-UASE & on-COA & $k$-NN \\ \hline
50 & 50 & 2 & 4   & PFA & 0.00 & 0.00 &0.04  &  0.00& 0.00 &0.01   \\
 & & &            & PFN & 0.00 & 0.00 & 0.00 &  0.12& 0.33 &0.98   \\
 & & &            & Delay& 9.35& 10.08 & 9.97 & 28.37& 35.34 & 49.94  \\ [3pt]
 50 & 50 & 3 & 4 & PFA & 0.00 & 0.00 & 0.00 &  0.00&  0.02  & 0.00 \\
 & & &          & PFN & 0.00 & 0.00& 0.00 &  0.00&  0.00  & 0.45\\
 & & &          & Delay&0.85 & 0.86 & 0.90 &  1.14&  0.50 & 42.56 \\ [3pt]
 50 & 50 & 4 & 4  & PFA& 0.01  & 0.01&0.03 &  0.01 & 0.00 & 0.04   \\
 & & &            & PFN & 0.00 & 0.00& 0.00 &  0.00 & 0.10 & 0.43   \\
 & & &            & Delay& 3.09& 3.11&  3.16&  10.16& 19.99 & 28.89   \\ 
 \hline
\end{tabular}
\label{table_directed_diri}
\end{center}
\end{table}

\medskip

In \textbf{Scenario 1}, at each time point, we allow the misalignment of nodes and after the change point, the order of layers flips. In \textbf{Scenarios 2} and \textbf{3}, the parameter vector of the distributions of latent positions changes after the change point. All simulation results can be found in \Cref{simu-sec5.3}. \Cref{table_directed_diri} presents  selected results from \textbf{Scenario 3}.

As we discussed in Sections~\ref{sec-single-mrdpg} and \ref{sec-dynamic}, TH-PCA algorithm detailed in \Cref{thpca} outperforms other competitors for tensor estimation.  Considering computational complexity and Tucker rank estimation, and especially with the support of theoretical discussions, in \Cref{sec-dynamic}, we adopt the HOSVD in \Cref{alg-hosvd} as the tensor estimation subroutine in the online change point detection procedure (\Cref{online_hpca}) for cases with random latent positions.  Simulation results in \Cref{table_directed_diri} demonstrate that the across-the-board, Alg.~\ref{online_hpca}-Def.~\ref{def-cusum}, \Cref{online_hpca} with the scan statistics defined in \Cref{def-cusum} and the tensor estimation routine specified in \Cref{alg-hosvd},  as well as on-TH-PCA, the competitor we created by replacing the tensor estimation routine with the TH-PCA detailed in \Cref{thpca}, are nearly equally effective. The runners-up in terms of performance are on-TWIST, on-USAE and on-COA.

\subsection{Real data analysis}\label{sec-real-data}
Our analysis incorporates two real data sets, the worldwide agricultural trade network data set presented here and the U.S.~air transport network data set in \Cref{simu-sec5.4}.

When nodes are labelled, we operate under the assumption of fixed latent positions. Utilizing this assumption, we employ the same competitor algorithms and tuning parameter selection methods as previously outlined in \Cref{change_point_section_f}. In particular, we apply  \Cref{online_hpca} with the scan statistics defined in \Cref{def-cusum-f} as well as the four competitors (on-HOSVD, on-TWIST,  on-UASE and on-COA) that we have developed, and $k$-NN. On the other hand, in cases where there are misalignments of nodes across time, we assume that latent positions are random, prompting us to implement the same competitors and tuning parameter selection methods as discussed in \Cref{change_point_section}, including \Cref{online_hpca} with the scan statistics defined in \Cref{def-cusum}, as well as the four competitors (on-TH-PCA, on-TWIST, on-UASE and on-COA) that we have introduced, and $k$-NN.

\medskip 
\noindent \textbf{The worldwide agricultural trade network data set} is collected annually by the Food and Agriculture Organisation of the United Nations from 1986 to 2020 ($35$ years).  The data are publicly available at \cite{FAOSTAT}. The data set has been previously analysed in the multilayer network literature using the data from 2010 \cite[e.g.][]{de2015structural, jing2021community, macdonald2022latent}. Each node represents a country and each layer corresponds to an agricultural product. Directed edges within each layer represent the export-import relationship between countries of a specific agricultural product.

To construct the multilayer networks, we select the top $L = 4$ agricultural products based on trade volume (food preparations not elsewhere classified, crude materials, pastry and non-alcoholic beverages), $n_1 = 75$ countries with the most importing relationship (referred to as reporter countries) and $n_2=75$ countries with the most exporting relationship (referred to as partner countries). For each year, the $(i,j,l)$-th entry of the raw multilayer network represents the total import value of the $l$-th 
agricultural product from the $j$-th partner country to the $i$-th reporter countries. We obtain $T = 35$ directed bipartite multilayer networks. Data from 1986 to 1995 ($10$ years) are used as the training data set, with the rest being the test data set. 

Given the labelled nature of the nodes in these multilayer networks, we assume that latent positions are fixed and employ the same competitors and tuning parameter selection methods as those mentioned in \Cref{change_point_section_f}.  With $\alpha = 0.05$, all algorithms identify Year 2002 as the change point, except that $k$-NN identifies Year 1996 as the change point.  This change point correlates with global commodity price fluctuations.  As illustrated in \cite{WorldEconomicOutlookApril2012}, Year 2002 marked a low point and subsequent turning point in global commodity prices, following a high price period from the 1970s to the late 1990s.  These fluctuations likely influenced trade patterns as countries adjusted their import and export practices in response to changing prices.

In addition, we apply independent random permutations to the nodes of the multilayer networks at each time point. Under this setting, we treat latent positions as random and utilise the same competitor algorithms and tuning parameter selection methods as presented in Section \ref{change_point_section}. All algorithms again identified Year 2002 as the change point except for $k$-NN, which identified Year~1996 as the change point.

\section{Conclusions}

We have formulated and analyzed a change point detection framework involving dynamic multilayer random dot product graphs, with fixed and random latent positions.  We represent a multilayer network as a tensor, and propose a novel estimation procedure based on tensor methods. For estimating a single multilayer network, our algorithm outperforms state-of-the-art methods.  In dynamic settings, where a sequence of multilayer random dot product graphs is observed, we assume models based on both fixed and random latent positions. In particular, for the more challenging case of random latent positions,  we developed a novel kernel-based method for detecting a nonparametric distributional change in the dynamic multilayer network, where the change is characterized by the point-wise expectation of certain kernel density estimators. We derived a bound on the detection delay with guaranteed false alarm control.  To the best of our knowledge,  the type of multilayer network models considered in the paper, the type of tensor techniques for estimating a single multilayer network, and the deployment of point-wise expectation of kernel density estimators for nonparametric change point detection are novel contributions to this field.

Our results and analysis can be improved in several ways.  First, in the online change point detection settings with random latent positions we use a rank-one approximation without iteration, instead of the more powerful but computationally expensive TH-PCA algorithm (\Cref{thpca}). In our numerical experiments,  on-TH-PCA delivers a clear improvement.  It would be interesting in future analysis to provide more refined theoretical guarantees for on-TH-PCA with random latent positions.  Second,  we do not consider temporal dependence.  As a possible future direction, it would be very interesting to allow for temporal dependence and study the performance of our online change point procedures in this setting; see, e.g., \cite[e.g.][]{padilla2019optimal, xu2022change}.

As for the multilayer network modelling, we have excluded inter-layer edges, since we are primarily interested in modelling multiple types of connections within the same set of nodes.  In the more complex scenarios with inter-layer edges, one can rely on order-4 -- instead of order-3 -- tensor to include additional information.  The PCA-based methodology and theory can be extended but with subtle problem-specific nuance.

As a parting note, we reiterate that dynamic network modelling with fixed and random latent positions implies that the nodes are labelled and unlabelled, respectively.  These are two fundamentally different modelling strategies, which in turn lead to  different ways of characterizing change points and, as a result, of measuring the magnitude of the jump in Section~\ref{sec-dynamic-f} and \ref{sec-dynamic}, respectively.  While we have chosen the Frobenius norm and the functional supremum norm respectively, we acknowledge that other metrics for quantifying the magnitude of the change can be used. 

\section*{Acknowledgements}

Wang is supported by Chancellor's International Scholarship, University of Warwick. Li, Madrid Padilla and Rinaldo are partially supported by NSF DMS 2015489.  Yu is partially supported by the EPSRC and Leverhulme Trust.

\clearpage
\bibliographystyle{plainnat}
\bibliography{references.bib}

\begin{thebibliography}{72}
\providecommand{\natexlab}[1]{#1}
\providecommand{\url}[1]{\texttt{#1}}
\expandafter\ifx\csname urlstyle\endcsname\relax
  \providecommand{\doi}[1]{doi: #1}\else
  \providecommand{\doi}{doi: \begingroup \urlstyle{rm}\Url}\fi

\bibitem[Airoldi et~al.(2008)Airoldi, Blei, Fienberg, and Xing]{airoldi2008mixed}
Edoardo~M Airoldi, David~M Blei, Stephen~E Fienberg, and Eric~P Xing.
\newblock Mixed membership stochastic blockmodels.
\newblock \emph{Journal of Machine Learning Research}, 9\penalty0 (Sep):\penalty0 1981--2014, 2008.

\bibitem[Amini and Razaee(2021)]{amini2021concentration}
Arash~A Amini and Zahra~S Razaee.
\newblock Concentration of kernel matrices with application to kernel spectral clustering.
\newblock \emph{The Annals of Statistics}, 49\penalty0 (1):\penalty0 531--556, 2021.

\bibitem[Arroyo et~al.(2021{\natexlab{a}})Arroyo, Athreya, Cape, Chen, Priebe, and Vogelstein]{MASE}
Jes{\'u}s Arroyo, Avanti Athreya, Joshua Cape, Guodong Chen, Carey~E Priebe, and Joshua~T Vogelstein.
\newblock \emph{MASE: Multiple Adjacency Spectral Embedding}, 2021{\natexlab{a}}.
\newblock URL \url{https://github.com/jesusdaniel/mase}.

\bibitem[Arroyo et~al.(2021{\natexlab{b}})Arroyo, Athreya, Cape, Chen, Priebe, and Vogelstein]{arroyo2021inference}
Jes{\'u}s Arroyo, Avanti Athreya, Joshua Cape, Guodong Chen, Carey~E Priebe, and Joshua~T Vogelstein.
\newblock Inference for multiple heterogeneous networks with a common invariant subspace.
\newblock \emph{Journal of machine learning research}, 22\penalty0 (142), 2021{\natexlab{b}}.

\bibitem[Athreya et~al.(2017)Athreya, Fishkind, Tang, Priebe, Park, Vogelstein, Levin, Lyzinski, and Qin]{athreya2017statistical}
Avanti Athreya, Donniell~E Fishkind, Minh Tang, Carey~E Priebe, Youngser Park, Joshua~T Vogelstein, Keith Levin, Vince Lyzinski, and Yichen Qin.
\newblock Statistical inference on random dot product graphs: a survey.
\newblock \emph{The Journal of Machine Learning Research}, 18\penalty0 (1):\penalty0 8393--8484, 2017.

\bibitem[Athreya et~al.(2022)Athreya, Lubberts, Park, and Priebe]{athreya2022discovering}
Avanti Athreya, Zachary Lubberts, Youngser Park, and Carey~E Priebe.
\newblock Discovering underlying dynamics in time series of networks.
\newblock \emph{arXiv preprint arXiv:2205.06877}, 2022.

\bibitem[Auddy et~al.(2024)Auddy, Xia, and Yuan]{auddy2024tensor}
Arnab Auddy, Dong Xia, and Ming Yuan.
\newblock Tensor methods in high dimensional data analysis: Opportunities and challenges.
\newblock \emph{arXiv preprint arXiv:2405.18412}, 2024.

\bibitem[Barigozzi et~al.(2010)Barigozzi, Fagiolo, and Garlaschelli]{barigozzi2010multinetwork}
Matteo Barigozzi, Giorgio Fagiolo, and Diego Garlaschelli.
\newblock Multinetwork of international trade: A commodity-specific analysis.
\newblock \emph{Physical Review E}, 81\penalty0 (4):\penalty0 046104, 2010.

\bibitem[Bazzi et~al.(2016)Bazzi, Porter, Williams, McDonald, Fenn, and Howison]{bazzi2016community}
Marya Bazzi, Mason~A Porter, Stacy Williams, Mark McDonald, Daniel~J Fenn, and Sam~D Howison.
\newblock Community detection in temporal multilayer networks, with an application to correlation networks.
\newblock \emph{Multiscale Modeling \& Simulation}, 14\penalty0 (1):\penalty0 1--41, 2016.

\bibitem[Berrett and Yu(2021)]{berrett2021locally}
Tom Berrett and Yi~Yu.
\newblock Locally private online change point detection.
\newblock \emph{Advances in Neural Information Processing Systems}, 34:\penalty0 3425--3437, 2021.

\bibitem[{Bureau of Transportation Statistics}(2022)]{BTS2022}
{Bureau of Transportation Statistics}.
\newblock T-100 domestic market (u.s. carriers), 2022.
\newblock URL \url{https://www.transtats.bts.gov/DL_SelectFields.aspx?gnoyr_VQ=GDL&QO_fu146_anzr=Nv4}.

\bibitem[Cai et~al.(2021)Cai, Zha, and Zhu]{cai2021slowly}
Xiaojing Cai, Hongyuan Zha, and Jun Zhu.
\newblock Slowly evolving community structures in dynamic networks.
\newblock \emph{Journal of the American Statistical Association}, 116\penalty0 (536):\penalty0 865--878, 2021.

\bibitem[Cardillo et~al.(2013)Cardillo, G{\'o}mez-Gardenes, Zanin, Romance, Papo, Pozo, and Boccaletti]{cardillo2013emergence}
Alessio Cardillo, Jes{\'u}s G{\'o}mez-Gardenes, Massimiliano Zanin, Miguel Romance, David Papo, Francisco~del Pozo, and Stefano Boccaletti.
\newblock Emergence of network features from multiplexity.
\newblock \emph{Scientific reports}, 3\penalty0 (1):\penalty0 1--6, 2013.

\bibitem[{Centers for Disease Control and Prevention }(2022)]{COVID}
{Centers for Disease Control and Prevention }.
\newblock Trends in number of covid-19 cases and deaths in the us reported to cdc, by state/territory, 2022.
\newblock URL \url{https://covid.cdc.gov/covid-data-tracker/#trends_totalcases_select_00}.

\bibitem[Chen(2019)]{chen2019sequential}
Hao Chen.
\newblock Sequential change-point detection based on nearest neighbors.
\newblock \emph{The Annals of Statistics}, 47\penalty0 (3):\penalty0 1381--1407, 2019.

\bibitem[Chen and Chu(2019)]{gStream}
Hao Chen and Lynna Chu.
\newblock \emph{gStream: Graph-Based Sequential Change-Point Detection for Streaming Data}, 2019.
\newblock URL \url{https://CRAN.R-project.org/package=gStream}.
\newblock R package version 0.2.0.

\bibitem[Chen et~al.(2022{\natexlab{a}})Chen, Wang, and Yan]{chen2022time}
Siqi Chen, Yaqing Wang, and Dong Yan.
\newblock Time-invariant multilayer stochastic blockmodel.
\newblock \emph{Journal of the American Statistical Association}, 117\penalty0 (536):\penalty0 300--315, 2022{\natexlab{a}}.

\bibitem[Chen et~al.(2022{\natexlab{b}})Chen, Wang, and Samworth]{chen2020high}
Yudong Chen, Tengyao Wang, and Richard~J Samworth.
\newblock High-dimensional, multiscale online changepoint detection.
\newblock \emph{Journal of the Royal Statistical Society Series B: Statistical Methodology}, 84\penalty0 (1):\penalty0 234--266, 2022{\natexlab{b}}.

\bibitem[Chu and Chen(2018)]{chu2018sequential}
Lynna Chu and Hao Chen.
\newblock Sequential change-point detection for high-dimensional and non-euclidean data.
\newblock \emph{arXiv preprint arXiv:1810.05973}, 2018.

\bibitem[De~Domenico et~al.(2015)De~Domenico, Nicosia, Arenas, and Latora]{de2015structural}
Manlio De~Domenico, Vincenzo Nicosia, Alexandre Arenas, and Vito Latora.
\newblock Structural reducibility of multilayer networks.
\newblock \emph{Nature communications}, 6\penalty0 (1):\penalty0 1--9, 2015.

\bibitem[De~Lathauwer et~al.(2000{\natexlab{a}})De~Lathauwer, De~Moor, and Vandewalle]{de2000best}
Lieven De~Lathauwer, Bart De~Moor, and Joos Vandewalle.
\newblock On the best rank-1 and rank-$(r_1, r_2,..., r_n)$ approximation of higher-order tensors.
\newblock \emph{SIAM journal on Matrix Analysis and Applications}, 21\penalty0 (4):\penalty0 1324--1342, 2000{\natexlab{a}}.

\bibitem[De~Lathauwer et~al.(2000{\natexlab{b}})De~Lathauwer, De~Moor, and Vandewalle]{de2000multilinear}
Lieven De~Lathauwer, Bart De~Moor, and Joos Vandewalle.
\newblock A multilinear singular value decomposition.
\newblock \emph{SIAM journal on Matrix Analysis and Applications}, 21\penalty0 (4):\penalty0 1253--1278, 2000{\natexlab{b}}.

\bibitem[Erd{\H{o}}s and R{\'e}nyi(1960)]{erdos1960evolution}
Paul Erd{\H{o}}s and Alfr{\'e}d R{\'e}nyi.
\newblock On the evolution of random graphs.
\newblock \emph{Publications of the Mathematical Institute of the Hungarian Academy of Sciences}, 5\penalty0 (1):\penalty0 17--61, 1960.

\bibitem[Fasy et~al.(2014)Fasy, Lecci, Rinaldo, Wasserman, Balakrishnan, and Singh]{fasy2014confidence}
Brittany~Terese Fasy, Fabrizio Lecci, Alessandro Rinaldo, Larry Wasserman, Sivaraman Balakrishnan, and Aarti Singh.
\newblock Confidence sets for persistence diagrams.
\newblock \emph{The Annals of Statistics}, pages 2301--2339, 2014.

\bibitem[{Food and Agricultural Organization of the United Nations}(2022)]{FAOSTAT}
{Food and Agricultural Organization of the United Nations}.
\newblock Food and agriculture data, 2022.
\newblock URL \url{https://www.fao.org/faostat/en/#data/TM}.

\bibitem[Garreau and Arlot(2018)]{garreau2018consistent}
Damien Garreau and Sylvain Arlot.
\newblock Consistent change-point detection with kernels.
\newblock \emph{Electronic Journal of Statistics}, 12\penalty0 (2):\penalty0 4440--4486, 2018.

\bibitem[Gin{\'e} and Guillou(1999)]{gine1999laws}
Evarist Gin{\'e} and Armelle Guillou.
\newblock Laws of the iterated logarithm for censored data.
\newblock \emph{The Annals of Probability}, 27\penalty0 (4):\penalty0 2042--2067, 1999.

\bibitem[Guo et~al.(2020)Guo, Chen, and Chen]{guo2020nonparametric}
Han Guo, Zihuai Chen, and Yichen Chen.
\newblock Nonparametric bayesian inference for evolving community structures in dynamic networks.
\newblock \emph{Journal of the American Statistical Association}, 115\penalty0 (530):\penalty0 677--689, 2020.

\bibitem[Han et~al.(2022)Han, Willett, and Zhang]{han2022optimal}
Rungang Han, Rebecca Willett, and Anru~R Zhang.
\newblock An optimal statistical and computational framework for generalized tensor estimation.
\newblock \emph{The Annals of Statistics}, 50\penalty0 (1):\penalty0 1--29, 2022.

\bibitem[He et~al.(2018)He, Xie, Wu, and Lin]{he2018sequential}
Xi~He, Yao Xie, Sin-Mei Wu, and Fan-Chi Lin.
\newblock Sequential graph scanning statistic for change-point detection.
\newblock In \emph{2018 52nd Asilomar Conference on Signals, Systems, and Computers}, pages 1317--1321. IEEE, 2018.

\bibitem[Holland et~al.(1983)Holland, Laskey, and Leinhardt]{holland1983stochastic}
Paul~W Holland, Kathryn~Blackmond Laskey, and Samuel Leinhardt.
\newblock Stochastic blockmodels: First steps.
\newblock \emph{Social networks}, 5\penalty0 (2):\penalty0 109--137, 1983.

\bibitem[IMF(2012)]{WorldEconomicOutlookApril2012}
Research~Dept. IMF.
\newblock \emph{World Economic Outlook, April 2012: Growth Resuming, Dangers Remain}.
\newblock International Monetary Fund, USA, 2012.
\newblock ISBN 9781616352462.
\newblock \doi{10.5089/9781616352462.081}.
\newblock URL \url{https://www.elibrary.imf.org/view/book/9781616352462/9781616352462.xml}.

\bibitem[Jing et~al.(2021)Jing, Li, Lyu, and Xia]{jing2021community}
Bing-Yi Jing, Ting Li, Zhongyuan Lyu, and Dong Xia.
\newblock Community detection on mixture multilayer networks via regularized tensor decomposition.
\newblock \emph{The Annals of Statistics}, 49\penalty0 (6):\penalty0 3181--3205, 2021.

\bibitem[Jones and Rubin-Delanchy(2020)]{jones2020multilayer}
Andrew Jones and Patrick Rubin-Delanchy.
\newblock The multilayer random dot product graph.
\newblock \emph{arXiv preprint arXiv:2007.10455}, 2020.

\bibitem[Karrer and Newman(2011)]{karrer2011stochastic}
Brian Karrer and MEJ Newman.
\newblock Stochastic blockmodels and community structure in networks.
\newblock \emph{Physical Review E}, 83\penalty0 (1):\penalty0 016107, 2011.

\bibitem[Kim et~al.(2019)Kim, Shin, Rinaldo, and Wasserman]{kim2019uniform}
Jisu Kim, Jaehyeok Shin, Alessandro Rinaldo, and Larry Wasserman.
\newblock Uniform convergence rate of the kernel density estimator adaptive to intrinsic volume dimension.
\newblock In \emph{International Conference on Machine Learning}, pages 3398--3407. PMLR, 2019.

\bibitem[Kolda and Bader(2009)]{kolda2009tensor}
Tamara~G Kolda and Brett~W Bader.
\newblock Tensor decompositions and applications.
\newblock \emph{SIAM review}, 51\penalty0 (3):\penalty0 455--500, 2009.

\bibitem[Lei et~al.(2023)Lei, Zhang, and Zhu]{lei2023computational}
Jing Lei, Anru~R Zhang, and Zihan Zhu.
\newblock Computational and statistical thresholds in multi-layer stochastic block models.
\newblock \emph{arXiv preprint arXiv:2311.07773}, 2023.

\bibitem[Li et~al.(2018)Li, Bien, and Wells]{rTensor}
James Li, Jacob Bien, and Martin~T. Wells.
\newblock {rTensor}: An {R} package for multidimensional array (tensor) unfolding, multiplication, and decomposition.
\newblock \emph{Journal of Statistical Software}, 87\penalty0 (10):\penalty0 1--31, 2018.
\newblock \doi{10.18637/jss.v087.i10}.

\bibitem[Li et~al.(2022)Li, Berrett, and Yu]{li2022network}
Mengchu Li, Tom Berrett, and Yi~Yu.
\newblock Network change point localisation under local differential privacy.
\newblock \emph{Advances in Neural Information Processing Systems}, 35:\penalty0 15013--15026, 2022.

\bibitem[Linton and H{\"a}rdle(1996)]{linton1996estimation}
OB~Linton and W~H{\"a}rdle.
\newblock Estimation of additive regression models with known links.
\newblock \emph{Biometrika}, 83\penalty0 (3):\penalty0 529--540, 1996.

\bibitem[MacDonald(2022)]{multiness}
Peter~W. MacDonald.
\newblock \emph{multiness: MULTIplex NEtworks with Shared Structure}, 2022.
\newblock URL \url{https://github.com/peterwmacd/multiness}.
\newblock R package version 0.0.0.9000.

\bibitem[MacDonald et~al.(2022)MacDonald, Levina, and Zhu]{macdonald2022latent}
Peter~W MacDonald, Elizaveta Levina, and Ji~Zhu.
\newblock Latent space models for multiplex networks with shared structure.
\newblock \emph{Biometrika}, 109\penalty0 (3):\penalty0 683--706, 2022.

\bibitem[McDiarmid et~al.(1989)]{mcdiarmid1989method}
Colin McDiarmid et~al.
\newblock On the method of bounded differences.
\newblock \emph{Surveys in combinatorics}, 141\penalty0 (1):\penalty0 148--188, 1989.

\bibitem[Padilla and Yu(2022)]{padilla2022dynamic}
Oscar Hernan~Madrid Padilla and Yi~Yu.
\newblock Dynamic and heterogeneous treatment effects with abrupt changes.
\newblock \emph{arXiv preprint arXiv:2206.09092}, 2022.

\bibitem[Padilla et~al.(2021{\natexlab{a}})Padilla, Yu, Wang, and Rinaldo]{padilla2019optimal}
Oscar Hernan~Madrid Padilla, Yi~Yu, Daren Wang, and Alessandro Rinaldo.
\newblock Optimal nonparametric change point analysis.
\newblock \emph{Electronic Journal of Statistics}, 15:\penalty0 1154--1201, 2021{\natexlab{a}}.

\bibitem[Padilla et~al.(2021{\natexlab{b}})Padilla, Yu, Wang, and Rinaldo]{padilla2021optimal}
Oscar Hernan~Madrid Padilla, Yi~Yu, Daren Wang, and Alessandro Rinaldo.
\newblock Optimal nonparametric multivariate change point detection and localization.
\newblock \emph{IEEE Transactions on Information Theory}, 68\penalty0 (3):\penalty0 1922--1944, 2021{\natexlab{b}}.

\bibitem[Padilla et~al.(2022)Padilla, Yu, and Priebe]{padilla2019change}
Oscar Hernan~Madrid Padilla, Yi~Yu, and Carey~E. Priebe.
\newblock Change point localization in dependent dynamic nonparametric random dot product graphs.
\newblock \emph{Journal of Machine Learning Research}, 23\penalty0 (234):\penalty0 1--59, 2022.
\newblock URL \url{http://jmlr.org/papers/v23/20-643.html}.

\bibitem[{R Core Team}(2022)]{R}
{R Core Team}.
\newblock \emph{R: A Language and Environment for Statistical Computing}.
\newblock R Foundation for Statistical Computing, Vienna, Austria, 2022.
\newblock URL \url{https://www.R-project.org/}.

\bibitem[Raskutti et~al.(2014)Raskutti, Wainwright, and Yu]{raskutti2014early}
Garvesh Raskutti, Martin~J Wainwright, and Bin Yu.
\newblock Early stopping and non-parametric regression: an optimal data-dependent stopping rule.
\newblock \emph{The Journal of Machine Learning Research}, 15\penalty0 (1):\penalty0 335--366, 2014.

\bibitem[Ren et~al.(2024)Ren, Li, Lyu, and Xia]{rMultiNet}
Chenyu Ren, Ting Li, Zhongyuan Lyu, and Dong Xia.
\newblock \emph{rMultiNet: rMultiNet: An R Package For Multilayer Networks Analysis}, 2024.
\newblock R package version 0.1.0.

\bibitem[Richard and Montanari(2014)]{richard2014statistical}
Emile Richard and Andrea Montanari.
\newblock A statistical model for tensor {P}{C}{A}.
\newblock \emph{Advances in neural information processing systems}, 27, 2014.

\bibitem[Rossi et~al.(2013)Rossi, Ahmed, Neville, and Henderson]{rossi2013modularity}
Ryan~A Rossi, Nesreen~K Ahmed, Jennifer Neville, and Keith Henderson.
\newblock Modularity and community detection in bipartite networks.
\newblock \emph{Journal of Machine Learning Research}, 14\penalty0 (1):\penalty0 1723--1757, 2013.

\bibitem[Song and Chen(2022)]{song2022new}
Hoseung Song and Hao Chen.
\newblock New kernel-based change-point detection.
\newblock \emph{arXiv preprint arXiv:2206.01853}, 2022.

\bibitem[Sussman et~al.(2012)Sussman, Tang, Fishkind, and Priebe]{sussman2012consistent}
Daniel~L Sussman, Minh Tang, Donniell~E Fishkind, and Carey~E Priebe.
\newblock A consistent adjacency spectral embedding for stochastic blockmodel graphs.
\newblock \emph{Journal of the American Statistical Association}, 107\penalty0 (499):\penalty0 1119--1128, 2012.

\bibitem[Vershynin(2010)]{vershynin2010introduction}
Roman Vershynin.
\newblock Introduction to the non-asymptotic analysis of random matrices.
\newblock \emph{arXiv preprint arXiv:1011.3027}, 2010.

\bibitem[Vershynin(2018)]{vershynin2018high}
Roman Vershynin.
\newblock \emph{High-dimensional probability: An introduction with applications in data science}, volume~47.
\newblock Cambridge university press, 2018.

\bibitem[Wang et~al.(2020)Wang, Yu, and Rinaldo]{wang2020univariate}
Daren Wang, Yi~Yu, and Alessandro Rinaldo.
\newblock Univariate mean change point detection: Penalization, {CUSUM} and optimality.
\newblock \emph{Electronic Journal of Statistics}, 14\penalty0 (1):\penalty0 1917--1961, 2020.

\bibitem[Wang et~al.(2021)Wang, Yu, and Rinaldo]{wang2021optimal}
Daren Wang, Yi~Yu, and Alessandro Rinaldo.
\newblock Optimal change point detection and localization in sparse dynamic networks.
\newblock \emph{The Annals of Statistics}, 49\penalty0 (1):\penalty0 203--232, 2021.

\bibitem[Weyl(1912)]{weyl1912asymptotische}
Hermann Weyl.
\newblock Das asymptotische verteilungsgesetz der eigenwerte linearer partieller differentialgleichungen (mit einer anwendung auf die theorie der hohlraumstrahlung).
\newblock \emph{Mathematische Annalen}, 71\penalty0 (4):\penalty0 441--479, 1912.

\bibitem[Xu et~al.(2022)Xu, Wang, Zhao, and Yu]{xu2022change}
Haotian Xu, Daren Wang, Zifeng Zhao, and Yi~Yu.
\newblock Change point inference in high-dimensional regression models under temporal dependence.
\newblock \emph{arXiv preprint arXiv:2207.12453}, 2022.

\bibitem[Xu(2018)]{xu2018rates}
Jiaming Xu.
\newblock Rates of convergence of spectral methods for graphon estimation.
\newblock In \emph{International Conference on Machine Learning}, pages 5433--5442. PMLR, 2018.

\bibitem[Young and Scheinerman(2007)]{young2007random}
Stephen~J Young and Edward~R Scheinerman.
\newblock Random dot product graph models for social networks.
\newblock In \emph{International Workshop on Algorithms and Models for the Web-Graph}, pages 138--149. Springer, 2007.

\bibitem[Yu et~al.(2015)Yu, Wang, and Samworth]{yu2015useful}
Yi~Yu, Tengyao Wang, and Richard~J Samworth.
\newblock A useful variant of the {D}avis--{K}ahan theorem for statisticians.
\newblock \emph{Biometrika}, 102\penalty0 (2):\penalty0 315--323, 2015.

\bibitem[Yu et~al.(2020)Yu, Padilla, Wang, and Rinaldo]{yu2020note}
Yi~Yu, Oscar Hernan~Madrid Padilla, Daren Wang, and Alessandro Rinaldo.
\newblock A note on online change point detection.
\newblock \emph{arXiv preprint arXiv:2006.03283}, 2020.

\bibitem[Yu et~al.(2021)Yu, Padilla, Wang, and Rinaldo]{yu2021optimal}
Yi~Yu, Oscar Hernan~Madrid Padilla, Daren Wang, and Alessandro Rinaldo.
\newblock Optimal network online change point localisation.
\newblock \emph{arXiv preprint arXiv:2101.05477}, 2021.

\bibitem[Zhang and Han(2019)]{zhang2019optimal}
Anru Zhang and Rungang Han.
\newblock Optimal sparse singular value decomposition for high-dimensional high-order data.
\newblock \emph{Journal of the American Statistical Association}, 114\penalty0 (528):\penalty0 1708--1725, 2019.

\bibitem[Zhang and Han(2022)]{STATSVD}
Anru Zhang and Rungang Han.
\newblock \emph{STATSVD: STATSVD}, 2022.
\newblock URL \url{https://github.com/Rungang/STATSVD}.

\bibitem[Zhang and Xia(2018)]{zhang2018tensor}
Anru Zhang and Dong Xia.
\newblock Tensor svd: Statistical and computational limits.
\newblock \emph{IEEE Transactions on Information Theory}, 64\penalty0 (11):\penalty0 7311--7338, 2018.

\bibitem[Zhang et~al.(2022)Zhang, Cai, and Wu]{zhang2018heteroskedastic}
Anru~R Zhang, T~Tony Cai, and Yihong Wu.
\newblock Heteroskedastic pca: Algorithm, optimality, and applications.
\newblock \emph{The Annals of Statistics}, 50\penalty0 (1):\penalty0 53--80, 2022.

\bibitem[Zhang et~al.(2019)Zhang, Xu, and Wang]{zhang2019detecting}
Xiang Zhang, Qian Xu, and Hao Wang.
\newblock Detecting abrupt changes in dynamic networks using a two-stage procedure.
\newblock \emph{Journal of the American Statistical Association}, 114\penalty0 (527):\penalty0 1489--1501, 2019.

\bibitem[Zhang et~al.(2020)Zhang, Xue, and Zhu]{zhang2020flexible}
Xuefei Zhang, Songkai Xue, and Ji~Zhu.
\newblock A flexible latent space model for multilayer networks.
\newblock In \emph{International Conference on Machine Learning}, pages 11288--11297. PMLR, 2020.

\end{thebibliography}

\appendix

\section*{Appendices}\label{appendices}

All the technical details of this paper are documented in the Appendices, with the proofs of Theorems~\ref{random_theorem}, \ref{main_theorem_f} and \ref{main_theorem} collected in Appendices~\ref{proof-sec2}, \ref{proof-sec3} and \ref{proof-sec4}, respectively. 

We only presented the undirected MRDGPs with fixed and random latent positions in the main text.  As for the fixed latent position case, the directed and undirected MRDPGs can be treated by exactly the same methods and share exactly the same theory.  We therefore omit it.  This is not true for the random latent positions.  If one assumes the tail and head vertex collections are the same in the directed case, then it can be treated the same as the undirected case.  If one considers a bipartite graph, with different tail and head vertex collections, then simpler methodology and theory are in order.  For completeness, we collect corresponding definitions, statistics and theoretical results in 
\Cref{sec-directed-edges}.  We would like to emphasize that this is an easier case than the undirected one we present in the main text.

An analysis of another real data set, the U.S.~air transport network, is mentioned in \Cref{sec-real-data} and detailed in \Cref{supp-sec-num}, along with additional simulation results including tuning parameter sensitivity analysis.

\section[]{Proof of Theorem \ref{random_theorem}}\label{proof-sec2}
We present the proof of \Cref{random_theorem} in \Cref{sec-A2}, with an introduction of additional notation in \Cref{sec-A0} and all necessary auxiliary results, including those from \cite{han2022optimal}, in \Cref{sec-A1}.

\subsection{Additional notation}\label{sec-A0}

For any matrix $A\in \R^{p_1 \times p_2}$,  let $\sigma_{\min}(A)$ denote the smallest non-zero singular value of $A$ if $A \neq 0$ and zero if $A = 0$.  Let $\otimes$ be the matrix Kronecker product operator.

\subsection{Auxiliary results}\label{sec-A1}

The proof of Theorem~4.1 in \cite{han2022optimal} includes an estimation error upper bound on their initial estimator, which is in fact the output of our \Cref{thpca}.  Also, their result can be straightforwardly extended to the undirected case, by adjusting (D.1) in their proofs. We include their results and extended results for the undirected case in \Cref{thm-han-4.1}, on which the proof of \Cref{random_theorem} relies.  The notation in \Cref{thm-han-4.1} is translated to the notation used in this paper.

\begin{theorem}[\citealp{han2022optimal}]\label{thm-han-4.1}
Assume that $\mathbf{A} \in \mathbb{R}^{n_1 \times n_2 \times L}$ satisfies
    \[
        \mathbb{E}\big(\mathbf{A}\big) = \mathbf{P} = \widetilde{\mathbf{S}} \times_1 U_X \times_2 U_Y \times_3 U_Q,
    \]
    with $n_1 = n_2$ and $U_X=  U_Y$ for the undirected case, where each entry of $P$ is fixed, $\widetilde{\mathbf{S}} \in \R^{d \times d \times m}$, $U_X\in \R^{n_1 \times d}, U_Y \in \R^{n_2 \times d}$ and $U_Q \in \R^{L \times m}$ are matrices with orthonormal columns and 
    \[
    \left( \mathrm{rank} \left(\mathcal{M}_1(\mathbf{P})\right), \mathrm{rank} \left(\mathcal{M}_2(\mathbf{P})\right),
    \mathrm{rank} \left(\mathcal{M}_3(\mathbf{P})\right)\right) = (d, d, m).
    \]
    For the undirected case, assume that for all $l \in [L]$, $\mathbf{A}_{:,:,l}$ are symmetric and for all $i, j\in [n_1]$, $i \leq j$, $l \in [L]$, $(\mathbf{A} - \mathbf{P})_{i, j, l}$ are independent mean-zero sub-Gaussian random variables with sub-Gaussian parameter upper bounded by $\sigma > 0$.  For the directed case, assume that all entries of $(\mathbf{A} - \mathbf{P})$ are independent mean-zero sub-Gaussian random variables with sub-Gaussian parameter upper bounded by $\sigma > 0$. 

For $s \in [3]$, let $\bar{\lambda}_s = \|\mathcal{M}_s(\mathbf{P})\|$ and $\underline{\lambda}_s = \sigma_{\min}(\mathcal{M}_s(\mathbf{P}))$.  Let $\widetilde{\mathbf{P}}$ be the intermediate output of \Cref{thpca} with $(r_1, r_2 , r_3) = (d, d, m)$.  It holds that
\begin{align*}
    & \P \Big\{\sigma^{-2}\big\|\widetilde{\mathbf{P}} - \mathbf{P}\big\|_{\mathrm{F}}^2 \leq C\big(\underline{\lambda}_1^{-4} \bar{\lambda}_1^2 (\bar{\lambda}_1^2 n_1 d + n_1 n_2 L d) + \underline{\lambda}_2^{-4} \bar{\lambda}_2^2 (\bar{\lambda}_2^2 n_2 d + n_1 n_2 L d)   \\
    & \hspace{1cm} + \underline{\lambda}_3^{-4} \bar{\lambda}_3^2 (\bar{\lambda}_3^2 Lm + n_1 n_2 L m) + d^2m + n_1d + n_2d + Lm \big)  \Big\} \geq 1 - C_1\exp\{-c_1(n_1 \vee n_2 \vee L)\},  
\end{align*}
where $C, C_1, c > 0$ are absolute constants.
\end{theorem}

\begin{lemma}\label{fix_lemma1}
Given fixed matrices $X \in \mathbb{R}^{n \times d}$ and $Q \in \mathbb{R}^{L \times d^2}$, assume $\mathrm{rank}(X) = d$ and $\mathrm{rank}(Q) = m$.  Write $\mathbf{P} = \mathbf{S} \times_1 X \times_2 X \times_3 Q$, with $\mathbf{S}$ defined in \eqref{matrix_S}.  For $s \in [3]$, let $r_s = \mathrm{rank} \left(\mathcal{M}_s(\mathbf{P})\right)$, $\overline{\lambda}_s = \|\mathcal{M}_s(\mathbf{P})\|$ and $\underline{\lambda}_s = \sigma_{\min} (\mathcal{M}_i(\mathbf{P}))$.  It holds that  $\left( r_1, r_2, r_3 \right)  = \left(d, d, m \right)$,
\[
 \sqrt{d} \sigma^2_d(X) \sigma_{m} (Q) \leq  \underline{\lambda}_1  \leq \overline{\lambda}_1 \leq  \sqrt{d} \sigma^2_1(X) \sigma_1(Q), \quad
 \sqrt{d} \sigma_d^2(X)\sigma_{m} (Q) \leq  \underline{\lambda}_2  \leq \overline{\lambda}_2 \leq  \sqrt{d} \sigma^2_1(X)\sigma_1(Q),
\] 
and 
\[
 \sigma_d^2(X) \sigma_{m} (Q) \leq  \underline{\lambda}_3  \leq \overline{\lambda}_3 \leq  \sigma^2_1(X)\sigma_1(Q). 
\] 
\end{lemma}

\begin{proof}[\textbf{Proof of \Cref{fix_lemma1}.}]
We write singular decomposition formulae of $X$ and $Q$ as
\[
    X = U_X D_X V_X^{\top} \quad \mbox{and} \quad Q = U_Q D_Q V_Q^{\top},
\]
with $U_X, U_Q, V_X, V_Q$ being matrices with orthonormal columns, $D_X \in \mathbb{R}^{d \times d}$ and $D_Q \in \mathbb{R}^{m \times m}$.  We can then write 
\[
      \mathbf{P} = \widetilde{\mathbf{S}} \times_1 U_X \times_2 U_X \times_3 U_Q,
\]
where 
\begin{align}
    \widetilde{\mathbf{S}} & = \mathbf{S} \times_1 (D_X V_X^{\top}) \times_2 (D_X V_X^{\top}) \times_3 (D_Q V_Q^{\top}) \nonumber \\
    & = \big(\mathbf{S} \times_1 V_X^{\top} \times_2 V_X^{\top} \times_3 V_Q^{\top}\big) \times_1 D_X \times_2 D_X \times_3 D_Q = \mathbf{R} \times_1 D_X \times_2 D_X \times_3 D_Q. \label{eq-tilde-S-def}
\end{align}

\medskip
\noindent \textbf{Step 1.}  By the tensor matricisation definition and properties of tensor matricisation \citep[e.g. Lemma 4 in ][]{zhang2018tensor}, we see that 
\[
    \mathcal{M}_1(\mathbf{P}) = U_X \mathcal{M}_1 (\widetilde{\mathbf{S}})(U_X \otimes U_Q)^{\top}.
\]
We are to show that $\mathrm{rank}(\mathcal{M}_1(\mathbf{P})) = \mathrm{rank}(\mathcal{M}_1(\widetilde{\mathbf{S}}))$.  It follows from \Cref{lem-rank-id}  that it suffices to show that $U_X \otimes U_Q \in \mathbb{R}^{Ln \times md}$ has $\mathrm{rank}(U_X \otimes U_Q) = md$.  Note that the columns of $U_X \otimes U_Q$ are all of the form $\mathbf{v} \otimes \mathbf{u}$, $\mathbf{v}$ and $\mathbf{u}$ being columns of $U_X$ and $U_Q$, respectively.  For any two different columns $\mathbf{v_1} \otimes \mathbf{u_1}$ and $\mathbf{v_2} \otimes \mathbf{u_2}$, we have that $(\mathbf{v_1} \otimes \mathbf{u_1})^{\top}(\mathbf{v_2} \otimes \mathbf{u_2}) = \mathbf{u_1}^{\top}\mathbf{u_2} \mathbf{v_1}^{\top}\mathbf{v_2} = 0$.  This implies that $\mathrm{rank}(U_X \otimes U_Q) = md$.

Due to the fact that both $U_X$ and $U_Q$  consist of orthonormal columns, 
we have that $\overline{\lambda}_1 \leq \|\mathcal{M}_1(\widetilde{\mathbf{S}})\|$ and $\underline{\lambda}_1 \geq \sigma_{\min}(\mathcal{M}_1(\widetilde{\mathbf{S}}))$, where the latter is also due to \Cref{lemma_smallest_singular_value_1}.

Same arguments lead to that for $s \in [3]$,
\begin{equation}\label{eq-lambda-upper-lower-bound}
    r_s = \mathrm{rank}(\mathcal{M}_s(\mathbf{P})) = \mathrm{rank}(\mathcal{M}_s(\widetilde{\mathbf{S}})),  \quad \underline{\lambda}_s \geq \sigma_{\min}(\mathcal{M}_s(\widetilde{\mathbf{S}})) \quad \mbox{and} \quad  \overline{\lambda}_s \leq \|\mathcal{M}_s(\widetilde{\mathbf{S}})\|. 
\end{equation}

\medskip
\noindent \textbf{Step 2.}  It follows from \eqref{eq-tilde-S-def} and the tensor matricisation definition, we have that 
\[
    \mathcal{M}_1(\widetilde{\mathbf{S}}) = D_X \mathcal{M}_1(\mathbf{R})(D_X \otimes D_Q)^{\top}.
\]
We are to show that $\mathrm{rank}(\mathcal{M}_1(\widetilde{\mathbf{S}})) = \mathrm{rank}(\mathcal{M}_1(\mathbf{R}))$.  It follows from \Cref{lem-rank-id} and the fact that $D_X$ is a full-rank diagonal matrix, it suffices to show that $\mathrm{rank}(D_X \otimes D_Q) = dm$.  Denoting $\{\mathbf{e}^{i, (d)}\}_{i = 1}^d$ and $\{\mathbf{e}^{j, (m)}\}_{j = 1}^m$ as the $0/1$ basis of $\mathbb{R}^d$ and $\mathbb{R}^m$, respectively, we have that the columns of $D_X \otimes D_Q$ are of the form $\mathbf{e}^{i, (d)} \otimes \mathbf{e}^{j, (m)}$, which implies that $\mathrm{rank}(D_X \otimes D_Q) = dm$.  We therefore have that $\mathrm{rank}(\mathcal{M}_1(\widetilde{\mathbf{S}})) = \mathrm{rank}(\mathcal{M}_1(\mathbf{R}))$.

As for the matrix operator norm, we have that
\[
    \|\mathcal{M}_1(\widetilde{\mathbf{S}})\| \leq \|D_X\| \|\mathcal{M}_1(\mathbf{R})\| \|D_Y \otimes D_Q\| = \sigma_1^2(X)\sigma_1(Q) \left\| \mathcal{M}_1(\mathbf{R}) \right\|.
\]
It follows from \Cref{lemma_smallest_singular_value_1} that 
\[
    \sigma_{\min} (\mathcal{M}_1(\widetilde{\mathbf{S}})) \geq \sigma_{\min} (D_X) \sigma_{\min}(\mathcal{M}_1(\mathbf{R})) \sigma_{\min} (D_X \otimes D_Q) = \sigma_d^2(X)  \sigma_m (Q) \sigma_{\min} \left(\mathcal{M}_1(\mathbf{R})\right).
\]
Same arguments and \eqref{eq-lambda-upper-lower-bound} imply that, for $s \in [3]$, $r_s = \mathrm{rank}(\mathcal{M}_s(\mathbf{R}))$,
\[
    \underline{\lambda}_s \geq \sigma_d^2(X) \sigma_m(Q) \sigma_{\min} (\mathcal{M}_s(\mathbf{R})) \quad \mbox{and} \quad \overline{\lambda}_s \leq \sigma_1^2(X) \sigma_1(Q) \|\mathcal{M}_s(\mathbf{R})\|. 
\]  

\medskip
\noindent \textbf{Step 3.} It follows from \eqref{eq-tilde-S-def} and the tensor matricisation definition, we have that 
\[
    \mathcal{M}_1 (\mathbf{R}) = V_{X}^{\top} \mathcal{M}_1 (\mathbf{S}) (V_{X} \otimes V_{Q}).
\]
It follows from \Cref{lem-rank-id} and \Cref{lemma_smallest_singular_value_2} that $\mathrm{rank}(\mathcal{M}_1(\mathbf{R})) = \mathrm{rank}(\mathcal{M}_1(\mathbf{S}) (V_{X} \otimes V_{Q})) = d$.  
From \Cref{lemma_smallest_singular_value_2}, we have that $ \|\mathcal{M}_1(\mathbf{S})\| = \sigma_{\min} (\mathcal{M}_1(\mathbf{S})) = \sqrt{d}$.
We therefore have $\|\mathcal{M}_1(\mathbf{R}) \| \leq \big\|V_{X}^{\top}\big\| \|\mathcal{M}_1(\mathbf{S})\| \|V_{X} \otimes V_{Q}\| = \sqrt{d}$.    It follows from \Cref{lemma_smallest_singular_value_1} and \Cref{lemma_smallest_singular_value_2} that  
\[
    \sigma_{\min}(\mathcal{M}_1(\mathbf{R})) \geq  \sigma_{\min} \big(V_{X}^{\top}\big)  \sigma_{\min} (\mathcal{M}_1(\mathbf{S})(V_{X} \otimes V_{Q})) \geq \sqrt{d}. 
\]
This means that $\|\mathcal{M}_1(\mathbf{R})\| = \sigma_{\min}(\mathcal{M}_1(\mathbf{R})) = \sqrt{d}$.  The same arguments also lead to that 
\[
    \mathrm{rank}(\mathcal{M}_2(\mathbf{R})) = d \quad \mbox{and} \quad \|\mathcal{M}_2(\mathbf{R})\| = \sigma_{\min}(\mathcal{M}_2(\mathbf{R})) = \sqrt{d}.
\]

As for $\mathcal{M}_3(R)$, note that due to \Cref{lemma_smallest_singular_value_2},
\[
    \mathcal{M}_3(\mathbf{R}) = V_{Q}^{\top} \mathcal{M}_3 (\mathbf{S})(V_{X} \otimes V_{X}) = V_{Q}^{\top}(V_{X} \otimes V_{X}),
\]
with 
\[
    \mathrm{rank}(\mathcal{M}_3(\mathbf{R})) = m \quad \mbox{and} \quad \|\mathcal{M}_3(\mathbf{R})\| = \sigma_{\min}(\mathcal{M}_3(\mathbf{R})) = 1, 
\]
where the first result is also due to \Cref{lem-rank-id}.

\medskip
\noindent \textbf{Step 4.}  Gathering all the previous three steps, we have that 
\[
    \underline{\lambda}_1 \geq  \sqrt{d}  \sigma^2_d(X) \sigma_{m} (Q),  \quad \underline{\lambda}_2 \geq  \sqrt{d}  \sigma^2_d(X) \sigma_{m} (Q), \quad \underline{\lambda}_3 \geq   \sigma_d^2(X) \sigma_{m} (Q), 
\]  
and
\[
    \overline{\lambda}_1   \leq  \sqrt{d} \sigma_1^2(X) \sigma_1(Q), \quad \overline{\lambda}_2   \leq  \sqrt{d} \sigma_1^2(X) \sigma_1(Q), \quad \overline{\lambda}_3   \leq   \sigma_1^2(X) \sigma_1(Q).
\]
\end{proof}

\begin{lemma}\label{lemma_smallest_singular_value_2}
For $\mathbf{S} \in \R^{d \times d \times d^2}$ defined in \eqref{matrix_S}, it holds that
\begin{equation}\label{eq-s-matrix-ranks}
    \mathcal{M}_1(\mathbf{S}) \mathcal{M}_1(\mathbf{S})^{\top} = \mathcal{M}_2(\mathbf{S}) \mathcal{M}_2(\mathbf{S})^{\top} = dI_d \quad \mbox{and} \quad \mathcal{M}_3(\mathbf{S}) = I_{d^2}.
\end{equation}
Letting $O \in \R^{d \times d}$ and $R \in \R^{d^2 \times m}$ be matrices with orthonormal columns, we have 
\begin{align}\label{eq_min_singular_1}
    \mathrm{rank} \left( \mathcal{M}_1 (\mathbf{S}) \left( O \otimes R \right) \right) = d,  \quad \sigma_{\min} \left( \mathcal{M}_1 (\mathbf{S}) \left( O \otimes R \right)  \right) \geq \sqrt{d},
\end{align}
\begin{align}\label{eq_min_singular_2}
    \mathrm{rank} \left( \mathcal{M}_2 (\mathbf{S}) \left( R \otimes O \right) \right) = d  \quad \mbox{and} \quad \sigma_{\min} \left( \mathcal{M}_2 (\mathbf{S}) \left( R \otimes O \right)  \right) \geq \sqrt{d}. 
\end{align}
\end{lemma}

\begin{proof}[\textbf{Proof of \Cref{lemma_smallest_singular_value_2}.}]
Recall that $\mathbf{S} \in \R^{d \times d \times d^2}$ is defined in \eqref{matrix_S}. We have that $\mathcal{M}_1(\mathbf{S}) \in \R^{d \times d^3}$ satisfy that for any $i \in [d]$, the non-zero entries of $\left(\mathcal{M}_1(\mathbf{S})\right)^i \in \R^{d^3}$ have indices
\[
    (1 + (i-1)d, d^2+2 + (i-1)d, 2d^2+3 + (i-1)d, \ldots, (d-1)d^2 + d + (i-1)d).
\]
This leads to that 
\[
    \mathcal{M}_1(\mathbf{S}) \mathcal{M}_1(\mathbf{S})^{\top} =  dI_d. \nonumber
\]
Similarly, we have 
\[
\mathcal{M}_2(\mathbf{S}) \mathcal{M}_2(\mathbf{S})^{\top} = dI_d \quad \mbox{and} \quad \mathcal{M}_3(\mathbf{S}) = I_{d^2}. \nonumber
\]
 We conclude the proof of \eqref{eq-s-matrix-ranks}.  We then have that $\sigma_{\min}\left( \mathcal{M}_1(\mathbf{S}) \right) = \sqrt{d}$.

To show \eqref{eq_min_singular_1} and \eqref{eq_min_singular_2}, let $B = O \otimes R \in \mathbb{R}^{d^3 \times dm}$. As for $B$, note that the columns of $B$ have the form $\mathbf{v} \otimes \mathbf{u}$, $\mathbf{v}$ and $\mathbf{u}$ being columns of $O$ and $R$, respectively.  For any column of B, it holds that $\|\mathbf{v} \otimes \mathbf{u}\| = (\mathbf{v} \otimes \mathbf{u})^{\top}(\mathbf{v} \otimes \mathbf{u}) = \mathbf{v}^{\top}\mathbf{v} \mathbf{u}^{\top}\mathbf{u} = 1$.  For any two different columns $\mathbf{v_1} \otimes \mathbf{u_1}$ and $\mathbf{v_2} \otimes \mathbf{u_2}$, it holds that $(\mathbf{v_1} \otimes \mathbf{u_1})^{\top}(\mathbf{v_2} \otimes \mathbf{u_2}) = \mathbf{v_1}^{\top}\mathbf{v_2} \mathbf{u_1}^{\top}\mathbf{u_2} = 0$.  We therefore have that $\mathrm{rank}(B) = dm$ and $\sigma_{\min}(B) = 1$.

As for $\mathcal{M}_1(\mathbf{S})B \in \mathbb{R}^{d \times dm}$, we have that 
\begin{align*}
    & \mathcal{M}_1(\mathbf{S})B =\\
    & \begin{pmatrix}
        O_1^{\top}R_{1:d, 1} & \cdots & O_1^{\top}R_{1:d, m} & \cdots & O_d^{\top}R_{1:d, 1} & \cdots & O_d^{\top}R_{1:d, m} \\
        \vdots & \vdots & \vdots & \vdots & \vdots & \vdots & \vdots \\
        O_1^{\top}R_{((d-1)d+ 1):d^2, 1} & \cdots & O_1^{\top}R_{((d-1)d+ 1):d^2, m} & \cdots & O_d^{\top}R_{((d-1)d+ 1):d^2, 1} & \cdots & O_d^{\top}R_{((d-1)d+ 1):d^2, m}
    \end{pmatrix}.
\end{align*}
To show that $\mathrm{rank}(\mathcal{M}_1(\mathbf{S}) B) = d$, we are to show the column space of $\mathcal{M}_1(\mathbf{S}) B$ has dimension at least $d$.  

We let $S^{(1)} \in \mathbb{R}^{d \times d}$ be a submatrix of $\mathcal{M}_1(\mathbf{S})B$, taking columns with indices $\{j: \, j = (i-1)m + 1, \, i \in [d]\}$.  For any $u_1 \neq u_2$, it holds that
\begin{align*}
    \big(S^{(1)}_{u_1}\big)^{\top} S^{(1)}_{u_2} & = \left( O_{u_1}^{\top} R^{(1)}_1, \cdots, O_{u_1}^{\top} R^{(d)}_1 \right) \left(O_{u_2}^{\top} R^{(1)}_1, \cdots, O_{u_2}^{\top} R^{(d)}_1 \right)^{\top} \\
    & = O_{u_1}^{\top} \left(R^{(1)}_1, \cdots, R^{(d)}_1 \right) \left(R^{(1)}_1, \cdots, R^{(d)}_1 \right)^{\top} O_{u_2} = 0,
\end{align*}
where $R^{(k)} = R_{((k-1)d+1): kd, :} \in \R^{d \times m}$, $k \in [d]$. We, therefore, have that the column space of $\mathcal{M}_1(\mathbf{S}) B$ is at least of dimension $d$.  Since $\mathcal{M}_1(\mathbf{S})B \in \mathbb{R}^{d \times dm}$, it holds that $\mathrm{rank}(\mathcal{M}_1(\mathbf{S}) B) = d$. Then following \Cref{lemma_smallest_singular_value_1}, we have $\sigma_{\min} \left( \mathcal{M}_1 (\mathbf{S}) B\right) \geq  \sigma_{\min} \left( \mathcal{M}_1 (\mathbf{S})  \right)  \sigma_{\min} \left(   B \right) \geq \sqrt{d}$, due to 
which proves \eqref{eq_min_singular_1}. Similarly, we can prove \eqref{eq_min_singular_2}.
\end{proof}

\begin{lemma}\label{lem-rank-id}
For any matrices $A \in \mathbb{R}^{p_1 \times p_2}$ and $B \in \mathbb{R}^{p_2 \times p_3}$, if $\mathrm{rank}(B) = p_2$, then $\mathrm{rank}(AB)= \mathrm{rank}(A)$; if $\mathrm{rank}(A) = p_2$, then $\mathrm{rank}(AB)= \mathrm{rank}(B)$.
\end{lemma}

\begin{proof}[\textbf{Proof of \Cref{lem-rank-id}.}]
We only show the first statement and the second statement follows from the same argument.  We first note that
\begin{equation}\label{eq-lem-rank-id-pf-1}
    \mathrm{rank}(AB) \leq \mathrm{rank}(A).
\end{equation}
Since $\mathrm{rank}(B) = p_2$, there exists a Moore--Penrose inverse of $B$, namely $B^+ = B^{\top} (BB^{\top})^{-1}$ such that $BB^+ = I$.  We then have 
\begin{equation}\label{eq-lem-rank-id-pf-2}
    \mathrm{rank}(A) = \mathrm{rank}(ABB^+) \leq \mathrm{rank}(AB).
\end{equation}
Combining \eqref{eq-lem-rank-id-pf-1} and \eqref{eq-lem-rank-id-pf-2} leads to that $\mathrm{rank}(A) = \mathrm{rank}(AB)$.
\end{proof}

\begin{lemma}\label{lemma_smallest_singular_value_1}
Let $A \in\R^{p_1 \times p_2}$ and $B \in \R^{p_2 \times p_3}$. Under either of the following three
        \begin{itemize}
            \item $\mathrm{rank}(A)= p_2$ and $B \neq 0$, or
            \item $\mathrm{rank}(B)= p_2$ and $A \neq 0$, or
            \item $\mathrm{rank}(AB) = \mathrm{rank}(A) = p_1$ and $\mathrm{rank}(B)=p_3$,
        \end{itemize}
        it holds that $\sigma_{\min}(AB) \geq \sigma_{\min}(A)\sigma_{\min}(B)$.
\end{lemma}

\begin{proof}[\textbf{Proof of \Cref{lemma_smallest_singular_value_1}.}]
We only prove the first statement and the second can be proved using the same argument.  Since $\mathrm{rank}(A) = p_2$, it follows from \Cref{lem-rank-id} that $\mathrm{rank}(AB) = \mathrm{rank}(B) = r$.  For any matrix $M \in \mathrm{R}^{n \times m}$, we have that
\begin{equation}\label{eq-lem-rank-id-2-pf-1}
    \sigma_k(M) = \sqrt{\lambda_k(M^{\top}M)}, \quad k \in [m].
\end{equation}
It follows from the min-max theorem that 
\begin{align}
    \lambda_r(B^{\top}A^{\top}AB) & = \min_U \bigg\{\max_x \bigg\{\frac{x^{\top}B^{\top}A^{\top}ABx}{\|x\|_2^2} \bigg| x \in U, \, x \neq 0\bigg\} \bigg| \mathrm{dim}(U) = r\bigg\} \nonumber \\
    & \geq \min_U \bigg\{\max_x \bigg\{\lambda_{\min}(A^{\top} A) \frac{x^{\top}B^{\top}Bx}{\|x\|_2^2} \bigg| x \in U, \, x \neq 0\bigg\} \bigg| \mathrm{dim}(U) = r\bigg\} \nonumber \\
    & = \lambda_{\min}(A^{\top} A) \min_U \bigg\{\max_x \bigg\{\frac{x^{\top}B^{\top}Bx}{\|x\|_2^2} \bigg| x \in U, \, x \neq 0\bigg\} \bigg| \mathrm{dim}(U) = r\bigg\} \nonumber \\
    & = \lambda_{\min}(A^{\top}A) \lambda_r(B^{\top}B). \label{eq-lem-rank-id-2-pf-2}
\end{align}
Combining \eqref{eq-lem-rank-id-2-pf-1} and \eqref{eq-lem-rank-id-2-pf-2} concludes the proof of the first statement.

Now, we prove the third statement. Since $\mathrm{rank}(AB)  = p_1$, we have $p_1 \leq p_3$ and $\sigma_{\min}(AB) = \sigma_{p_1}(AB)$. Then following the min-max theorem, we have
\begin{align}\label{eq-lem-rank-id-2-pf-3}
    \lambda_{p_1}(B^{\top}A^{\top}AB) & = \min_U \bigg\{\max_x \bigg\{\frac{x^{\top}B^{\top}A^{\top}ABx}{\|x\|_2^2} \bigg| x \in U, \, x \neq 0\bigg\}  \bigg| \mathrm{dim}(U) = p_1\bigg\} \nonumber \\
    & =\min_U \bigg\{\max_x \bigg\{ \frac{\left( Bx\right)^{\top} A^{\top}A\left( Bx \right)}{\|Bx\|_2^2} \frac{x^{\top}B^{\top}Bx}{\|x\|_2^2} \bigg| x \in U, \, x \neq 0\bigg\} \bigg|  \mathrm{dim}(U) = p_1 \bigg\} \nonumber \\
    & \geq \lambda_{\min}(B^{\top} B) \min_{\widetilde{U}} \bigg\{\max_y \bigg\{\frac{y^{\top}A^{\top}Ay}{\|y\|_2^2} \bigg| y \in \widetilde{U}, \, x \neq 0\bigg\} \bigg| \mathrm{dim}(\widetilde{U}) = p_1\bigg\} \nonumber \\
    & \geq \lambda_{\min}(B^{\top}B) \lambda_{p_1}(A^{\top}A),
\end{align}
where the first inequality follows from $\mathrm{rank}(B)=p_3 \geq p_1$.
Combining \eqref{eq-lem-rank-id-2-pf-1} and \eqref{eq-lem-rank-id-2-pf-3}, then we complete the proof.
\end{proof}

\subsection{Proof of Theorem \ref{random_theorem}}\label{sec-A2}
\begin{proof}[\textbf{Proof of \Cref{random_theorem}.}]

This proof consists of four steps. In \textbf{Step 1}, we consider the cases with random latent positions, and we study a high-probability event within which the following proof is conducted.  In \textbf{Step 2}, we check all the necessary conditions in \Cref{thm-han-4.1} using \Cref{fix_lemma1}.  Using \Cref{fix_lemma1}, in \textbf{Step 3}, we simplify the presentation and then complete the proof for the cases with fixed latent positions. We conclude the proof for the cases with random latent positions in \textbf{Step 4}.

\medskip
\noindent \textbf{Step 1.} In this step, we only consider random latent positions.  Writing $X = (X_1, \ldots, X_{n})^{\top} \in \mathbb{R}^{n \times d}$, due to \Cref{ass_X_Y}$(a)$, $\mathbb{P}\{\mathcal{A}\} > 1 - \epsilon$ holds for any $\epsilon > 0$ 
with
\[
    \mathcal{A} = \Big\{\big\|n^{-1} X^{\top}X - \Sigma_X\big\| \leq C\sqrt{n^{-1} ( d +  \log(1/\epsilon))} \Big\},
\]
with an absolute constant $C > 0$, following from matrix Bernstein's inequality \citep[e.g.~Remark~5.40 in][]{vershynin2010introduction}.
In the event $\mathcal{A}$, it follows from Weyl's inequality \citep{weyl1912asymptotische} that
\begin{equation} \label{random_X_singular_min}
    \sigma_d(X) \geq \sqrt{n \mu_{X, d} - C\sqrt{n (d + \log(1/\epsilon))}} \geq\sqrt{ n \mu_{X, d}/2} > 0
\end{equation}
then similarly we have $\sigma_{1}(X)\leq \sqrt{3 n \mu_{X, d}/2}$,
for any 
\begin{align}\label{epsilon_check}
   \epsilon \geq \exp\big\{- 4^{-1}C^{-2} n \mu_{X, d}^2  +d \big\}. 
\end{align}

We therefore have that in the event $\mathcal{A}$, if $\epsilon$ satisfy \eqref{epsilon_check} , then $\mathrm{rank}(X)  = d$. 

\medskip
\noindent \textbf{Step 2.}  In this step, we apply \Cref{thm-han-4.1}. In the case of random latent positions, we assume that the event $\mathcal{A}$ holds.  Recall the tensor representation of $\mathbf{P}$ defined in \eqref{Y_tucker_rep} that $\mathbf{P} = \mathbf{S} \times_1 X \times_2 X \times_3 Q$.  We first derive a similar decomposition of $\mathbf{P}$ as \Cref{thm-han-4.1}.

Let singular decomposition formulae of $X$ and $Q$ be
\[
    X = U_X D_X V_X^{\top} \quad \mbox{and} \quad Q = U_Q D_Q V_Q^{\top},
\]
with $U_X, U_Q, V_X, V_Q$ being matrices with orthonormal columns, $D_X \in \mathbb{R}^{d \times d}$ and $D_Q \in \mathbb{R}^{m \times m}$.  We can then write
\[
    \mathbf{P} = \widetilde{\mathbf{S}} \times_1 U_X \times_2 U_X \times_3 U_Q,
\]
where $\widetilde{\mathbf{S}} \in \mathbb{R}^{d \times d \times m}$ is defined as $\widetilde{\mathbf{S}} = \mathbf{S} \times_1 (D_X V_X^{\top}) \times_2 (D_X V_X^{\top}) \times_3 (D_Q V_Q^{\top})$. From \Cref{fix_lemma1}, 
we  have that  $ \mathrm{rank} \left(\mathcal{M}_1(\mathbf{P})\right) = d$,  $ \mathrm{rank} \left(\mathcal{M}_2(\mathbf{P})\right) = d$, and  $ \mathrm{rank} \left(\mathcal{M}_3(\mathbf{P})\right) = m$.

For $s \in [3]$, let $\bar{\lambda}_s = \|\mathcal{M}_s(\mathbf{P})\|$ and $\underline{\lambda}_s = \sigma_{\min}(\mathcal{M}_s(\mathbf{P}))$.  
 It follows from \Cref{thm-han-4.1} that for the cases with fixed latent positions,
\begin{align}
    & \P \Big\{\big\|\widetilde{\mathbf{P}} - \mathbf{P}\big\|_{\mathrm{F}}^2 \leq \underline{\lambda}_1^{-4} \bar{\lambda}_1^2 (\bar{\lambda}_1^2 n d + n^2 L d) + \underline{\lambda}_2^{-4} \bar{\lambda}_2^2 (\bar{\lambda}_2^2 n d + n^2 L d) + \underline{\lambda}_3^{-4} \bar{\lambda}_3^2 (\bar{\lambda}_3^2 Lm + n^2 L m) \nonumber \\
    & \hspace{1cm} + C_2(d^2m + nd + Lm)  \Big\} \geq 1 - C_1\exp\{-c_1(n \vee L)\},  \label{eq-thm-2-pf-long}
\end{align}
 or for the cases with random latent positions,
\begin{align}
    & \P \Big\{\big\|\widetilde{\mathbf{P}} - \mathbf{P}\big\|_{\mathrm{F}}^2 \leq \underline{\lambda}_1^{-4} \bar{\lambda}_1^2 (\bar{\lambda}_1^2 n d + n^2 L d) + \underline{\lambda}_2^{-4} \bar{\lambda}_2^2 (\bar{\lambda}_2^2 n d + n^2 L d) + \underline{\lambda}_3^{-4} \bar{\lambda}_3^2 (\bar{\lambda}_3^2 Lm + n^2 L m) \nonumber \\
    & \hspace{1cm} + C_2(d^2m + nd + Lm) \big| \mathcal{A} \Big\} \geq 1 - C_1\exp\{-c_1(n \vee L)\},  \label{eq-thm-1-pf-long}
\end{align}
where $C_1, C_2, c > 0$ are absolute constants.

\medskip
\noindent \textbf{Step 3.}  In this step we are to simplify the upper bound in \eqref{eq-thm-2-pf-long} and \eqref{eq-thm-1-pf-long}.  For \eqref{eq-thm-2-pf-long}, we have that with an absolute constant $C_3>0$,
\begin{align*}
    & \underline{\lambda}_1^{-4} \bar{\lambda}_1^2 (\bar{\lambda}_1^2 n d + n^2 L d) 
    \leq  d^{-1} \sigma_d^{-8}(X)  \sigma_m^{-4}(Q) \sigma_1^4(X)  \sigma_1^2(Q)\left(d\sigma_1^4(X)  \sigma_1^2(Q)n d + n^2Ld \right)  \leq C_3 nd + C_3L,
\end{align*}
where the first inequality follows from \Cref{fix_lemma1}, the second inequality follows from \Cref{ass_X_Y_u_f}$(a)$ and  \Cref{ass_X_Y_u_f}$(b)$. The same arguments lead to
\[
    \underline{\lambda}_2^{-4} \bar{\lambda}_2^2 (\bar{\lambda}_2^2 n d + n^2L d) \leq C_3 nd + C_3L \quad
\mbox{and} \quad
\underline{\lambda}_3^{-4} \bar{\lambda}_3^2 (\bar{\lambda}_3^2 Lm + n^2 L m)  \leq C_3Lm.
\]

For \eqref{eq-thm-1-pf-long}, we have that 
\begin{align*}
    & \underline{\lambda}_1^{-4} \bar{\lambda}_1^2 (\bar{\lambda}_1^2 n d + n^2 L d) \\
    \leq & d^{-1} \sigma_d^{-8}(X)  \sigma_m^{-4}(Q) \sigma_1^4(X) \sigma_1^2(Q)\left(d\sigma_1^4(X)  \sigma_1^2(Q)n d + n^2 Ld \right) \\
    \leq & C_3 n d + C_3\frac{n^2L}{\sigma_d^4(X)  \sigma_m^2(Q)} \leq C_3 nd + C_3L,
\end{align*}
where the first inequality follows from \Cref{fix_lemma1}, the second inequality follows from \Cref{ass_X_Y}$(a)$ and $(c)$, the third inequality follows from \eqref{random_X_singular_min}, \Cref{ass_X_Y}$(a)$ and $(c)$. The same arguments lead to
\[
    \underline{\lambda}_2^{-4} \bar{\lambda}_2^2 (\bar{\lambda}_2^2 n d + n^2 L d) \leq C_3 nd + C_3L.
\]
We also have that 
\begin{align*}
    & \underline{\lambda}_3^{-4} \bar{\lambda}_3^2 (\bar{\lambda}_3^2 Lm + n^2 L m) \leq \sigma_d^{-4}(X)  \sigma_m^{-2}(Q) (\sigma_1^4(X)  \sigma_1^2(Q) Lm + n^2 Lm) \leq C_3Lm,
\end{align*}
where the first inequality follows from \Cref{fix_lemma1}, and the second inequality follows from \eqref{random_X_singular_min}, \Cref{ass_X_Y}$(a)$ and $(c)$.

Gathering all the terms together and noticing the fact that $\|\widehat{\mathbf{P}} - \mathbf{P}\|_{\mathrm{F}} \leq \|\widetilde{\mathbf{P}} - \mathbf{P}\|_{\mathrm{F}}$ by \Cref{thpca}, we have that  with an absolute constant $C>0$, for the cases with fixed latent positions,
\[
    \P \Big\{\big\|\widehat{\mathbf{P}} - \mathbf{P}\big\|_{\mathrm{F}}^2 \leq C(d^2m + nd + Lm)   \Big\} \geq 1 - C_1\exp\{-c_1(n \vee L)\},
\]
completing the proof of \eqref{upper_bound_theorem_unified} for the cases with fixed latent positions.  For the cases with random latent positions, we can similarly have that
\[
    \P \Big\{\big\|\widehat{\mathbf{P}} - \mathbf{P}\big\|_{\mathrm{F}}^2 \leq C(d^2m + n d + Lm) \big|\mathcal{A}  \Big\} \geq 1 - C_1\exp\{-c_1(n \vee L)\}.
\]

\medskip
\noindent \textbf{Step 4.}  Setting $\epsilon = C_1 n^{-c_1}$, which satisfies \eqref{epsilon_check} by 
 \Cref{ass_X_Y}$(a)$, we have that
\begin{align*}
    & \mathbb{\mathbf{P}}\Big\{\big\|\widehat{\mathbf{P}} - \mathbf{P}\big\|_{\mathrm{F}}^2 > C(d^2m + n d + Lm)\Big\} \nonumber\\
    =& \mathbb{E}\Big\{\P \Big\{\big\|\widehat{\mathbf{P}} - \mathbf{P}\big\|_{\mathrm{F}}^2 > C( d^2m + nd + Lm) \big| \mathcal{A}\Big\}\Big\} \nonumber\\
    \leq & C_1\exp\{-c_1(n \vee L)\} \mathbb{P}\{\mathcal{A}\} + 1 - \mathbb{P}\{\mathcal{A}\} \leq C_1\exp\{-c_1(n \vee L)\} + C_1n^{-c_1},
\end{align*}
 which completes the proof of \eqref{upper_bound_theorem_unified}.
\end{proof}

\section[]{Proof of Theorem \ref{main_theorem_f}}\label{proof-sec3}
We present the proof of \Cref{main_theorem_f} in \Cref{subsec-proof-thm2}, following all necessary auxiliary results in \Cref{sec-B1}.

\subsection{Auxiliary results}\label{sec-B1}

\begin{lemma}\label{lemma_psi_f}
Let the data be a sequence of adjacency tensors $\{\mathbf{A}(t)\}_{t \in \mathbb{N}^*}$ as defined in \Cref{def-umrdpg-f-dynamic}, satisfying \Cref{ass_X_Y_u_f}. Let $ \widehat{\mathbf{P}}^{\cdot, \cdot}(\cdot)$ be defined in \Cref{def-cusum-f}. 

For any $\alpha \in (0, 1)$, it holds that
\begin{align*}
    \mathbb{P} & \bigg\{\exists s, t \in \mathbb{N},   \, 0 \leq s < t \mbox{ satisfying that there is no change point in } [s+1, t): \,  \\
    & \hspace{1cm}   \big\|    \widehat{\mathbf{P}}^{s, t}   -  \mathbf{P}(t) \big\|_{\mathrm{F}} >  C \sqrt{  \frac{ (d^2m + nd+Lm ) \log(t / \alpha)} {t- s}}\bigg\} \leq  \alpha,
\end{align*}
where $C > 0$ is an absolute constant.
\end{lemma}

\begin{proof}[\textbf{Proof of \Cref{lemma_psi_f}.}]
To simplify notation, for any integer pair $0 \leq s < t$, let 
\[
\xi_{s, t}  = C \sqrt{  \frac{ (d^2m + nd+Lm ) \log(t / \alpha)} {t- s}}.
\]
For any sequence of positive real values $\{\xi_{s, t}\}_{0 \leq s < t}$, it holds that 
\begin{align}\label{fix_1}
      &\mathbb{P}  \bigg\{\exists s, t \in \mathbb{N},   \, 0 \leq s < t  \mbox{ satisfying that there is no}  \nonumber\\
     & \hspace{1cm}  \mbox{change point in }[s+1, t): \big\|    \widehat{\mathbf{P}}^{s, t}   -  \mathbf{P}(t) \big\|_{\mathrm{F}} >  \xi_{s, t}  \Big\} \nonumber \\
    \leq & \sum_{u =1}^{\infty} \mathbb{P}_{\infty} \left\{\bigcup_{2^u \leq t < 2^{u+1}} \bigcup_{0 \leq s < t} \big\|    \widehat{\mathbf{P}}^{s, t}   -  \mathbf{P}(t) \big\|_{\mathrm{F}} >  \xi_{s, t} \right\} \nonumber  \\
    \leq & \sum_{u =1}^{\infty} 2^u \max_{ 2^u \leq t < 2^{u+1}} \mathbb{P}_{\infty} \left\{   \bigcup_{0 \leq s < t} \big\|    \widehat{\mathbf{P}}^{s, t}   -  \mathbf{P}(t) \big\|_{\mathrm{F}} >  \xi_{s, t} \right\}  \nonumber \\
    \leq & \sum_{u =1}^{\infty} 2^u \max_{ 2^u \leq t < 2^{u+1}} t \max_{0 \leq s < t}  \mathbb{P}_{\infty} \bigg\{\big\|    \widehat{\mathbf{P}}^{s, t}   -  \mathbf{P}(t) \big\|_{\mathrm{F}} >  \xi_{s, t} \bigg\}  \nonumber \\
    \leq & \sum_{u =1}^{\infty} 2^{2u+1} \max_{ 2^u \leq t < 2^{2u+1}} \max_{0 \leq s < t}  \mathbb{P}_{\infty} \bigg\{\big\|    \widehat{\mathbf{P}}^{s, t}   -  \mathbf{P}(t) \big\|_{\mathrm{F}} >  \xi_{s, t} \bigg\}. 
\end{align}

Since there is no change point in $[s+1, t)$, we have that
\begin{align}
     (t-s)^{-1} \sum_{u=s+1}^t \mathbf{P}(u) = \mathbf{P}(t). \nonumber
\end{align}
As indicated in  \Cref{sec-A1}, the proof of \Cref{random_theorem} is based on \Cref{thm-han-4.1}. 
\Cref{thm-han-4.1} concludes the result stated in the proof of Theorem~4.1 in \cite{han2022optimal}, providing an estimation error bound for their initial estimator, which  in fact corresponds to the output of our \Cref{thpca}. 
By adjusting (D.1) in the poof of Theorem~4.1 in \cite{han2022optimal} and then following the proof of \Cref{random_theorem}, for any $\delta \in (0,1)$ to be specified later, we have that 
\begin{align}
     \P \left\{   \big\|  \widehat{\mathbf{P}}^{s,t}  - \mathbf{P}(t) \big\|_{\mathrm {F}} >    C \sqrt{  \frac{ d^2m + nd+Lm +\log(1/\delta) } {t- s}} \right\} \leq \delta. \nonumber 
\end{align}
Then for any  $ \delta \in (0, e^{-1}]$ to be specified later, it holds that
\begin{align}\label{fix_2}
     \P \left\{   \big\|  \widehat{\mathbf{P}}^{s,t}  - \mathbf{P}(t) \big\|_{\mathrm {F}} >    C \sqrt{  \frac{ (d^2m + nd+Lm) \log(1/\delta) } {t- s}} \right\} \leq \delta.
\end{align}
Let $\delta  = \frac{\alpha \log^2 2}{2 \left(\log t +\log 2\right)^2 t^2 } \leq \frac{\alpha}{2t^2}$. Combining \eqref{fix_1} and \eqref{fix_2}, it holds that
\begin{align}
   &\mathbb{P}  \bigg\{\exists s, t \in \mathbb{N},   \, 0 \leq s < t  \mbox{ satisfying that there is no change point in [s+1, t)}:\big\|    \widehat{\mathbf{P}}^{s, t}   -  \mathbf{P}(t) \big\|_{\mathrm{F}} >  \xi_{s, t}  \Big\} \nonumber \\
  & \leq   \sum_{u=1}^{\infty} 2^{2u+1}  \max_{2^u \leq t <2^{u+1}}  \frac{\alpha \log^2 2}{2 (\log t +\log 2 )^2 t^2 }
 \leq \alpha  \sum_{u=1}^{\infty} \frac{1}{(u+1)^2} 
\leq \alpha  \sum_{u=1}^{\infty} \frac{1}{u(u+1)}
  = \alpha, \nonumber
\end{align}
which completes the proof. 
\end{proof}

\subsection[]{Proof of Theorem \ref{main_theorem_f}}\label{subsec-proof-thm2}
\begin{proof}[\textbf{Proof of \Cref{main_theorem_f}.}]
The proof is conducted in the event $\mathcal{B}$, with $\mathcal{B}$  defined as
\begin{align*}
    \mathcal{B} = & \bigg\{\exists s, t \in \mathbb{N}, 0 \leq s < t \mbox{ satisfying that there is no change point in [s+1, t)}: \,  \\
    & \hspace{1cm}   \big\|   \widehat{\mathbf{P}}^{s, t}   -  \mathbf{P}(t) \big\|_{\mathrm{F}} \leq C \sqrt{  \frac{ (d^2m + nd+Lm ) \log(t / \alpha)} {t- s}}\bigg\} ,
\end{align*}
where $C>0$ is an absolute constant.  By \Cref{lemma_psi_f}, it holds that   
\begin{equation}\label{eq-b-1-b-2-cap-f}
    \P\left\{ \mathcal{B} \right\} \geq 1 - \alpha.   
\end{equation}

 For any $s, t \in \mathbb{N}, 0 \leq s < t$, let 
\[
    \varepsilon_{s, t} = C \sqrt{  \frac{ (d^2m + nd+Lm ) \log(t / \alpha)} {t- s}}.
\]
In the event $\mathcal{B}$, for any $ 1 \leq s < t \leq \Delta$, 
\begin{align}
   \widehat{D}_{s, t}  =   \big\| \widehat{\mathbf{P}}^{0, s} - \widehat{\mathbf{P}}^{s, t} \big\|_{\mathrm{F}} \leq \big\|  \widehat{\mathbf{P}}^{0, s} - \mathbf{P}(s)\big\|_{\mathrm{F}} +  \big\| \widehat{\mathbf{P}}^{s, t}  - \mathbf{P}(t)\big\|_{\mathrm{F}} 
   \leq  \varepsilon_{0, s} + \varepsilon_{s, t} \leq \tau_{s, t}, \nonumber
\end{align}  
where the first inequality follows from the triangle inequality and  for any $ 1 \leq s < t \leq \Delta$, $\mathbf{P}(s) = \mathbf{P}(t)$, the second inequality follows from the definition of $\mathcal{B}$ and the final inequality follows from the definition of $\tau_{s, t}$ in \eqref{eq-tau-def-thm_f}.  We, therefore, have that,
    \[
        \mathcal{B} \subset \bigcap_{t \leq \Delta} \big\{\widehat{\Delta} > t \big\}. 
    \]

Under \Cref{ass_no_change_point_u_f} and due to \eqref{eq-b-1-b-2-cap-f}, it holds that
\begin{equation}\label{eq-thm2-pf-state-1_f}
    \mathbb{P}_{\infty} \{\widehat{\Delta} < \infty\} =  \mathbb{P}_{\infty} \{\exists t \in \mathbb{N}: \, \widehat{\Delta} \leq  t \} \leq \mathbb{P}\{\mathcal{B}^c\} \leq \alpha.
\end{equation}
Under \Cref{ass_change_point_u_f}, it holds that 
\begin{equation}\label{eq-no-false-alarm-change-point-f}
    \big(\widehat{\Delta} < \Delta\big) \subset \mathcal{B}^c.
\end{equation}    

We then let 
\[
    \widetilde{\Delta} = \Delta + C_\epsilon  \frac{(d^2m + nd +Lm) \log (\Delta /\alpha )}{\kappa^2},
\]
where $ C_\epsilon > 0$ is an absolute constant. 
In the event $\mathcal{B}$,  we have that
\begin{align}\label{eq-thm-2-proof-change-exists-f}
  &  \widehat{D}_{\Delta, \widetilde{\Delta}} = \big\|  \widehat{\mathbf{P}}^{0, \Delta} -  \widehat{\mathbf{P}}^{\Delta, \widetilde{\Delta}} \big\|_{\mathrm{F}} = \big\|    \widehat{\mathbf{P}}^{0, \Delta} -  \mathbf{P}(\Delta)  -    \widehat{\mathbf{P}}^{\Delta, \widetilde{\Delta}} + \mathbf{P} (\widetilde{\Delta}) +  \mathbf{P}(\Delta) -  \mathbf{P}(\Delta+1)  \big\|_{\mathrm{F}}  \nonumber\\
  \geq &    \big\|\mathbf{P}(\Delta) -  \mathbf{P}(\Delta+1)\big\|_{\mathrm{F}} -  \big\|  \widehat{\mathbf{P}}^{0, \Delta} - \mathbf{P}(\Delta) \big\|_{\mathrm{F}} - \big\|   \widehat{\mathbf{P}}^{\Delta, \widetilde{\Delta}} - \mathbf{P}(\widetilde{\Delta})  \big\|_{\mathrm{F}}  
  \geq  \kappa - \varepsilon_{0, \Delta} - \varepsilon_{\Delta, \widetilde{\Delta}}   \nonumber \\
  \geq & \frac{C_{\mathrm{SNR}}}{4} \sqrt{(d^2m + nd +Lm) \log (\Delta /\alpha )} \left(\sqrt{\frac{1}{ \Delta}} + \sqrt{\frac{1}{\widetilde{\Delta} - \Delta}} \right)> \tau_{\Delta, \widetilde{\Delta}}, 
\end{align}
where
\begin{itemize}
    \item the second equality follows from $\mathbf{P} (\widetilde{\Delta}) = \mathbf{P} (\Delta+1) $, 
    \item the second inequality follows from the definition of $\mathcal{B}$,  
    \item the third inequality follows from \Cref{snr_ass_change_point_f} and the definitions of $\varepsilon_{0, \Delta}$ and $\varepsilon_{\Delta, \widetilde{\Delta}}$, with a sufficiently large $C_{\mathrm{SNR}}$,
    \item and the last inequality follows from the definition of $\tau_{\Delta, \widetilde{\Delta}}$ in \eqref{eq-tau-def-thm_f} and with a sufficiently large $C_{\mathrm{SNR}}$. 
\end{itemize}  

Recalling the design of \Cref{online_hpca}, \eqref{eq-thm-2-proof-change-exists-f} implies that $ \mathcal{B}\subset \big(\widehat{\Delta} \leq \widetilde{\Delta}\big)$, which together with \eqref{eq-no-false-alarm-change-point-f} leads to that, under \Cref{ass_change_point_u_f},
\begin{equation}\label{eq-thm2-state-2_f}
   \mathbb{P}_{\Delta} \left\{ \Delta < \widehat{\Delta} \leq \Delta +C_\epsilon  \frac{(d^2m + nd +Lm) \log (\Delta /\alpha )}{\kappa^2}  \right\} \geq 1 - \alpha. 
\end{equation}

In view of \eqref{eq-thm2-pf-state-1_f} and \eqref{eq-thm2-state-2_f}, we complete the proof.
\end{proof}

\section[]{Proof of Theorem \ref{main_theorem}}\label{proof-sec4}

We present the proof of \Cref{main_theorem} in \Cref{sec-C3}, preceded by additional notation in \Cref{sec-C1}, a sketch of the proof in \Cref{sec-C2} and all necessary auxiliary results in \Cref{sec-C4}.

\subsection{Additional notation}\label{sec-C1}

We first start introducing additional notation required to proceed with the proofs.  Recalling the notation introduced in \Cref{def-cusum}, for any integer pair $1 < s < t$, any $k \in [2n- 1]$ and any $z \in [0, 1]^L$, let 
\[
    \mathbf{P}^{s, t} = (t-s)^{-1} \sum_{u = s+1}^t \mathbf{P}(u), \quad \mathbf{P}^{s, t}_{\mathcal{S}_k, :} = \widetilde{S}_k^{-1} \sum_{(i, j) \in \mathcal{S}_k} \mathbf{P}^{s, t}_{i, j, :},
\]
and 
\begin{equation}\label{eq-def-D-tilde}
    \widetilde{D}_{s, t}(z) =  \bigg( \sum_{k \in [2n -1]} \widetilde{S}_k  \bigg)^{-1} h^{-L} \sum_{k \in [2n -1]} \widetilde{S}_k \bigg[\mathbb{E} \bigg\{\mathcal{K}\left(\frac{z - \mathbf{P}^{0, s}_{\mathcal{S}_k, :}}{h}\right) \bigg\} -  \mathbb{E} \bigg\{\mathcal{K}\left(\frac{z - \mathbf{P}^{s, t}_{\mathcal{S}_k, :}}{h}\right)\bigg\} \bigg].
\end{equation}

\subsection{Sketch of the proof of Theorem \ref{main_theorem}}\label{sec-C2}

The proof essentially consists of two goals: with probability at least $1 - \alpha$, $(a)$ when there is no change point in $[1, t)$, 
\[
    \max_{1 < t_1 \leq t} \max_{1 \leq s < t_1} \widehat{D}_{s, t_1} \leq \tau_{s, t_1},
\]
and $(b)$ when there exists a change point $\Delta$, 
\begin{equation}
    \widehat{D}_{\Delta, \Delta + \varepsilon} > \tau_{\Delta, \Delta + \varepsilon}, \quad \mbox{with } \epsilon = C_\epsilon h^{-2L-2} \frac{(L^2 \vee d) \log \left((n \vee \Delta) /\alpha \right)}{\kappa^2 n}, \nonumber
\end{equation}
where $C_{\epsilon} > 0$ is an absolute constant.

To achieve these two goals, the building blocks of our proofs are the following: we are to 
\begin{itemize}
    \item with high probability upper bound 
    \[
    \sup_{z \in [0, 1]^L} \big|  \widehat{D}_{s, t}(z)   
 - \widetilde{D}_{s, t}(z)  \big|,
    \]
      for any integer pair $1 \leq s < t$, satisfying that there is no change point in either $[1, s)$ or $[s+1, t)$, which is done in \Cref{lemma_psi}.  In order to conduct the proof \Cref{lemma_psi}, three supplementary lemmas have been proposed in \Cref{lemm_eigen}, \Cref{prop_1} and  \Cref{lemma_S_u};

    \item with high probability upper bound  
     \[
    \Bigg|\sup_{z \in [0, 1]^L} \big| \widetilde{D}_{s, t}(z) \big| - \max_{m=1}^{M_{\alpha, t}}\big| \widetilde{D}_{s, t}(z_m) \big|\Bigg|,
    \] 
     for any integer pair $1 \leq s < t$, which is done in \Cref{lemma_z};
     \item upper bound
     \[
        \sup_{z \in [0, 1]^L} \big|  \widetilde{D}_{s, t}(z)  \big|,
     \]
      for any integer pair $1 \leq s < t$, satisfying that there is no change point in $[1, t)$, which is done in \Cref{lemma_no_change_point}; and
      \item lower bound
     \[
        \sup_{z \in [0, 1]^L} \big|  \widetilde{D}_{\Delta, \Delta+\epsilon}(z)  \big|;
     \]
     when the change point $\Delta < \infty$ exists, which is done in \Cref{lemma_one_change_point}.
\end{itemize}

\subsection{Auxiliary results}\label{sec-C4}

\begin{lemma}\label{lemma_psi}
Let the data be a sequence of adjacency tensors  $\{\mathbf{A}(t)\}_{t \in \mathbb{N}^*} \subset \mathbb{R}^{n \times n \times L}$ as defined in \Cref{def-umrdpg-dynamic} satisfying \Cref{ass_X_Y}.  Let $\widehat{D}_{\cdot, \cdot}(\cdot)$  be defined in \Cref{def-cusum} and $\widetilde{D}_{\cdot, \cdot}(\cdot)$ be defined in and \eqref{eq-def-D-tilde}, with the kernel function $\mathcal{K}(\cdot)$ satisfying \Cref{kernel_function_ass}. 

For any $\alpha \in (0, 1)$, it holds that
\begin{align*}
    \mathbb{P} & \bigg\{\exists s, t \in \mathbb{N},   \, 1 \leq s < t \mbox{ satisfying that there is no change point in either [1, s) or [s+1, t)}: \,  \\
    & \hspace{0.2cm} \sup_{z \in [0, 1]^L}  \bigg|  \widehat{D}_{s, t}(z)   
 - \widetilde{D}_{s, t}(z)  \bigg| >  C h^{-L-1} \sqrt{ \frac{(L^2 \vee d) \log\{ (n \vee t) / \alpha\}}{n}} \bigg( \frac{1}{\sqrt{s}}+ \frac{1}{\sqrt{t- s}}\bigg)\bigg\} \leq  \alpha/2,
\end{align*}
where $C > 0$ is an absolute constant.
\end{lemma}

\begin{proof}[\textbf{Proof of \Cref{lemma_psi}.}]
This proof consists of a few steps.  In \textbf{Step 1}, we decompose our target quantity into a few additive terms.  In \textbf{Steps 2}, \textbf{3} and \textbf{4}, we deal with these terms separately.  In \textbf{Step 5}, we gather all pieces and conclude the proof. 

\medskip
\noindent \textbf{Step 1.}  To simplify notation, for any integer pair $1 \leq s < t$ satisfy that there is no change point in either $[1, s)$ or $[s+1, t)$ and any $z \in [0, 1]^L$, let $\Psi_{s, t} (z) = \widehat{D}_{s, t}(z)     - \widetilde{D}_{s, t}(z)$. For any $s, t \in \mathbb{N}$ and $1 \leq s < t$, let
\[
    \widetilde{ \widehat{\mathbf{P}}}^{s, t} = \mathbf{S} \times_1 \left(\frac{1}{t-s} \sum_{u = s+1}^t X(u)\right) \times_2 \left(\frac{1}{t-s} \sum_{u = s+1}^t X(u)\right) \times_3 Q.
\]

For any sequence of positive real values $\{\varepsilon_{s, t}\}_{1 \leq s < t}$, it holds that 
\begin{align*}
      \mathbb{P} & \bigg\{\exists s, t \in \mathbb{N},   \, 1 \leq s < t: \sup_{z \in [0, 1]^L} \left|\Psi_{s, t}(z)\right| >  \varepsilon_{s, t} \Big\} \\
    \leq & \sum_{u =1}^{\infty} \mathbb{P}_{\infty} \left\{\bigcup_{2^u \leq t < 2^{u+1}} \bigcup_{1 \leq s < t} \sup_{z \in [0, 1]^L} \left|\Psi_{s, t}(z)\right| > \varepsilon_{s, t} \right\} \\
    \leq & \sum_{u =1}^{\infty} 2^u \max_{ 2^u \leq t < 2^{u+1}} \mathbb{P}_{\infty} \left\{   \bigcup_{1 \leq s < t} \sup_{z \in [0, 1]^L} \left| \Psi_{s, t}(z)  \right| >  \varepsilon_{s, t} \right\}  \\
    \leq & \sum_{u =1}^{\infty} 2^u \max_{ 2^u \leq t < 2^{u+1}} t \max_{1 \leq s < t}  \mathbb{P}_{\infty} \bigg\{\sup_{z \in [0, 1]^L} \left| \Psi_{s, t}(z)  \right| >  \varepsilon_{s, t} \bigg\} \\
    \leq & \sum_{u =1}^{\infty} 2^{2u+1} \max_{ 2^u \leq t < 2^{2u+1}} \max_{1 \leq s < t}  \mathbb{P}_{\infty} \bigg\{\sup_{z \in [0, 1]^L} \left| \Psi_{s, t}(z) \right| > \varepsilon_{s, t} \bigg\}. 
\end{align*}

To proceed, we are to investigate the sequence of events $\{\sup_{z \in [0, 1]^L} \left| \Psi_{s, t}(z) \right| > \varepsilon_{s, t}\}_{1 \leq s < t}$.  For any integer pair $1 \leq s < t$, we have that
\begin{align}\label{Term}
    & \sup_{z \in [0, 1]^L} \left| \Psi_{s, e}^t(z) \right| = \sup_{z \in [0, 1]^L} \Big|\widehat{D}_{s, t}(z) - \widetilde{D}_{s, t}(z) \Big| \nonumber\\ 
    = & \sup_{z \in [0, 1]^L} \bigg( \sum_{k \in [2n -1]} \widetilde{S}_k  \bigg)^{-1} h^{-L}\bigg|\sum_{k \in [2n - 1]} \widetilde{S}_k \bigg[\mathcal{K}\left(\frac{z -  \widehat{\mathbf{P}}^{0, s}_{\mathcal{S}_k, :}}{h}\right) - \mathcal{K}\left(\frac{z -  \widehat{\mathbf{P}}^{s, t}_{\mathcal{S}_k, :}}{h}\right) \nonumber \\
    & \hspace{1cm} - \mathbb{E}\bigg\{\mathcal{K}\left(\frac{z - \mathbf{P}^{0, s}_{\mathcal{S}_k, :}}{h}\right)\bigg\} + \mathbb{E}\bigg\{\mathcal{K}\left(\frac{z - \mathbf{P}^{s, t}_{\mathcal{S}_k, :}}{h}\right)\bigg\}\bigg]\bigg|\nonumber \\
    \leq & \sup_{z \in [0, 1]^L} \frac{C h^{-L}}{n^2}\bigg|\sum_{k \in [2n - 1]} \widetilde{S}_k \bigg[\mathcal{K}\left(\frac{z -  \widehat{\mathbf{P}}^{0, s}_{\mathcal{S}_k, :}}{h}\right) - \mathcal{K}\left(\frac{z -  \widehat{\mathbf{P}}^{s, t}_{\mathcal{S}_k, :}}{h}\right) \nonumber \\
    & \hspace{1cm} - \mathbb{E}\bigg\{\mathcal{K}\left(\frac{z - \mathbf{P}^{0, s}_{\mathcal{S}_k, :}}{h}\right)\bigg\} + \mathbb{E}\bigg\{\mathcal{K}\left(\frac{z - \mathbf{P}^{s, t}_{\mathcal{S}_k, :}}{h}\right)\bigg\}\bigg]\bigg|\nonumber \\
    \leq & \sup_{z \in [0, 1]^L} \frac{Ch^{-L}}{n^2} \bigg|\sum_{k \in [2n- 1]} \widetilde{S}_k \bigg\{\mathcal{K}\left(\frac{z -  \widehat{\mathbf{P}}^{0, s}_{\mathcal{S}_k, :}}{h}\right) - \mathcal{K}\left(\frac{z - \widetilde{ \widehat{\mathbf{P}}}^{0, s}_{\mathcal{S}_k, :}}{h}\right)\bigg\} \bigg| \nonumber \\
    & \hspace{1cm} +  \sup_{z \in [0, 1]^L} \frac{Ch^{-L}}{n^2} \bigg|\sum_{k \in [2n - 1]} \widetilde{S}_k \bigg\{\mathcal{K}\left(\frac{z -  \widehat{\mathbf{P}}^{s, t}_{\mathcal{S}_k, :}}{h}\right) - \mathcal{K}\left(\frac{z - \widetilde{ \widehat{\mathbf{P}}}^{s, t}_{\mathcal{S}_k, :}}{h}\right)\bigg\} \bigg| \nonumber \\
     & \hspace{1cm} + \sup_{z \in [0, 1]^L} \frac{Ch^{-L}}{n^2} \bigg|\sum_{k \in [2n - 1]} \widetilde{S}_k \bigg[\mathcal{K}\left(\frac{z - \widetilde{ \widehat{\mathbf{P}}}^{0, s}_{\mathcal{S}_k, :}}{h}\right) - \mathcal{K}\left(\frac{z - \mathbf{P}^{0, s}_{\mathcal{S}_k, :}}{h}\right)\bigg] \bigg| \nonumber \\
    & \hspace{1cm} + \sup_{z \in [0, 1]^L} \frac{Ch^{-L}}{n^2} \bigg|\sum_{k \in [2n - 1]} \widetilde{S}_k \bigg[\mathcal{K}\left(\frac{z - \widetilde{ \widehat{\mathbf{P}}}^{s, t}_{\mathcal{S}_k, :}}{h}\right) - \mathcal{K}\left(\frac{z - \mathbf{P}^{s, t}_{\mathcal{S}_k, :}}{h}\right)\bigg] \bigg| \nonumber \\ 
    & \hspace{1cm} + \sup_{z \in [0, 1]^L} \frac{Ch^{-L}}{n^2} \bigg|\sum_{k \in [2n - 1]} \widetilde{S}_k \bigg[\mathcal{K}\left(\frac{z - \mathbf{P}^{0, s}_{\mathcal{S}_k, :}}{h}\right) - \mathbb{E}\bigg\{\mathcal{K}\left(\frac{z - \mathbf{P}^{0, s}_{\mathcal{S}_k, :}}{h}\right)\bigg\}\bigg] \bigg| \nonumber \\
    & \hspace{1cm} + \sup_{z \in [0, 1]^L} \frac{Ch^{-L}}{n^2} \bigg|\sum_{k \in [2n - 1]} \widetilde{S}_k \bigg[\mathcal{K}\left(\frac{z - \mathbf{P}^{s, t}_{\mathcal{S}_k, :}}{h}\right) - \mathbb{E}\bigg\{\mathcal{K}\left(\frac{z - \mathbf{P}^{s, t}_{\mathcal{S}_k, :}}{h}\right)\bigg\}\bigg] \bigg| \nonumber \\
    = & (I.1) + (I.2)+ (II.1)+ (II.2) + (III.1) + (III.2),
\end{align}
where $C> 0$ is an absolute constant and the first inequality follows from  \Cref{lemma_S_u}.

\medskip
\noindent \textbf{Step 2.}  In this step, we deal with the terms $(I.1)$ and $(I.2)$ in \eqref{Term}.  

As for $(I.1)$, note that
\begin{align}\label{TermI.1}
    (I.1) = & \sup_{z \in [0, 1]^L} \frac{Ch^{-L}}{n^2} \bigg|\sum_{k \in [2n - 1]} \widetilde{S}_k \bigg\{\mathcal{K}\left(\frac{z -  \widehat{\mathbf{P}}^{0, s}_{\mathcal{S}_k, :}}{h}\right) - \mathcal{K}\left(\frac{z - \widetilde{ \widehat{\mathbf{P}}}^{0, s}_{\mathcal{S}_k, :}}{h}\right)\bigg\} \bigg|\nonumber \\
    \leq & \sup_{z \in [0, 1]^L} \frac{Ch^{-L}}{n^2} \sum_{k \in [2n - 1]} \widetilde{S}_k \bigg|\mathcal{K}\left(\frac{z -  \widehat{\mathbf{P}}^{0, s}_{\mathcal{S}_k, :}}{h}\right) - \mathcal{K}\left(\frac{z - \widetilde{ \widehat{\mathbf{P}}}^{0, s}_{\mathcal{S}_k, :}}{h}\right)\bigg|\nonumber \\
    \leq & \frac{Ch^{-L}}{n^2} \sum_{k \in [2n -1]} \widetilde{S}_k \frac{C_{\mathrm{Lip}}}{h} \big\| \widehat{\mathbf{P}}^{0, s}_{\mathcal{S}_k, :} - \widetilde{ \widehat{\mathbf{P}}}^{0, s}_{\mathcal{S}_k, :}\big\| \nonumber \\
    = & \frac{CC_{\mathrm{Lip}}}{n^2 h^{L+1}} \sum_{k \in [2n -1]} \widetilde{S}_k \sqrt{\sum_{l = 1}^L \left\{\frac{1}{\widetilde{S}_k}\sum_{(i, j) \in \mathcal{S}_k}  \left(  \widehat{\mathbf{P}}^{0, s}_{i, j, l} - \widetilde{ \widehat{\mathbf{P}}}^{0, s}_{i, j, l}\right)\right\}^2} \nonumber \\
    \leq & \frac{CC_{\mathrm{Lip}}}{n^2 h^{L+1}} \sum_{k \in [2n -1]} \sqrt{ \widetilde{S}_k \sum_{l = 1}^L \sum_{(i, j) \in \mathcal{S}_k}  \left(  \widehat{\mathbf{P}}^{0, s}_{i, j, l} - \widetilde{ \widehat{\mathbf{P}}}^{0, s}_{i, j, l}\right)^2},
\end{align}
where the second inequality follows from \Cref{kernel_function_ass}, the third inequality follows from the fact that $\|\mathbf{v}\|_1 \leq \sqrt{p}\|\mathbf{v}\|$, for any $\mathbf{v} \in \mathbb{R}^p$.

By \Cref{prop_1}, we have that  for any $2 \exp\{- C'  n \mu_{X, 1}^2 \}\leq \delta <1$, with an absolute constant $C'>0$, it holds that
\begin{align}\label{tensor_s_t_probability} 
  &\P\left\{ n^{-2}\sum_{k \in [2n -1]} \sqrt{ \widetilde{S}_k \sum_{l = 1}^L \sum_{(i, j) \in \mathcal{S}_k}  \left(  \widetilde{\mathbf{P}}^{0, s}_{i, j, l} -  \widetilde{ \widehat{\mathbf{P}}}^{0, s}_{i, j, l}\right)^2} >
  C_1 \sqrt{\frac{d\log\{(n \vee s)/\delta\}}{ns}} \right\} \leq \delta,
\end{align}
where $C_1> 0$ is an absolute constant. Combining \eqref{TermI.1} and \eqref{tensor_s_t_probability}, we have
\begin{align}\label{Term(I.1)}
		\mathbb{P} \left\{  (I.1) >   CC_1C_{\mathrm{Lip}} h^{-L-1} \sqrt{\frac{d\log\{(n \vee t)/\delta\}}{ns}} \right\} \leq \delta.
\end{align}
Similarly, we have 
\begin{align}\label{Term(I.2)}
		\mathbb{P} \left\{ (I.2) > C C_1C_{\mathrm{Lip}} h^{-L-1} \sqrt{\frac{d\log\{(n \vee t)/\delta\}}{n(t-s)}} \right\} \leq \delta.
\end{align}

\medskip
\noindent \textbf{Step 3.} In this step, we deal with the terms $(II.1)$ and $(II.2)$ in \eqref{Term}.  As for $(II.1)$, 
\begin{align}\label{eq_ii_i}
    (II.1) = &  \sup_{z \in [0, 1]^L} \frac{Ch^{-L}}{n^2} \bigg|\sum_{k \in [2n - 1]} \widetilde{S}_k \bigg[\mathcal{K}\bigg(\frac{z - \widetilde{ \widehat{\mathbf{P}}}^{0, s}_{\mathcal{S}_k, :}}{h}\bigg) - \mathcal{K}\bigg(\frac{z - \mathbf{P}^{0, s}_{\mathcal{S}_k, :}}{h}\bigg)\bigg] \bigg| \nonumber\\
    \leq &  \sup_{z \in [0, 1]^L} \frac{Ch^{-L}}{n^2} \sum_{k \in [2n - 1]} \widetilde{S}_k  \bigg| \bigg[\mathcal{K}\left(\frac{z - \widetilde{ \widehat{\mathbf{P}}}^{0, s}_{\mathcal{S}_k, :}}{h}\right) - \mathcal{K}\left(\frac{z - \mathbf{P}^{0, s}_{\mathcal{S}_k, :}}{h}\right)\bigg] \bigg| \nonumber \\
    \leq &  CC_{\mathrm{Lip}} \frac{h^{-L-1}}{n^2} \sum_{k \in [2n- 1]} \widetilde{S}_k \bigg\|\widetilde{ \widehat{\mathbf{P}}}^{0, s}_{\mathcal{S}_k, :} - \mathbf{P}^{0, s}_{\mathcal{S}_k, :} \bigg\| \nonumber \\
    \leq &C C_{\mathrm{Lip}}   \frac{h^{-L-1} \sqrt{L}}{n^2} \sum_{k \in [2n - 1]} \widetilde{S}_k \max_{l \in [L]}\bigg|\widetilde{ \widehat{\mathbf{P}}}^{0, s}_{\mathcal{S}_k, l} - \mathbf{P}^{0, s}_{\mathcal{S}_k, l} \bigg| \nonumber \\
    = & CC_{\mathrm{Lip}}   \frac{h^{-L-1} \sqrt{L}}{n^2} \nonumber \\
    & \hspace{1cm} \times \sum_{k \in [2n - 1]} \widetilde{S}_k  \max_{l \in [L]}\bigg| \widetilde{S}_k^{-1} s^{-2} \sum_{(i, j) \in \mathcal{S}_k} \sum_{u, v=1}^{s} \left[ \left( X_i(u) \right)^{\top} W_{(l)} X_j(v) - \left(  X_i(u) \right)^{\top} W_{(l)} X_j(u)\right] \bigg| \nonumber \\ 
    \leq &C C_{\mathrm{Lip}} \frac{h^{-L-1} \sqrt{L}}{n^2} \sum_{k \in [2n - 1]} \widetilde{S}_k   \nonumber \\
    & \hspace{1cm} \times \max_{l \in [L]}  \bigg|  s^{-1}\sum_{v=1}^s\widetilde{S}_k^{-1} (s-1)^{-1} \sum_{(i, j) \in \mathcal{S}_k} \sum_{u \in [s] \backslash \{v\}} \left[  \left( X_i(u) \right)^{\top} W_{(l)} \left( X_j(v) - X_j(u) \right)\right] \bigg| \nonumber \\
    \leq & CC_{\mathrm{Lip}} \frac{h^{-L-1} \sqrt{L}}{ n \sqrt{n}}  
 \sum_{k \in [2n - 1]}  \widetilde{S}_k^{1/2} \max_{l \in [L]} \left| s^{-1} \sum_{v = 1}^s \widetilde{T}_{v, k, l} \right|,
\end{align}
where 
\begin{itemize}
    \item the second inequality follows from \Cref{kernel_function_ass},
    \item the third inequality follows from the fact that $\|\mathbf{v} \| \leq \sqrt{p} \| \mathbf{v}\|_{\infty}$, for any $\mathbf{v} \in \R^{p}$,
    \item the second identity follows from the definitions 
    \[
        \mathbf{P}^{0, s}_{\mathcal{S}_k, l} = \widetilde{S}_k^{-1} s^{-1} \sum_{(i, j) \in \mathcal{S}_k} \sum_{u = 1}^s (X_i(u))^{\top} W_{(l)} X_j(u) \mbox{ and } \widetilde{ \widehat{\mathbf{P}}}^{0, s}_{\mathcal{S}_k, l} = \widetilde{S}_k^{-1} s^{-2} \sum_{u = 1}^s \sum_{v = 1}^s (X_i(u))^{\top} W_{(l)} X_j(v), 
    \]    
    \item and the final inequality follows from the cardinality of $\mathcal{S}_k$ implied by  \Cref{lemma_S_u}.
\end{itemize}

As for the $\widetilde{T}_{v, k, l}$, $v \in [s]$, $k \in [2n - 2]$ and $l \in [L]$, we introduced in \eqref{eq_ii_i}, we have that
\begin{align*}
    \widetilde{T}_{v, k, l} & =   \widetilde{S}_k^{-1} (s-1)^{-1} \sum_{(i, j) \in \mathcal{S}_k} \sum_{u \in [s] \backslash \{v\}} \left[  \left( X_i(u) \right)^{\top} W_{(l)} \left( X_j(v) - X_j(u) \right)\right] \nonumber \\
     & =  \widetilde{S}_k^{-1} (s-1)^{-1} \sum_{(i, j) \in \mathcal{S}_k} \sum_{u \in [s] \backslash \{v\}} \widetilde{Q}_{i, j,u, l}^{(v)} .
\end{align*}
Note that $\{\widetilde{Q}_{i, j, u, l}^{(v)}, \, (i, j) \in \mathcal{S}_k, u \in [s] \backslash \{ v\}\} \subset [-1,1]$ are mutually independent given $X(v)$.  We are to deploy Talagrand's concentration inequality \citep[e.g.~Theorem 5.2.16 in][]{vershynin2018high} to upper bound $\widetilde{T}_{v, k, l}$.  In order to do so, we first derive the Lipschitz constant of $\widetilde{T}_{v, k, l}$ as a function of $\widetilde{Q}^{v}_{i, j, u, l}$.  For any sequences 
\[
    \{x_{u, i, j}, \, (i, j) \in \mathcal{S}_k, u \in [s] \setminus \{v\}\}, \, \{y_{u, i, j}, \, (i, j) \in \mathcal{S}_k, u \in [s] \setminus \{v\}\} \subset [-1, 1],
\]
it holds that
\begin{align}
	& \left| \widetilde{S}_k^{-1} (s-1)^{-1}\sum_{(i, j) \in \mathcal{S}_k} \sum_{u \in [s]\backslash\{ v \} }x_{u, i, j}   -  \widetilde{S}_k^{-1} (s-1)^{-1}\sum_{(i, j) \in \mathcal{S}_k} \sum_{u \in [s]\backslash\{ v \} } y_{u, i, j} \right|\nonumber \\
	\leq  &  \widetilde{S}_k^{-1} (s-1)^{-1}\sum_{(i, j) \in \mathcal{S}_k} \sum_{u \in [s]\backslash\{ v \}} \left| x_{u, i, j}  -  y_{u, i, j}  \right|.
 \nonumber
\end{align}
Then by concentration inequality of Lipschitz functions with respect to $l_1$-metric, we have that for any $\delta_1 > 0$,
\begin{align}
    \mathbb{P}\left\{\left. \left| \widetilde{T}_{v, k, l}     -   \mathbb{E}  \left\{ \widetilde{T}_{v, k, l} | X(v)\right\}  \right| \geq \delta_1 \right|  X(v)  \right\}  \leq 2 \exp \left\{ -2\delta_1^2\widetilde{S}_k (s-1)\right\}. \nonumber
\end{align}
Consequently, we have 
\begin{align}
    & \mathbb{P}\left\{ \left| \widetilde{T}_{v, k, l}     -   \mathbb{E}  \left\{ \widetilde{T}_{v, k, l} | X(v)\right\}   \right| \geq \delta_1   \right\}     \nonumber \\
     = & \E \left\{ \mathbb{P}\left\{\left. \left| \widetilde{T}_{v, k, l}     -   \mathbb{E}  \left\{ \widetilde{T}_{v, k, l} | X(v)\right\}  \right| \geq \delta_1 \right|  X(v)  \right\} \right\}   \leq  2 \exp \left\{ -2\delta_1^2\widetilde{S}_k (s-1)\right\}. \nonumber
\end{align}
Let $\delta_1 = C_2 \sqrt{ \log\{(s  \vee n)  /\delta \} / \widetilde{S}_k (s-1)}$ and have  
\begin{align}
    \mathbb{P}\left\{  \left| \widetilde{T}_{v, k, l}     -   \mathbb{E}  \left\{ \widetilde{T}_{v, k, l} | X(v)\right\}  \right|  \geq C_2 \sqrt{\frac{ \log\{(s  \vee n )/\delta \}}{  \widetilde{S}_k (s-1) }}\right\}  \leq \delta (s  \vee n)^{-c_2}/2, \nonumber
\end{align}
where $C_2, c_2 > 0$ is an absolute constant.
Then by a union bound argument, we have
\begin{align}
    \mathbb{P}\left\{  \left| s^{-1}\sum_{v=1}^s \widetilde{T}_{v, k, l}     -   s^{-1}\sum_{v=1}^s  \mathbb{E}  \left\{ \widetilde{T}_{v, k, l} | X(v)\right\}  \right|  \geq C_2 \sqrt{\frac{ \log\{(s  \vee n) /\delta \}}{  \widetilde{S}_k (s-1) }}\right\}  \leq \delta (s  \vee n)^{-c_3}/2, \nonumber
\end{align}
where $c_3 > 0$ is an absolute constant.
Using a union bound argument twice, we have 
\begin{align}\label{eq_II_i_main_1}
   &  \mathbb{P}\left\{ \sum_{k \in [2n - 1]} \widetilde{S}_k^{1/2} \max_{l \in [L]}\left|s^{-1}\sum_{v=1}^s \widetilde{T}_{v, k, l}     -  s^{-1}\sum_{v=1}^s \mathbb{E}  \left\{ \widetilde{T}_{v, k, l} | X(v)\right\}  \right| \right. \nonumber\\
   & \hspace{1cm} \left.
   \geq  C_2n\sqrt{\frac{ \log\{(s  \vee n) /\delta \}}{   s-1 }}\right\}  
    \leq \delta/2.
\end{align}
Note that 
\begin{align}\label{mean_1}
   & s^{-1}\sum_{v=1}^s  \mathbb{E}  \left\{ \widetilde{T}_{v, k, l} | X(v)\right\} \nonumber \\
   = &   s^{-1}\sum_{v=1}^s \mathbb{E}  \left\{ \left[ \widetilde{S}_k^{-1} (s-1)^{-1}\sum_{(i, j) \in \mathcal{S}_k} \sum_{u \in [s]\backslash\{ v \} }   \left(  \left( X_i(u) \right)^{\top} W_{(l)} X_j(v) - \left( X_i(u) \right)^{\top} W_{(l)} X_j(u) \right)   \right] \Bigg| X(v)\right\} \nonumber \\
    = &   s^{-1} \widetilde{S}_k^{-1} \sum_{v=1}^s     \sum_{(i, j) \in \mathcal{S}_k} \mathbb{E}  \left\{  X_i(1) \right\}^{\top} W_{(l)}  X_j(v)  -  \mathbb{E}  \left\{  X_1(1) \right\}^{\top} W_{(l)}  \mathbb{E}  \left\{  X_1(1) \right\}, 
\end{align}
where we get the final equality since $\cup_{(i, j) \in \mathcal{S}_k}\left\{ X_i(t), X_j (t)\right\}$ are mutually independent random vectors. Consequently, we have
\begin{align}\label{mean_2}
   & \E \left\{ s^{-1}\sum_{v=1}^s  \mathbb{E}  \left\{ \widetilde{T}_{v, k, l} | X(v)\right\}  \right\}\nonumber \\
    = & \E \left\{  s^{-1} \widetilde{S}_k^{-1} \sum_{v=1}^s     \sum_{(i, j) \in \mathcal{S}_k} \mathbb{E}  \left\{  X_i(1) \right\}^{\top} W_{(l)}  X_j(v)  -  \mathbb{E}  \left\{  X_1(1) \right\}^{\top} W_{(l)}  \mathbb{E}  \left\{  X_1(1) \right\} \right\} =0. 
\end{align}
Let 
\[
   \mathcal{A} = \left\{ \sum_{k \in [2n - 1]}\widetilde{S}_k^{1/2} \max_{l \in [L]} \Bigg| s^{-1} \sum_{v=1}^s  \mathbb{E}  \left\{ \widetilde{T}_{v, k, l} | X(v)\right\}  \Bigg|\leq C_3 n \sqrt{\frac{\log(n/\delta)}{s}}\right\},
\]
where $C_3 > 0$ is an absolute constant. Note that for any $k \in [2n -1]$, for all pairs $(i, j )\in \mathcal{S}_k$, $j$'s are all distinct.
Combining \eqref{mean_1} and \eqref{mean_2}, by Hoeffding's inequality for general bounded random variables \citep[e.g.~Theorem 2.2.6 in][]{vershynin2018high}, we have that 
\begin{align}
    \P \left\{ \Bigg| s^{-1}\sum_{v=1}^s \mathbb{E}  \left\{ \widetilde{T}_{v, k, l} | X(v)\right\} \Bigg| \geq C_2 \sqrt{ \frac{\log(n/\delta)}{s \widetilde{S}_k}}\right\} \leq  \delta n^{-c_2}/2. \nonumber
\end{align}
By using a union bound argument twice, we have
\begin{align}\label{eq_II_i_main_2}
    \mathbb{P}\left\{  \mathcal{A} \right\}  \geq  1- \delta/2.  
\end{align}
Then combining \eqref{eq_II_i_main_1} and \eqref{eq_II_i_main_2}, we have
\begin{align}\label{eq_II_i_main_3}
   &  \mathbb{P}\left\{ \sum_{k \in [2n - 1]} \widetilde{S}_k^{1/2} \max_{l \in [L]}\left| s^{-1}\sum_{v=1}^s \widetilde{T}_{v, k, l} \right|   \geq 
     C_4 n \sqrt{\frac{ \log\{(s  \vee n ) /\delta \}}{   s }}\right\} \leq \delta,
\end{align}
where $C_4 >0$ is an absolute constant.
Then combining \eqref{eq_ii_i} and \eqref{eq_II_i_main_3}, we have 
\begin{align}\label{Term(II.1)} 
    \P\left\{ (II.1)\leq  C C_{\mathrm{Lip}}C_4 h^{-L-1} \sqrt{\frac{L \log\{(s  \vee n )/\delta \}}{sn}} \right\} \leq \delta.
\end{align}
Similarly, we have 
\begin{align}\label{Term(II.2)} 
    \P\left\{ (II.2)\leq  C C_{\mathrm{Lip}}C_4  h^{-L-1} \sqrt{\frac{L \log\big[ \{(t-s)  \vee n \}/\delta \big]}{(t-s)n}} \right\} \leq \delta.
\end{align}

\medskip
\noindent \textbf{Step 4.} In this step, we deal with the terms $(III.1)$ and $(III.2)$ in \eqref{Term}.  As for $(III.1)$, 
\begin{align}\label{term_iii}
    (III.1) & = \sup_{z \in [0, 1]^L} \frac{Ch^{-L}}{n^2} \bigg|\sum_{k \in [2n - 1]} \widetilde{S}_k \bigg[\mathcal{K}\left(\frac{z - \mathbf{P}^{0, s}_{\mathcal{S}_k, :}}{h}\right) - \mathbb{E}\bigg\{\mathcal{K}\left(\frac{z - \mathbf{P}^{0, s}_{\mathcal{S}_k, :}}{h}\right)\bigg\}\bigg] \bigg| \nonumber \\
	& \leq  \sup_{z \in [0, 1]^L} C \left\{ \sum_{k \in [2n- 1]}\frac{\widetilde{S}_k}{n^2}  \left|   h^{-L} \mathcal{K} \left( \frac{z-  \mathbf{P}^{0, s}_{\mathcal{S}_k, :}  }{h} \right) -  \mathbb{E} \left\{h^{-L} \mathcal{K} \left( \frac{z-  \mathbf{P}^{0, s}_{\mathcal{S}_k, :}   }{h} \right) \right\} \right| \right\} \nonumber \\
    & = \sup_{z \in [0, 1]^L} C \sum_{k \in [2n - 1]}\frac{\widetilde{S}_k}{n^2} T_k^{0, s}(z) \leq  \sup_{z \in [0, 1]^L} C\max_{k \in [2n -1]} \sqrt{\frac{\widetilde{S}_k}{n}} T^{0, s}_k(z),
\end{align}
where the final inequality follows from the fact that $\max_{k \in [2n- 1]} |\widetilde{S}_k| \leq n$.

For any $z \in [0, 1]^L$ and $k\in [2n -1]$, we are to upper bound $T_k^{0, s}(z)$ and then to upper bound term $(III.1)$ using \eqref{term_iii}.  Note that $\{\mathbf{P}_{i, j, :}(u), \, u \in [s], (i, j) \in \mathcal{S}_k\}$ are mutually independent.  For any sequences $\{x_{u, i, j}, \, u \in [s], (i, j) \in \mathcal{S}_k\}, \{y_{u, i, j}, \, u \in [s], (i, j) \in \mathcal{S}_k\}\subset [0, 1]^L$, differing only by one element, we have 
\begin{align}
	& \left| h^{-L} \mathcal{K} \left( \frac{z-  \frac{1}{s \widetilde{S}_k}\sum_{(i, j) \in \mathcal{S}_k}\sum_{u=1}^{s}   x_{u, i, j}   }{h} \right) - h^{-L} \mathcal{K} \left( \frac{z-  \frac{1}{s \widetilde{S}_k}\sum_{(i, j) \in \mathcal{S}_k}\sum_{u=1}^{s}   y_{u, i, j}   }{h} \right) \right|\nonumber \\
	\leq  & C_{\mathrm{Lip}}  h^{-L-1}  \left\|\frac{1}{s \widetilde{S}_k}\sum_{(i, j) \in \mathcal{S}_k}\sum_{u = 1}^{s}  \left( x_{u, i, j}  -  y_{u, i, j} \right)  \right\| \leq  C_{\mathrm{Lip}} h^{-L-1}\frac{\sqrt{L}}{s \widetilde{S}_k}, \nonumber
\end{align}
where the first inequality follows from \Cref{kernel_function_ass} and the second inequality follows from the fact that $\{x_{u, i, j}\}$ and $\{y_{u, i, j}\}$ only differ by one $L$-dimensional vector.  It then follows from the independent bounded differences inequality \citep[e.g.][]{mcdiarmid1989method} that, for any $\delta_2 > 0$,
\begin{align}
	\mathbb{P}\left\{ \left|T_k^{0, s}(z)\right|> \delta_2 \right\}   \leq 2\exp\left\{\frac{-2
	\delta_2^2 s \widetilde{S}_k }{ C_{\mathrm{Lip}}^2 h^{-2L-2} L} \right\}. \nonumber
\end{align}
We then let $\delta_2 = C_3 h^{-L-1}  \sqrt{\frac{ L  \log ( n /\delta_3 ) }{s \widetilde{S}_k }}$, with $\delta_3 > 0$ to be specified, due to a union bound argument we have that, for any $z \in [0, 1]^L$,
\begin{align}\label{eq_iii_i_main}
	\mathbb{P}\left\{  \max_{k \in [2n -1]} \sqrt{\frac{\widetilde{S}_k}{n}} T_k^{0, s}(z)  > C_3 h^{-L-1}  \sqrt{\frac{ L\log (n /\delta_3) }{sn}}\right\}   \leq \delta_3. 
\end{align}

For any $\epsilon \in (0, 1)$, note that $\mathcal{N}([0, 1]^L, \|\cdot\|, \epsilon)$, the cardinality of an $\epsilon$-net of $[0, 1]^L$ with respect to the $\ell_2$-norm, satisfies that 
\[
    \left(\frac{1}{\epsilon}\right)^L \leq \mathcal{N} \left( [0,1]^L, \| \cdot \|, \epsilon \right) \leq  \left(\frac{3}{\epsilon}\right)^L.
\]
let $\epsilon = C_3 \sqrt{\log(n \vee s) / \{ns\}} $, then we have 
\[
 Q =  \mathcal{N} \left( [0,1]^L, \| \cdot \|_2, \epsilon \right) \leq \left(\frac{3}{C_3} \sqrt{\frac{ ns }{\log ( n \vee s )}} \right)^{L}.
\]
We then let $\left\{z_1, \ldots,  z_Q \right\}$ be an $\epsilon$-covering of $ [0,1]^L$ and $I_j = \left\{z\in[0, 1]^L: \|z-z_j \|_2 \leq \epsilon \right\}$, $j \in [Q]$.  It follows from \eqref{term_iii} that 
\begin{align}\label{term_ii}
	 & (III.1) \leq \sup_{z \in [0,1]^L} \max_{k \in [2n -1]} \sqrt{\frac{\widetilde{S}_k}{n}} \left| T_k^{0, s}(z)\right|  \nonumber \\
  \leq   &  \max_{j \in [Q]} \left\{ \max_{k \in [2n -1]} \sqrt{\frac{\widetilde{S}_k}{n}} \left| T_k^{0, s}(z_j)\right|    + \sup_{z \in I_j} \max_{k \in [2n -1]} 
 \sqrt{\frac{\widetilde{S}_k}{n}} \left| T_k^{0, s}(z_j) - T_k^{0, s}(z)\right|\right\} \nonumber \\
    \leq  &  \max_{j \in [Q]} \max_{k \in [2n -1]} \sqrt{\frac{\widetilde{S}_k}{n}} \left| T_k^{0, s}(z_j)\right|    +  \max_{j \in [Q]} \sup_{z \in I_j}\max_{k \in [2n -1]} 
 \sqrt{\frac{\widetilde{S}_k}{n}} \left| T_k^{0, s}(z_j) - T_k^{0, s}(z)\right|. 
\end{align} 

By \eqref{eq_iii_i_main} and a union bound argument we have 
\begin{align}
   &  \mathbb{P} \left\{ \max_{j \in [Q]} \max_{k \in [2n -1]} \sqrt{\frac{\widetilde{S}_k}{n}} \left| T_k^{0, s}(z_j)\right|  \geq  C_3 h^{-L-1}  \sqrt{\frac{ L\log(n/\delta_3) }{sn}}   \right\} \leq Q\delta_3. \nonumber 
\end{align}
For $\delta > 0$ to be specified, letting $\delta_3 = \delta/Q$, we have 
\begin{align}\label{term_ii_1_1}
   &  \mathbb{P} \left\{ \max_{j \in [Q]} \max_{k \in [2n -1]} \sqrt{\frac{\widetilde{S}_k}{n}} \left| T_k^{0, s}(z_j)\right|  \geq  C_4  h^{-L-1} L \sqrt{\frac{ \log \left( (n \vee s )  /\delta\right) }{ sn} }  \right\} \leq \delta,
\end{align}
where $C_4 > 0$ is an absolute constant.

Note that  
 \begin{align}\label{term_ii_1_2}
 	&  \max_{j \in [Q]} \sup_{z \in I_j}\max_{k \in [2n -1]} 
 \sqrt{\frac{\widetilde{S}_k}{n}} \left| T_k^{0, s}(z_j) - T_k^{0, s}(z)\right| \nonumber \\
 	\leq & \max_{j \in [Q]} \sup_{z \in I_j}\max_{k \in [2n -1]}  \sqrt{\frac{\widetilde{S}_k}{n}}\left|  h^{-L} \mathcal{K} \left( \frac{z_j-  \mathbf{P}^{0, s}_{\mathcal{S}_k, :}  }{h} \right) -  h^{-L} \mathcal{K} \left( \frac{z-  \mathbf{P}^{0, s}_{\mathcal{S}_k, :}  }{h} \right) \right|  \nonumber \\
 	&  \hspace{0.2cm}+    \max_{j \in [Q]} \sup_{z \in I_j} \sum_{p \in [2n -1]} \frac{\widetilde{S}_p}{n^2} \left| \mathbb{E} \left\{h^{-L} \mathcal{K} \left( \frac{z_j-  \mathbf{P}^{0, s}_{\mathcal{S}_p, :}   }{h} \right) \right\} -    \mathbb{E} \left\{h^{-L} \mathcal{K} \left( \frac{z-  \mathbf{P}^{0, s}_{\mathcal{S}_p, :}   }{h} \right) \right\}  \right| \nonumber  \\
 	\leq &\max_{j \in [Q]} \sup_{z \in I_j} 2 C_{\mathrm{Lip}} h^{-L-1} \left\|z_j - z \right\|_2
 	\leq 2 C_{\mathrm{Lip}} C_3 h^{-L-1} \sqrt{\frac{\log(n \vee s)}{ns}}.  
 \end{align}
Combining \eqref{term_ii}, \eqref{term_ii_1_1} and \eqref{term_ii_1_2}, we have
\begin{align}\label{Term(III.1)}
    \mathbb{P}	 \left\{ (III.1)> C_5 h^{-L-1}L\sqrt{\frac{\log\{ (n \vee s)/\delta \}}{ns}}\right\} \leq  \delta,
\end{align}
where $C_5 > 0$ is an absolute constant.
Similarly, we have
\begin{align}\label{Term(III.2)}
    \mathbb{P}	 \left\{ (III.2)\leq C_5 h^{-L-1}L\sqrt{\frac{\log [  \{ n \vee (t-s)\}/\delta ]}{n(t-s)} } \right\} \leq  \delta. 
\end{align}

\medskip
\noindent \textbf{Step 5.}
Combining \eqref{Term}, \eqref{Term(I.1)}, \eqref{Term(I.2)}, \eqref{Term(II.1)},  \eqref{Term(II.2)}, \eqref{Term(III.1)} and \eqref{Term(III.2)}, we have
\begin{align}
   \mathbb{P} \left\{   \sup_{z \in [0, 1]^L} \left| \Psi_{s, e}^t(z) \right| >  C_{\Psi} h^{-L-1} \sqrt{ \frac{(L^2 \vee d) \log\{ (n \vee t) / \delta \}}{n} } \left( \frac{1}{\sqrt{s}}+ \frac{1}{\sqrt{t- s}}\right)  \right\}  \leq 6 \delta. \nonumber 
\end{align}
where $C_{\Psi} > 0$ is an absolute constant.

Let 
\[
\varepsilon_{s, t} = C_{\Psi} h^{-L-1} \sqrt{ \frac{(L^2 \vee d) \log\{ (n \vee t) / \delta \}}{n} } \left( \frac{1}{\sqrt{s}}+ \frac{1}{\sqrt{t- s}}\right) 
\]
and
$\delta = \frac{\alpha \log^2 2 }{ 24(\log t + \log 2)^2 t^2 } \leq \frac{\alpha}{48t^2}$, then we have that 
\begin{align}
     &   \mathbb{P}  \left\{ \exists s, t \in \mathbb{N}, t > 1, s \in [1, t): \sup_{z \in [0, 1]^L} \left| \Psi_{s, t}(z)  \right|    >  \varepsilon_{s, t} \right\}  \nonumber \\
     \leq & \sum_{u =1}^{\infty} 2^{2u+1} \max_{ 2^u \leq t < 2^{2u+1}}  \frac{\alpha \log^2 2 }{ 4(\log t + \log 2)^2 t^2 }  \leq \frac{\alpha}{2} \sum_{u =1}^{\infty} \frac{1}{(u+1)^2}  \leq  \frac{\alpha}{2} \sum_{u =1}^{\infty} \frac{1}{u(u+1)} = \alpha/2, \nonumber 
\end{align}
which completes the proof.
\end{proof}

\begin{lemma}\label{lemma_z}
Let the data be a sequence of adjacency tensors  $\{\mathbf{A}(t)\}_{t \in \mathbb{N}^*} \subset \mathbb{R}^{n \times n \times L}$ as defined in \Cref{def-umrdpg-dynamic} satisfying \Cref{ass_X_Y}.  Let  $\widetilde{D}_{\cdot, \cdot}(\cdot)$ be defined in and \eqref{eq-def-D-tilde}, with the kernel function $\mathcal{K}(\cdot)$ satisfying \Cref{kernel_function_ass}. 

Let $\left\{ z_m \right\}_{m=1}^{M_{s, t}}$ be a collection of grid points randomly sampled from a density $g: [0, 1]^L \to \R$ such that $\inf_{z \in [0, 1]^L}g(z) \geq c_g$, where $c_g >0$ is an absolute constant.
Let $\alpha \in (0, 1)$
For any $s, t \in \mathbb{N}$ and $1 \leq s < t$  provided that 
\[
M_{s, t}  \geq   C_M L^{-L} [ n \{ s \wedge (t-s)\} ] ^{L/2}   [ \log\{(n \vee t )/\alpha\}  ]^{-L/2+1},
\]
it holds that
\begin{align}
   \P & \left\{ \exists s, t \in \mathbb{N}, t > 1, s \in [1, t):  \sup_{z \in [0, 1]^L}  \left| \widetilde{D}_{s, t} \left(z\right)\right| -  \max_{m=1}^{M_{s, t}} \left| \widetilde{D}_{s, t}\left( z_m\right) \right| \right.
  \nonumber \\
   &\hspace{1cm} \left. >  C_{M} h^{-L-1} \sqrt{ \frac{L^2 \log\{ (n \vee t)/\alpha \}}{n}} \left( \frac{1}{\sqrt{s}}+ \frac{1}{\sqrt{t- s}}\right) \right\}   \leq \alpha/2, \nonumber
\end{align}
where $C_M >0 $ is an absolute constant.
\end{lemma}

\begin{proof}[\textbf{Proof of \Cref{lemma_z}.}]
Let $z^{*}_{s, t} = \argsup_{z \in [0, 1]^L} |  \widetilde{D}_{s, t}(z)|$.  Recalling the definition of $\widetilde{D}_{s, t}(z)$ in \eqref{eq-def-D-tilde}, we note that $z^*_{s, t}$ is non-random.  Let $\delta > 0 $ to be specified later and 
\[
\varepsilon_{s, t} =  C_{M} h^{-L-1} \sqrt{ \frac{L^2 \log\{ (n \vee t) / \delta\}}{n }} \left( \frac{1}{\sqrt{s}}+ \frac{1}{\sqrt{t- s}}\right).
\]
Note that
\begin{align}
    &   \left| \widetilde{D}_{s, t} \left(z^{*}_{s, t}\right)\right| -  \max_{m=1}^{M_{s, t}} \left| \widetilde{D}_{s, t}\left( z_m\right) \right|  \leq \min_{m=1}^{M_{s, t}} \left| \widetilde{D}_{s, t} \left(z^{*}_{s, t}\right)  -  \widetilde{D}_{s, t}\left( z_m \right) \right|  \nonumber \\
    \leq & \min_{m=1}^{M_{s, t}} \sum_{k \in [2n - 1]}\frac{C\widetilde{S}_k}{n^2}\bigg| \mathbb{E} \left\{h^{-L}  \mathcal{K} \left( \frac{ z^{*}_{s, t} -  \mathbf{P}^{0, s}_{\mathcal{S}_k, :}    }{h} \right) - h^{-L}  \mathcal{K} \left( \frac{z_m -  \mathbf{P}^{0, s}_{\mathcal{S}_k, :}    }{h} \right) \right\} \nonumber \\
    & \hspace{0.5cm}  +  \mathbb{E} \left\{h^{-L} \mathcal{K} \left( \frac{z_m-  \mathbf{P}^{s, t}_{\mathcal{S}_k, :}   }{h} \right) - h^{-L} \mathcal{K} \left( \frac{z^{*}_{s, t} -  \mathbf{P}^{s, t}_{\mathcal{S}_k, :}   }{h} \right)\right\} \bigg|\nonumber\\
    \leq &  \min_{m=1}^{M_{s, t}}  2C C^{\prime}h^{-L-1} C_{\mathrm{Lip}} \left\| z_m - z^{*}_{s, t}  \right\|_2,\nonumber
\end{align}
 with absolute constants $C, C^{\prime}> 0$, where the last inequalities follow from \Cref{lemma_S_u} and \Cref{kernel_function_ass}.  For $\varepsilon_{s, t} > 0$ to be defined, we then have that 
\begin{align}
    &\P \left\{ \sup_{z \in [0, 1]^L} \left| \widetilde{D}_{s, t} \left(z\right)\right| -  \max_{m=1}^{M_{s, t}} \left| \widetilde{D}_{s, t}\left( z_m\right) \right|  > \varepsilon_{s, t} \right\} \leq  \P \left\{  \min_{m=1}^{M_{s, t}}  2C C^{\prime}h^{-L-1} C_{\mathrm{Lip}} \left\| z_m - z^{*}_{s, t}  \right\|_2  >\varepsilon_{s, t} \right\} \nonumber \\
    =  & \P \left\{  \left\{ z_m \right\}_{m = 1}^{M_{s, t}} \notin  B \left(z^{*}_{s, t}, \,  C_1\sqrt{ \frac{L^2 \log\{ (n \vee t)/\delta \}}{n}} \left( \frac{1}{\sqrt{s}}+ \frac{1}{\sqrt{t- s}}\right) \right)  \right\} \nonumber \\
    = & \left[ 1 - \P\left\{ z_1  \in  B \left(z^{*}_{s, t}, \, C_1\sqrt{ \frac{L^2  \log\{ (n \vee t)/\delta \})}{n}} \left( \frac{1}{\sqrt{s}}+ \frac{1}{\sqrt{t- s}}\right) \right) \right\} \right]^{M_{s, t}}
    \nonumber \\
    \leq  & \left[ 1 - \left( C_2 \sqrt{\frac{L^2 \log\{ (n \vee t)/\delta \}}{n}}  \max\left\{ \frac{1}{\sqrt{s}},  \frac{1}{\sqrt{t- s}}\right\} \right)^L  \right]^{M_{s, t}} \nonumber\\
    \leq  & \exp \left\{ -M_{s, t} \left( C_2 \sqrt{\frac{ L^2  \log\{ (n \vee t)/\delta \}}{n}}  \max\left\{ \frac{1}{\sqrt{s}},  \frac{1}{\sqrt{t- s}}\right\} \right)^L  \right\}, \nonumber
\end{align}
where $C_1, C_2 > 0$ are absolute constants,  the second inequality follows from the assumption that  the density $\inf_{z \in [0, 1]^L}g(z) \geq c_g$, with $c_g >0$ being an absolute constant.  Thus, if 
\[
M_{s, t}  \geq  C_2^{-L} \left( \frac{n \{s \wedge (t-s)\} }{ L^2 \log\{ (n \vee t)/\delta\} }   \right)^{L/2} \log(1/\delta),
\]
we have 
\[
 \P \left\{ \sup_{z \in [0, 1]^L} \left| \widetilde{D}_{s, t} \left(z\right)\right| -  \max_{m=1}^{M_{s, t}} \left| \widetilde{D}_{s, t}\left( z_m\right) \right|  > \varepsilon_{s, t} \right\} \leq \delta.
\]
Let $\delta = \frac{\alpha \log^2(2) }{ 4\{\log (t) + \log (2)\}^2 t^2 } \leq \frac{\alpha}{16t^2}$, then we have that 
\begin{align}
     &   \mathbb{P}  \left\{ \exists s, t \in \mathbb{N}, t > 1, s \in [1, t):  \sup_{z \in [0, 1]^L} \left| \widetilde{D}_{s, t} \left(z\right)\right| -  \max_{m=1}^{M_{s, t}} \left| \widetilde{D}_{s, t}\left( z_m\right) \right|  > \varepsilon_{s, t}  \right\}  \nonumber \\
     \leq & \sum_{u =1}^{\infty} 2^{2u+1} \max_{ 2^u \leq t < 2^{2u+1}}  \frac{\alpha \log^2 (2) }{ 4\{\log (t) + \log (2)\}^2 t^2 }  \leq  \frac{\alpha}{2} \sum_{u =1}^{\infty} \frac{1}{(u+1)^2} \leq  \frac{\alpha}{2} \sum_{u =1}^{\infty} \frac{1}{u(u+1)} = \alpha/2, \nonumber 
\end{align}
which completes the proof.
\end{proof}

\begin{lemma}\label{lemma_no_change_point}
Let the data be a sequence of adjacency tensors  $\{\mathbf{A}(t)\}_{t \in \mathbb{N}^*} \subset \mathbb{R}^{n \times n \times L}$ as defined in \Cref{def-umrdpg-dynamic} satisfying \Cref{ass_X_Y} and \Cref{ass_no_change_point}. Let $\widetilde{D}_{\cdot, \cdot}(\cdot)$ be defined in and \eqref{eq-def-D-tilde}, with the kernel function $\mathcal{K}(\cdot)$ satisfying \Cref{kernel_function_ass}. 

For any $s, t \in \mathbb{N}$, $1 \leq s< t $, it holds that 
\[
\sup_{z \in [0, 1]^L} \left| \widetilde{D}_{s, t}(z)\right| \leq  C h^{-L-1}  \sqrt{\frac{\{ 1 + \log(L)\} L}{n}} \left( \frac{1}{\sqrt{s}} + \frac{1}{\sqrt{t-s}} \right),
\]
where $C > 0$ is an absolute constant.
\end{lemma}

\begin{proof}[\textbf{Proof of \Cref{lemma_no_change_point}.}]
Note that
\begin{align}\label{no_change_point_0}
    & \sup_{z \in [0, 1]^L} \left| \widetilde{D}_{s, t}(z) \right|   \nonumber\\
     \leq &  \sup_{z \in [0, 1]^L} \frac{C^{\prime}}{n^2}\left|  \sum_{k \in [2n - 1]}\widetilde{S}_k
      \left[ \mathbb{E} \left\{h^{-L}  \mathcal{K} \left( \frac{z-  \mathbf{P}^{0, s}_{\mathcal{S}_k, :}    }{h} \right) \right\}
    -  \mathbb{E} \left\{h^{-L} \mathcal{K} \left( \frac{z-  \mathbf{P}^{s, t}_{\mathcal{S}_k, :}   }{h} \right)\right\} \right]\right|  \nonumber \\
\leq  &  C^{\prime} C_{\mathrm{Lip}} \sum_{k \in [2n - 1]}\frac{\widetilde{S}_k}{n^2}  h^{-L-1}  \mathbb{E} \left\|   \mathbf{P}^{0,s}_{\mathcal{S}_k, :}     - \mathbf{P}^{s, t}_{\mathcal{S}_k, :}  \right\| \nonumber \\
\leq &  C^{\prime}C_{\mathrm{Lip}} \sum_{k \in [2n - 1]}\frac{\widetilde{S}_k}{n^2}  h^{-L-1} \sqrt{L} \mathbb{E} \max_{l \in [L]}\left|   \mathbf{P}^{0, s}_{\mathcal{S}_k, l}   - \mathbf{P}^{s, t}_{\mathcal{S}_k, l}  \right|,
\end{align}
with an absolute constant $C^{\prime}> 0$, where the first inequality follows from \Cref{lemma_S_u}, the second inequality follows from \Cref{kernel_function_ass} and the third inequality follows from the fact that $\|\mathbf{v}\| \leq \sqrt{p} \| \mathbf{v}\|_{\infty}$, for any $\mathbf{v} \in \R^p$.

Now we consider the term $\mathbb{E} \max_{l \in [L]}\left|   \mathbf{P}^{0, s}_{\mathcal{S}_k, l}   - \mathbf{P}^{s, t}_{\mathcal{S}_k, l}  \right|$.  Note that   $\left\{ \mathbf{P}_{i, j, l}^u, \, u \in [s], \, (i, j) \in \mathcal{S}_k\right\} \subset [0, 1]$ are mutually independent. For any sequences 
\[
    \left\{ x^{u, i, j}, \, u \in [s], \, (i, j) \in \mathcal{S}_k\right\}, \,\left\{ y^{u, i, j}, \, u \in [s], \, (i, j) \in \mathcal{S}_k\right\} \subset [0, 1],
\]
we have that for any $l \in [L]$,
\begin{align}
	& \left|\left( \frac{1}{s \widetilde{S}_k}\sum_{(i, j) \in \mathcal{S}_k} \sum_{u = 1}^{s}x^{u, i, j}  -  \mathbf{P}^{s, t}_{\mathcal{S}_k, l}\right) -  \left( \frac{1}{s \widetilde{S}_k}\sum_{(i, j) \in \mathcal{S}_k} \sum_{u = 1}^{s}y^{u, i, j} -   \mathbf{P}^{s, t}_{\mathcal{S}_k, l} \right)\right|\nonumber \\
	\leq  &  \frac{1}{s \widetilde{S}_k}\sum_{(i, j) \in \mathcal{S}_k} \sum_{u = 1}^{s} \left|x^{u, i, j} -  y^{u, i, j}  \right|. \nonumber 
\end{align}

Then by concentration inequality of Lipschitz functions with respect to $l_1$-metric, we have for any $\delta_1 >0$
\begin{align}
   & \mathbb{P}\left\{\left. \left| \left( \mathbf{P}^{0, s}_{\mathcal{S}_k, l}   -  \mathbf{P}^{s, t}_{\mathcal{S}_k, l}\right)  -   \mathbb{E}  \left\{  \mathbf{P}^{0, s}_{\mathcal{S}_k, l}   -  \mathbf{P}^{s, t}_{\mathcal{S}_k, l} \bigg| \left\{ \mathbf{P}_{i, j, l}(u) \right\}_{\substack{ s+1 \leq u \leq {
 t} \\ (i, j) \in \mathcal{S}_k}} \right\}  \right| \geq \delta_1 \right|  \left\{ \mathbf{P}_{i, j, l}(u) \right\}_{\substack{ s+1 \leq u \leq {
 t} \\ (i, j) \in \mathcal{S}_k}} \right\}  \nonumber\\
 & \hspace{0.5cm} \leq 2 \exp \left\{ -2\delta_1^2s \widetilde{S}_k \right\}. \nonumber
\end{align}
Then we have 
\begin{align}
   & \mathbb{P}\left\{ \left| \left( \mathbf{P}^{0, s}_{\mathcal{S}_k, l}   -  \mathbf{P}^{s, t}_{\mathcal{S}_k, l}\right)  -   \mathbb{E}  \left\{  \mathbf{P}^{0, s}_{\mathcal{S}_k, l}   -  \mathbf{P}^{s, t}_{\mathcal{S}_k, l} \bigg| \left\{ \mathbf{P}_{i, j, l}(u) \right\}_{\substack{ s+1 \leq u \leq {
 t} \\ (i, j) \in \mathcal{S}_k}} \right\}  \right| \geq \delta_1  \right\}
 \nonumber \\ 
   = &\E \left\{  \mathbb{P}\left\{\left. \left| \left( \mathbf{P}^{0, s}_{\mathcal{S}_k, l}   -  \mathbf{P}^{s, t}_{\mathcal{S}_k, l}\right)  -   \mathbb{E}  \left\{  \mathbf{P}^{0, s}_{\mathcal{S}_k, l}   -  \mathbf{P}^{s, t}_{\mathcal{S}_k, l} \bigg| \left\{ \mathbf{P}_{i, j, l}(u) \right\}_{\substack{ s+1 \leq u \leq {
 t} \\ (i, j) \in \mathcal{S}_k}} \right\}  \right| \geq \delta_1 \right|  \left\{ \mathbf{P}_{i, j, l}(u) \right\}_{\substack{ s+1 \leq u \leq {
 t} \\ (i, j) \in \mathcal{S}_k}} \right\}  \right\} \nonumber \\
  \leq & 2 \exp \left\{ -2\delta_1^2s \widetilde{S}_k \right\}. \nonumber
\end{align}
Let $\delta_1 =\delta_2 / \sqrt{s \widetilde{S}_k }$ with $\delta_2 > 0$ to be specified and by a union bound argument, we have
\begin{align}\label{no_change_point_mean_1}
   & \mathbb{P}\left\{ \max_{ l \in [L]} \left| \left( \mathbf{P}^{0, s}_{\mathcal{S}_k, l}   -  \mathbf{P}^{s, t}_{\mathcal{S}_k, l}\right)  -   \mathbb{E}  \left\{  \mathbf{P}^{0, s}_{\mathcal{S}_k, l}   -  \mathbf{P}^{s, t}_{\mathcal{S}_k, l} \bigg| \left\{ \mathbf{P}_{i, j, l}(u) \right\}_{\substack{ s+1 \leq u \leq {
 t} \\ (i, j) \in \mathcal{S}_k}} \right\}  \right| \geq  (s \widetilde{S}_k )^{-1/2}\delta_2     \right\} \nonumber \\
 & \hspace{0.5cm} \leq 2L \exp \left\{ -2\delta_2^2\right\}. 
\end{align}
Now we consider the term $ \Bigg| \mathbb{E}  \left\{  \mathbf{P}^{0, s}_{\mathcal{S}_k, l}   -  \mathbf{P}^{s, t}_{\mathcal{S}_k, l} \bigg| \left\{ \mathbf{P}_{i, j, l}(u) \right\}_{\substack{ s+1 \leq u \leq {
 t} \\ (i, j) \in \mathcal{S}_k}} \right\} \Bigg| $. 
 Let
\begin{align}
\mathcal{A}_1 = \left\{ \Bigg| \mathbb{E}  \left\{  \mathbf{P}^{0, s}_{\mathcal{S}_k, l}   -  \mathbf{P}^{s, t}_{\mathcal{S}_k, l} \bigg| \left\{ \mathbf{P}_{i, j, l}(u) \right\}_{\substack{ s+1 \leq u \leq {
 t} \\ (i, j) \in \mathcal{S}_k}}  \right\} \Bigg|   \leq  ((t-s)\mathcal{S}_k)^{-1/2} \delta_2  \right\}, \nonumber 
 \end{align}
 and 
 \begin{align}\label{no_change_point_mean_2}
\mathcal{A}_2 = \left\{ \max_{l \in [L]}\Bigg| \mathbb{E}  \left\{  \mathbf{P}^{0, s}_{\mathcal{S}_k, l}   -  \mathbf{P}^{s, t}_{\mathcal{S}_k, l} \bigg| \left\{ \mathbf{P}_{i, j, l}(u) \right\}_{\substack{ s+1 \leq u \leq {
 t} \\ (i, j) \in \mathcal{S}_k}}  \right\} \Bigg|   \leq  ((t-s)\mathcal{S}_k)^{-1/2} \delta_2  \right\}.
 \end{align}
 Note that
 \begin{align}
   &  \mathbb{E}  \left\{  \mathbf{P}^{0, s}_{\mathcal{S}_k, l}   -  \mathbf{P}^{s, t}_{\mathcal{S}_k, l} \bigg| \left\{ \mathbf{P}_{i, j, l}(u) \right\}_{\substack{ s+1 \leq u \leq {
 t} \\ (i, j) \in \mathcal{S}_k}} \right\}  \nonumber \\
 = & \mathbb{E}  \left\{
 \frac{1}{s \widetilde{S}_k}
 \sum_{(i, j) \in \mathcal{S}_k}
 \sum_{u = 1}^{s} \mathbf{P}_{i, j, l}(u)
 - \frac{1}{(t-s)\widetilde{S}_k}
 \sum_{(i, j) \in \mathcal{S}_k}
 \sum_{u = s+1}^{t}
 \mathbf{P}_{i, j, l}(u) \Bigg| \left\{ \mathbf{P}_{i, j, l}(u) \right\}_{\substack{ s+1 \leq u \leq  t \\ (i, j) \in \mathcal{S}_k }}
 \right\}  \nonumber \\
 = & \mathbb{E}  \left\{ \mathbf{P}_{1,2,l}(u) \right\} - \frac{1}{(t-s)\widetilde{S}_k}\sum_{(i, j) \in \mathcal{S}_k} \sum_{u = s+1}^{t}\mathbf{P}_{i, j, l}(u).   \nonumber 
 \end{align}
 Thus we have 
 \[
   \mathbb{E}  \left\{ \mathbb{E}  \left\{  \mathbf{P}^{0, s}_{\mathcal{S}_k, l}   -  \mathbf{P}^{s, t}_{\mathcal{S}_k, l} \bigg| \left\{ \mathbf{P}_{i, j, l}(u) \right\}_{\substack{ s+1 \leq u \leq {
 t} \\ (i, j) \in \mathcal{S}_k}} \right\} \right\}  =0,
 \]
an by Hoeffding's inequality for general bounded random variables \citep[e.g. Theorem 2.2.6 in][]{vershynin2018high}, we have that for any $\delta_2 > 0$,
\[
    \P \left\{\mathcal{A}_1 \right\} \geq 1 - 2\exp(-2\delta_2^2) .
\]
Then by a union bound argument, we have
\begin{align}\label{no_change_point_mean_3}
\P \left\{\mathcal{A}_2 \right\} \geq 1 - 2L \exp(-2\delta_2^2).
\end{align}
Combining \eqref{no_change_point_mean_1}, \eqref{no_change_point_mean_2} and \eqref{no_change_point_mean_3}, we have 
\begin{align}
   &  \mathbb{P}\left\{ \max_{l \in [L]} \left|  \mathbf{P}^{0, s}_{\mathcal{S}_k, l}   -  \mathbf{P}^{s, t}_{\mathcal{S}_k, l}   \right| \geq  \delta_2  \widetilde{S}_k^{-1/2} \left(\frac{1}{\sqrt{s}} +\frac{1}{\sqrt{t-s}}  \right)   \right\} \nonumber\\
   \leq &   2L \exp \left\{ -2\delta_2^2\right\}\P \left\{ \mathcal{A}_2\right\}  + 1- \P \left\{ \mathcal{A}_2 \right\} \leq 4  L\exp \left\{ -2\delta_2^2\right\}.\nonumber
\end{align}
We are ready to upper bound the term $\mathbb{E} \max_{l \in [L]}\left|   \mathbf{P}^{0, s}_{\mathcal{S}_k, l}   - \mathbf{P}^{s, t}_{\mathcal{S}_k, l}  \right|$
\begin{align}
    & \mathbb{E}\left\{\frac{\widetilde{S}_k^{1/2}\{ s^{-1/2} + (t-s)^{-1/2}\}^{-1}} {\sqrt{1 + \log(L)}}  \max_{l \in [L]}\left|  \mathbf{P}^{0, s}_{\mathcal{S}_k, l}   -  \mathbf{P}^{s, t}_{\mathcal{S}_k, l}  \right| \right\}   \nonumber\\
    = &  \int_{0}^{\infty} \mathbb{P}\left\{\frac{\widetilde{S}_k^{1/2}  \{ s^{-1/2} + (t-s)^{-1/2}\}^{-1} }{\sqrt{1+ \log(L)}}\max_{l \in [L]}\left|  \mathbf{P}^{0, s}_{\mathcal{S}_k, l}   - \mathbf{P}^{s, t}_{\mathcal{S}_k, l}  \right|\geq \delta   \right\} \, \mathrm{d}\delta \nonumber \\
    \leq & \int_{0}^{1} 1     d\delta +  \int_{1}^{\infty}\mathbb{P}\left\{ \frac{ \widetilde{S}_k^{1/2} \{ s^{-1/2} + (t-s)^{-1/2}\}^{-1} }{\sqrt{1+ \log(L)}} \left|  \mathbf{P}^{0, s}_{\mathcal{S}_k, l}   - \mathbf{P}^{s, t}_{\mathcal{S}_k, l}  \right|\geq \delta   \right\}  d\delta\nonumber \\
     \leq & 1  +  \int_{1}^{\infty}  2L\exp \left\{- \delta^2 \big(1 +\log (L)\big) \right\}  \, \mathrm{d}\delta \leq 1  +  \int_{1}^{\infty}  2\exp \left\{- \delta^2\right\} \, \mathrm{d}\delta  \leq C_1, \nonumber 
\end{align}
where $C_1 > 0$ is an absolute constant. Thus 
\begin{align}\label{no_change_point_1}
     \mathbb{E}\left\{ \max_{l \in [L]}\left|  \mathbf{P}^{s, t}_{\mathcal{S}_k, l}   - \mathbf{P}^{t, e}_{\mathcal{S}_k, l}   \right| \right\} \leq C_1 \sqrt{\frac{1 + \log(L)}{\widetilde{S}_k}}\left( \frac{1}{\sqrt{s}} + \frac{1}{\sqrt{t-s}} \right).
\end{align}
Then combining \eqref{no_change_point_0} and  \eqref{no_change_point_1}, we have  with absolute constants $C_2, C_3 > 0$ that
\begin{align}
    \sup_{z \in [0, 1]^L} \left| \widetilde{D}_{s, t}(z) \right| \leq  &  C_2 \sum_{k \in [2n - 1]}\frac{\sqrt{\widetilde{S}_k}}{n^2}  h^{-L-1} \sqrt{\{ 1 + \log(L)\} L}\left( \frac{1}{\sqrt{s}} + \frac{1}{\sqrt{t-s}} \right)   \nonumber\\
     \leq  & C_3 h^{-L-1}  \sqrt{\frac{ \{ 1 + \log(L)\} L}{n}} \left( \frac{1}{\sqrt{s}} + \frac{1}{\sqrt{t-s}} \right) , \nonumber
\end{align}
where the final inequality follows from  \Cref{lemma_S_u}. We conclude the proof.
\end{proof}

\begin{lemma}\label{lemma_one_change_point}

Let the data be a sequence of adjacency tensors  $\{\mathbf{A}(t)\}_{t \in \mathbb{N}^*} \subset \mathbb{R}^{n \times n \times L}$ as defined in \Cref{def-umrdpg-dynamic} satisfying Assumptions~\ref{ass_X_Y}, \ref{ass_change_point} and \ref{snr_ass_change_point}.  Let  $\widetilde{D}_{\cdot, \cdot}(\cdot)$ be defined in and \eqref{eq-def-D-tilde}, with the kernel function $\mathcal{K}(\cdot)$ satisfying \Cref{kernel_function_ass}.  For $\alpha \in (0, 1)$, let 
\begin{align}\label{tilde_Delta}
   \widetilde{\Delta} = \Delta+ C_\epsilon h^{-2L-2} \frac{(L^2 \vee d) \log \{(n \vee \Delta) /\alpha \}}{\kappa^2 n},
\end{align}
with $C_{\epsilon} > 0$ being an absolute constant.
 It holds that
\[
  \sup_{z \in [0, 1]^L} \left| \widetilde{D}_{\Delta, \widetilde{\Delta}} (z) \right|> \frac{C_{\mathrm{SNR}} h^{-L-1}}{4}  \sqrt{\frac{(L^2 \vee d)\log\{ (n \vee \Delta)/\alpha \} }{n}} \left( \sqrt{\frac{1}{\Delta}}   + \sqrt{\frac{1}{\widetilde{\Delta} - \Delta}}\right),
\]
\end{lemma}

\begin{proof}[\textbf{Proof of \Cref{lemma_one_change_point}.}]
Recall that for any $t \in \mathbb{N}^*$, $G_t(\cdot)$ and $\widetilde{G}_t(\cdot)$ are defined in \Cref{ass_change_point}.  Note that 
\begin{align}\label{one_change_point_0}
    & \sup_{z \in [0, 1]^L} \left| \widetilde{D}_{\Delta, \widetilde{\Delta}} (z) \right| \nonumber \\
    \geq &  C^{\prime}\sup_{z \in [0, 1]^L} \left| \sum_{k \in [2n - 1]}\frac{\widetilde{S}_k}{n^2}\left[ \mathbb{E} \left\{h^{-L}  \mathcal{K} \left( \frac{z-  \mathbf{P}^{0, \Delta}_{\mathcal{S}_k, :}    }{h} \right) \right\} -  \mathbb{E} \left\{h^{-L} \mathcal{K} \left( \frac{z-  \mathbf{P}^{\Delta, \widetilde{\Delta}}_{\mathcal{S}_k, :}   }{h} \right)\right\} \right] \right| \nonumber \\
    \geq &  \sup_{z \in [0, 1]^L} C^{\prime} \left\{ \Bigg|\sum_{k \in [2n - 1]}\frac{\widetilde{S}_k}{n^2} \mathbb{E} \left\{h^{-L}  \mathcal{K} \left( \frac{z-  \mathbf{P}_{1, 2, :}(\Delta)    }{h} \right) \right\} - \mathbb{E} \left\{h^{-L}  \mathcal{K} \left( \frac{z-  \mathbf{P}_{1, 2, :}(\Delta+1)    }{h} \right) \right\} \Bigg| \right.\nonumber \\
    & \hspace{0.2cm}-   \sum_{k \in [2n - 1]}\frac{\widetilde{S}_p}{n^2}\left[ \left| \mathbb{E} \left\{h^{-L}  \mathcal{K} \left( \frac{z-  \mathbf{P}^{0, \Delta}_{\mathcal{S}_k, :}    }{h} \right) \right\} - h^{-L}  \mathcal{K} \left( \frac{z-  \E\left\{ \mathbf{P}_{1, 2, :} (\Delta)\right\}  }{h} \right) \right|\right. \nonumber \\
    &  \hspace{0.2cm}+  \left| \mathbb{E} \left\{h^{-L} \mathcal{K} \left( \frac{z-  \mathbf{P}^{\Delta, \widetilde{\Delta}}_{\mathcal{S}_k, :}   }{h} \right)\right\} - h^{-L}  \mathcal{K} \left( \frac{z-  \E\left\{ \mathbf{P}_{1, 2, :}(\Delta+1)  \right\}  }{h} \right) \right| \nonumber\\
    & \hspace{0.2cm}+    \left| \mathbb{E} \left\{h^{-L}  \mathcal{K} \left( \frac{z-   \mathbf{P}_{1, 2, :}(\Delta)    }{h} \right) \right\} -h^{-L}  \mathcal{K} \left( \frac{z-  \E\left\{ \mathbf{P}_{1, 2, :}(\Delta) \right\}  }{h} \right)  \right| \nonumber \\
    & \hspace{0.2cm} \left.+   \left. \left| \mathbb{E} \left\{h^{-L}  \mathcal{K} \left( \frac{z-  \mathbf{P}_{1, 2, :}(\Delta + 1)     }{h} \right) \right\} -h^{-L}  \mathcal{K} \left( \frac{z-  \E\left\{ \mathbf{P}_{1, 2, :}(\Delta+1) \right\}  }{h} \right)  \right| \right] \right\} \nonumber \\
    \geq & \sup_{z\in [0, 1]^L} C^{\prime}\Bigg \{ \kappa(z) -\sum_{k \in [2n - 1]} \frac{\widetilde{S}_k}{n^2} \bigg[ C_{\mathrm{Lip}} h^{-L-1}  \mathbb{E} \left\|   \mathbf{P}^{0, \Delta}_{\mathcal{S}_k, :}   - \E \left\{ \mathbf{P}_{1, 2, :}(\Delta+1)  \right\} \right\| \nonumber \\
     & \hspace{0.2cm} + C_{\mathrm{Lip}} h^{-L-1} \mathbb{E} \left\|   \mathbf{P}^{ \Delta, \widetilde{\Delta}}_{\mathcal{S}_k, :}   - \E \left\{  \mathbf{P}_{1, 2, :} (\Delta+1)  \right\} \right\|  + \left| G_{\Delta}(z) - \widetilde{G}_{\Delta}(z)\right|  + \left| G_{\Delta + 1}(z) - \widetilde{G}_{\Delta + 1}(z)\right| \bigg] \Bigg\}\nonumber \\
    \geq &C^{\prime}\Bigg \{  \sup_{z\in [0, 1]^L}  \kappa(z) -\sum_{k \in [2n - 1]}\frac{ \widetilde{S}_k}{n^2} \bigg[ C_{\mathrm{Lip}} h^{-L-1} \bigg( \mathbb{E} \left\|   \mathbf{P}^{0, \Delta}_{\mathcal{S}_k, :}   - \E \left\{ \mathbf{P}_{1, 2, :}(\Delta)  \right\} \right\|   \nonumber \\ 
    & \hspace{0.2cm}  + \mathbb{E} \left\|   \mathbf{P}^{ \Delta, \widetilde{\Delta}}_{\mathcal{S}_k, :}   - \E \left\{  \mathbf{P}_{1, 2, :}(\Delta+1)  \right\} \right\| \bigg)  \bigg]  
      -  \sup_{z' \in [0, 1]^L}\left| G_{\Delta}(z') - \widetilde{G}_{\Delta}(z')\right|  \nonumber\\ 
      &\hspace{0.2cm} -  \sup_{z' \in [0, 1]^L}\left| G_{\Delta + 1}(z') - \widetilde{G}_{\Delta + 1}(z')\right|\Bigg \}  \nonumber \\
        \geq  & C^{\prime}\Bigg \{  \kappa -\sum_{k \in [2n - 1]}\frac{C_{\mathrm{Lip}} h^{-L-1}\widetilde{S}_k}{n^2}   \sqrt{L}\bigg( \mathbb{E} \max_{l \in [L]}\left|   \mathbf{P}^{0, \Delta}_{\mathcal{S}_k, l}   - \E \left\{ \mathbf{P}_{1, 2, l}(\Delta)  \right\}  \right|  
        \nonumber \\
       & \hspace{0.2cm}   + \mathbb{E} \max_{l \in [L]} \left|   \mathbf{P}^{ \Delta, \widetilde{\Delta}}_{\mathcal{S}_k, l}   - \E \left\{ \mathbf{P}_{1, 2, l}(\Delta+1) \right\}  \right| \bigg) -   \sup_{z' \in [0, 1]^L}\left| G_{\Delta}(z') - \widetilde{G}_{\Delta}(z')\right| \nonumber \\
      & \hspace{0.2cm}   -\sup_{z' \in [0, 1]^L}\left| G_{\Delta + 1}(z') - \widetilde{G}_{\Delta + 1}(z')\right| \Bigg\}  \nonumber \\
    \geq  &  C^{\prime}\Bigg \{  \frac{\kappa}{2} -\sum_{k \in [2n - 1]}\frac{C_{\mathrm{Lip}} h^{-L-1}\widetilde{S}_k }{n^2}   \sqrt{L}\bigg( \mathbb{E} \max_{l \in [L]}\left|   \mathbf{P}^{0, \Delta}_{\mathcal{S}_k, l}   - \E \left\{ \mathbf{P}_{1, 2, l} (\Delta) \right\}  \right| \nonumber \\
    & \hspace{0.2cm} 
    + \mathbb{E} \max_{l \in [L]} \left|   \mathbf{P}^{ \Delta, \widetilde{\Delta}}_{\mathcal{S}_k, l}   - \E \left\{ \mathbf{P}_{1, 2, l}(\Delta+1)  \right\}  \right| \bigg) \Bigg\}, 
\end{align}
with an absolute constant $C^{\prime}> 0$, where
\begin{itemize}
    \item the first and fourth inequalities follow from  \Cref{lemma_S_u}, 
    \item the third inequality follows from \Cref{kernel_function_ass},
    \item the fifth inequality follows from  the fact that for any $\mathbf{v} \in \R^p$,  $\|\mathbf{v}\| \leq \sqrt{p}\|\mathbf{v} \|_{\infty}$,
     \item and the final inequality follows from \Cref{ass_change_point}.
\end{itemize}

To further lower bound \eqref{one_change_point_0},
We are to deploy Talagrand's concentration inequality \citep[e.g.~Theorem 5.2.16 in][]{vershynin2018high} to upper bound the term $ \mathbb{E} \max_{l \in [L]}\left|   \mathbf{P}^{0, \Delta}_{\mathcal{S}_k, l}   - \E\left\{ \mathbf{P}_{1, 2, l} (\Delta)\right\} \right|  $.   Note that $\{ \mathbf{P}_{i, j, l}(u), u \in [\Delta], (i, j) \in \mathcal{S}_k\} \subset [0, 1]$ are mutually independent. For any sequences $\{ x^{u, i, j}, u \in [\Delta], (i, j) \in \mathcal{S}_k\}, \{ y^{u, i, j}, u \in [\Delta], (i, j) \in \mathcal{S}_k\} \subset [0, 1]$,
we have 
\begin{align}
	& \left|\left( \frac{1}{\widetilde{S}_k\Delta}\sum_{(i, j) \in \mathcal{S}_k} \sum_{u = 1}^{\Delta}x^{u, i, j}  - \E \left\{ \mathbf{P}_{1, 2, l}(\Delta)  \right\} \right) -  \left( \frac{1}{\widetilde{S}_k\Delta}\sum_{(i, j) \in \mathcal{S}_k} \sum_{u = 1}^{\Delta}y^{u, i, j} - \E \left\{ \mathbf{P}_{1, 2, l} (\Delta) \right\} \right)\right|\nonumber \\
	\leq  &  \frac{1}{\widetilde{S}_k\Delta}\sum_{(i, j) \in \mathcal{S}_k} \sum_{u = 1}^{\Delta} \left|x^{u, i, j} -  y^{u, i, j}  \right|. \nonumber
\end{align}
Then by concentration inequality of Lipschitz functions with respect to $l_1$-metric, we have that for any $\delta_1 > 0$ to be specified later
\begin{align}
    \mathbb{P}\left\{ \left| \left( \mathbf{P}^{0, \Delta}_{\mathcal{S}_k, l}   - \E \left\{ \mathbf{P}_{1, 2, l}(\Delta)  \right\} \right)  -   \mathbb{E}  \left\{  \mathbf{P}^{0, \Delta}_{\mathcal{S}_k, l}   - \E \left\{ \mathbf{P}_{1, 2, l}(\Delta)  \right\}  \right\}  \right| \geq \delta_1  \right\}  \leq 2 \exp \left\{ -2\delta_1^2\widetilde{S}_k \Delta\right\}, \nonumber
\end{align}
 Letting $\delta_2 =\delta_1 / \sqrt{\widetilde{S}_k \Delta}$ and since $ \mathbb{E} \{  \mathbf{P}^{0, \Delta}_{\mathcal{S}_k, l}   - \E \{ \mathbf{P}_{1, 2, l}(\Delta)\} \} = 0$, we have
\begin{align}
    \mathbb{P}\left\{  \sqrt{\widetilde{S}_k \Delta} \left|  \mathbf{P}^{0, \Delta}_{\mathcal{S}_p, l}   -\E \left\{ \mathbf{P}_{1, 2, l}(\Delta)  \right\}   \right| \geq  \delta_2  \right\}  \leq 2 \exp \left\{ -2\delta_2^2\right\}.   \nonumber
\end{align}
Consequently, we have
\begin{align}
    & \mathbb{E}\left\{\sqrt{\frac{\widetilde{S}_k \Delta}{1 + \log(L)}} \max_{l \in [L]}\left|  \mathbf{P}^{0, \Delta}_{\mathcal{S}_k, l}   - \E \left\{ \mathbf{P}_{1, 2, l}(\Delta)  \right\}  \right| \right\}   \nonumber \\
    = &   \int_{0}^{\infty} \mathbb{P}\left\{\sqrt{\frac{\widetilde{S}_k \Delta}{1+ \log(L)}} \max_{l \in [L]}\left|  \mathbf{P}^{0, \Delta}_{\mathcal{S}_k, l}   -\E \left\{ \mathbf{P}_{1, 2, l}(\Delta)  \right\}  \right|\geq \delta   \right\}  \,\mathrm{d}\delta \nonumber \\
    \leq & \int_{0}^{1} 1   \,  \mathrm{d}\delta +  \int_{1}^{\infty} \sum_{l \in [L]}\mathbb{P}\left\{ \sqrt{\frac{\widetilde{S}_k \Delta}{1+ \log(L)}} \left|  \mathbf{P}^{0, \Delta}_{\mathcal{S}_k, l}   - \E \left\{ \mathbf{P}_{1, 2, l}(\Delta)  \right\}  \right|\geq \delta   \right\} \, \mathrm{d}\delta\nonumber \\
    \leq & 1  +  2L \int_{1}^{\infty} \exp \left\{- \delta^2(1 +\log (L)) \right\}  \, \mathrm{d}\delta \leq 1  +  2\int_{1}^{\infty}  \exp \left\{- \delta^2\right\}  \, \mathrm{d}\delta  \leq  C_1, \nonumber 
\end{align}
where $C_1 > 0$ is an absolute constant.  We thus have that 
\begin{align}\label{one_change_point_1}
     \mathbb{E}\left\{ \max_{l \in [L]}\left|  \mathbf{P}^{0, \Delta}_{\mathcal{S}_k, l}   -\E \left\{ \mathbf{P}_{1, 2, l}(\Delta)  \right\}   \right| \right\} \leq C_1 \sqrt{\frac{1 + \log(L)}{\widetilde{S}_k \Delta}}.
\end{align}
Following almost identical arguments, we have that
\begin{align}\label{one_change_point_2}
     \mathbb{E} \left\{ \max_{l \in [L]} \left|   \mathbf{P}^{ \Delta, \widetilde{\Delta}}_{\mathcal{S}_k, l}   - \E \left\{ \mathbf{P}_{1, 2, l}(\Delta + 1)  \right\} \right| \right\}\leq C_1 \sqrt{\frac{1 + \log (L)}{\widetilde{S}_k\left(\widetilde{\Delta} - \Delta \right) }}.
\end{align}
Combining \eqref{one_change_point_0}, \eqref{one_change_point_1} and \eqref{one_change_point_2}, we have with $C_2 > 0$ being an absolute constant that
\begin{align}\label{one_change_point_3}
     \sup_{z \in [0, 1]^L} \left|  \widetilde{D}_{\Delta, \widetilde{\Delta}}(z) \right| \geq  &C^{\prime}\Bigg \{   \frac{\kappa}{2} -\sum_{k \in [2n - 1]}\frac{\sqrt{\widetilde{S}_k}}{n^2}  C_2 h^{-L-1} \{ 1 + \log(L)\} \sqrt{L}\left(\sqrt{\frac{1}{\Delta}}   + \sqrt{\frac{1}{\widetilde{\Delta} - \Delta}}\right)   \Bigg\}\nonumber\\
    \geq & C^{\prime}\Bigg \{ \frac{\kappa}{2} - \frac{ C_2 h^{-L-1} \{ 1 + \log(L)\} \sqrt{L}}{\sqrt{n}}\left(\sqrt{\frac{1}{\Delta}}   + \sqrt{\frac{1}{\widetilde{\Delta} - \Delta}}\right) \Bigg\},
\end{align}    
where the final inequality due to the fact that $\max_{k \in [2n -1]} |\widetilde{S}_k| \leq n$ and \Cref{lemma_S_u}.

By the definition of $\widetilde{\Delta}$ in \eqref{tilde_Delta}, we have that
\begin{align}\label{one_change_point_4}
    \kappa  = \sqrt{C_{\epsilon}} h^{-L-1}  \sqrt{\frac{(L^2 \vee d)\log\{ (n \vee \Delta)/\alpha \} }{(\widetilde{\Delta} - \Delta) n }}.
\end{align} 

Combining \eqref{one_change_point_3} and \eqref{one_change_point_4}, due to \Cref{snr_ass_change_point}, we complete the proof. 
\end{proof}

\begin{lemma}\label{lemm_eigen}
Let the data be a sequence of adjacency tensors  $\{\mathbf{A}(t)\}_{t \in \mathbb{N}^*} \subset \mathbb{R}^{n \times n \times L}$ as defined in \Cref{def-umrdpg-dynamic} satisfying \Cref{ass_X_Y}. Let $\theta_X$ be defined \Cref{ass_X_Y} and denote
 \[
    \widetilde{ \widehat{\mathbf{P}}}^{s, t} = \mathbf{S} \times_1  \bigg((t-s)^{-1}\sum_{u = s+1}^t X(u)\bigg)  \times_2 \bigg((t-s)^{-1}\sum_{u = s+1}^t X(u)\bigg) \times_3 Q, \quad s, t \in \mathbb{N}, \, 1 \leq s < t.
\]
Let 
\begin{align}\label{lemma_12_event}
    \mathcal{A}= \Bigg\{ & 
\left( \mathrm{rank} \left(\mathcal{M}_1\left( \widetilde{ \widehat{\mathbf{P}}}^{s, t} \right)\right) , \mathrm{rank} \left(\mathcal{M}_2\left( \widetilde{ \widehat{\mathbf{P}}}^{s, t}\right)\right) , \mathrm{rank} \left(\mathcal{M}_3\left( \widetilde{ \widehat{\mathbf{P}}}^{s, t} \right)\right) \right)  = \left(d, d, m \right),
\nonumber \\
     & \sigma_1\left( \mathcal{M}_1\left( \widetilde{ \widehat{\mathbf{P}}}^{s, t}\right)\right) \leq C_2 n \sqrt{d} \|\theta_{X}\|^2 \sigma_1(Q),
\nonumber \\
     & 
   \sigma_1\left( \mathcal{M}_2\left( \widetilde{ \widehat{\mathbf{P}}}^{s, t}\right)\right) \leq C_2 n \sqrt{d }\|\theta_{X}\|^2 \sigma_1(Q),
\nonumber \\
   &  \sigma_1\left( \mathcal{M}_3\left( \widetilde{ \widehat{\mathbf{P}}}^{s, t}\right)\right) \leq C_2 n \|\theta_{X}\|^2 \sigma_1(Q),
\nonumber \\
&\sigma_1^2\left( \mathcal{M}_1\left( \widetilde{ \widehat{\mathbf{P}}}^{s, t}\right)\right)  - \sigma_2^2\left( \mathcal{M}_1\left( \widetilde{ \widehat{\mathbf{P}}}^{s, t}\right)\right) \geq C_3 dn^2 \|\theta_{X}\|^4 \sigma_1(Q)^2,
\nonumber \\
&\sigma_1^2\left( \mathcal{M}_2\left( \widetilde{ \widehat{\mathbf{P}}}^{s, t}\right)\right)  - \sigma_2^2\left( \mathcal{M}_2\left( \widetilde{ \widehat{\mathbf{P}}}^{s, t}\right)\right) \geq C_3 dn^2 \|\theta_{X}\|^4 \sigma_1(Q)^2,
\nonumber \\ 
& \mbox{and } \sigma_1^2\left( \mathcal{M}_3\left( \widetilde{ \widehat{\mathbf{P}}}^{s, t}\right)\right)  - \sigma_2^2\left( \mathcal{M}_3\left( \widetilde{ \widehat{\mathbf{P}}}^{s, t}\right)\right) \geq C_3 n^2 \|\theta_{X}\|^4 \sigma_1(Q)^2 \Bigg\},
\end{align}
 with absolute constants $C_2 >C_3 > 0$.
For any $0\leq s<t$ satisfying that there is no change point in $[s+1, t)$, 
it holds for any 
\begin{align}
    \exp\big\{- C_1  n \mu_{X, 1}^2 \big\}   \nonumber \leq \delta < 1,
\end{align} 
with an absolute constant $C_1 >0$,  that
\[
  \P \{ \mathcal{A} \} \geq 1 - \delta.
\]

\end{lemma}

\begin{proof}[\textbf{Proof of \Cref{lemm_eigen}.}]
Let
\[
    \overline{X}^{s, t} = \frac{1}{t-s} \sum_{u = s+1}^t X(u) \in \mathbb{R}^{n \times d} 
\]
This proof consists of two steps. In \textbf{Step 1}, we discuss
the singular values of $\overline{X}^{s, t}$. In \textbf{Step 2}, we complete the proof using \Cref{fix_lemma1} and results obtained from \textbf{Step 1}.

\medskip
\noindent \textbf{Step 1.}
Define $\Sigma^{s, t}_X = \mathbb{E}[ \{(\overline{X}^{s, t})^1\}^{\top} (\overline{X}^{s, t})^{1} ]$, where $(\overline{X}^{s, t})^{1}$  denote the first rows of $\overline{X}^{s, t}$.  Note that, for $(i, j) \in [d] \times [d]$, 
\begin{align}
    & \left( \Sigma^{s, t}_{X} \right)_{i, j}  = \left[ \mathbb{E} \left\{ \left( (t-s)^{-1}\sum_{u = s+1}^t X^{1}(u)  \right)^{\top} \left( (t-s)^{-1}\sum_{u = s+1}^t X^{1}(u)  \right) \right\} \right]_{i, j}  \nonumber\\
    = &  \mathbb{E} \left\{ \left( \frac{1}{t-s}\sum_{u = s+1}^t (X^{1}(u))_{i}  \right) \left( \frac{1}{t-s}\sum_{u = s+1}^t (X^{1}(u))_{j}  \right) \right\} \nonumber \\
    =&  \mathbb{E} \left\{ (t-s)^{-2}\sum_{u, u' = s+1}^t (X^{1}(u))_{i} (X^{1}(u'))_{j} \right\}  \nonumber \\
    = &  (t-s)^{-1} \left( \Sigma_{X} \right)_{i, j} + (t-s)^{-1}(t-s-1) (\theta_{X})_i (\theta_{X})_j. \nonumber 
\end{align}
We thus have that
\begin{align}\label{Sigma_st_X}
    \Sigma^{s, t}_{X} = (t-s)^{-1} \Sigma_X + (t-s)^{-1}(t-s-1)   \theta_{X} \left( \theta_{X} \right)^{\top}.
\end{align}
Since $\Sigma_X$ is positive definite by \Cref{ass_X_Y}$(a)$   and $\theta_{X} \left( \theta_{X} \right)^{\top}$ is positive semi-definite, we have that $\Sigma^{s, t}_{X}$ is positive definite by Weyl's inequality \citep{weyl1912asymptotische} and consequently, $\mathrm{rank}(\Sigma^{s, t}_{X}) = d$.  

By \eqref{Sigma_st_X} and Weyl's inequality \citep{weyl1912asymptotische}, we have that 
\begin{align}\label{theta_mu_1}
    \frac{1}{t-s} \mu_{X, d} + \frac{1}{2}\sigma_1   \left( \theta_{X} \theta_{X} ^{\top} \right)  \leq \sigma_1 (\Sigma^{s, t}_{X}) \leq \frac{1}{t-s} \mu_{X, 1} + \sigma_1   \left( \theta_{X} \theta_{X} ^{\top} \right). 
\end{align}
Combining \eqref{theta_mu_1} and \Cref{ass_X_Y}$(b)$, we have that with absolute constants $C_1 >  C_2 >0$
\[
   C_2 \|\theta_{X}\|^2 \leq \sigma_1 (\Sigma^{s, t}_{X}) \leq C_ 1\|\theta_{X}\|^2. 
\]
By \eqref{Sigma_st_X} and Weyl's inequality \citep{weyl1912asymptotische}, we also have that  
\[
     \sigma_2 ( \Sigma^{s, t}_{X}) \leq  (t-s)^{-1} \mu_{X, 1} +  (t-s)^{-1}(t-s-1) \sigma_2 \left(\theta_{X} \theta_{X}^{\top} \right) = (t-s)^{-1} \mu_{X, 1},
\]
where the equality follows from the fact that $\theta_{X} \theta_{X}^{\top}$ is a rank-1 matrix.  

Since $F$ is a sub-Gaussian distribution, we have $\mathbb{P}\{\mathcal{A}_1\} > 1 - \delta$ holds that for any $\delta > 0$ satisfying $d < \log(1/\delta)$ with
\begin{align}
    \mathcal{A}_1 &= \Big\{\big\|n^{-1} \left(  \overline{X}^{s, t} \right)^{\top} \overline{X}^{s, t} - \Sigma_X^{s, t}\big\| \leq C(t-s)^{-1}\sqrt{n^{-1} \log(1/\delta)} 
    \Big\}. \nonumber
\end{align}
where $C > 0$ is an absolute constant.  The aforementioned result follows  matrix Bernstein's inequality \citep[e.g.~Remark~5.40(1) in][]{vershynin2010introduction}.  To be specific, the constant $C_K$ in Remark 5.40 in \cite{vershynin2010introduction} is specified in Step 3 in the proof of Theorem 5.39 therein, satisfying that $C_K \asymp K^2$, where $K$ is the sub-Gaussian parameter of the rows of $\overline{X}^{s, t}$.  In our case, $K \asymp (t-s)^{-1/2}$.

In the event $\mathcal{A}_1$, with with an absolute constant $C_3 >0$, for any 
\begin{align}
   \delta \geq \exp\big\{- C_3  n \mu_{X, 1}^2,\big\},   \nonumber
\end{align}
it follows from Weyl's inequality \citep{weyl1912asymptotische} that
\begin{align}
    & \sigma_2 \left( \overline{X}^{s, t} \right) \leq \sqrt{n  \sigma_2 \left( \Sigma^{s, t}_{X} \right) + C(t-s)^{-1}\sqrt{n \log(1/\delta)}}   \nonumber \\
    \leq &  \sqrt{n  (t-s)^{-1} \mu_{X, 1} + C(t-s)^{-1}\sqrt{n \log(1/\delta)}} \leq  \sqrt{\frac{3n \mu_{X, 1}}{2(t-s)}}, \nonumber
\end{align}
\begin{align}
    & \sigma_1 \left( \overline{X}^{s, t} \right) \leq \sqrt{n  \sigma_1 \left( \Sigma^{s, t}_{X} \right) + C(t-s)^{-1}\sqrt{n \log(1/\delta)}}   \nonumber \\
    \leq &  \sqrt{n  C_1 \|\theta_{X}\|^2  + C(t-s)^{-1}\sqrt{n \log(1/\delta)}}  \leq  \sqrt{3C_1 n /2 }\|\theta_{X}\| , \nonumber
\end{align}
and
\begin{align}
    & \sigma_1 \left( \overline{X}^{s, t} \right) \geq \sqrt{n  \sigma_1 \left( \Sigma^{s, t}_{X} \right) - C(t-s)^{-1}\sqrt{n \log(1/\delta)}}   \nonumber \\
    \geq &  \sqrt{n C_2\|\theta_{X}\|^2 - C(t-s)^{-1}\sqrt{n \log(1/\delta)}}  \geq  \sqrt{ nC_2 /2}\|\theta_{X}\|. \nonumber
\end{align}
\medskip
\noindent\textbf{Step 2.}
By \Cref{fix_lemma1}, in the event $\mathcal{A}_1$, we have that 
\begin{align}
\left( \mathrm{rank} \left(\mathcal{M}_1\left( \widetilde{ \widehat{\mathbf{P}}}_{s, t} \right)\right) , \mathrm{rank} \left(\mathcal{M}_2\left( \widetilde{ \widehat{\mathbf{P}}}_{s, t}  \right)\right) , \mathrm{rank} \left(\mathcal{M}_3\left( \widetilde{ \widehat{\mathbf{P}}}_{s, t} \right)\right) \right)  = \left(d, d, m \right). \nonumber
\end{align}
\[
   \sigma_1\left( \mathcal{M}_1\left( \widetilde{ \widehat{\mathbf{P}}}^{s, t}\right)\right) \leq  3C_1n \sqrt{ d }\|\theta_{X}\|^2 \sigma_1(Q)/2
\]
\[
    \sigma_1\left( \mathcal{M}_2\left( \widetilde{ \widehat{\mathbf{P}}}^{s, t}\right)\right) \leq   3C_1n\sqrt{ d}\|\theta_{X}\|^2 \sigma_1(Q)/2,
\]  
and 
\[
\sigma_1\left( \mathcal{M}_3\left( \widetilde{ \widehat{\mathbf{P}}}^{s, t}\right)\right) \leq 3 C_1 n \|\theta_{X} \|^2\sigma_1(Q)/2.
\] 

Let $\Theta_{X} = (\theta_X \dots \theta_X )^{\top} \in \R^{n \times d}$. Note that
\begin{align}\label{eq_eigen_two}
   &   \widetilde{ \widehat{\mathbf{P}}}^{s, t} - \mathbf{S} \times_1 \Theta_X \times_2 \Theta_X \times_3 Q   \nonumber\\
    = &     \mathbf{S} \times_1 \left( \overline{X}_{s, t}  -\Theta_X \right) \times_2 \left( \overline{X}_{s, t}-\Theta_X \right) \times_3 Q  +   \mathbf{S}\times_1 \left( \overline{X}_{s, t}-\Theta_X \right) \times_2 \Theta_X  \times_3 Q\nonumber\\
  & \hspace{0.5cm} +   \mathbf{S} \times_1 \Theta_X \times_2 \left( \overline{X}_{s, t}-\Theta_X \right) \times_3 Q. 
\end{align}
Combining \eqref{eq_eigen_two} and  the proof of \Cref{fix_lemma1}, it holds that 
\begin{align}
  &  \left\| \mathcal{M}_1\left(\widetilde{ \widehat{\mathbf{P}}}^{s, t} -  \mathbf{S} \times_1 \Theta_X \times_2 \Theta_X \times_3 Q \right)\right\| 
   \nonumber\\
   \leq &  \sqrt{d} \sigma_1^2(\overline{X}_{s, t} - \Theta_X) \sigma_1(Q)  +  2 \sqrt{d} \sigma_1(\overline{X}_{s, t} - \Theta_X) \sigma_1(\Theta_X )\sigma_1(Q)
   \nonumber\\
    \leq &  3(t-s)^{-1}n \sqrt{d } \mu_{X, 1}\sigma_1(Q)/2 + 2 
    (t-s)^{-1/2}n \sqrt{6dn \mu_{X, 1}} \|\theta_X\|\sigma_1(Q)
       \nonumber\\
     \leq & C_4 n \sqrt{d}\|\theta_X\|^2\sigma_1(Q), \nonumber
\end{align}
with an absolute constant $C_4>0$, 
where the second inequality follows from \textbf{Step 1} and  the proof of \Cref{fix_lemma1}, the third inequality follows from \Cref{ass_X_Y}$(b)$.
Following the proof of \Cref{fix_lemma1}, we also have that with an large enough absolute constant $C_5 >0$,
\begin{align}
    \left\|  \mathcal{M}_1 \left( \mathbf{S} \times_1 \Theta_X \times_2 \Theta_X \times_3 Q  \right) \right\| \geq C_5n \sqrt{d} \|\theta_X\|^2 \sigma_d(Q), \nonumber
\end{align}
and 
\begin{align}
  \mathrm{rank} \left(  \mathcal{M}_1 \left( \mathbf{S} \times_1 \Theta_X \times_2 \Theta_X \times_3 Q  \right) \right) =1. \nonumber
\end{align}
By \eqref{eq_eigen_two}, Weyl's inequality \citep{weyl1912asymptotische} and 
\Cref{ass_X_Y}$(b)$, 
we have that with absolute constants $C_6 > C_4> 0$,
\begin{align}
\sigma_1\left( \mathcal{M}_1\left( \widetilde{ \widehat{\mathbf{P}}}^{s, t}\right)\right)  \geq &  \left\|  \mathcal{M}_1 \left( \mathbf{S} \times_1 \Theta_X \times_2 \Theta_X \times_3 Q  \right) \right\| - \left\| \mathcal{M}_1\left(\widetilde{ \widehat{\mathbf{P}}}^{s, t} -  \mathbf{S} \times_1 \Theta_X \times_2 \Theta_X \times_3 Q \right)\right\| \nonumber\\
\geq & C_6n \sqrt{ d} \|\theta_X\|^2\sigma_1(Q), \nonumber
\end{align}
and
\begin{align}
\sigma_2\left( \mathcal{M}_1\left( \widetilde{ \widehat{\mathbf{P}}}^{s, t}\right)\right) \leq &   \sigma_2 \left(  \mathcal{M}_1 \left( \mathbf{S} \times_1 \Theta_X \times_2 \Theta_X \times_3 Q\right) \right) + \left\| \mathcal{M}_1\left(\widetilde{ \widehat{\mathbf{P}}}^{s, t} -  \mathbf{S} \times_1 \Theta_X \times_2 \Theta_X \times_3 Q \right)\right\| \nonumber\\
\leq &  C_4 n \sqrt{d}\|\theta_X\|^2 \sigma_1(Q). \nonumber 
\end{align} 
Consequently, we have with an absolute constant $C_7 >0$ that
\[
    \sigma_1^2\left( \mathcal{M}_1\left( \widetilde{ \widehat{\mathbf{P}}}^{s, t}\right)\right)   -  \sigma_2^2\left( \mathcal{M}_1\left( \widetilde{ \widehat{\mathbf{P}}}^{s, t}\right)\right) \geq  C_7  dn^2 \|\theta_X\|^4 \sigma_1(Q)^2.
\]
Similarly, we have 
\[
     \sigma_1^2\left( \mathcal{M}_2\left( \widetilde{ \widehat{\mathbf{P}}}^{s, t}\right)\right)   -  \sigma_2^2\left( \mathcal{M}_2\left( \widetilde{ \widehat{\mathbf{P}}}^{s, t}\right)\right)\geq  C_7dn^2 \|\theta_X\|^4 \sigma_1(Q)^2,
\]
and
\[
     \sigma_1^2\left( \mathcal{M}_3\left( \widetilde{ \widehat{\mathbf{P}}}^{s, t}\right)\right)   -  \sigma_2^2\left( \mathcal{M}_3\left( \widetilde{ \widehat{\mathbf{P}}}^{s, t}\right)\right) \geq  C_7  n^2 \|\theta_X\|^4\sigma_1(Q)^2,
\]
completing the proof.
\end{proof}

\begin{lemma}\label{prop_1}
Let the data be a sequence of adjacency tensors  $\{\mathbf{A}(t)\}_{t \in \mathbb{N}^*} \subset \mathbb{R}^{n \times n \times L}$ as defined in \Cref{def-umrdpg-dynamic} satisfying \Cref{ass_X_Y}. 
 
 Let $\theta_X$ be defined \Cref{ass_X_Y}, 
 \[
    \widetilde{ \widehat{\mathbf{P}}}^{s, t} = \mathbf{S} \times_1  \bigg((t-s)^{-1}\sum_{u = s+1}^t X(u)\bigg)  \times_2 \bigg((t-s)^{-1}\sum_{u = s+1}^t X(u)\bigg) \times_3 Q, \quad s, t \in \mathbb{N}, \, 0 \leq s < t
\]
and
    \[
        \bar{\mathbf{A}}^{s, t} = \frac{1}{t-s}\sum_{u = s+1}^{t} \mathbf{A}(u), \quad s, t \in \mathbb{N}, \, 0 \leq s < t.
    \] 
For $v \in [3]$, let $\widehat{U}_v^{s, t}$ be the eigenvector vector of  $\mathcal{M}_v( \bar{\mathbf{A}}^{s, t}) \mathcal{M}_v( \bar{\mathbf{A}}^{s, t})^{\top}$ corresponding to the largest eigenvalue.
Then let
\[
       \widetilde{\mathbf{P}}^{s, t} =  \bar{\mathbf{A}}^{s, t}  \times_1 \widehat{U}_1^{s, t} (\widehat{U}_1^{s, t})^{\top} \times_2 \widehat{U}_2^{s, t} (\widehat{U}_2^{s, t})^{\top} \times_3 \widehat{U}_3^{s, t}  (\widehat{U}_3^{s, t})^{\top},
\]
and $ \widehat{\mathbf{P}} \in \R^{n \times n \times L}$ satisfy for any $i,j \in [n]$ and $l \in [L]$
\[
   \widehat{\mathbf{P}}_{i, j, l}= \mathbbm{1}\{\widetilde{\mathbf{P}}_{i, j, l} > 1\} + \widetilde{\mathbf{P}}_{i, j, l}\mathbbm{1}\{\widetilde{\mathbf{P}}_{i, j, l} \in [0, 1]\}.
\]
For any $0\leq s<t$ satisfying that there is no change point in $[s+1, t)$, it holds for any 
\begin{align}
    2\exp\big\{- C  n \mu_{X, 1}^2\big\}\leq \delta <1,   \nonumber
\end{align} 
that with probability at least $ 1- \delta$,   
    \[
      n^{-2}\sum_{k \in [2n -1]} \sqrt{ \widetilde{S}_k \sum_{l = 1}^L \sum_{(i, j) \in \mathcal{S}_k}  \left(   \widehat{\mathbf{P}}^{s, t}_{i, j, l} -  \widetilde{ \widehat{\mathbf{P}}}^{s, t}_{i, j, l}\right)^2} \leq C\sqrt{\frac{d\log(n \vee (t-s))/\delta)}{n(t-s)}},
    \]
where $C> 0 $ is an absolute constant.
\end{lemma}

\begin{proof}[\textbf{Proof of \Cref{prop_1}.}]
This proof consists of multiple steps.  
\begin{itemize}
     \item In \textbf{Step 1}, we discuss the tail behaviours of  $|\bar{\mathbf{A}}^{s, t}_{i, j, l} - \widetilde{ \widehat{\mathbf{P}}}^{s, t}_{i, j, l}|$ given $\{ X(u)\}_{u=s+1}^{t}$, and $|\bar{\mathbf{A}}^{s, t}_{i, j, l} - \widetilde{ \widehat{\mathbf{P}}}^{s, t}_{i, j, l}|$.
     \item In \textbf{Step 2}, we decompose our target into a few addictive terms.
    \item In \textbf{Step 3}, we are to upper bound $\|\sin\Theta (\widehat{U}_v^{s, t}, U_v^{s, t})\|$ for any $v \in [3]$ using the Davis--Kahan theorem, where  $U_v^{s, t}$ denotes the eigenvector of  $\mathcal{M}_v(\widetilde{ \widehat{\mathbf{P}}}^{s, t}) \mathcal{M}_v( \widetilde{ \widehat{\mathbf{P}}}^{s, t})^{\top}$ corresponding to the largest eigenvalue, for any $v \in [3]$,.
    \item In \textbf{Step 4}, we are to keep exploiting the sub-Gaussianity of the entries of $\bar{\mathbf{A}}^{s, t} - \widetilde{ \widehat{\mathbf{P}}}^{s, t}$, in order to provide a uniform upper bound of the inner product of $\bar{\mathbf{A}}^{s, t} - \widetilde{ \widehat{\mathbf{P}}}^{s, t}$ and any unit low-rank tensors.
    \item In \textbf{Step 5}, based on all the ingredients we collect in the previous four steps, we conclude the proof.
\end{itemize}
For simplicity, we drop the superscript $s, t$ for all notation. Let $\mathbf{Z}  =\bar{\mathbf{A}} - \widetilde{ \widehat{\mathbf{P}}}$.

\medskip
\noindent \textbf{Step 1.}  
We first discuss the tail behaviour of the entries of $\mathbf{Z}$.  Note that
    \begin{align}
        & |\mathbf{Z}_{i, j, l}| = \bigg|\frac{1}{t-s } \sum_{ u= s+1}^{t} \{\mathbf{A}(u)\}_{i, j, l} - 
        \frac{1}{(t-s)^2} \sum_{u,v = s+1}^{t} \left( X_i(u) \right)^{\top} W_{(l)}  X_j(v)\bigg| \nonumber \\
        \leq & \left|\frac{1}{t-s} \sum_{t = s+ 1}^t \{\mathbf{A}(u)\}_{i, j, l} - \mathbb{E}\left[\frac{1}{t-s} \sum_{t = s+1}^t \{\mathbf{A}(u)\}_{i, j, l}\bigg| \{X(u)\}_{u = s+1}^t\right]\right| \nonumber \\
        & \hspace{1cm} + \left|\mathbb{E}\left[\frac{1}{t-s} \sum_{t = s+1}^t \{\mathbf{A}(u)\}_{i, j, l}\bigg| \{X(u)\}_{u = s+1}^t\right] - \frac{1}{(t-s)^2} \sum_{u,v = s+1}^t \left( X_i(u) \right)^{\top} W_{(l)}  X_j(v)\right| \nonumber \\
        = & (I) + (II).\nonumber
    \end{align}

\medskip 
\noindent \textbf{Step 1.1.} As for $(I)$,  by Hoeffding's inequality for general bounded random variables \citep[e.g.~Theorem 2.2.6 in][]{vershynin2018high},  for any $0 < \delta_1 < 1$, we have that with absolute constants $C_1, c_1 >0$
    \begin{align}
        & \mathbb{P}\Bigg\{ \left|\frac{1}{t-s} \sum_{u = s+1}^t \{\mathbf{A}(u)\}_{i, j, l} - \mathbb{E}\left[\frac{1}{t-s} \sum_{u = s+1}^t \{\mathbf{A}(u)\}_{i, j, l} \bigg| \{X(u)\}_{u = s+1}^t\right]\right|  > \delta_1 \bigg| \{X(u)\}_{u = s+1}^t\Bigg\} \nonumber \\
        \leq & C_1\exp \left\{-\frac{c_1 \delta_1^2 (t-s)^2}{\sum_{u = s+1}^t \{X(u)\}_i W_{(l)} \{X(u)\}_j^{\top}}\right\} \leq C_1\exp\{-c_1\delta_1^2 (t-s)\}, \nonumber
    \end{align}
    where the second inequality is due to the definition of the inner product distribution pair detailed in \Cref{ipd}. Consequently, we have
    \begin{align}\label{step_1.1_result}
     &  \mathbb{P}\left\{\left|\frac{1}{t-s} \sum_{u = s+1}^t \{\mathbf{A}(u)\}_{i, j, l} - \mathbb{E}\left[\frac{1}{t-s} \sum_{u = s+1}^t \{\mathbf{A}(u)\}_{i, j, l} \bigg| \{X(u)\}_{u = s+1}^t\right]\right| > \delta_1 \right\} \nonumber\\
     = & \mathbb{E} \Bigg\{  \mathbb{P}\Bigg\{\bigg|\frac{1}{t-s} \sum_{u = s+1}^t \{\mathbf{A}(u)\}_{i, j, l} - \mathbb{E}\bigg[\frac{1}{t-s} \sum_{u = s+1}^t \{\mathbf{A}(u)\}_{i, j, l} \nonumber \\
     & \hspace{1.5cm} \bigg| \{X(u)\}_{u = s+1}^t\bigg]\bigg| > \delta_1  \bigg| \{X(u)\}_{u = s+1}^t \Bigg\} \Bigg\}  = C_1\exp\{-c_1\delta_1^2 (t-s)\} 
       \end{align}

\medskip 
\noindent \textbf{Step 1.2.}  As for $(II)$, we have that 
    \begin{align} \label{eq-prop-proof-term-II-all}
        (II) & = \left|\frac{1}{(t-s)^2} \sum_{u,v = s+1}^t \left( X_i(u) \right)^{\top} W_{(l)}  X_j(v) - \frac{1}{t-s} \sum_{u = s+1}^t \left( X_i(u)  \right)^{\top} W_{(l)} X_j(u)\right| \nonumber \\
        & \leq \left|\frac{1}{(t-s)(t-s-1)} \sum_{v=s+1}^t \sum_{u \in \{s+1, \dots, t\}\backslash\{v\}} \left( X_i(u) \right)^{\top} W_{(l)} \left( X_j(v)  -X_j(u) \right)\right| \nonumber \\
         & = \left|\frac{1}{(t-s)(t-s-1)} \sum_{ v =s+ 1}^{v} \sum_{u\in \{s+1, \dots, t\}\backslash\{v\}}  \widetilde{Q}_{i, j, l, u}^{(v)} \right|,
    \end{align}
    Note that $\left\{ \widetilde{Q}_{i, j, l, u}^{(v)}, u\in \{s+1, \dots, t\}\backslash\{v\} \right\}$ are mutually independent given $X_j(v)$. We are to deploy Talagrand's concentration inequality \citep[e.g.~Theorem 5.2.16 in][]{vershynin2018high} to upper bound \eqref{eq-prop-proof-term-II-all}. In order to do so, we first derive the Lipschitz constant. For any sequence
\[
    \{x_{u}, \,  u\in \{s+1, \dots, t\}\backslash\{v\}, \, \{y_{t},  \, u\in \{s+1, \dots, t\}\backslash\{v\}\} \subset [0, 1],
\]
it holds that
\begin{align}
	& \left| (t-s-1)^{-1} \sum_{u\in \{s+1, \dots, t\}\backslash\{v\}}x_{u}   -  (t-s-1)^{-1} \sum_{u\in \{s+1, \dots, t\}\backslash\{v\} } y_{u} \right|   \nonumber\\
 & \leq   (t-s-1)^{-1}  \sum_{u\in \{s+1, \dots, t\}\backslash\{v\}} \left| x_{u}  -  y_{u}  \right|. \nonumber
\end{align}
Then by concentration inequality of Lipschitz functions with respect to $l_1$-metric, we have that for any $\delta_1 > 0$,
\begin{align}
   &  \mathbb{P}\Bigg\{ \bigg| \frac{1}{t-s-1}  \sum_{u\in \{s+1, \dots, t\}\backslash\{v\}}  \widetilde{Q}_{i, j, l, u}^{(v)}     -   \mathbb{E}  \Big\{ \frac{1}{t-s-1}  \sum_{u\in \{s+1, \dots, t\}\backslash\{v\}}  \widetilde{Q}_{i, j, l, u}^{(v)} \Big|  X_j(v)  \Big\} \bigg| \geq  \delta_1 \Bigg| X_j(v)\Bigg\}\nonumber \\
   \leq & 2 \exp \left\{ -2\delta_1^2 (t-s-1)\right\}. \nonumber
\end{align}
Consequently, we have 
\begin{align}
   & \mathbb{P}\Bigg\{ \bigg| \frac{1}{t-s-1}  \sum_{u\in \{s+1, \dots, t\}\backslash\{v\}}  \widetilde{Q}_{i, j, l, u}^{(v)}     -   \mathbb{E}  \Big\{ \frac{1}{t-s-1}  \sum_{u\in \{s+1, \dots, t\}\backslash\{v\}}  \widetilde{Q}_{i, j, l, u}^{(v)}  \Big|  X_j(v)  \Big\} \bigg| \geq \delta_1 \Bigg\}  \nonumber \\
   = &  \E \Bigg\{  \mathbb{P}\bigg\{ \bigg| \frac{1}{t-s-1}  \sum_{u\in \{s+1, \dots, t\}\backslash\{v\}}  \widetilde{Q}_{i, j, l, u}^{(v)}     -   \mathbb{E}  \Big\{ \frac{1}{t-s-1}  \sum_{u\in \{s+1, \dots, t\}\backslash\{v\}}  \widetilde{Q}_{i, j, l, u}^{(v)} \Big|  X_j(v)  \Big\} \bigg| \geq \delta_1  \bigg\vert X_j(v) \bigg\}  \Bigg\}   \nonumber\\
   \leq & 2 \exp \left\{ -2\delta_1^2 (t-s-1)\right\} 
. \nonumber
\end{align}
By using a union bound argument, we have 
\begin{align}\label{pro_13_5}
   & \mathbb{P}\Bigg\{ \Big| \frac{1}{(t-s)(t-s-1)}  \sum_{v=s+1}^t \sum_{u\in \{s+1, \dots, t\}\backslash\{v\}}  \widetilde{Q}_{i, j, l, u}^{(v)}     \nonumber\\
    & \hspace{0.8cm} -  \frac{1}{t-s}  \sum_{v=s+1}^t   \mathbb{E}  \Big\{ \frac{1}{t-s-1}  \sum_{u\in \{s+1, \dots, t\}\backslash\{v\}}  \widetilde{Q}_{i, j, l, u}^{(v)}  \Big|  X_j(v)  \Big\} \Big| \geq \delta_1 \Bigg\}  \nonumber\\
  \leq &  2 (t-s)\exp \left\{ -2 \delta_1^2 (t-s-1)\right\}.  
\end{align}
Note that 
\begin{align}\label{pro_13_1}
   &\frac{1}{t-s}  \sum_{v=s+1}^t   \mathbb{E}  \Big\{ \frac{1}{t-s-1}  \sum_{u\in \{s+1, \dots, t\}\backslash\{v\}}  \widetilde{Q}_{i, j, l, u}^{(v)}  \Big|  X_j(v)  \Big\}  \nonumber \\
   = &  \frac{1}{t-s} \sum_{v=s+1}^t  \mathbb{E}  \left\{ \left[  (t-s-1)^{-1}  \sum_{u\in \{s+1, \dots, t\}\backslash\{v\}}   \left(  \left( X_i(u) \right)^{\top} W_{(l)} X_j(v) - \left( X_i(u) \right)^{\top} W_{(l)} X_j(u) \right)   \right] \Bigg| X_j(v)\right\} \nonumber \\
    = &   (t-s)^{-1}  \sum_{v= s+1}^t     \mathbb{E}  \left\{  X_i(t) \right\}^{\top} W_{(l)}  X_j(v)  -  \mathbb{E}  \left\{  X_1(t) \right\}^{\top} W_{(l)}  \mathbb{E}  \left\{ X_1(t) \right\}
\end{align}
where we get the final equality since 
 $\cup_{(i, j) \in \mathcal{S}_k}\left\{ X_i(t), X_j (t)\right\}$ are mutually independent random vectors. 
Consequently, we have
\begin{align}\label{pro_13_2}
   & \E \left\{ (t-s)^{-1}  \sum_{v= s+1}^t     \mathbb{E}  \Big\{ \frac{1}{t-s-1}  \sum_{u\in \{s+1, \dots, t\}\backslash\{v\}}  \widetilde{Q}_{i, j, l, u}^{(v)}  \Big|  X_j(v)  \Big\}  \right\}  =0 
\end{align}
Let 
\[
   \mathcal{A}_1 = \left\{ \Bigg| (t-s)^{-1}  \sum_{v= s+1}^t     \mathbb{E}  \Big\{ \frac{1}{t-s-1}  \sum_{u\in \{s+1, \dots, t\}\backslash\{v\}}  \widetilde{Q}_{i, j, l, u}^{(v)}  \Big|  X_j(v)  \Big\}  \Bigg| \leq  \delta_1 \right\}.
\]
Combining \eqref{pro_13_1} and \eqref{pro_13_2}, by Hoeffding's inequality for general bounded random variables \citep[e.g.~Theorem 2.2.6 in][]{vershynin2018high}, we have that 
\begin{align}\label{pro_13_4}
    \P \left\{ \mathcal{A}_1\right\} \geq 1- 2\exp \left\{-2\delta_1^2(t-s) \right\}.
\end{align}
Then combining \eqref{pro_13_5} and \eqref{pro_13_4}, we have
\begin{align}\label{step_1.2_result}
   & \mathbb{P}\left\{ \left| \frac{1}{(t-s)(t-s-1)}  \sum_{v=s+1}^{t}\sum_{ u\in \{s+1, \dots, t\}\backslash\{v\}}  \widetilde{Q}_{i, j, l, u}^{(v)}  \right| \geq 2\delta_1 \right\}\nonumber\\ 
   \leq &  2 (t-s)\exp \left\{ -2\delta_1^2 (t-s-1)\right\} \P\left\{  \mathcal{A}_1\right\} + 1- \P\left\{  \mathcal{A}_1\right\} \leq C_1 (t-s)\exp \left\{  -c_1\delta_1^2 (t-s)  \right\},  
\end{align}
with absolute constants $C_1, c_1 > 0$.

\medskip
\noindent \textbf{Step 1.3.} Combining  \eqref{step_1.1_result}, \eqref{eq-prop-proof-term-II-all} and \eqref{step_1.2_result},  we then have that for any $\delta_1 > 0$,
    \begin{align*}
          \mathbb{P}\{|\mathbf{Z}_{i, j, l}| > 3 \delta_1 \}  \leq C_1 (t-s) \exp \left\{  -c_1\delta_1^2 (t-s)  \right\}.
    \end{align*}
From \textbf{Step 1.1.}, we also have that for any $\delta_1 >  0$,
    \begin{align}
          \mathbb{P}\{|\mathbf{Z}_{i, j, l} - \E\{\mathbf{Z}_{i, j, l} \big| \{ X(u) \}_{u=s+1}^t\} \}| > \delta_1 \big| \{ X(u) \}_{u=s+1}^t\}  \leq C_2 \exp \left\{ -c_2\delta_1^2 (t-s) \right\}, \nonumber
    \end{align}
with
\[
    \E\{\mathbf{Z}_{i, j, l} \big| \{ X(u) \}_{u=s+1}^t\}\} =  \frac{1}{(t-s)^2} \sum_{u,v = s+1}^t \left( X_i(u) \right)^{\top} W_{(l)}  X_j(v) - \frac{1}{t-s} \sum_{u = s+1}^t \left( X_i(u)  \right)^{\top} W_{(l)} X_j(u) .
\]

\medskip
\noindent \textbf{Step 2.}  
In this step, we decompose our target into a few addictive terms.
Note that 
    \begin{align}\label{pro_13_main}
        & \frac{1}{n^2} \sum_{k \in [2n - 1]} \sqrt{\widetilde{S}_k \sum_{l = 1}^L \sum_{(i, j) \in \mathcal{S}_k} \big( \widehat{\mathbf{P}}_{i, j, l} - \widetilde{ \widehat{\mathbf{P}}}_{i, j, l}\big)^2} \nonumber \\
        \leq & \frac{\sqrt{n}}{n^2} \sum_{k \in [2n-1]} \sqrt{\sum_{l = 1}^L \sum_{(i, j) \in \mathcal{S}_k} \big(\widetilde{\mathbf{P}}_{i, j, l} - \widetilde{ \widehat{\mathbf{P}}}_{i, j, l}\big)^2} 
        \leq \frac{\sqrt{2}}{n} \big\|\widetilde{\mathbf{P}} - \widetilde{ \widehat{\mathbf{P}}}\big\|_{\mathrm{F}} \nonumber \\
        \leq &
         \frac{\sqrt{2}}{n} \big\|\widetilde{ \widehat{\mathbf{P}}} \times_1 \widehat{U}_1 \widehat{U}_1^{\top} \times_2 \widehat{U}_2\widehat{U}_2^{\top} \times_3 \widehat{U}_3 \widehat{U}_3^{\top} - \widetilde{ \widehat{\mathbf{P}}}\big\|_{\mathrm{F}}\nonumber \\
        & \hspace{1cm} + \frac{\sqrt{2}}{n} \big\|\big(\bar{\mathbf{A}} - \widetilde{ \widehat{\mathbf{P}}}\big) \times_1 \widehat{U}_1 \widehat{U}_1^{\top} \times_2 \widehat{U}_2 \widehat{U}_2\times_3 \widehat{U}_3  \widehat{U}_3\big\|_{\mathrm{F}} \nonumber\\
        = & (III) + (IV),
    \end{align}
where the first inequality follows from
\Cref{lemma_S_u}, and the second inequality follows from the Cauchy--Schwarz inequality.

\noindent \textbf{Step 3.}   In this step, we deal with term $(III)$ in \eqref{pro_13_main}.
It follows from the proof of Theorem 4.1 in \cite{han2022optimal} that
\begin{align}\label{eq_step_3_main}
    (III) \leq \sqrt{2}n^{-1} \sum_{v \in [3]}\|\sin\Theta(\widehat{U}_{v}, U_v)\|  \|\mathcal{M}_v (\widetilde{ \widehat{\mathbf{P}}}) \|,
\end{align} 
where  $U_v$ denotes the eigenvector of  $\mathcal{M}_v(\widetilde{ \widehat{\mathbf{P}}}) \mathcal{M}_v( \widetilde{ \widehat{\mathbf{P}}})^{\top}$ corresponding to the largest eigenvalue, for any $v \in [3]$.
We are to apply the Davis--Kahan theorem presented in \citep{yu2015useful} to upper bound $\|\sin \Theta (\widehat{U}_s, U_s)\|$, i.e
\begin{align}
    \|\sin\Theta(\widehat{U}_{1}, U_1)\| \|\mathcal{M}_1 (\widetilde{ \widehat{\mathbf{P}}}) \| &\leq \frac{2\|\mathcal{M}_1(\bar{\mathbf{A}}) \mathcal{M}_1(\bar{\mathbf{A}})^{\top} - \mathcal{M}_1(\widetilde{ \widehat{\mathbf{P}}})\mathcal{M}_1(\widetilde{ \widehat{\mathbf{P}}})^{\top} \| }{\sigma_1^2(\mathcal{M}_1(\widetilde{ \widehat{\mathbf{P}}})) - \sigma_2^2(\mathcal{M}_1(\widetilde{ \widehat{\mathbf{P}}}))} \sigma_1(\mathcal{M}_1(\widetilde{ \widehat{\mathbf{P}}})\nonumber 
\end{align}
Let $\mathcal{A}$ be defined in \eqref{lemma_12_event} in \Cref{lemm_eigen} and from now on assume that the event $\mathcal{A}$ holds. Then we have 
\begin{align}\label{eq_1}   
\|\sin\Theta(\widehat{U}_{1}, U_1)\| \|\mathcal{M}_1 (\widetilde{ \widehat{\mathbf{P}}}) \| 
    &\leq \frac{C_{\sigma}\|\mathcal{M}_1(\bar{\mathbf{A}}) \mathcal{M}_1(\bar{\mathbf{A}})^{\top} - \mathcal{M}_1(\widetilde{ \widehat{\mathbf{P}}})\mathcal{M}_1(\widetilde{ \widehat{\mathbf{P}}})^{\top} \| }{n \sqrt{d}\| \theta_X\|^2  \sigma_1(Q) }, 
\end{align}
with an absolute constant $C_{\sigma} >0$.

We are to upper bound
    \[
        \|\mathcal{M}_1(\bar{\mathbf{A}}) \mathcal{M}_1(\bar{\mathbf{A}})^{\top} - \mathcal{M}_1(\widetilde{ \widehat{\mathbf{P}}})\mathcal{M}_1(\widetilde{ \widehat{\mathbf{P}}})^{\top} \|.
    \]
Note that 
\begin{align}\label{eq_2}
     \|\mathcal{M}_1(\bar{\mathbf{A}}) \mathcal{M}_1(\bar{\mathbf{A}})^{\top} - \mathcal{M}_1(\widetilde{ \widehat{\mathbf{P}}})\mathcal{M}_1(\widetilde{ \widehat{\mathbf{P}}})^{\top} \| 
   \leq & 2 \|\mathcal{M}_1(\mathbf{Z}) \mathcal{M}_1(\widetilde{ \widehat{\mathbf{P}}})^{\top} \| +
   \|\mathcal{M}_1(\mathbf{Z}) \mathcal{M}_1(\mathbf{Z})^{\top} \| \nonumber \\
   =  & (III.1) + (III.2).
\end{align}
\medskip
\noindent \textbf{Step 3.1.}
As for $(III.1)$ in \eqref{eq_2}, by \textbf{Step 1} and  Lemma 2 in \cite{zhang2018heteroskedastic}, for any $ 0 < \delta_2 < 1$ to be specified later, we have that with absolute constants $C_2>0$, 
\[
  \P \left\{  \|\mathcal{M}_1(\mathbf{Z}) \mathcal{M}_1(\widetilde{ \widehat{\mathbf{P}}})^{\top} \| 
 \geq C_2 \sigma_1\left( \mathcal{M}_1(\widetilde{ \widehat{\mathbf{P}}})\right) \sqrt{n/(t-s)} \log(1/ \delta_2) \big| \left\{X(u) \right\}_{u=s+1}^t\right\} \leq \delta_2.
\]
Then by the event $\mathcal{A}$, we have 
\[
  \P \left\{  \|\mathcal{M}_1(\mathbf{Z}) \mathcal{M}_1(\widetilde{ \widehat{\mathbf{P}}})^{\top} \| 
 \geq C_2  n\sqrt{dn/(t-s)} \| \theta_X\|^2  \sigma_1(Q) \log(1/ \delta_2) \big| \left\{X(u) \right\}_{u=s+1}^t\right\} \leq \delta_2.
\]
Consequently, we have 
\begin{align}\label{eq_3}
 & \P \left\{  \|\mathcal{M}_1(\mathbf{Z}) \mathcal{M}_1(\widetilde{ \widehat{\mathbf{P}}})^{\top} \| 
 \geq C_2  n\sqrt{dn /(t-s)} \| \theta_X\|^2 \sigma_1(Q) \log(1/ \delta_2) \right\}\nonumber\\
 = & \E \left\{ \P \left\{  \|\mathcal{M}_1(\mathbf{Z}) \mathcal{M}_1(\widetilde{ \widehat{\mathbf{P}}})^{\top} \| 
 \geq C_2  n\sqrt{dn /(t-s)} \| \theta_X\|^2 \sigma_1(Q) \log(1/ \delta_2) \big| \left\{X(u) \right\}_{u=s+1}^t\right\} \right\} \nonumber\\
 \leq & \delta_2.
\end{align}
\medskip
\noindent \textbf{Step 3.2.}
 As for $(III.2)$ in \eqref{eq_2}, following the proof of \textbf{Step 1}, the Hoeffding inequality \citep[e.g.~Theorem 2.6.3 in ][]{vershynin2018high}, a union bound argument and the proof of Theorem 4.4.5 in \cite{vershynin2018high}, for any  $ 0 < \delta_2 < 1$ to be specified later, we have that 
\begin{align}\label{eq_7}
     \P \left\{ (III.2) \geq   C_3 \sqrt{\frac{n^3 \log(n \vee (t-s))/\delta_2 )}{t-s}}  \right\} \leq \delta_2,
\end{align}
where $C_3>0$ is an absolute constant.

\medskip
\noindent \textbf{Step 3.3.}
Combining \eqref{eq_1}, \eqref{eq_2}, \eqref{eq_3} and \eqref{eq_7}, we have 
\begin{align}\label{eq_8}
 \P  \left\{  \|\sin\Theta(\widehat{U}_{1}, U_1)\| \|\mathcal{M}_1 (\widetilde{ \widehat{\mathbf{P}}}) \| \geq C_4 \sqrt{\frac{n   \log(n \vee (t-s))/\delta_2 )}{t-s} }  \right\} \leq 2 \delta_2,
\end{align}
where $C_4 > 0$ is an absolute constant.
Similarly, we have 
\begin{align}\label{eq_9}
 \P  \left\{  \|\sin\Theta(\widehat{U}_{2}, U_2)\| \|\mathcal{M}_2 (\widetilde{ \widehat{\mathbf{P}}}) \| \geq C_4 \sqrt{ \frac{n   \log(n\vee (t-s))/\delta_2 )}{t-s} }  \right\} \leq 2 \delta_2,
\end{align}
and
\begin{align}\label{eq_10}
 \P  \left\{  \|\sin\Theta(\widehat{U}_{3}, U_3)\| \|\mathcal{M}_3 (\widetilde{ \widehat{\mathbf{P}}}) \| \geq  C_4 \sqrt{  \frac{d n   \log(n\vee (t-s))/\delta_2 )}{t-s} }  \right\} \leq 2 \delta_2.
\end{align}
Combining \eqref{eq_step_3_main}, \eqref{eq_8}, \eqref{eq_9},  \eqref{eq_10}, we have 
\begin{align}\label{pro_13_main_1}
  \P \left\{   (III) \geq C_5 \sqrt{ \frac{   d \log(n \vee (t-s))/\delta_2 ) }{n(t-s)}} \right\} \leq 6 \delta_2,
\end{align}
where $C_5 > 0$ is an absolute constant.

\medskip
\noindent \textbf{Step 4.}
 This step deals with term $(IV)$ in \eqref{pro_13_main}. Following the proof of Theorem 4.1 in \cite{han2022optimal}, we have 
\begin{align}\label{step_4_1}
     \big\|   \mathbf{Z}\times_1 \widehat{U}_1 \widehat{U}_1^{\top} \times_2 \widehat{U}_2 \widehat{U}_2^{\top} \times_3 \widehat{U}_3  \widehat{U}_3^{\top} \big\|_{\mathrm{F}} \leq  \sup_{\substack{\mathbf{E} \in \R^{n \times n \times L}, \\ \mathrm{rank}(\mathbf{E}) \leq (1,1,1)\\ \|\mathbf{E}\|_{\mathrm{F}} \leq 1} } \sum_{i, j, l = 1}^{n, n, L}\mathbf{Z}_{i, j, l} \mathbf{E}_{i, j, l}.
\end{align}

 We are to apply the tail behaviours of $|\mathbf{Z}_{i, j, l}|$ we proved in the \textbf{Step 1}.  Following the proof of \textbf{Step 1}, the Hoeffding inequality \citep[e.g.~Theorem 2.6.3 in ][]{vershynin2018high}, a union bound argument and the proof of Lemma E.5 in \cite{han2022optimal},
we have that  for any $0 <\delta_3 <1$ to be specified later,
 \begin{align}\label{step_4_2}
     & \P \bigg\{ \sup_{\substack{\mathbf{E} \in \R^{n \times n \times L}, \\ \mathrm{rank}(\mathbf{E}) \leq (1,1,1)\\ \|\mathbf{E}\|_{\mathrm{F}}}} \sum_{i, j, l = 1}^{n, n, L}\mathbf{Z}_{i, j, l} \mathbf{E}_{i, j, l}
 \leq C_6 \sqrt{\frac{ n \log(n \vee (t-s))/\delta_3)}{t-s}}   \bigg\} \leq \delta_3,
 \end{align}
 where $C_6>0$ is an absolute constant.
Combining \eqref{step_4_1} and \eqref{step_4_2}, then we have 
\begin{align}
     & \P \bigg\{  \big\|   \mathbf{Z}\times_1 \widehat{U}_1 \widehat{U}_1^{\top} \times_2 \widehat{U}_2 \widehat{U}_2^{\top} \times_3 \widehat{U}_3  \widehat{U}_3^{\top} \big\|_{\mathrm{F}}
 \leq C_6 \sqrt{\frac{ n \log(n \vee (t-s))/\delta_3)}{t-s}}   \bigg\} \leq \delta_3.  \nonumber
\end{align}
Consequently, we have that with an absolute constant $C_7>0$
\begin{align}\label{pro_13_main_2}
     & \P \bigg\{  (IV)
 \leq C_7 \sqrt{\frac{ \log(n \vee (t-s))/\delta_3)}{n(t-s)}}   \bigg\} \leq \delta_3. 
\end{align}

\medskip
\noindent \textbf{Step 5.} By \Cref{lemm_eigen}, we have that for any 
\begin{align}
    2 \exp\big\{- C_8  n \mu_{X, 1}^2 \big\} \leq \delta < 1,   \nonumber
\end{align} 
with an absolute constant $C_8 >0$, 
\[
  \P \{ \mathcal{A} \} \geq 1 - \delta/2.
\]
Combining \eqref{pro_13_main}, \eqref{pro_13_main_1} and \eqref{pro_13_main_2}, let $\delta_2 = \delta/24$ and $\delta_3 = \delta/4$, then we have that 
\begin{align}
   & \P \Bigg\{ n^{-2}\sum_{k \in [2n -1]} \sqrt{ \widetilde{S}_k \sum_{l = 1}^L \sum_{(i, j) \in \mathcal{S}_k}  \left(   \widehat{\mathbf{P}}^{s, t}_{i, j, l} -  \widetilde{ \widehat{\mathbf{P}}}^{s, t}_{i, j, l}\right)^2} > C_9\sqrt{\frac{d\log(n \vee (t-s))/\delta)}{n(t-s)}} \Bigg\} \nonumber\\
   &  \leq \delta/2  \P \{ \mathcal{A} \} + 1- \P \{ \mathcal{A} \}  \leq \delta \nonumber
  \end{align}
where $C_9 >0$ is an absolute constant, completing the proof.    
\end{proof}

\begin{lemma}\label{lemma_S_u}
For any $k \in [2n-1]$, let $\mathcal{S}_k$ be defined in \eqref{eq-mathcal-S-def_u} with $\widetilde{S}_k = \vert  \mathcal{S}_k \vert$. For any $n > 2$,  we have that $\widetilde{S}_k \leq n$ and
\[
   C_2 n^2  \leq \sum_{k \in [2n-1 ]}\widetilde{S}_k \leq C_1 n^2,
\]
where  $C_1 > C_2 >0$ are absolute constants.
\end{lemma}
\begin{proof}[\textbf{Proof of \Cref{lemma_S_u}.}]
By the construction of $\mathcal{S}_k$ defined in \eqref{eq-mathcal-S-def_u}, we have that $\widetilde{S}_1 = 0$, $\widetilde{S}_k = \widetilde{S}_{k-n+1}$ for any $k \in [2n-1] \setminus [n]$, and 
\begin{equation}\label{card_1}
    \widetilde{S}_{k+1} = k \lfloor \frac{n-k}{2k} \rfloor + \min \left\{k, n-k - 2k\lfloor \frac{n-k}{2k} \rfloor \right\},
\end{equation}
for any $k \in [1:n-1]$. By \eqref{card_1}, it holds for any $k \in [1:n-1]$ that 
\begin{equation}
  \frac{n-k}{2} \leq   \widetilde{S}_{k+1} \leq n. \nonumber
\end{equation}
Thus, we have
\begin{align}
  2 \sum_{k \in [n-1 ]}  \frac{n-k}{2}  \leq  \sum_{k \in [2n-1 ]}\widetilde{S}_k  \leq 2(n-1)n, \nonumber
\end{align}
which leads to 
\begin{align}\label{card_2}
    \frac{n(n-1)}{2}  \leq  \sum_{k \in [2n-1 ]}\widetilde{S}_k  \leq 2n(n-1). 
\end{align}
By \eqref{card_2} and $n>2$, there exist absolute constants $C_1 > C_2 >0$ such that 
\[
   C_2 n^2  \leq \sum_{k \in [2n-1 ]}\widetilde{S}_k \leq C_1 n^2,
\]
which completes the proof.
\end{proof}

\subsection{Proof of Theorem \ref{main_theorem}}\label{sec-C3}

\begin{proof}[\textbf{Proof of \Cref{main_theorem}.}]
The proof is conducted in the event $\mathcal{B}_1 \cap \mathcal{B}_2$, with $\mathcal{B}_1$ and $\mathcal{B}_2$ defined as
\begin{align*}
    \mathcal{B}_1 = & \bigg\{ \forall s, t \in \mathbb{N}, \, 1 \leq s < t \mbox{ satisfying that there is no change point in either [1, s) or [s+1, t)}:   \\
    & \hspace{0.5cm} \sup_{z \in [0, 1]^L} \big|\widehat{D}_{s, t}(z) - \widetilde{D}_{s, t}(z)  \big| 
  \leq  Ch^{-L-1}  \sqrt{\frac{(L^2 \vee d) \log\{ (n  \vee t) / \alpha\}}{n}}\left( \frac{1}{\sqrt{s}}+ \frac{1}{\sqrt{t- s}}\right)\bigg\},
\end{align*}
and 
\begin{align*}
   \mathcal{B}_2 = & \bigg\{  \forall s, t \in \mathbb{N}, 1 \leq s < t: \, \sup_{z \in [0, 1]^L} \left| \widetilde{D}_{s, t} \left(z\right)\right| -  \max_{m=1}^{M_{\alpha, t}} \left| \widetilde{D}_{s, t}\left( z_m\right) \right| 
    \\
   & \hspace{1cm} 
   \leq  C h^{-L-1} \sqrt{ \frac{L^2  \log\{ (n  \vee t) / \alpha\}}{n}} \left( \frac{1}{\sqrt{s}}+ \frac{1}{\sqrt{t- s}}\right) \bigg\},
\end{align*}
respectively, $C>0$ being an absolute constant.  By Lemmas~\ref{lemma_psi} and \ref{lemma_z}, we have  
\begin{equation}\label{eq-b-1-b-2-cap}
    \P\left\{ \mathcal{B}_1 \cap \mathcal{B}_2  \right\} \geq 1 - \alpha.   
\end{equation}
 
Let 
\[
    \varepsilon_{1, s, t} = Ch^{-L-1}  \sqrt{\frac{(L^2 \vee d) \log\{ (n  \vee t) / \alpha\}}{n}}\left( \frac{1}{\sqrt{s}}+ \frac{1}{\sqrt{t- s}}\right)
\]
and
\[
    \varepsilon_{2, s, t} = C h^{-L-1} \sqrt{ \frac{L^2 \log\{ (n  \vee t) / \alpha\}}{n }} \left( \frac{1}{\sqrt{s}}+ \frac{1}{\sqrt{t- s}}\right).
\]

In the event $\mathcal{B}_1 \cap \mathcal{B}_2$, for any $ 1 \leq s < t \leq \Delta$, 
\begin{align}
    \max_{m=1}^{M_{\alpha, t}} \left| \widehat{D}_{s, t}\left( z_m\right)\right|
    \leq &  \max_{m=1}^{M_{\alpha, t}}  \left| \widetilde{D}_{s, t}\left( z_m\right) \right| + \varepsilon_{1, s, t} \leq \sup_{z \in [0, 1]^L}  \left| \widetilde{D}_{s, t}\left(z\right) \right| + \varepsilon_{1, s, t} \nonumber \\
    \leq & C_1 h^{-L-1}\sqrt{\frac{\{1 + \log(L)\} L}{ns}} +   \varepsilon_{1, s, t}
    \leq \tau_{s, t}, \nonumber
\end{align}  
with an absolute constant $C_1 >0$, where the first inequality follows from the definition of $\mathcal{B}_1$, the third inequality follows from \Cref{lemma_no_change_point} and the final inequality follows from the definition of $\tau_{s, t}$ in \eqref{eq-tau-def-thm-2}.  We, therefore, have that,
    \[
        \big(\mathcal{B}_1 \cap \mathcal{B}_2\big) \subset \bigcap_{t \leq \Delta} \big\{\widehat{\Delta} > t \big\}. 
    \]

Under \Cref{ass_no_change_point} and due to \eqref{eq-b-1-b-2-cap}, it holds that
\begin{equation}\label{eq-thm2-pf-state-1}
    \mathbb{P}_{\infty} \{\widehat{\Delta} < \infty\} =  \mathbb{P}_{\infty} \{\exists t \in \mathbb{N}: \, \widehat{\Delta} \leq  t \} \leq \mathbb{P}\{\mathcal{B}_1^c \cup \mathcal{B}_2^c\} \leq \alpha.
\end{equation}
Under \Cref{ass_change_point}, it holds that 
\begin{equation}\label{eq-no-false-alarm-change-point}
    \big(\widehat{\Delta} < \Delta\big) \subset (\mathcal{B}_1^c \cup \mathcal{B}_2^c).
\end{equation}    

We then let 
\[
    \widetilde{\Delta} = \Delta + C_\epsilon h^{-2L-2} \frac{(L^2 \vee d) \log \{(n \vee \Delta) /\alpha \}}{\kappa^2 n},
\]
where $ C_\epsilon > 0$ is an absolute constant.  It follows from \Cref{lemma_one_change_point} that
\[
 \sup_{z \in [0, 1]^L} \left| \widetilde{D}_{\Delta, \widetilde{\Delta}}(z)  \right| > \frac{C_{\mathrm{SNR}}}{2} h^{-L-1}  \sqrt{\frac{(L^2 \vee d)\log\{(n \vee \Delta)/\alpha \} }{n}} \left(\sqrt{\frac{1}{ \Delta}} + \sqrt{\frac{1}{\widetilde{\Delta} - \Delta}} \right).
\]
In the event $\mathcal{B}_1 \cap \mathcal{B}_2$, we have that
\begin{align}
  & \max_{m=1}^{M_{\alpha, t}} \left| \widehat{D}_{\Delta, \widetilde{\Delta}}\left( z_m\right)\right| \geq   \max_{m=1}^{M_{\alpha, t}} \left| \widetilde{D}_{\Delta, \widetilde{\Delta}}\left( z_m\right)\right|  - \varepsilon_{1, \Delta, \widetilde{\Delta}}  \geq \sup_{z \in [0, 1]^L}\left| \widetilde{D}_{\Delta, \widetilde{\Delta}}\left( z\right)\right|  - \varepsilon_{1, \Delta, \widetilde{\Delta}} - \varepsilon_{2, \Delta, \widetilde{\Delta}} \nonumber \\
  \geq & \frac{C_{\mathrm{SNR}} h^{-L-1} }{8}  \sqrt{\frac{(L^2 \vee d)\log\{ (n \vee \Delta)/\alpha \} }{n}} \left(\sqrt{\frac{1}{ \Delta}} + \sqrt{\frac{1}{\widetilde{\Delta} - \Delta}} \right)> \tau_{\Delta, \widetilde{\Delta}}, \label{eq-thm-2-proof-change-exists}
\end{align}
where the first inequality follows from the definition of $\mathcal{B}_1$, the second inequality follows from the definition of $\mathcal{B}_2$, the third inequality follows from \Cref{lemma_one_change_point} and the definitions of $\varepsilon_{1, \Delta, \widetilde{\Delta}}$ and $\varepsilon_{2, \Delta, \widetilde{\Delta}}$, with a sufficiently large $C_{\mathrm{SNR}}$, and the last inequality follows from the definition of $\tau_{\Delta, \widetilde{\Delta}}$ in \eqref{eq-tau-def-thm-2} and with a sufficiently large $C_{\mathrm{SNR}}$.  

Recalling the design of \Cref{online_hpca}, \eqref{eq-thm-2-proof-change-exists} implies that $\big(\mathcal{B}_1 \cap \mathcal{B}_2\big) \subset \big(\widehat{\Delta} \leq \widetilde{\Delta}\big)$, which together with \eqref{eq-no-false-alarm-change-point} leads to that, under \Cref{ass_change_point},
\begin{equation}\label{eq-thm2-state-2}
   \mathbb{P}_{\Delta} \left\{ \Delta < \widehat{\Delta} \leq \Delta + C_\epsilon h^{-2L-2} \frac{(L^2 \vee d) \log \{(n \vee \Delta) /\alpha \}}{\kappa^2 n} \right\} \geq 1 - \alpha. 
\end{equation}

In view of \eqref{eq-thm2-pf-state-1} and \eqref{eq-thm2-state-2}, we complete the proof.
\end{proof}

\section{Directed multilayer networks}\label{sec-directed-edges}

In this section, we include the counterparts of the undirected MRDPG with random latent positions, in the directed multilayer network cases.  We collect the corresponding methodology and theory.  Note that the proofs of the theoretical results collected here follow from exactly the same proofs of their counterparts.  We thus only present the results without repetitive proofs.

\Cref{mrdpg} is the counterpart of \Cref{umrdpg}.
\begin{definition}[Directed multilayer random dot product graphs with random latent positions, $\mbox{directed}$-$\mbox{MRDPG}$]\label{mrdpg}
Let $(F, \widetilde{F})$ be an inner product distribution pair with weight matrices $\{W_{(l)}\}_{l\in [L]}$ satisfying \Cref{ipd}.  Let $\{X_i\}_{i=1}^{n_1}$, $\{Y_j\}_{j=1}^{n_2} \subset \mathbb{R}^d$ be mutually independent random vectors generated from $F$ and $\widetilde{F}$, respectively.  

Tensor $\mathbf{A} \in \R^{n_1 \times n_2 \times L}$ is an adjacency tensor of a directed multilayer random dot product graph with random latent positions $\{X_i\}_{i=1}^{n_1}$ and $\{Y_j\}_{j=1}^{n_2}$ if the conditional distribution satisfies that
\begin{align*}
    \mathbb{P} \big\{\mathbf{A} \vert \{X_i\}_{i = 1}^{n_1}, \{Y_j\}_{j = 1}^{n_2}\big\} & = \prod_{i, j, l = 1}^{n_1, n_2, L} \mathbf{P}_{i, j, l}^{\mathbf{A}_{i, j, l}} (1- \mathbf{P}_{i, j, l})^{1- \mathbf{A}_{i, j, l}} \\
    & = \prod_{i, j, l = 1}^{n_1, n_2, L} \big(X_i^{\top} W_{(l)} Y_j\big)^{A_{i, j, l}} \big(1 - X_i^{\top} W_{(l)} Y_j\big)^{1 - A_{i, j, l}}.
\end{align*}
We write $\mathbf{A} \sim \mathrm{directed}$-$\mathrm{MRDPG}(F, \widetilde{F}, \{W_{(l)}\}_{l\in [L]}, n_1, n_2)$ and write $\mathbf{P} = (\mathbf{P}_{i, j, l}) \in \mathbb{R}^{n_1 \times n_2 \times L}$ as the probability tensor. 
\end{definition}

The following is the counterpart of \eqref{Y_tucker_rep}. The probability tensor introduced in \Cref{mrdpg} can be written in a Tensor representation, i.e.
\[
      \mathbf{P} = \mathbf{S} \times_1 X \times_2 Y \times_3 Q,
\]
where $X = (X_1, \ldots, X_{n_1})^{\top} \in \mathbb{R}^{n_1 \times d}$, $Y = (Y_1, \ldots, Y_{n_2})^{\top} \in \mathbb{R}^{n_2 \times d}$, $\mathbf{S} \in \R^{d \times d \times d^2}$ defined in \eqref{matrix_S} and $Q \in \R^{L \times d^2}$ defined in \eqref{matrix Q}. For $l \in [L]$, the $l$th layer of the probability tensor is $\mathbf{P}_{:,:,l} = XW_{(l)}Y^{\top}$.

\Cref{ass_X_Y-app} is the counterpart of \Cref{ass_X_Y}.
\begin{assumption}[Random latent positions]\label{ass_X_Y-app}
Consider a $\mathrm{directed}$-$\mathrm{MRDPG}$ $(F, \widetilde{F}, \{W_{(l)}\}_{l = 1}^L, n_1, n_2)$ defined in \Cref{mrdpg}. Let $X_1 \sim F$ and $Y_1 \sim \widetilde{F}$ be sub-Gaussian random vectors with $\Sigma_X = \mathbb{E}(X_1X_1^{\top}) \in \mathbb{R}^{d \times d}$ and $\Sigma_Y = \mathbb{E}(Y_1Y_1^{\top}) \in \mathbb{R}^{d \times d}$.  Let $\mu_{X, 1} \geq \cdots \geq \mu_{X, d} > 0$ and $\mu_{Y, 1} \geq \cdots \geq \mu_{Y, d} > 0$ be the eigenvalues of $\Sigma_X$ and $\Sigma_Y$, respectively.  
\begin{enumerate}
    \item [$(a)$] Assume that there exist absolute constants $C_{\mathrm{gap}}, C_{\sigma} > 0$, such that $\min\{\mu_{X, d}, \, \mu_{Y, d}\} > C_{\mathrm{gap}}$ and $\max\{\mu_{X, 1}/\mu_{X, d}, \mu_{Y, 1}/\mu_{Y, d}\} \leq C_{\sigma}$.
    \item [$(b)$] Further let $\theta_X = \E\{ X_1 \}$ and $\theta_Y =  \E\{ Y_1 \}$.  Assume that there exists an absolute constant $C_{\theta} > 0$ such that $\max\{\mu_{X, 1}/\| \theta_X\|^2, \, \mu_{Y, 1}/\| \theta_Y\|^2\} \leq C_{\theta}$.
    \item [$(c)$] For $Q \in \mathbb{R}^{L \times d^2}$ defined in \eqref{matrix Q}, let $m = \mathrm{rank}(Q)$ and a singular value decomposition of~$Q$ be $Q = U_Q D_Q V_Q^{\top}$, with $U_Q \in \mathbb{R}^{L \times m}$, $D_Q \in \mathbb{R}^{m \times m}$ and $V_Q \in \mathbb{R}^{d^2 \times m}$.  Assume that $ \sigma_1(Q)/\sigma_m(Q) \leq C_{\sigma}$ and  $\sigma_{m}(Q) \geq C_{\mathrm{gap}}$, where $C_{\mathrm{gap}}, C_{\sigma}$ are the same as those in $(a)$.
\end{enumerate}
\end{assumption}

\Cref{random_theorem-app} is the counterpart of \Cref{random_theorem}.
\begin{theorem}\label{random_theorem-app}
Suppose that \Cref{ass_X_Y-app}$(a)$ and $(c)$ hold with random latent positions.
Let $\widehat{\mathbf{P}}$  be the output of TH-PCA (\Cref{thpca}) with inputs
\begin{enumerate}[(a)]
    \item [$(a)$] an adjacency tensor $\mathbf{A}$ satisfying that $\mathbf{A} \sim \mathrm{directed}$-$\mathrm{MRDPG}(F, \widetilde{F}, \{W_{(l)}\}_{l = 1}^L, n_1, n_2)$ described in \Cref{mrdpg}, and
    \item [$(b)$] Tucker ranks $(r_1, r_2, r_3) = (d, d, m)$, where $d$ and $m$ are defined in \Cref{ass_X_Y-app}.
\end{enumerate}

Letting $\mathbf{P}$ be the probability tensor of the adjacency tensor $\mathbf{A}$, it then holds that 
\[
    \mathbb{P}\big\{\|\widehat{\mathbf{P}} - \mathbf{P}\|_{\mathrm{F}}^2 \leq  C(d^2m + n_1d + n_2d + Lm)\big\} > 1 - C(n_1 \vee n_2 \vee L)^{-c},
\]
where $C, c > 0$  are constants depending on the constants $C_{\mathrm{gap}}, C_{\sigma}>0 $, which are defined in \Cref{ass_X_Y-app}.
\end{theorem}

\Cref{def-mrdpg-dynamic} is the counterpart of \Cref{def-umrdpg-dynamic}.
\begin{definition}[Dynamic directed multilayer random dot product graphs, dynamic MRDPG]\label{def-mrdpg-dynamic}
Let $\{W_{(l)}(t)\}_{l \in [L], t\in \mathbb{N}^*} \subset \mathbb{R}^{d\times d}$ be weight matrix sequence and $\{(F_t, \widetilde{F}_t)\}_{t \in \mathbb{N}^*}$ be latent position distributions pair sequence.  We say that $\{\mathbf{A}(t)\}_{t \in \mathbb{N}^*}$ is a sequence of adjacency tensors of dynamic directed MRDPGs with random latent positions, if they are mutually independent and
\[
    \mathbf{A}(t) \sim \mathrm{directed}\mbox{-}\mathrm{MRDPG}(F_t, \widetilde{F}_t, \{W_{(l)}(t)\}_{l \in [L]}, n_1, n_2), \quad t \in \mathbb{N}^*,
\]
as defined in \Cref{mrdpg}.
\end{definition}

\Cref{ass_no_change_point-app} is the counterpart of \Cref{ass_no_change_point}.
\begin{assumption}[No change point]\label{ass_no_change_point-app}
Let $\{\mathbf{A}(t)\}_{t \in \mathbb{N}^*}$ be defined in \Cref{def-mrdpg-dynamic}.  Assume that 
\[
    F_1 = F_2 = \cdots, \quad \widetilde{F}_1 = \widetilde{F}_2 = \cdots \quad \mbox{and} \quad \{W_{(l)}(1)\}_{l = 1}^L =  \{W_{(l)}(2)\}_{l = 1}^L = \cdots.
\]    
\end{assumption}     

\Cref{def-cusum-S-app} is the counterpart of \Cref{def-cusum-S}.
\addtocounter{definition}{1}  
\begin{subdefinition}\label{def-cusum-S-app}
For any $k \in [n_1+n_2-1]$, let
\[
    \mathcal{S}_k = \begin{cases}
        \big\{(i, i + k - 1): \, i \in [(n_2 + 1 - k) \wedge n_1]\big\}, & k \in [n_2],\\
        \big\{(i + k - n_2, i): \, i \in [(n_1 + n_2 - k) \wedge n_2)]\big\}, & k \in [n_1 + n_2 -1] \setminus [n_2].
    \end{cases}
\]
\end{subdefinition}

\Cref{def-cusum-app} is the counterpart of \Cref{def-cusum}.

\addtocounter{definition}{-1}
\begin{definition}\label{def-cusum-app}
For any $\alpha \in (0, 1)$, let $\widehat{D}_{s, t} = \max_{z \in \mathcal{Z}_{\alpha, t}}|\widehat{D}_{s, t}(z)|$, where
\begin{itemize}
    \item $\mathcal{Z}_{\alpha, t} = \{z_v\}_{v = 1}^{M_{\alpha, t}} \sim \mathrm{Uniform}([0, 1]^L)$  is a collection of independent random vectors  with 
        \[
            M_{\alpha, t}  =   C_M L^{-L}\left(t ( n_1 \wedge n_2) \right) ^{L/2}   \left( \log((n_1 \vee n_2 \vee t )/\alpha)   \right)^{-L/2+1},
        \]
    for $C_M > 0$ an absolute constant;
    \item for $z \in [0, 1]^L$, 
        \[
            \widehat{D}_{s, t}(z) =  \bigg( \sum_{k \in [n_1 + n_2 -1]} \widetilde{S}_k  \bigg)^{-1} h^{-L} \sum_{k \in [n_1 + n_2 -1]} \widetilde{S}_k \bigg\{\mathcal{K} \left(\frac{z - \widehat{P}^{0, s}_{\mathcal{S}_k, :}}{h} \right) - \mathcal{K} \left(\frac{z - \widehat{P}^{s, t}_{\mathcal{S}_k, :}}{h} \right) \bigg\},
        \]
        with $\{\mathcal{S}_k, \widetilde{S}_k = |\mathcal{S}_k|\}_{k \in [n_1 + n_2 -1]}$ defined in \Cref{def-cusum-S-app}; and
    \item for any integer pair $(s, t)$, $ 0 \leq s < t$, with $\mathrm{HOSVD}$ detailed in \Cref{alg-hosvd},
        \[
            \widehat{\mathbf{P}}^{s, t} = \mathrm{HOSVD} \bigg((t-s)^{-1}\sum_{u=s+1}^t \mathbf{A}(u)\bigg) \quad \mbox{and} \quad \widehat{\mathbf{P}}^{s, t}_{\mathcal{S}_k, :} = \widetilde{S}_k^{-1} \sum_{(i, j) \in \mathcal{S}_k } \widehat{\mathbf{P}}^{s, t}_{i, j, :},
        \]
\end{itemize}

\end{definition}

\Cref{snr_ass_change_point-app} is the counterpart of \Cref{snr_ass_change_point}.

\begin{assumption}[Signal-to-noise ratio condition] \label{snr_ass_change_point-app}
Under \Cref{ass_X_Y}, for any $\alpha \in (0, 1)$, assume that there exists a large enough absolute constant $C_{\mathrm{SNR}} > 0$ such that
\[
    \kappa\sqrt{\Delta}  > C_{\mathrm{SNR}} h^{-L-1} \sqrt{\frac{(L^2 \vee d) \log \{(n_1 \vee n_2 \vee \Delta)/\alpha \}}{  n_1 \wedge n_2}}.
\]   
\end{assumption}

\Cref{main_theorem-app} is the counterpart of \Cref{main_theorem}.
\begin{theorem}\label{main_theorem-app}
Let $\widehat{\Delta}$ be the output of the dynamic multilayer random dot product graph online change point detection algorithm detailed in \Cref{online_hpca}, with 
\begin{enumerate}[(a)]
    \item [$(a)$] adjacency tensor sequence $\{\mathbf{A}(t)\}_{t \in \mathbb{N}^*}$ be defined in \Cref{def-mrdpg-dynamic}, satisfying \Cref{ass_X_Y-app};    
    \item [$(b)$] tolerance level $\alpha \in (0, 1)$ and threshold sequence $\{\tau_{s, t}\}_{1 \leq s < t}$ with
    \[
        \tau_{s, t} = C_{\tau}h^{-L-1} \sqrt{\frac{ (L^2 \vee d) \log( (n_1 \vee n_2  \vee t) / \alpha)}{n_1 \wedge n_2}}\left( \frac{1}{\sqrt{s}}+ \frac{1}{\sqrt{t- s}}\right), \quad 1 < s < t,    
    \]
    where $C_{\tau} > 0$ is an absolute constant; and
    \item [$(c)$] scan statistic sequence $\{\widehat{D}_{s, t}\}_{1 \leq s < t}$ is defined in \Cref{def-cusum-app} with the kernel function $\mathcal{K}(\cdot)$ satisfying \Cref{kernel_function_ass}.
\end{enumerate}

We have the following.
\begin{enumerate}[(i)]
    \item[$(i)$] Under \Cref{ass_no_change_point-app}, it holds that $\mathbb{P}_{\infty} \{\widehat{\Delta} < \infty\} \leq \alpha$.
    \item[$(ii)$] Under Assumptions~\ref{ass_change_point} and \ref{snr_ass_change_point-app}, with $ C_\epsilon > 0$ being an absolute constant, it holds that
    \[
        \mathbb{P}_{\Delta} \left\{\Delta < \widehat{\Delta} \leq \Delta + C_\epsilon  \frac{(L^2 \vee d) \log \left((n_1 \vee n_2 \vee \Delta) /\alpha \right)}{\kappa^2 h^{2L+2} (n_1 \wedge n_2)} \right\} \geq 1 - \alpha.
    \]
\end{enumerate}
\end{theorem}

\section[]{Additional details and results in Section \ref{sec-numerical}} \label{supp-sec-num}

\subsection{Additional results in Section \ref{change_point_section_f}}\label{simu-sec5.2}
The simulation results for \textbf{Scenarios 1, 2} and \textbf{3} in \Cref{change_point_section_f} are presented in Tables~\ref{table_fix_supp-1}, \ref{table_fix_supp-2} and~\ref{table_fix_supp-3}, respectively.  

We also include a sensitivity analysis of $C_{\tau}$ in \Cref{table_fix_supp-4}.  Instead of assuming access to training data to tune $C_{\tau}$, we directly conduct analysis based on a range of $C_{\tau}$.  In terms of PFA and PFN, our proposed method along with all competitors are robust against the choice of $C_{\tau}$.  In terms of the detection delay, unfortunately, we as well as all the competitors are sensitive.

\begin{table}[ht]
\caption{Scenario 1 in \Cref{change_point_section_f}: $n$, the number of nodes; $L$, the number of layers; PFA, the proportion of false alarm; PFN, the proportion of false negatives; Delay, the average of the detection delays; 
Alg.~\ref{online_hpca}-Def.~\ref{def-cusum-f}, on-HOSVD, on-TWIST, on-UASE and  on-COA, \Cref{online_hpca} with the scan statistics defined in \Cref{def-cusum-f} and the tensor estimation routine using \Cref{thpca}, HOSVD \citep{de2000multilinear},  TWIST \citep{jing2021community}, UASE \citep{jones2020multilayer} and  COA \citep{macdonald2022latent} respectively; $k$-NN, \cite{chen2019sequential}.}
\begin{center}
\begin{tabular}{ccccccccc}
\hline
   $n$ & $L$ & Metric & Alg.~\ref{online_hpca}-Def.~\ref{def-cusum-f} & on-HOSVD & on-TWIST&  on-UASE & on-COA & $k$-NN \\ \hline
    50 &  3  & PFA   & 0.07 & 0.08 & 0.10  & 0.07 &  0.00 & 0.01   \\
    &        & PFN   & 0.00 & 0.00 & 0.00  & 0.00 &  0.99 & 0.00   \\
    &        & Delay & 1.66 & 1.86 & 0.59 & 3.55 &  48.54 & 4.51  \\ [3pt]
    25 & 3 & PFA  & 0.00 & 0.00 &  0.01  & 0.06 &  0.06 & 0.02   \\
 &        & PFN   & 0.00 & 0.00 & 0.00  & 0.00 &  0.90 & 0.00  \\
 &        & Delay & 2.60 & 2.77 &  2.59  & 3.66 &  48.14 & 5.21  \\ [3pt]
 75& 3   & PFA    & 0.03 & 0.03 & 0.04 &  0.01 &  0.02 & 0.01 \\
 &        & PFN   & 0.00 & 0.00 & 0.00  &  0.00 &  0.02 & 0.00 \\
 &        & Delay & 1.00 & 1.00 & 0.90 &  1.00 &  17.52 & 2.27 \\ [3pt]
 50 & 2  & PFA &  0.05 &  0.04 &  0.07  &  0.02& 0.06 & 0.00\\
 &        & PFN & 0.00  & 0.00 & 0.00& 0.00   & 0.00 & 0.00 \\
 &       & Delay & 2.07 & 2.20 & 2.20 & 3.06  & 0.00 & 4.09  \\ [3pt]
 50 & 4  & PFA    & 0.00 & 0.00 &  0.17 &  0.01 &  0.00 & 0.02 \\
 &        & PFN   & 0.00 & 0.00 & 0.00 &  0.00 &  0.98 & 0.00\\
 &       &  Delay & 1.00 & 1.00 & 0.00 &  1.00 &  48.27 & 2.36\\\hline
\end{tabular}\label{table_fix_supp-1}
\end{center}
\end{table}

\begin{table}[ht]
\caption{Scenario 2 in \Cref{change_point_section_f}: $n$, the number of nodes; $L$, the number of layers; PFA, the proportion of false alarm; PFN, the proportion of false negatives; Delay, the average of the detection delays; 
Alg.~\ref{online_hpca}-Def.~\ref{def-cusum-f}, on-HOSVD, on-TWIST, on-UASE and  on-COA, \Cref{online_hpca} with the scan statistics defined in \Cref{def-cusum-f} and the tensor estimation routine using \Cref{thpca}, HOSVD \citep{de2000multilinear},  TWIST \citep{jing2021community}, UASE \citep{jones2020multilayer} and  COA \citep{macdonald2022latent} respectively; $k$-NN, \cite{chen2019sequential}.}
\begin{center}
\begin{tabular}{ccccccccc}
\hline
   $n$ & $L$ & Metric & Alg.~\ref{online_hpca}-Def.~\ref{def-cusum-f} & on-HOSVD & on-TWIST & on-UASE & on-COA & $k$-NN \\ \hline
50 & 3  & PFA & 0.13  & 0.11  & 0.08 &0.07 &  0.11 &  0.99    \\
 &        & PFN & 0.00 & 0.00 & 0.00 &0.00 & 0.79 &   0.00   \\
 &       & Delay & 0.99 & 1.00 & 0.14 & 1.01 & 46.36 &  2.00\\ [3pt]
25 & 3 & PFA  & 0.05  & 0.06 & 0.16  & 0.05 & 0.18  & 0.99  \\
 &     & PFN  & 0.00 & 0.00 & 0.00 & 0.00  & 0.75  & 0.00  \\
 &     & Delay& 2.00 & 2.01 & 1.94 & 2.69 &  46.84 & 4.00  \\ [3pt]
 75& 3   & PFA   & 0.01  & 0.01 & 0.06  & 0.01  &  0.08  & 0.99  \\
 &        & PFN   & 0.00  & 0.00 & 0.00 & 0.00 & 0.18  & 0.00  \\
 &        & Delay & 0.03 & 0.05  & 0.00 & 0.99  & 28.26  & 2.00   \\[3pt]
50 &  2  & PFA   &  0.04  & 0.04 & 0.21  &0.07  & 0.02   & 0.99   \\
&        & PFN   & 0.00 & 0.00 & 0.00 & 0.00 & 0.89  &  0.00  \\
 &        & Delay  & 0.96 & 0.99 & 0.00 & 1.00  & 46.18   &  2.00  \\
    [3pt]
 50 & 4  & PFA   & 0.04 & 0.01  & 0.23  & 0.02 & 0.06  &   0.99 \\
 &        & PFN   & 0.00 & 0.00 & 0.00 & 0.00 &  0.78 &  0.00 \\
 &       &  Delay & 0.93 & 0.99 & 0.00  & 1.00 &  43.80  &  2.00 \\\hline
\end{tabular}\label{table_fix_supp-2}
\end{center}

\end{table}

\begin{table}[ht]
\caption{Scenario 3 in \Cref{change_point_section_f}: $n$, the number of nodes; $L$, the number of layers; PFA, the proportion of false alarm; PFN, the proportion of false negatives; Delay, the average of the detection delays; 
Alg.~\ref{online_hpca}-Def.~\ref{def-cusum-f}, on-HOSVD, on-TWIST, on-UASE and  on-COA, \Cref{online_hpca} with the scan statistics defined in \Cref{def-cusum-f} and the tensor estimation routine using \Cref{thpca}, HOSVD \citep{de2000multilinear},  TWIST \citep{jing2021community}, UASE \citep{jones2020multilayer} and  COA \citep{macdonald2022latent} respectively; $k$-NN, \cite{chen2019sequential}.}
\begin{center}
\begin{tabular}{ccccccccc}
\hline
   $n$ & $L$ & Metric & Alg.~\ref{online_hpca}-Def.~\ref{def-cusum-f} & on-HOSVD  & on-TWIST & on-UASE & on-COA & $k$-NN \\ \hline
    50 &  3  & PFA   & 0.07  & 0.08  & 0.10  & 0.07  & 0.00   & 0.01  \\
    &        & PFN   & 0.00 &  0.00 & 0.00 & 0.00 & 0.00  &  0.00 \\
    &        & Delay  & 1.40 & 1.46 & 0.40 & 3.31  & 0.73  &  27.36 \\
    [3pt]
25 & 3 & PFA  & 0.00 & 0.00 &  0.01  &0.06  & 0.06  & 0.02   \\
&      & PFN   & 0.00 & 0.00 & 0.00 & 0.00 & 0.00  & 0.06  \\
&      & Delay & 7.32 & 8.19 & 9.94 & 15.97 & 2.86  & 32.12  \\ [3pt]
75& 3   & PFA    & 0.03 & 0.03 & 0.14 &0.01  & 0.02  & 0.01  \\
&        & PFN   & 0.00 & 0.00 & 0.00 &0.00  & 0.00  & 0.00  \\
&        & Delay & 1.00 & 1.00 & 0.01 &1.02  & 0.04  & 23.96  \\[3pt]
50 & 2  & PFA & 0.05 & 0.04  & 0.10  & 0.02 & 0.06  & 0.00   \\
&        & PFN & 0.00 & 0.00 & 0.00 & 0.00 & 0.00  & 0.00  \\
&       & Delay & 1.18 & 1.18 & 1.06 & 1.42  & 0.53  & 20.90  \\ [3pt]
 50 & 4  & PFA   & 0.00  & 0.00  & 0.26  & 0.01 &  0.00  & 0.02   \\
 &        & PFN   &0.00  & 0.00 & 0.00  & 0.00 &  0.00  & 0.00  \\
 &       &  Delay &  1.28  &  1.27 & 0.04  & 3.52  & 0.29   & 25.65   \\\hline
\end{tabular}\label{table_fix_supp-3}
\end{center}

\end{table}

\begin{table}[ht]
\caption{Scenario 1 in \Cref{change_point_section_f} with no additional sample, the number of nodes $n=50$ and the number of layers $L = 2$:  $L$, the number of layers; $C_{\tau}$, the constant related to the thresholds;   PFA, the proportion of false alarm; PFN, the proportion of false negatives; Delay, the average of the detection delays; Alg.~\ref{online_hpca}-Def.~\ref{def-cusum-f}, on-HOSVD, on-TWIST, on-UASE and on-COA, \Cref{online_hpca} with the scan statistics defined in \Cref{def-cusum-f} and the tensor estimation routine using \Cref{thpca}, HOSVD \citep{de2000multilinear},  TWIST \citep{jing2021community}, UASE \citep{jones2020multilayer} and  COA \citep{macdonald2022latent} respectively.}
\begin{center}
\begin{tabular}{ccccccccc}
\hline
$C_{\tau}$ & Metric & Alg.~\ref{online_hpca}-Def.~\ref{def-cusum-f} & on-HOSVD & on-TWIST&  on-UASE & on-COA \\ \hline
0.16 & PFA   & 0.00  & 0.00 & 0.02  & 0.03  & 0.00   \\
     & PFN  & 0.00 & 0.00 &  0.00 & 0.00 &  0.00  \\
     & Delay & 4.18 & 4.07  &1.90 & 1.10 &  11.57   \\ [3pt]
0.17 & PFA  & 0.00 & 0.00  & 0.05  & 0.07  & 0.00   \\ 
     & PFN  & 0.00 &0.00  & 0.00  & 0.00 &  0.00  \\
     & Delay &6.04  &5.91  &  4.07 & 2.63 &  14.12  \\ [3pt]
0.18 & PFA   & 0.00  & 0.00  & 0.04 &0.09  & 0.00     \\ 
     & PFN  & 0.00  & 0.00 &0.00 &0.00  &  0.00   \\
     & Delay& 8.09 & 8.02  & 6.14 &4.59 &  17.07    \\ [3pt]
0.19  &PFA & 0.00 & 0.00  & 0.04 &0.10  & 0.00     \\
     & PFN & 0.00 & 0.00 & 0.00 &0.00  & 0.00    \\
     & Delay &10.46 & 10.33 & 8.10 & 6.93  & 20.71    \\[3pt]
0.20  & PFA  &0.00  & 0.00  & 0.07 & 0.12  &0.00      \\
     & PFN   &0.00  & 0.00 & 0.00 &0.00 &  0.00    \\
     & Delay &13.13 & 13.06 &10.55 & 9.51  & 24.94  \\\hline
\end{tabular}\label{table_fix_supp-4}
\end{center}
\end{table}

\subsection{Additional results in Section \ref{change_point_section}}\label{simu-sec5.3}
The simulation results based on \textbf{Scenarios 1}, \textbf{2} and \textbf{3} in \Cref{change_point_section} are collected in Tables~\ref{table_1_random_sbm_supp}, \ref{table_undirected_diri_supp} and \ref{table_directed_diri_supp}, respectively.

\begin{table}[ht]
\caption{Scenario 1 in \Cref{change_point_section}: $n$, the number of nodes; $L$, the number of layers; PFA, the proportion of false alarm; PFN, the proportion of false negatives; Delay, the average of the detection delays;  Alg.~\ref{online_hpca}-Def.~\ref{def-cusum}, on-TH-PCA, on-TWIST,
on-UASE and on-COA, \Cref{online_hpca} with the scan statistics defined in \Cref{def-cusum} and the tensor estimation routine using \Cref{alg-hosvd},   \Cref{thpca}, TWIST \citep{jing2021community}, UASE \citep{jones2020multilayer}  and COA \citep{macdonald2022latent}, respectively; $k$-NN, \cite{chen2019sequential}.}
\begin{center}
\begin{tabular}{ccccccccc}
\hline
   $n$ & $L$ & Metric & Alg.~\ref{online_hpca}-Def.~\ref{def-cusum} & on-TH-PCA  & on-TWIST & on-UASE  & on-COA & $k$-NN \\ \hline
    50 &  3  & PFA   & 0.01 & 0.00 & 0.04   & 0.03 &  0.00 & 0.03   \\
    &        & PFN   & 0.00 & 0.00 & 0.00 & 0.00 &  1.00 & 0.00   \\
    &        & Delay & 0.00 & 0.00 & 0.00 & 0.00 &  49.00 & 4.34  \\ [3pt]
    25 & 3 & PFA   & 0.00 & 0.01 & 0.04 & 0.06 &  0.06 & 0.01   \\
 &         & PFN   & 0.00 & 0.00 & 0.00& 0.00 &  0.90 & 0.00  \\
 &         & Delay & 0.41 & 0.40 & 0.38 & 0.28 &  48.41 & 5.64  \\ [3pt]
 75& 3    & PFA   & 0.04 & 0.04  & 0.08 & 0.02 & 0.01  & 0.01 \\
 &        & PFN   & 0.00 & 0.00  & 0.00 & 0.00 &  0.00 & 0.00 \\
 &        & Delay & 0.00 & 0.00  & 0.00  & 0.00 & 6.61  & 2.19 \\ [3pt]
 50 & 2  & PFA    & 0.01 & 0.01  &  0.05 & 0.00 &  0.04 & 0.05\\
 &       & PFN    & 0.00 & 0.00  & 0.00 & 0.00 &  0.00 & 0.00 \\
 &       & Delay  & 0.02 & 0.01  & 0.01& 0.06 &  0.00 & 4.77  \\ [3pt]
 50 & 4   & PFA    & 0.02 & 0.02  &0.07 & 0.01 &  0.03 & 0.00 \\
 &        & PFN    & 0.00 & 0.00  & 0.00& 0.00 &  0.91 & 0.00 \\
 &        &  Delay & 0.00 & 0.00  &0.00 &  0.00 &  47.53 & 2.32 \\\hline
\end{tabular}\label{table_1_random_sbm_supp}
\end{center}

\end{table}

\begin{table}[ht]
\caption{Scenario 2 in \Cref{change_point_section}: $n$, the number of nodes; $L$, the number of layers; $d$, the dimension of the latent position; PFA, the proportion of false alarm; PFN, the proportion of false negatives; Delay, the average of the detection delays;   Alg.~\ref{online_hpca}-Def.~\ref{def-cusum}, on-TH-PCA, on-TWIST,
on-UASE and on-COA, \Cref{online_hpca} with the scan statistics defined in \Cref{def-cusum} and the tensor estimation routine using \Cref{alg-hosvd},   \Cref{thpca}, TWIST \citep{jing2021community}, UASE \citep{jones2020multilayer}  and COA \citep{macdonald2022latent}, respectively; $k$-NN, \cite{chen2019sequential}.}
\begin{center}
\begin{tabular}{cccccccccccc}
\hline
$n$ & $L$ & $d$ & Metric & Alg.~\ref{online_hpca}-Def.~\ref{def-cusum} & on-TH-PCA & on-TWIST &  on-UASE & on-COA & $k$-NN \\ \hline
50  & 3 & 4 & PFA   & 0.05 & 0.04 & 0.05 &  0.02     & 0.13 & 0.02 \\
 &  &           & PFN  & 0.00 & 0.00 & 0.00   &  0.00 & 0.00 & 0.81\\
 &  &           & Delay& 1.56 & 1.56 & 1.58 &  1.91 & 0.70 & 47.87 \\ [3pt]
25  & 3 & 4 & PFA &  0.08   & 0.05 &  0.11 &  0.06     &   0.14  & 0.01   \\
 &  &           & PFN &  0.55  &  0.71 & 0.37 &  0.71  &   0.58  & 0.98   \\
 &  &           & Delay& 39.91 & 49.92 & 33.26 & 44.55  &   40.22  & 49.67  \\ [3pt]
 75  & 3 & 4  & PFA & 0.01  & 0.01 & 0.01 &  0.02    & 0.11 & 0.01  \\
 &  &           & PFN &  0.00  &  0.00 & 0.00 &  0.00 & 0.00 & 0.98   \\
 &  &           & Delay & 3.20 & 3.43 &  4.45 &  7.01 & 5.74 & 49.74   \\ [3pt]
50  & 2 & 4  & PFA & 0.14  & 0.08 & 0.14 &  0.08     &  0.05 & 0.04   \\
 &  &           & PFN & 0.03  & 0.05 & 0.03\textbf{} &  0.24  &  0.01 & 0.96   \\
 &  &           & Delay& 13.70& 16.03 & 14.16 & 25.43   &  15.81& 50.00  \\ [3pt]
 50  & 4 & 4  & PFA& 0.02  & 0.04 & 0.07 &  0.06&     0.10  & 0.00   \\
 &  &            & PFN & 0.00 & 0.00 & 0.00 &  0.00&  0.01  & 0.71   \\
 &  &            & Delay& 5.67& 5.60 & 5.20 &  8.74&  9.68  & 40.02   \\[3pt]
50  & 3 & 2  & PFA & 0.00 & 0.00 & 0.00 &  0.00     & 0.00  & 0.00   \\
 &  &           & PFN & 0.00 & 0.00  & 0.00&  0.00  & 0.00  & 0.24   \\
 &  &           & Delay& 0.86 & 0.86 & 0.53 &  0.95 & 1.13  & 40.26   \\[3pt]
50  & 3 & 6 & PFA   & 0.00 & 0.00 & 0.01 &  0.00    & 0.01 & 0.02   \\
 &  &          & PFN   &  0.00& 0.00 & 0.00 &  0.00 & 0.00 & 0.43\\
 &  &          & Delay & 5.71 & 6.01 & 5.78 &  12.24 & 7.65 & 42.61  \\ \hline
\end{tabular}
\label{table_undirected_diri_supp}
\end{center}
\end{table}

\begin{table}[ht]
\caption{Scenario 3 in \Cref{change_point_section}: $n_1$, $n_2$, the number of nodes; $L$, the number of layers; $d$, the dimension of the latent position; PFA, the proportion of false alarm; PFN, the proportion of false negatives; Delay, the average of the detection delays;   Alg.~\ref{online_hpca}-Def.~\ref{def-cusum}, on-TH-PCA, on-TWIST,
on-UASE and on-COA, \Cref{online_hpca} with the scan statistics defined in \Cref{def-cusum} and the tensor estimation routine using \Cref{alg-hosvd},   \Cref{thpca}, TWIST \citep{jing2021community}, UASE \citep{jones2020multilayer}  and COA \citep{macdonald2022latent}, respectively; $k$-NN, \cite{chen2019sequential}; NA, not applicable.}
\begin{center}
\begin{tabular}{ccccccccccc}
\hline
$n_1$ & $n_2$ & $L$ & $d$ & Metric & Alg.~\ref{online_hpca}-Def.~\ref{def-cusum} & on-TH-PCA & on-TWIST  &  on-UASE & on-COA & $k$-NN \\ \hline
50 & 50 & 3 & 4 & PFA & 0.00 & 0.00 & 0.00 &  0.00&  0.02  & 0.00 \\
 & & &          & PFN & 0.00 & 0.00 & 0.00 &  0.00&  0.00  & 0.45\\
 & & &          & Delay&0.85 & 0.86 & 0.90&  1.14&  0.50 & 42.56 \\ [3pt]
25 & 50 & 3 & 4 & PFA &  0.08  & 0.08 & NA &  0.06 & NA  &0.03   \\
 & & &          & PFN &   0.01 &  0.01 & NA  &  0.07 & NA  &0.95   \\
 & & &          & Delay&  15.24&  15.08 & NA & 21.95 & NA  & 49.08  \\ [3pt]
 75 & 50 & 3 & 4  & PFA & 0.12  & 0.13 & NA  &  0.04 & NA  & 0.03  \\
 & & &            & PFN &  0.00 &  0.00 &NA  &   0.00& NA  & 0.95   \\
 & & &            & Delay & 0.95& 0.97 & NA &   2.92& NA  & 49.54   \\ [3pt]
50 & 50 & 2 & 4   & PFA & 0.00 & 0.00  & 0.04&  0.00& 0.00 &0.01   \\
 & & &            & PFN & 0.00 & 0.00 & 0.00 &  0.12& 0.33 &0.98   \\
 & & &            & Delay& 9.35& 10.08 & 9.97 & 28.37& 35.34 & 49.94  \\ [3pt]
 50 & 50 & 4 & 4  & PFA& 0.01  & 0.01 & 0.03 & 0.01 & 0.00 & 0.04   \\
 & & &            & PFN & 0.00 & 0.00 &0.00 & 0.00 & 0.10 & 0.43   \\
 & & &            & Delay& 3.09& 3.11 & 3.16 & 10.16& 19.99 & 28.89   \\[3pt]
50 & 50 & 3 & 2   & PFA & 0.00 & 0.00 & 0.04 &0.02 & 0.04 & 0.02   \\
 & & &            & PFN & 0.00 & 0.00 & 0.00 & 0.00& 0.00 & 0.23   \\
 & & &            & Delay& 0.29& 0.33 & 0.00 & 0.29& 0.40 & 40.03   \\[3pt]
50 & 50 & 3 & 6   & PFA & 0.00  & 0.00 & 0.25 & 0.00& 0.41 & 0.01   \\
 & & &            & PFN&  0.00  & 0.00 &0.00  & 0.00& 0.00 & 0.25\\
 & & &            & Delay & 1.75& 1.74 &0.87 & 5.69&  3.15 & 39.59  \\ \hline
\end{tabular}
\label{table_directed_diri_supp}
\end{center}
\end{table}

\subsection{Additional details and results in Section \ref{sec-real-data}}\label{simu-sec5.4}
In this section, we analyse the U.S.~air transport network data set.

\medskip
\noindent \textbf{The U.S.~air transportation network data set} includes monthly data from January 2015 to July 2022 ($91$ months) and is available at \cite{BTS2022}. Each node here represents an airport and each layer a commercial airline.   Directed edges at each layer represent direct flights of a specific commercial airline connecting two airports. 

To construct multilayer networks, we select  $L \in \{3, 4\}$ airlines with the most flights (Southwest Airlines, Sky West Airlines, United Airlines and Delta Airlines),  $n_1 \in \{50, 75\}$ airports with the most departing flights and $n_2 \in \{50, 75\}$ airports with the most arriving flights. We obtain $T = 91$ multilayer networks, each corresponding to data within a month.  Data from January 2015 to June 2017 ($30$ months) are used as the training data set, with the rest being the test set. 

Since all airports are named and the multilayer networks are labelled multilayer networks,
we assume that latent positions are fixed and adopt the same competitors and tuning parameter selection methods as those mentioned in \Cref{change_point_section_f}. 
With $\alpha = 0.05$ and with different combinations of $n$ and $L$, \Cref{online_hpca} with the scan statistics defined in \Cref{def-cusum-f}, along with on-HOSVD, on-TWIST and on-UASE, identify May 2020 as the change point, while on-COA identifies April 2020 as the change point, and $k$-NN method detects a change point of July 2018 or December 2017 depending on the combination of $(n, L)$ used. The May 2020 change point is related to the outbreak of Covid-19. According to \cite{COVID}, Covid-19 cases in the U.S.~exceeded $10,000$ on 23 March 2020 and the cumulative number of cases in the U.S.~exceeded $200,000$ on 1 April 2020.

In addition, we apply independent random permutations to nodes of multilayer networks at each time point. We then treat latent positions as random and apply the same competitors and tuning parameter selection methods as those presented in \Cref{change_point_section}. In this setting, \Cref{online_hpca} with the scan statistics defined in \Cref{def-cusum}, on-TH-PCA, on-TWIST and on-UASE all identify April 2020 as the change point, while on-COA identifies May 2020 as the change point, and $k$-NN does not identify any change points.

\end{document}